\documentclass{report}

\usepackage{conf/tikz-cd}
\usepackage{conf/custom}
\usepackage{conf/stat}
\usepackage{conf/forArxiv}

\definecolor{rouge}{HTML}{a02030}

\usepackage[
    colorlinks=true,
    citecolor=blue
]{hyperref}

\begin{document}

\interfootnotelinepenalty=10000

\title{\larger{
Message-Passing Algorithms 
and Homology\\ }
\vspace{4em}
\Large{
\it From Thermodynamics 
to Statistical Learning 
\vspace{3em}
}
}
\author{Olivier Peltre}

\date{2020}

\maketitle

\phantomsection
\addcontentsline{toc}{chapter}{Introduction}
\chapter*{Introduction}

The problem of describing the statistics of a 
large number $x_i, x_j$,... of interacting random variables 
emerged in physics with Boltzmann's efforts to lay
principles of thermodynamics on statistical grounds, 
and high dimensional statistics are now expected to provide 
with reasonable and tractable models in artificial 
intelligence and biology. 
Aimed at modelling the emergence of collective behaviours in
large assemblies of constituents, 
the prism of statistics hence shows deep analogies between atoms in a crystal, 
and neurons in a network.

A probability distribution $p(x)$ 
on the joint variable $x = (x_j)_{j \in \Om}$ 
is usually assumed to capture all collective phenomena,
although a dimensional curse prohibits the computation 
of expectation values. 
Local effects on a small
subset $\aa \incl \Om$ of variables 
may nonetheless be estimated,
as the statistics of the local variable $x_\aa = (x_i)_{i \in \aa}$
only involve the marginal distribution\footnote{
    Throughout the following  
    $f_\aa(x_\aa)$ stands for $f_{i_1\dots i_n}(x_{i_1}, \dots, x_{i_n})$ 
    with $\aa = \{ i_1, \dots, i_n \} \incl \Om$.
} 
$p_\aa(x_\aa)$. 
Spontaneous magnetisation, for instance, is given by 
the expectation value of a single atomic dipole $x_i = \pm 1$,
subject to interactions within an arbitrary large crystalline network.
Accessing marginals is also a crucial step 
of statistical learning:
usually appearing in the gradient of a loss function, 
they are necessary to guide the update of model parameters.
The design of efficient algorithms for marginal estimation is 
hence of great practical importance. 
{\it Message-passing algorithms} 
estimate marginals through a parallelised and asynchronous computing scheme, 
in which a collection of 
local units communicate until they eventually reach a consensual state.
Understanding their connections with 
algebraic topology was the first motivation of this thesis. 

{\it Gibbs random fields}\footnotemark{} 
are probabilistic models with a local structure described by a 
collection $X$ of subsets 
$\aa, \bb,\cc,$... of $\Om$, over which the global distribution $p(x)$ 
factorises as a product of local functions. 
We write  $p \in {\cal G}(X)$ when
there exists a collection of positive factors $(f_\aa)$ such that: 
\footnotetext{
    Or {\it Gibbs distributions}, which are also {\it Markov random fields} 
    according to the Hammersley-Clifford correspondence. 
    Factorisability however yields a finer characterisation 
    of ${\cal G}(X)$ than Markov properties, hence the preferred terminology. 
}
\begin{equation} \label{GRF}
p(x) = \frac 1 Z \prod_{\aa \in X} f_\aa (x_\aa) 
\end{equation}
Distributions of this form are more often called {\it graphical models} 
in the computer science literature. 
The hypergraph $X \incl \Part(\Om)$ is then represented by
the so-called factor graph, depicted in figure 1, 
formed by joining variable nodes $(x_i)$ 
with their associated factor nodes $(f_\aa)$. 
This factorisation  is more conveniently 
viewed at the level of energies, where the {\it hamiltonian} 
$H$ is defined as a sum of local {\it interaction potentials} $(u_\aa)$, 
related to the factors by $u_\aa = - \ln f_\aa$:
\begin{equation}
p = \frac {\e^{-H}} Z \quad\txt{where}\quad 
H(x) = \sum_{\aa \in X} u_\aa (x_\aa) 
\end{equation}
The fundamental Legendre duality between the energy function $H(x)$ and 
its Gibbs distribution $p(x)$ is related to variational principles on 
entropy and free energy, 
of which message-passing algorithms yield approximate solutions.

\begin{figure}[h]
\sbox0{\includegraphics[width=0.6\textwidth]{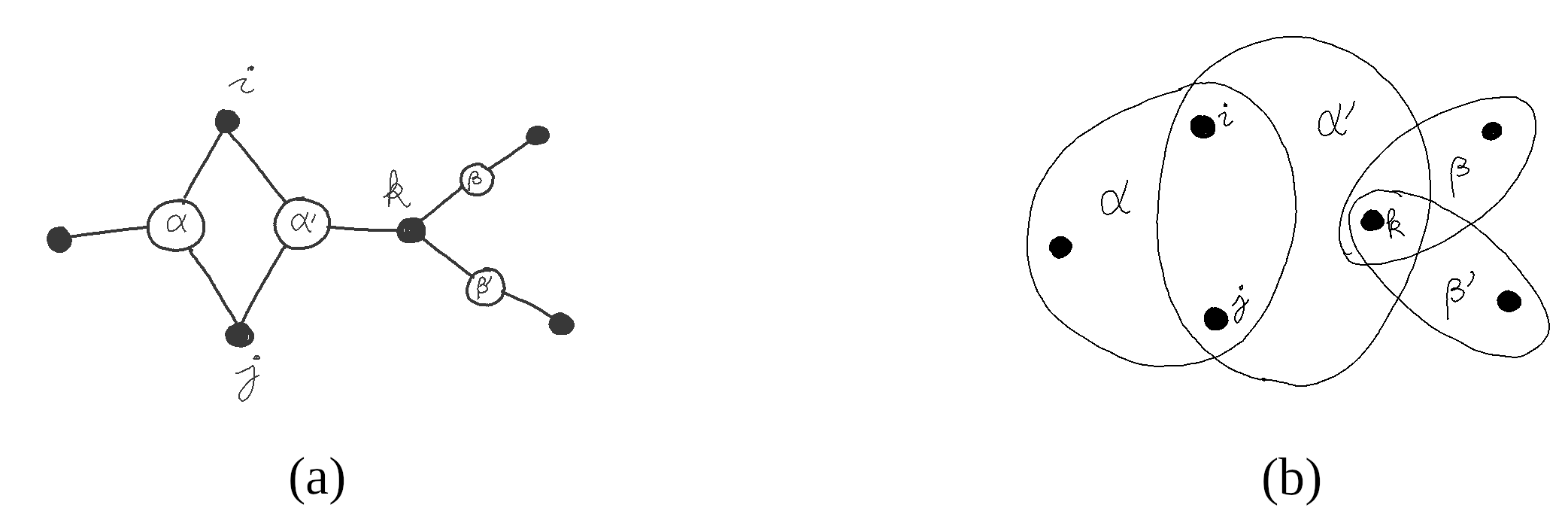}}
\begin{center}
\begin{minipage}{\wd0+0.15\textwidth}
\centering
\usebox0
\caption{
(a) factor graph and (b) hypergraph representations of $X \incl \Part(\Om)$.
}
\end{minipage}
\end{center}
\end{figure}

One of our contributions is to view 
the Gibbs random field $p \in {\cal G}(X)$ 
as a {\it homology class} of factors. 
Introducing mutual dependence of variables, 
overlapping subsets in $X$ 
also make the parametrisation of $p(x)$ given by equation (\ref{GRF})
ambiguous,
two collections of factors $(f_\aa)$ and $(f'_\aa)$ 
defining the same Gibbs random field 
whenever 
$\prod_\aa f_\aa \simeq \prod_\aa f'_\aa$ up to a scaling factor. 
This ambiguity is resolved by introducing {\it messages} as collections
$(m_{\aa\bb})$ of local functions $m_{\aa\bb}(x_\bb)$
for every ordered pair $\aa \cont \bb$ in $X$, 
and a {\it boundary operator} $\bord$ 
defining factors from messages through:
\begin{equation} \label{bord-m}
(\bord m)_\bb =
\prod_{\aa \cont \bb} m_{\aa\bb} \bigg/ \prod_{\cc \incl \bb} m_{\bb\cc} 
\end{equation}
Assuming $X$ is closed under intersection, we show 
that $(f_\aa)$ and $(f'_\aa)$ define the same Gibbs random field 
if and only if 
there exists $(m_{\aa\bb})$ such that $f' \simeq f \cdot \bord m$ 
up to scaling.
In this view, message-passing algorithms explore a homology class 
of factors by iterating over messages, 
the homological constraint expressing conservation of
the global distribution $p \in {\cal G}(X)$.

{\it Beliefs} $(q_\aa)$ are intended to estimate the local marginals $(p_\aa)$
of the global probability distribution. 
These local probabilities should in particular satisfy {\it consistency} conditions 
which require that  $q_\bb$ is the marginal of $q_\aa$ 
for every $\bb \incl \aa$.
Of cohomological nature, this constraint shall take the form 
$dq = 0$ and is expressed by the following set of equations:
\begin{equation} \label{consistency} 
q_\bb(x_\bb) = \sum_{y \in E_{\aa \setminus \bb}} q_\aa(x_\bb, y) 
\end{equation}
Defining local probabilities through the local analog of equation (\ref{GRF}):
\begin{equation} 
    q_\aa(x_\aa) = \frac 1 {Z_\aa} 
    \prod_{\bb \incl \aa} f'_\bb(x_\bb) 
\end{equation}
the specificity of message-passing algorithms is hence to search for consistent 
beliefs that derive from homologous factors $f' = f \cdot \bord m$.
The most general message-passing scheme, 
{\it belief propagation}, assumes the following update rule:
\begin{equation} \label{m-update}
m_{\aa\bb}(x_\bb) \wa  m_{\aa\bb}(x_\bb) 
\cdot \frac {\sum_y q_\aa(x_\bb, y)} {q_\bb(x_\bb)}
\end{equation}

Our main contribution is to introduce
diffusion equations of the form $\dot u = \div \Phi(u)$ 
on interaction potentials, 
which allow to view existing message-passing algorithms 
as coarse numerical integrators of continuous-time 
differential equations. 
The operator $\div$ is the first-degree boundary 
of a natural homology theory, 
acting on a collection of 
energy fluxes $\ph_{\aa\bb}(x_\bb)$ by: 
\begin{equation}
    \div_\bb \ph (x_\bb) = \sum_{\aa \cont \bb} \ph_{\aa\bb}(x_\bb) 
- \sum_{\bb \cont \cc} \ph_{\bb\cc}(x_\cc) 
\end{equation}
In addition to revealing their deeply homological character, 
this approach should dramatically improve the stability\footnote{
    Belief propagation has for instance been reported 
    to start converging poorly after several epochs of training 
    restricted Boltzmann 
    machines, 
    a brutal phenomenon that has been compared 
    to phase transitions of the Hopfield model. 
}
of message-passing algorithms.  
Showing belief propagation equivalent to a time-step-one 
explicit Euler scheme of
$\dot u = \div \Phi(u)$, 
a first and highly advisable improvement is to 
use a smaller time step $\lambda < 1$,
which would act as an exponent 
on the geometric increment of $m_{\aa\bb}$ in 
equation (\ref{m-update}). 
As another direction of improvement, 
we propose a combinatorial correction of messages 
eliminating their redundancies
by extending Möbius inversion formulas to higher degrees.
A practical question which shall remain open is whether there exists
a notion of optimal transport on $\Phi$ bringing interaction potentials 
to equilibrium. 

Our approach reveals that stationary states of message-passing algorithms 
lie at the intersection of two constraint surfaces, of 
homological and cohomological nature respectively. 
The homological constraint is linear at the level of interaction potentials, 
and expresses conservation of the {\it total energy}: 
\begin{equation} H(x) = \sum_{\aa \in X} u_\aa(x_\aa) \end{equation}
The cohomological constraint $dq = 0$, however, is linear at the level of 
the effective {\it Gibbs states}: 
\begin{equation} \label{q-u} 
q_\aa = \frac {\e^{-U_\aa}} {Z_\aa} 
\txt{with} U_\aa(x_\aa) = \sum_{\bb \incl \aa} u_\bb(x_\bb) 
\end{equation}
The problem of describing this intersection is hence highly non-linear,
and trying to understand how the geometry of the underlying hypergraph $X$ 
affects the geometry of message-passing equilibria 
will lead to difficult questions meeting both 
algebraic topology and singularity theory.

\vfill
\[
*\quad*\quad*
\]
\vfill

Out of the six chapters contained in this thesis, 
chapters 1 to 3 review and develop the algebraic theory 
we shall rely upon. 
Energy and information functionals are covered in chapter 4, 
providing background for the central theorem \ref{cvm} characterising
solutions of Kikuchi's cluster variational method \cite{Kikuchi-51} 
i.e. consistent collections of local probabilities 
which are critical for a generalised Bethe free energy.

Message-passing algorithms are then addressed in chapter 5, 
as discrete integrators of continuous-time\footnote{
    The time step of the integrator may be tuned 
    $<1$ to improve stability, analogously to a learning rate.
} 
diffusion equations $\dot u = \delta \Phi(u)$. 
The homological picture allows us to give a rigorous proof 
of the correspondence theorem \ref{correspondence} between stationary states 
of belief propagation and solutions of the cluster variational method,
as suggested by Yedidia {\it et al.} \cite{Yedidia-2005}.
Our approach more generally characterises all the flux functionals $\Phi$ 
for which such a correspondence holds,
while the combinatorics developed in chapter 3 lead us to propose
another regularisation of the generalised belief propagation
algorithm by a degree-one Möbius inversion\footnote{
    Möbius inversion eliminates redundancies 
    otherwise counted in the heat flux $\Phi$.
} 
on the flux functional $\Phi$. 

The geometry of message-passing equilibria is finally studied in chapter 6. 
We describe a class of {\it retractable} hypergraphs for which 
message-passing always converges to the exact marginals of the global 
probabilistic model to estimate. 
In general, multiple equilibria may coexist, whose bifurcations 
are related to singularities of the projection of a smooth manifold 
of consistent potentials onto their homology classes, 
and may be tracked in the spectrum of a linearised diffusion operator. 

Our first efforts consisted in looking for a formalism in which 
the elementary operations of message-passing algorithms would fit. 
The reader is therefore expected to run into some unusual notations and 
properties, which the following few pages attempt to summarise efficiently.  
With these in mind, we hope that an informed reader mostly curious of applications 
might jump directly to chapter 5.  

\newpage

\phantomsection
\addcontentsline{toc}{section}{Systems, Fields and Diffusion}
{\bf \hfill Statistical Systems \hfill}

{\smaller

A system will be defined by a collection of random variables $x_i, x_j, x_k, \dots$ 
indexed by labels $i \in \Om$. 

The set of labels $\Om$ in general has an additional
geometric structure describing interactions (e.g. graph or hypergraph).

Subsets $\aa, \bb, \cc, \dots$ in $\Part(\Om)$ form a 
partial order for inclusion, usually denoted in descending alphabetical order: 
$$\aa \cont \bb \cont \cc$$ 

\textbf{Local Spaces} 

Given a finite set of microstates $E_i$ for every atom/neuron/bit $i \in \Om$ 
and a subset of atoms $\aa \incl \Om$: 
\bi 
\iii $E_\aa = \prod_{i \in \aa} E_i$ set of local microstates:
$ x_\aa = (x_i)_{i \in \aa} \in E_\aa $ 

\iii $A_\aa = \R^{E_\aa}$ algebra of local observables:
$ f_\aa(x_\aa) \in A_\aa$

\iii $A^*_\aa \simeq \R^{E_\aa}$ vector space of local measures: 
$q_\aa : f_\aa \mapsto \croc{q_\aa}{f_\aa} \in A^*_\aa$ 

\iii $\Delta_\aa \incl A^*_\aa$ convex set of positive local probabilities: 
$p_\aa > 0$ and $\sum_{x_\aa} 
p_\aa(x_\aa) = 1 $
\ei

\textbf{Local Operators}

For every $\aa \incl \Om$ and $\bb \incl \aa$:
\bi
\iii $\pi^{\bb\aa} : E_\aa \law E_\bb$ natural projection 
= restriction $x_\aa \mapsto x_\bb$  

    $$\pi^{\bb\aa} : (x_i)_{i \in \aa} \mapsto (x_i)_{i \in \bb}$$

\iii $j_{\aa\bb} : A_\bb \law A_\aa$ natural extension\footnote{
    $\:j_{\aa\bb}$ coincides with the identity map w.r.t. the inclusion
    $A_\bb \incl A_\aa$, we shall therefore simply write 
    $f_\bb$ for $j_{\aa\bb}(f_\bb)$.
} = inclusion $A_\bb \incl A_\aa$

    $$j_{\aa\bb}(f_\bb)(x_\aa) = f_\bb(x_\bb)$$ 

\iii $\Sigma^{\bb\aa} : A^*_\aa \law A^*_\bb$ partial integration 
     $=$ marginal projection

    $$ \Sigma^{\bb\aa}(q_\aa)(x_\bb)  
    = \sum_{x' \in E_{\aa \setminus \bb}} q_\aa(x_\bb, x') $$

\iii $\E^{\bb\aa}_{p_\aa} : A_\aa \law A_\bb$ conditional expectation 
w.r.t. $p_\aa \in \Delta_\aa$ given $\bb$

\[ 
\E^{\bb\aa}_{p_\aa}(f_\aa)(x_\bb) 
    = \E_{p_\aa}[f_\aa \st x_\bb] 
    = \sum_{x' \in E_{\aa \setminus \bb}} 
\frac {p_\aa(x_\bb, x') \: f_\aa(x_\bb, x')} {p_\bb(x_\bb)} 
\]

\iii $\Fh^{\bb\aa} : A_\aa \overset{C^\infty}{\law} A_\bb$ conditional free energy of $\aa$ given $\bb$
= effective energy

\[ 
\Fh^{\bb\aa}(f_\aa)(x_\bb) 
= - \ln \sum_{x' \in E_{\aa \setminus \bb}} \e^{ - f_\aa(x_\bb, x')} 
\] 
\ei

\textbf{Local Duality}

\bi
\iii natural duality bracket $\croc{-}{-} : A^*_\aa \otimes A_\aa \law \R$
$$ \croc{q_\aa}{f_\aa} =  \sum_{x_\aa \in E_\aa} q_\aa(x_\aa) \, f_\aa(x_\aa)$$ 
\iii covariance metric $\croc{-}{-}_{p_\aa} : A_\aa \otimes A_\aa \law \R$  
induced by a local probability $p_\aa \in \Delta_\aa$
$$ \Croc{f_\aa}{g_\aa}_{p_\aa} = \E_{p_\aa}\big[ f_\aa \, g_\aa \big] 
= \sum_{x_\aa \in E_\aa} p_\aa(x_\aa)\, f_\aa(x_\aa)\, g_\aa(x_\aa) 
$$

\ei

\textbf{Properties} 

\bi
\iii Adjunction of $\Sigma^{\bb\aa}$ and $j_{\aa\bb}$ for the natural 
duality brackets: 
\[ \croc{\Sigma^{\bb\aa}(q_\aa)}{f_\bb} 
= \croc{q_\aa}{j_{\aa\bb}(f_\aa)}
\]
\iii Adjunction\footnote{
    $\;\E^{\bb\aa}$ is the orthogonal projection of $A_\aa$ 
    onto $A_\bb \incl A_\aa$ for the covariance metric $\croc{-}{-}_{p_\aa}$.
} 
of $\E^{\bb\aa}_{p_\aa}$ and $j_{\aa\bb}$ for 
the metric induced by $p_\aa$ on $A_\aa$ 
\[  \Croc{\E^{\bb\aa}(f_\aa)}{g_\bb}_{\Sigma^{\bb\aa}(p_\aa)} = 
\Croc{f_\aa}{j_{\aa\bb}(g_\bb)}_{p_\aa} \]  
\iii Gibbs state conditional expectations\footnote{
    Gibbs state expectation $\E^\aa_{p_\aa} = \E^{\vide\aa}_{p_\aa}$ 
    is the differential of the free energy functional 
    $d\Fh^\aa_{H_\aa} = d \Fh^{\vide\aa}_{H_\aa}$, 
} from effective energies:
$\E^{\bb\aa}_{p_\aa} = d\Fh^{\bb\aa}_{H_\aa}$ 
\[ 
\E_{p_\aa}[f_\aa | x_\bb] = \Fh^{\bb\aa}(H_\aa + f_\aa) - \Fh^{\bb\aa}(H_\aa) 
+ o(f_\aa) \txt{for} p_\aa = \frac 1 {Z_\aa} \e^{-H_\aa} \]
\ei 

}


\newpage
{\bf \hfill Fields \hfill}

{\smaller 
Suppose given a covering $X = \{ \aa, \bb, \cc, \dots \}$ of $\Om$ by subsets 
s.t. $\aa \cap \bb \in X$ for every $\aa, \bb \in X$.

$n$-Fields\footnote{
    The graded vector space $A_\bullet(X)$ of fields 
    is an analog of the space $\Om^\bullet(\R^3)$ of 
    scalar, vector, ... fields on $\R^3$, 
    or of the space $C_\bullet(K)$ of chains in a simplicial complex $K$, 
    except here fields have functional coefficients $f_\aa(x_\aa)$ instead 
    of scalar coefficients. 
} 
are collections 
$\{ f_{\aa_0 \dots \aa_n} \in \R^{E_{\aa_n}} \,|\,
\aa_0 \contst \dots \contst \aa_n \in X \}$ 
of local observables indexed by $n$-chains 
in $X$.

\textbf{Field Spaces} 

\bi 
\iii $A_0(X) = \prod_\aa \R^{E_\aa}$ space of potentials 

\iii $A_1(X) = \prod_{\aa \contst \bb }  \R^{E_\bb}$ space of currents 

\iii $A_n(X) = \prod_{\aa_0 \contst \dots \contst \aa_n} \R^{E_{\aa_n}}$ 
space of local observable $n$-fields

\iii $\Delta_0(X) = \prod_{\aa} \Delta_\aa \incl A^*_0(X)$ convex space of positive 
beliefs

\iii $\Gamma(X) \incl \Delta_0(X)$ convex subset of consistent beliefs: 
$p \in \Gamma(X)$ iff $p_\bb = \Sigma^{\bb\aa}(p_\aa)$ for all $\aa \cont \bb$.
\ei

\textbf{Differential Operators\footnote{
    Differential here means that $\delta : A_{n+1}(X) \aw A_n(X)$ 
    and its adjoint $d : A^*_n(X) \aw A^*_{n+1}(X)$ extend to  
    square-null operators ($\delta^2 = \delta \circ \delta = 0$ and $d^2 = 0$) 
    on the whole complexes $A_\bullet(X)$ and $A^*_\bullet(X)$.
    They play the role of discrete spatial differentiation operators,
    and the terminology chosen to reflect their analogs on 
    the space of smooth fields $\Om^\bullet(\R^3)$. 
    See equations (\ref{boundary}) and (\ref{differential}) 
    for the actions of $\delta$ and $d$ on higher degrees. 
}} 

\bi
\iii $\delta : A_1(X) \law A_0(X)$ divergence 

\[ \delta(\ph)_\bb(x_\bb) = \sum_{\aa' \cont \bb} \ph_{\aa'\bb}(x_\bb)
- \sum_{\bb \cont \cc'} \ph_{\bb\cc'}(x_{\cc'}) 
\] 

\iii $d : A^*_0(X) \law A^*_1(X)$ differential 

\[ d(q)_{\aa\bb}(x_\bb) = 
q_\bb(x_\bb) - \Sigma^{\bb\aa}(q_\aa)(x_\bb) \]

\iii $\nabla_p : A_0(X) \law A_1(X)$ gradient 
w.r.t. to a consistent belief $p \in \Gamma(X)$ 

\[ \nabla_p(f)_{\aa  \bb}(x_\bb) = 
f_\bb(x_\bb) - \E^{\bb\aa}_{p_\aa}(f_\aa)(x_\bb) \]

\iii $\DF : A_0(X) \law A_1(X)$ effective energy gradient  

\[ \DF(f)_{\aa\bb}(x_\bb) = f_\bb(x_\bb) - \Fh^{\bb\aa}(f_\aa)(x_\bb) \]

\ei

\textbf{Field Duality} 

\bi
\iii natural duality bracket 
$\croc{-}{-} : A^*_n(X) \otimes A_n(X) \law  \R$ 

\[ \croc{q}{f} = \sum_{\aa_0 \contst \dots \contst \aa_n} 
\croc{q_{\aa_0 \dots \aa_n}}{f_{\aa_0 \dots \aa_n}} 
\]

\iii covariance metric 
$\croc{-}{-}_p : A_n(X) \otimes A_n(X) \law \R$ 
induced by a consistent $p \in \Gamma(X)$ 

\[ \Croc{f}{g}_{p} = \sum_{\aa_0 \contst \dots \contst \aa_n} 
\Croc{f_{\aa_0 \dots \aa_n}}{g_{\aa_0 \dots \aa_n}}_{p_{\aa_n}} 
\]
\ei 

\textbf{Properties} 

\bi
\iii Adjunction of $d$ and $\delta$ for the natural duality bracket 
\[ \croc{dq}{\ph} = \croc{q}{\delta \ph} \]

\iii Adjunction\footnote{
    These two properties are discrete analogs  
    of the integration by parts formula 
    $\int_{\R^3} \vec{\rm grad}(f) \cdot \vec \ph = 
    - \int_{\R^3} f \: {\rm div}(\vec \ph)$. 
} of $\nabla_p$ and $\delta$ for the metric induced by 
$p \in \Gamma(X)$ 
\[ \Croc{\nabla_p(f)}{\ph}_p = \Croc{f}{\delta \ph}_p \]

\iii Gibbs State gradient operator $\nabla_p = d\DF_H$
\[ \nabla_p(f)_{\aa\bb} = \DF(H + f)_{\aa\bb} - \DF(H)_{\aa\bb}  + o(f) 
\txt{for} p_\aa = \frac{1}{Z_\aa} \e^{-H_\aa} \]
\ei 

\newpage

\textbf{Combinatorial Operators}

\bi 

\iii $\tilde \zeta_\Om : A_0(X) \law A_\Om$ total energy 

\[ \tilde \zeta(h)_\Om(x_\Om)  = \sum_\aa h_\aa(x_\aa) = H_\Om(x_\Om) \]

\iii $\zeta : A_0(X) \law A_0(X)$ zeta transform\footnote{
    << $\zeta(h)_\aa = \int_{\LL^\aa} h$ >> is analogous to 
    a discrete integral of the potential $h$ 
    on the cone $\LL^\aa = \{ \bb \incl \aa\}$ below $\aa$.
}
\[ \zeta(h)_\aa(x_\aa) = \sum_{\aa \cont \bb'} h_{\bb'}(x_{\bb'}) = H_\aa(x_\aa) \]

\iii $\mu : A_0(X) \law A_0(X)$ Möbius transform\footnote{
    The coefficients $\mu_{\aa\bb} \in \Z$ are computed inductively by 
    $\mu_{\aa\aa} = 1$ and $\mu_{\aa\cc} = 1 - 
    \sum_{\aa \contst \bb' \contst \cc} \mu_{\aa\bb'}$ 
    for every $\Om \cont \aa \cont \cc$. 
}
: $\mu = \zeta^{-1}$ 

\[ \mu(H)_\aa(x_\aa) = 
 \sum_{\aa \cont \bb'} \mu_{\aa\bb'} \: H_{\bb'}(x_{\bb'}) = h_\aa(x_\aa) \]

\iii $\zeta : A_n(X) \law A_n(X)$ extended zeta transform 

\[ \zeta(\ph)_{\aa_0 \dots \aa_n} = 
\sum_{\aa_0 \cont \bb_0 \not \incl \aa_1} \dots 
\sum_{\aa_n \cont \bb_n} \ph_{\bb_0 \dots \bb_n} 
= \Phi_{\aa_0 \dots \aa_n} 
\] 

\iii $\mu : A_n(X) \law A_n(X)$ extended Möbius transform: $\mu = \zeta^{-1}$

\[ \mu(\Phi)_{\aa_0 \dots \aa_n} 
= \sum_{\aa_n \cont \bb_n} \dots 
\sum_{\aa_0 \cont \bb_0 \not \incl \bb_1} 
\mu_{\aa_0 \bb_0} \dots \mu_{\aa_n \bb_n} 
\: \Phi_{\bb_0 (\bb_0 \cap \bb_1) \dots (\bb_0 \cap \dots \cap \bb_n)}  
= \ph_{\aa_0 \dots \aa_n} 
\] 
\ei

\textbf{Properties} 

\bi

\iii Möbius numbers\footnote{
    The coefficients $c_\aa \in \Z$ are computed inductively 
    $c_\bb = 1 - \sum_{\aa' \contst \bb} c_{\aa'} 
    = \sum_{\aa' \contst \bb} \mu_{\aa'\bb}$.
} $c_\aa \in \Z$ and total energy $H_\Om$ of a potential 
$h = \mu \cdot H \in A_0(X)$ 

\[ H_\Om = \sum_\aa h_\aa = \sum_\aa c_\aa \: H_\aa 
= \tilde \zeta(h)_\Om \] 

\iii Gauss formula\footnote{
    << $\zeta(\delta \ph)_\aa = \int_{\LL^\aa} \delta \ph = 
    \int_{d \LL^\aa} \ph$ >>
    is analogous to a discrete flux integral 
    of the current $\ph$ bound into $\LL^\aa$. 
} for a current $\ph \in A_1(X)$

\[ \zeta(\delta \ph)_\aa = \sum_{\aa \cont \bb'} \delta(\ph)_{\bb'} 
= \sum_{\aa' \not \incl \aa}\: \sum_{\aa \cont \bb'} \ph_{\aa'\bb'} 
= \tilde\zeta(\ph)_{\Om\aa} 
\] 

\iii Gauss formula\footnote{
    Möbius inversion on the effective energy gradient 
    $\ph = \mu \cdot \DF(U)$ before updating effective energies 
    by $\dot U = \zeta(\delta \ph)$ is one of our proposed regularisations 
    of message-passing schemes (chapter 5).
}
for $\ph = \mu \cdot \Phi \in A_1(X)$

\[ \zeta\big(\delta(\mu \cdot \Phi) \big)_\aa 
= \sum_{\aa' \not \incl \aa} 
c_{\aa'}\: \Phi_{\aa'(\aa \cap \aa')} 
= \tilde \zeta(\mu \cdot \Phi)_{\Om\aa} 
\] 
\ei 

\textbf{Main Theorem\footnote{
    This result will rely on the fundamental 
    yet not so widely known
    {\it interaction decomposition theorem} 
    \ref{int-decomposition} \cite{Kellerer-64, Matus-88}
    which consistently decomposes each 
    $A_\aa$ as a direct sum of interaction subspaces 
    $\oplus_\bb Z_\bb$ for $\bb \incl \aa$. 
}} 

\bi
\iii Homology and total energy: for every $u, h \in A_0(X)$
\[ \exists \ph \in A_1(X) \txt{s.t.} u = h + \delta \ph 
\quad \eqvl \quad 
\sum_\aa h_\aa = \sum_\aa u_\aa 
\] 
\ei 
}

\newpage

{\bf \hfill Energy and Information Functionals \hfill}
\vspace{0.5cm}

{\smaller

\textbf{Local Functionals} 

\bi 
\iii $\Fh^\aa : A_\aa \overset{C^\infty}{\law} \R$ local free energy 
\[ \Fh^\aa(H_\aa) = - \ln \sum_{x_\aa \in E_\aa} \e^{-H_\aa(x_\aa)} \] 

\iii $S_\aa : \Delta_\aa \overset{C^\infty}{\law} \R$ local entropy 
\[ S_\aa(p_\aa) = - \ln \sum_{x_\aa \in E_\aa} p_\aa(x_\aa) \: \ln p_\aa(x_\aa) 
\] 
\ei

\textbf{Legendre Duality} 

\bi
\iii $S_\aa$ is the Legendre transform of $\Fh^\aa$, and reciprocally
\[ \Fh^\aa(H_\aa) = \min_{p_\aa \in \Delta_\aa} \Big[ 
\croc{p_\aa}{H_\aa} - S_\aa(p_\aa) \Big] 
\] 

\iii $d\Fh^\aa : A_\aa \overset{C^\infty}{\law} \Delta_\aa$ 
maps hamiltonians to their Gibbs probability densities 
\[ d \Fh^\aa(H_\aa) = \croc{p_\aa}{{-}}  = \E_{p_\aa}[{-}] 
\txt{where} p_\aa = \frac 1 {Z_\aa} \e^{-H_\aa} \] 

\iii $dS_\aa : \Delta_\aa \overset{C^\infty}{\law} A_\aa / \R$ 
maps probability densities to their hamiltonians, defined up to additive constants
\[ d S_\aa(p_\aa) = \croc{{-}}{H_\aa} 
\txt{where} H_\aa \simeq - \ln p_\aa \mod \R 
\] 
\ei

\vspace{1cm}

\textbf{Global Functionals} 

\bi 
\iii ${\cal U}_\Om : \Delta_\Om \times A_\Om \overset{C^\infty}{\law} \R$ 
internal energy
\[ {\cal U}_\Om(p_\Om, H_\Om) = \croc{p_\Om}{H_\Om} = \E_{p_\Om}[H_\Om] \]

\iii $\Fg_\Om : \Delta_\Om \times A_\Om \overset{C^\infty}{\law} \R$ 
variational free energy 
\[ \Fg_\Om(p_\Om, H_\Om) = 
\croc{p_\Om}{H_\Om} - S_\Om(p_\Om) \]
\ei

\textbf{Bethe-Kikuchi Approximation} 

\bi
\iii $\Fb : \Delta_0(X) \times A_0(X) 
\overset{C^\infty}{\law} \R$ Bethe free energy: constrained 
to consistent beliefs $p \in \Gamma(X)$ 
\[ \Fb(p, H) = 
    \sum_{\aa \in X} c_\aa 
\Big[ \croc{p_\aa}{H_\aa} - S_\aa(p_\aa) \Big] 
\] 
\ei

\textbf{Main Theorems} 

\bi 
\iii Homological invariance of $\Fb(p, \,-\,)$: for every consistent belief $p \in \Gamma(X)$ 
and effective hamiltonian $H \in A_0(X)$ 
\[ \Fb\big(p, H + \zeta(\delta \ph)\big) = \Fb(p, H) \]

\iii $p \in \Gamma(X)$ 
is critical for $\Fb({-}, H)_{|\Gamma(X)}$ iff there 
exists a current $\ph \in A_1(X)$ s.t. for all $\aa \in X$ 
\[ 
- \ln p_\aa \simeq H_\aa + \zeta(\delta \ph)_\aa \mod \R 
\]
\ei
}

\newpage
{\bf \hfill Message-Passing as Diffusion \hfill} 

{\smaller 

\textbf{Heat Analogy} 

Energy density = time-dependent scalar field 
$u : \R \aw \Om^0(\R^3) = C^\infty(\R^3)$ 

Heat exchange = time-dependent vector field
$\vec \ph : \R \aw \Om^1(\R^3) = C^\infty(\R^3, \R^3)$ 
\bi
\iii energy conservation: $\dot u = {\rm div}(\vec \ph)$ 
\iii heat flux: $\vec \ph = - \lambda \: \vec{\rm grad} (T)$ 
\ei 

Characteristic relation $u = c \: T$ (condensed matter) or non-relationship
between temperature and energy.

\textbf{Message-Passing on Graphs\footnote{ 
    On acyclic graphs (trees) the algorithm 
    converges in finite time, 
    as already stated in Pearl's seminal paper \cite{Pearl-82}.
    Substitute  
    $u^{(t+1)} - u^{(t)}$ for $\frac {du}{dt}$ and translate to beliefs 
    to recover the usual belief propagation algorithm.
}}

The algorithm takes the form 
$\dot u = - \big( \div \circ \DF \circ \zeta \big)(u)$ on potentials, i.e. 
\[
    \ba{ll|l} 
    \dot u = \div \ph &\quad\txt{where}\;\;\quad 
    & \ph = - \DF(U) \\[0.4em]
    & & U = \zeta \cdot u 
\ea
\]

\bi 
\iii Energy conservation $\dot u = \delta \ph$ 
dictates the update of effective potentials $u \in A_0(X)$
\[ 
\frac{d}{dt} u_{ij}(x_i, x_j) = - \ph_{ij \to j}(x_j) - \ph_{ij \to i}(x_i) 
\txt{and} 
\frac{d}{dt} u_i(x_i) = \sum_{j' \sim i} \ph_{ij' \to i}(x_i) 
\] 

\iii Heat flux $\ph = - \DF(U) \in A_1(X)$ measures the lack 
of consistency of effective hamiltonians
\[
\ph_{ij \to j}(x_j) =  - U_j(x_j)
- \ln \sum_{x_i \in E_i} \e^{-U_i(x_i, x_j)}
\]

\iii Effective hamiltonians $U = \zeta(u) \in A_0(X)$ are given by
\[ U_{ij}(x_i, x_j) = u_{ij}(x_i, x_j) +  u_i(x_i) + u_j(x_j) 
 \txt{and} U_i(x_i) = u_i(x_i) 
\] 

\iii Beliefs
$q_{ij} = \frac 1 {Z_{ij}} \e^{-U_{ij}}$ and $q_i = \frac 1 {Z_i} \e^{-U_i}$
should be normalised at each iteration on graphs with loops\footnote{
    With loops, the dynamic on potentials is best understood 
    up to additive constants. 
}     
\ei

\textbf{Message-Passing on Hypergraphs\footnote{
    On {\it retractable} hypergraphs $X \incl \Part(\Om)$, 
    we show the algorithm to converge in finite time (chapter 6). 
    Note that Möbius inversion of the heat flux only affects 
    additive constants when $X$ is a graph, 
    hence the proposed regularisation only modifies the 
    generalised belief propagation (GBP) algorithm 
    of Yedidia et al. \cite{Yedidia-2005}
}} 

Möbius inversion on the heat flux 
reads $\dot u = - \big(\delta \circ \mu \circ \DF \circ \zeta\big)(u)$ 
on potentials, i.e.
\[
    \ba{ll|l} 
    \dot u = \div \ph &\quad\txt{where}\;\;\quad 
    & \ph = \mu \cdot \Phi\\[0.4em]
    & & \Phi = - \DF(U) \\[0.4em]
    & & U = \zeta \cdot u 
\ea
\]
\bi 
\iii Effective hamiltonians 
$U = \zeta \cdot u \in A_0(X)$ 
follow the energy conservation principle:
$\dot U = \delta^\zeta (\Phi)$ where $\delta^\zeta = 
\zeta \, \delta \, \zeta^{-1}$
\[ 
\frac d {dt} U_\aa(x_\aa) = 
\sum_{\aa' \not \incl \aa} c_{\aa'} \: 
\Phi_{\aa'(\aa \cap \aa')}(x_{\aa \cap \aa'}) 
\] 

\iii Extensive heat flux $\Phi = \zeta \cdot \ph \in A_1(X)$ 
flows against the effective energy gradient: $\Phi = - \DF(U)$ 
\[ 
\Phi_{\aa\bb}(x_\bb) = - U_\bb(x_\bb) 
- \ln \sum_{x' \in E_{\aa \setminus \bb}} \e^{-U_\aa(x_\bb, x')} 
\] 

\iii Beliefs $q_\aa = \frac 1 {Z_\aa}$ do not need to be normalised\footnote{ 
    When $\varnothing \in X$,
    Möbius inversion of fluxes $\Phi_{\aa \varnothing} \in \R$ 
    already takes care of regularising normalisation factors. 
}
\ei 

\textbf{Correspondence Theorem} 

Effective hamiltonians $U = H + \zeta \cdot \delta \ph \in A_0(X)$ are stationary 
under diffusion\\[0.4em]
$\eqvl \:\:$ 
Beliefs $q = \frac 1 Z \e^{-U} \in \Gamma(X)$ are consistent and
critical for the Bethe free energy $\Fb({-}, H)_{|\Gamma(X)}$

}

\newpage

\phantomsection
\addcontentsline{toc}{section}{Crystals, Codes and Networks}
{\bf \hfill Crystals, Codes and Networks \hfill} 

{\smaller 

\vfill 

The following 3 pages give a brief overview of some applications 
that motivated this thesis. They involve a rather wide spectrum 
of communities, including statistical physics, artificial intelligence 
and information theory. 
Although an extensive coverage of applications is far out of our scope, 
the book {\it Information, Physics and Computation} 
by Mézard and Montanari \cite{Mezard-Montanari} is an excellent reference 
for all three subjects. 

\vfill

\setlength{\parindent}{0pt}

\textbf{Crystals and Spin Glasses} 
    \vspace{-0.4cm}
\begin{figure}[H]
    \sbox0{\includegraphics[width=0.45\textwidth]{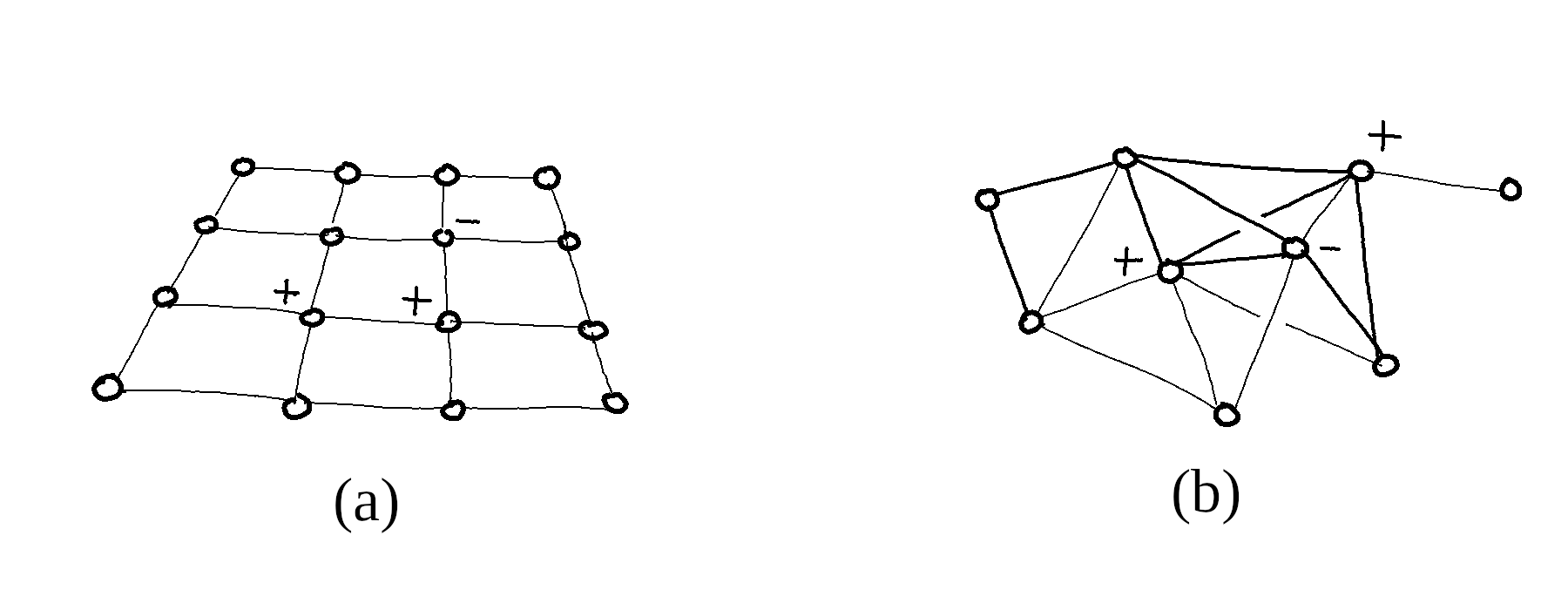}}
\begin{center}
\begin{minipage}{0.8\textwidth}
\centering
\usebox0
    \vspace{-0.5cm}
    \caption{{\smaller
    (a) Ising model 
    on the square lattice $\Z^n$
    (b) Spin glasses as in the Sherrington-Kirkpatrick 
    model have random magnetic couplings and 
    release the $\Z^n$ symmetry assumption. 
    }}
\end{minipage}
\end{center}
    \vspace{-0.3cm}
\end{figure}

A set of atoms $\Om$ carries a magnetic dipole $x_i \in \{ \pm 1 \}$ 
for all $i \in \Om$. 

A set of pairs in $\Om \times \Om$ relates neighbours $i \sim j$
whose interactions contribute to the total energy of the system.

\bi 
\iii Local magnetic field\footnote{
    The effect of $h_i(x_i) = \pm b_i$ is to  
    bring the dipole $x_i$ aligned with the field $b_i \in \R$, 
    with ground energy $- b_i$.  
} = bias $b_i \in \R$ 
\[ h_i(x_i) = - b_i \: x_i \]

\iii Local magnetic coupling\footnote{
    The effect of $h_{ij}(x_i, x_j) = \pm w_{ij}$ is to 
    bring dipole $x_j$ aligned with $x_i$ 
    when $w_{ij} > 0$, and opposed with $x_i$ otherwise.
} = weight $w_{ij} \in \R$ 
\[ h_{ij}(x_i, x_j) = - w_{ij} \: x_i x_j \] 

\iii Total hamiltonian $H : \{ \pm 1 \}^\Om \law \R$  
\[ H(x) = \sum_{i \in \Om} h_i(x_i) 
+ \sum_{i \sim j} h_{ij}(x_i, x_j) \] 

\iii Gibbs state $p^\theta : \{\pm 1\}^\Om \law \R$ describes 
statistical equilibrium at inverse temperature\footnote{
    We denote inverse temperature by $\theta$ instead of the 
    usual $\bb = 1 / k_B T$ to avoid future confusion with our 
    notation for subsets $\aa, \bb, \cc, \dots \incl \Om$. 
    Note that $\theta$ will generally be considered set to $\theta = 1$,
    the effect of temperature being viewed through dilations 
    on the hamiltonian $H \in \R^{\{\pm 1 \}^\Om}$.   
} $\theta = \frac 1 {k_B T}$
\[ p^\theta(x) = \frac { \e^{- \theta H(x)} } 
{ \sum_{x'} \e^{- \theta H(x')} } \]
\ei

{\it High temperature limit:} if $\theta \to 0$ then $p^\theta$ tends to the uniform 
measure on $\{\pm 1 \}^\Om$. 

{\it Low temperature limit:} if $\theta \to + \infty$ then $p^\theta$ tends to a 
barycenter of Dirac measures on the minima\footnote{
    In the ferromagnetic case i.e. $w_{ij} = 1$ for all $i \sim j$, 
    there are two ground-energy states $x = (+1, \dots,+1)$ and 
    $x = (-1, \dots, -1)$ to which the system crystallises 
    below the critical Curie temperature (spontaneous magnetisation). 
}
of $H$.

\vfill 

\newpage

\textbf{Boltzmann Machines and Neural Networks}
\vspace{0.4cm}

\begin{figure}[h]
    \sbox0{\includegraphics[width=0.7\textwidth]{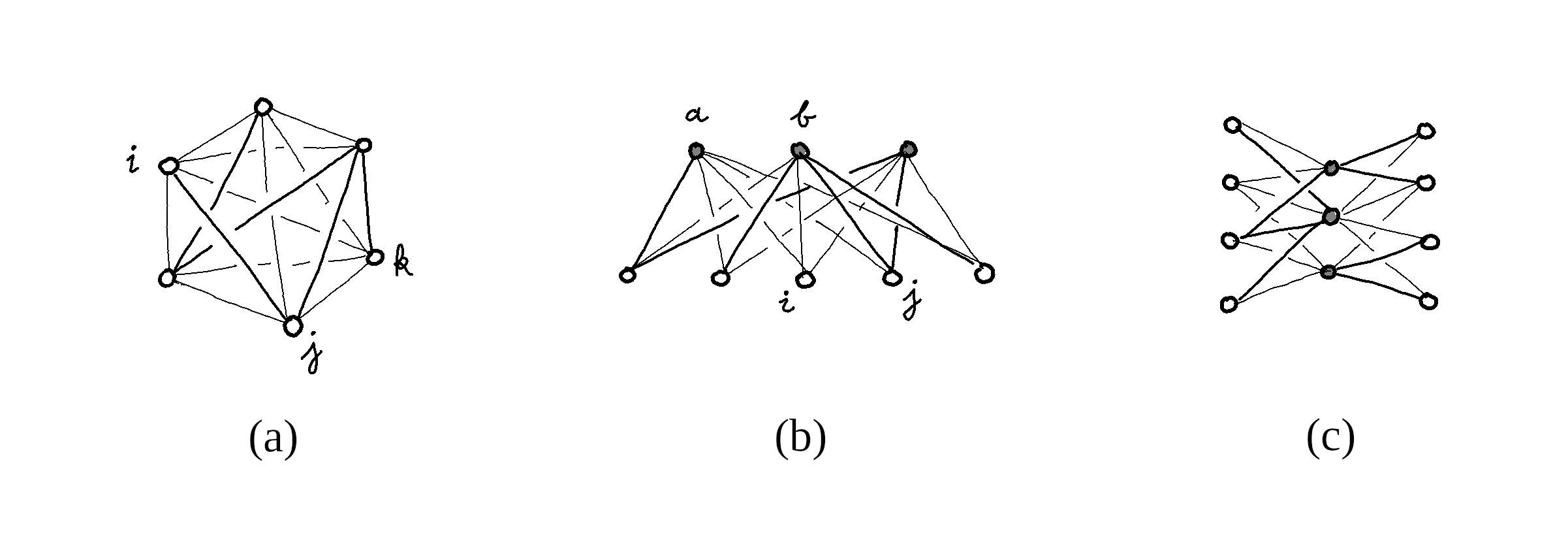}}
\begin{center}
\begin{minipage}{0.8\textwidth}
\centering
\usebox0
    \vspace{-0.3cm}
\caption{{\smaller 
    Fully-connected (a), restricted (b) and deep (c) Boltzmann machines. 
}}
\end{minipage}
\end{center}
\end{figure}

A restricted Boltzmann machine (RBM) consists of two neuron\footnote{
    $\: -1$ describes a steady neuron, 
    $+1$ a neuron firing at maximal rate, 
    and convex combinations i.e. probability densities on $\{\pm 1 \}$ 
    then correspond to intermediate firing rates.
} layers 
$x_1, \dots, x_N \in \{ \pm 1 \}$ and $y_1, \dots, y_P \in \{ \pm 1 \}$.

The network is trained to generate configurations 
on its {\it visible} layer $x$, interacting
with the {\it hidden} layer $y$.

Biases $a_i, b_j \in \R$ and couplings $w_{ij} \in \R$ 
are learned from observed samples
$\bar x^1, \dots, \bar x^n \in \{\pm 1\}^N$ on the visible layer.

\bi 
\iii Probability of a configuration $(x, y) \in \{\pm 1\}^{N + P}$ modeled by 
\[ p(x, y) = \frac 1 Z \e^{-H(x, y)} 
\txt{with} H(x, y) = - \sum_{i = 1}^N a_i x_i - \sum_{j = 1}^P b_j y_j 
- \sum_{i=1}^N \sum_{j=1}^P w_{ij}\; x_i y_j \]   

\iii Maximise the expected log-likelihood\footnote{
    The negative log-likelihood of $x$ 
    will be seen as the difference of an effective energy term 
    $\Fh[H|x] = - \ln \sum_y \e^{-H(x,y)}$ and 
    a free energy term $\Fh[H] = - \ln Z = - \ln \sum_{x,y} \e^{-H(x,y)}$. 
    See next pages and chapter 4, proposition \ref{Eba}.
} of a training sample 
$\bar x^1, \dots, \bar x^n$ 
\[
\ba{lllll} 
\E_{\bar x}[ - \ln p(\bar x) ] 
&=& \disp 
- \frac {1} {n} \sum_{s = 1}^n  
\ln  \sum_{y \in \{ \pm 1\}^P} \e^{-H(\bar x^s, y)}    
&+&  \disp 
\ln  \sum_{x \in \{\pm 1\}^N} \sum_{y \in \{ \pm 1 \}^P} 
\e^{- H(x, y)} 
\ea
\] 

\iii Estimate\footnote{
    The first term is exactly computable using the Markov properties 
    of the network, however the second term requires to estimate marginals 
    either by Hinton's contrastive divergence algorithm (CD) or by 
    belief propagation (BP).
} the loss-function gradient along 
a local variation\footnote{
    The potential $h$ varies with model parameters, e.g. 
    $h(x, y) = - \Delta w_{ij}\: x_i y_j$ 
} $h(x, y) = h(x_i, y_j)$ of the total energy $H(x, y)$ 
\[ 
\ba{lllll}
\disp \frac{\dr \E_{\bar x}[ - \ln p(\bar x)]} {\dr h} 
&=& \E_{\bar x}\big[ \, \E_p \big[\, h(\bar x_i, y_j) \,\big|\, \bar x \,\big] \big] 
&-& \E_p\big[ \, h(x_i, y_j) \, \big]  \\ 

&=& \disp \frac 1 n 
\sum_{s=1}^n \sum_{y_j \in \{\pm 1 \}} p(y_j | \bar x^s)\: h(\bar x_i, y_j) 
&+& \disp \sum_{x \in \{\pm 1 \}} \sum_{y \in \{ \pm 1 \} } 
p(x_i, y_j) \: h(x_i, y_j) 
\ea
\]  
\ei

{\it Markov properties:} conditional independence of 
$x_1, \dots, x_n$ given $y$ and reciprocally\footnote{
    $\: y_1, \dots, y_n$ are conditionally independent given $x$. 
    Here $\sigma$ denotes the logistic function 
    $\sigma(u) = \frac{1}{1 + \e^{- 2 u}} = \frac 1 Z \e^{-u}$. 
}
\[ p(x|y) = \prod_{i = 1}^N p(x_i | y) 
= \prod_{i = 1}^N \sigma \Big( a_i x_i - \sum_j w_{ij} x_i y_j \Big) 
\] 

{\it Gibbs sampling\footnote{
    Gibbs sampling is a type of Markov-chain Monte Carlo, 
    on which the contrastive divergence algorithm relies. 
    Belief propagation does not resort to Monte Carlo methods. 
}:} 
\begin{minipage}[t]{0.7\textwidth}
    start from a random configuration $y^{(0)} \in \{\pm 1\}$ \\
    ... draw $x^{(t + 1)} \in \{ \pm 1\}^N$ from $p(x \,|\, y^{(t)})$ \\
    ... draw $y^{(t + 1)} \in \{ \pm 1\}^P$ from $p(y \,|\, x^{(t+1)})$ \\ 
    average over time $t$ to get an estimate of $p(x, y)$ and its marginals
\end{minipage}

\newpage

\textbf{Low Density Parity-Check Codes\footnote{
    LDPCs were introduced in Gallager's 
    1960 thesis \cite{Gallager-63} along with the electronic decoding  
    apparatus, equivalent to the later called belief propagation << algorithm >>. 
    The sparse parity-check matrix and the parallelised decoding scheme 
    reach performance close to the Shannon capacity of the channel, 
    hence the revived interest shown by 5G networks. 
}} 
\begin{figure}[H]
    \sbox0{\includegraphics[width=0.3\textwidth]{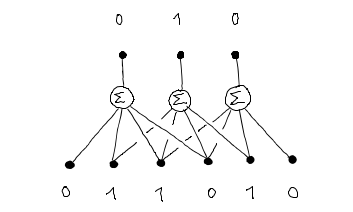}}
\begin{center}
\begin{minipage}{0.8\textwidth}
\centering
\usebox0
\caption{{\smaller
    Parity-check codes consist of two sequences
    $(x_i) \in \{0, 1\}^N$ of {\it signal} bits (bottom)
    and $(y_j) \in \{0, 1\}^P$ of {\it validation} bits (top)
    computing the sum mod $2$ of a subset of signal bits. 
}}
\end{minipage}
\end{center}
    \vspace{-0.5cm}
\end{figure} 

A compressed message is encoded as a binary sequence $(x_i) \in \{0,1\}^N$ 
followed by a parity-check sequence $(y_j) \in \{0,1\}^P$. 

Decoding consists in restoring parity-check consistency 
to rectify the errors induced by a noisy transmission channel. 

\bi
\iii Beliefs $q_i(x_i)$ and $q_j(y_j)$ initially depend on 
the received values\footnote{
    The parameter $\theta$ may be tuned to the transmission channel 
    noise/signal ratio, so that 
    $\frac 1 {1 + \e^{\theta}} = \P(x_i = 0 \st \bar x_i = 1)$. 
}
$\bar x_i$ and $\bar y_j$ e.g.
\[ 
q_i(x_i) = \frac 1 {Z_i} \e^{- h_i(x_i)}  \txt{with} 
h_i(x_i) = \theta \cdot {|x_i - \bar x_i|} 
\] 
\iii Local potentials $h_{\aa_j}$ for every validation bit $j \in \{1, \dots, P\}$
connected to a subset $\aa_j \incl \{ 1, \dots, N \}$ of signal bits
\[ h_{\aa_j}(x_{i_1}, \dots, x_{i_n}, y_j) 
= \lambda \cdot \bigg[ \Big( y_j - \sum_{i \in \aa_j} x_i \Big) \: \% \: 2 \bigg]  
\] 

\iii Local beliefs $q_{\aa_j}$ favorise parity-check 
consistency\footnote{
    Parity-check consistency is usually enforced as a {\it hard} constraint 
    i.e. in the limit $\lambda \to + \infty$, where inconsistent configurations 
    have zero-probability, which amounts to assume a fail-proof encoding. 
    In any case, consider $\lambda >> \theta$. 
} and agreement with received values 
\[ 
q_{\aa_j}(x_{i_1}, \dots, x_{i_n}, y_j) 
= \frac 1 {Z_{\aa_j}} \e^{-H_{\aa_j}(x_{i_1}, \dots, x_{i_n}, y_j)} 
\txt{with} 
H_{\aa_j} = h_{\aa_j} + h_j + \sum_{i \in \aa_j} h_i 
\] 
\ei

\vspace{0.4cm}

{\it Message-Passing:} 
\begin{minipage}[t]{0.8\textwidth} 
    initialise beliefs according to signal received\\[0.3em]
    ... compute a message $m_{\aa_j \to j}$ for all $\aa_j$ 
    and messages $m_{\aa_j \to i_p}$ for 
        all $i_p \in \aa_j$ 
    \[
    \ba{lll} 
        m_{\aa_j \to j}(y_j) 
        &= \disp \frac 1 {q_j(y_j)} 
        & \disp \sum_{x' \in \{0, 1\}^{\aa_j}} 
            q_{\aa_j}(x'_{i_1}, \dots, x'_{i_n}, y_j) 
    \\
        m_{\aa_j \to i_p}(x_{i_p}) 
            &= \disp \frac 1 {q_{i_p}(x_{i_p})} 
            &  \disp 
            \sum_{x' \in \{0, 1\}^{\aa_j \setminus i_p}}
            \: \sum_{y'_j \in \{0, 1\}} 
            q_{\aa_j}(x'_{i_1}, \dots, x_{i_p}, \dots, x'_{i_n}, y_j) 
    \ea 
    \]
    ... update beliefs\footnotemark{} according to incoming messages and normalise
    \[ \ba{ll} 
        q_j(y_j) 
            &\wa\; q_j(y_j) \: m_{\aa_j \to j}(y_j)  \\[0.5em]
        q_i(x_i) 
            &\wa\; \disp  q_i(x_i) \prod_{\aa_j \ni i} 
            m_{\aa_j \to i}(x_i) \\ 
        q_{\aa_j}(x_{i_1}, \dots, x_{i_n}, y_p) 
            &\wa\; \disp  q_{\aa_j}(x_{i_1}, \dots, x_{i_n}) 
            \prod_{i_p \in \aa_j} \prod_{\aa_k \ni i_p}
            m_{\aa_k \to i_p}(x_{i_p}) 
    \ea
    \]
    loop until computed messages are close to 1 
    $\eqvl$ belief consistency is 
    reached
\end{minipage}
\footnotetext{
    The dynamic variable is usually considered to be the set 
    of messages: we show that stationarity of beliefs implies 
    stationarity of messages (theorem \ref{faithful}) 
    so that considering the dynamic on beliefs instead 
    is equivalent from the point of view of message-passing equilibria. 
}

}

\section*{Index of Notations} 

{\smaller 
{\bf Functors:} 
\bi
\iii $X \incl \Part(\Om)$ base hypergraph  
\iii $(E, \pi)$ microstates, \S 2.1.1
\iii $(A, j)$ observables, \S 2.2.2
\iii $(A^*, \Sigma)$ measures, \S 2.2.3 
\iii $(\Delta, \Sigma)$ probability densities, \S 2.2.4
\ei  

{\bf Spaces:}
\bi 
\iii $A_\bullet(X)$ complex of local observables, def. \ref{complexes}
\iii $A^*_\bullet(X)$ complex of local measures, --
\iii $\Delta_\bullet(X) $ convex subspace of local probabilities, --
\iii $\Gamma(X)$ convex subspace of consistent local probabilities, 
def. \ref{Gamma}
\iii $\CU(X)$ manifold of consistent local hamiltonians, def. \ref{CU}
\iii ${\cal Z}(X)$ manifold consistent local potentials, def. \ref{Z}
\ei

{\bf Differential Operators:}
\bi 
\iii $\div$ boundary of $A_\bullet(X)$, \S 2.2.2
\iii $d$ differential of $A^*_\bullet(X)$, --
\iii $\DF$ effective energy gradient, \S 5.2.1
\iii $\nabla = \DF_*$ linearised effective energy gradient -- 
\ei

{\bf Combinatorial Operators:}
\bi 
\iii $\zeta$ zeta transform, sections 3.2 and 3.3
\iii $\mu = \zeta^{-1}$ Möbius transform, \S 3.2.1 and 3.3.3
\ei 

{\bf Diffusion Operators:}
\bi 
\iii $\Phi = - \DF \circ \zeta$ standard diffusion flux, \S 5.2.2 
\iii $\Tspt = \div \Phi$ standard diffusion vector field, --
\iii $\phi = - \mu \circ \DF \circ \zeta$ canonical diffusion flux, \S 5.3.2
\iii $\tau = \div \phi$ canonical diffusion vector field, --
\ei

{\bf Fields:}
\bi 
\iii $h, u \in A_0(X)$ interaction potentials,
 $h$ for reference and $u$ for evolution
\iii $H, U \in A_0(X)$ local hamiltonians, 
$H = \zeta \cdot h$ and $U = \zeta \cdot u$ 
\iii $\ph \in A_1(X)$ energy flux
\iii $q \in \Delta_0(X)$ local beliefs, $q_\aa = [\e^{-U_\aa}]$ 
\iii $p \in \Gamma(X)$ Gibbs state marginals $p_\aa = \Sigma^{\aa\Om}
\big(p_\Om \big)$
\ei

{\bf Information Functionals:} 
\bi 
\iii $\Fh^\aa$ free energy, \S 4.1.1
\iii $\Fh^{\bb\aa}$ effective energy, \S 4.1.2
\iii $S$ Shannon entropy, \S 4.2.1
\iii $\Fg$ variational free energy, \S 4.3.2
\iii $\Fb$ Bethe free energy, \S 4.3.3
\ei
}

\chapter*{Acknowledgements} 
I would first like to thank my director, 
Daniel Bennequin, whose curiosity and breadth of knowledge 
is only met by his patience, kindness and humility. 
Working by his side was an invaluable chance and an every day inspiration 
to look beyond scientific borders. 
I hope I can one day reach anywhere close to him in 
the span of these qualities. 

This work could never have been possible without the initial year-long 
turmoil of working groups 
and discussions I had with Grégoire Sergeant-Perthuis 
and Juan-Pablo Vigneaux. 
I feel very lucky that I had them by my side to teach and challenge me. 
I would like to thank Juan for having shared so much of his knowledge,
I hope to count him as a colleague and friend for long,
in spite of distances.  
I wish many luck to Grégoire in finishing his PhD that 
I am so curious to view complete.   
He knows how precious his support was to me, 
in so many different dimensions,  
and I hope I can make it up to him in the months to come. 

I have many precious friends that I could not thank individually. 
I would only like them to know that I care for them and want to 
spend more time cultivating the valuable relations too long neglected. 
They did not need to understand the motives behind this work 
to be helpful, and they all deserve credit for its completion.
I thank each of them warmly. 
I would however add a particular mention for my saturday 
library companions, who I was so glad to find and to whom 
I wish the strongest courage for the difficult work 
they carry on. 

I am obviously very grateful for the wonderful family I always have behind me. 
Much of the courage I used pursuing this long and demanding work comes from 
my parents, sisters and brothers in law.
I wish to be there for them as much as they have been for me. 
I would like to dedicate this work to the memory of my grandfather James, 
who was such an inspiring presence of calm and wisdom in my childhood 
that I picture him as patient and dedicated a teacher as he 
was with his grandchildren.

\tableofcontents 

\chapter{Homological Algebra}

This chapter consists of a short yet self-contained 
introduction to some of the most remarkable concepts 
of ${\rm XX^{th}}$ mathematics, which happened to take a definite form 
simultaneously and for the needs of one another: homology and categories. 
It still seems quite rare that they belong to a common language
with the physicist or the computer scientist, 
and we hope this chapter provides with more than the necessary 
material\footnotemark{}. 
For the more informed reader, 
this chapter's purpose is to relate our construction
to the general theory of simplicial groups. 
\footnotetext{
Apart from a few proofs, this work should demand little more than 
a good understanding of the notion of functor, 
and formulas defining the boundary operator and the differential could talk 
for themselves. 
}

Categories may be thought of as collections of points and arrows,
which describe mathematical objects and their relations,
while functors consistently transform categories into other categories. 
Section 1 reviews these elementary definitions and focuses on providing 
concrete examples 
such as the categories of groups, vector spaces, 
topological spaces, programmable types, {\it etc.}

The practical use of category theory language is specially remarkable 
in the characterisation of particular objects by universal properties. 
Section 2 focuses on the categorical concept of limit, 
which unifies many constructions such as union and product of sets, 
sums of vector spaces, inductive and projective limits, {\it etc.} 
It should familiarise the reader with commutative diagrams 
and will help describe homology groups in chapter 2. 

Homology provides a general procedure to extract algebraic invariants
from topological spaces, while 
cohomology may be thought of as an abstraction of differential calculus.
Section 3 provides with the basic definition of homology groups, 
which from a purely algebraic point of view, occur in the study 
of a square-null operator $d$ such that $d^2 = 0$. 
Geometry however best illustrates the purpose 
of homology, which unifies various integration by parts formulas 
under the Stokes theorem\footnotemark{}. The Gauss formula is a particular case 
of the latter, and its discrete analog \ref{gauss} encountered 
in chapter 2 will play a fundamental role in understanding the 
geometric structure of message-passing algorithms. 

\footnotetext{
    Homological thinking already emerged with Maxwell and Faraday, in formulating 
    physical principles for electromagnetism. They involve
    (i) the geometric operator $\bord$ mapping a subspace to its boundary,
    which has empty boundary, and (ii)  
    the differential $d$ acting on fields
    as gradient, curl, divergence, while $d^2$ vanishes as
    ${\rm div} \circ {\rm curl} = {\rm curl} \circ {\rm grad} = 0$.
}

\newpage 

\section{Categories and Functors}

\subsection{Categories}

\newcommand{\Ob}{{\rm Ob}}
\newcommand{\Cat}{{\bf C}}

Categories provide with a convenient abstraction of most mathematical constructions
and theories. 
They were introduced by Eilenberg and MacLane \cite{Eilenberg-MacLane}
to build homological algebra on a rigorous and flexible ground,
they have since proven useful in many diverse applications in mathematics, 
informatics and physics.

\begin{defn}
    A {\rm category} ${\bf C}$ is a class of objects $A, B, C, \dots$ denoted 
by $\Ob(\Cat)$ together with:
\bi 
\iii a set of arrows $\Hom(A,B)$ for every 
$A, B$ in $\Ob(\Cat)$,
\iii an identity arrow $1_A \in \Hom(A,A)$ for every $A$ in $\Ob(\Cat)$,
\iii a composite arrow $gf \in \Hom(A,C)$ for every $f \in \Hom(A,B)$ 
and $g \in \Hom(B,C)$
\ei 
satisfying the following axioms:
\bei
    \item {\rm Identity:} for every $f : A \aw B$ 
        \begin{equation} f = f \cdot 1_A = 1_B \cdot f \end{equation}
    \item {\rm Associativity:} for every $f : A \aw B$, $g: B\aw C$
        and $h: C \aw D$ 
        \begin{equation} h (g f) = (h g) f \end{equation}
\ee
\end{defn}

The first example of a category is $\Set$, the category
whose objects are sets and arrows are functions. 
When each object of $\Cat$ may be viewed as a set 
and each arrow $f \in \Hom(A, B)$ induces a function $f \in B^A$ 
of the underlying sets, the category $\Cat$ is called concrete, 
Equivalently, a concrete category $\Cat$ is a subcategory of $\Set$.
Although the following definitions make sense in any category, they
respectively correspond to bijections, injections, and surjections
in $\Set$ as in most examples of concrete categories.

\begin{defn} 
Let $\Cat$ be a category and $f : A \aw B$ a morphism.
\bi
\iii f is an {\rm isomorphism} if 
there exists $g: B \aw A$ such that $gf = 1_A$ and $fg = 1_B$.
    \iii f is a {\rm monomorphism} if for all $X$ and 
    $u, u' : X \aw A$, $fu = fu'$ implies $u = u'$.
\iii f is an {\rm epimorphism} if for all $Y$ and 
    $v, v' : B \aw Y$, $vf = v'f$ implies $v = v'$.
\ei
\end{defn}

A category may have {\it terminal} objects, 
satisfying one of the conditions of the following definition. 
These very special objects are also called {\it universal} as 
a terminal object of a given kind, when it exists, 
is always defined up to isomorphism. 
Universal objects are related to the existence of certain {\it limits}, 
and describe many fundamental constructions 
in algebra and geometry\footnote{
    Such constructions are called universal,
    a few of which being the object of section 1.2.
    The construction of the tensor algebra $( T(V), \otimes)$ 
    from a vector space $V$ is a classical example that is not exposed here. 
}.

\begin{defn} 
Let $\Cat$ be a category. 
\bi 
    \iii an object $I$ is {\rm initial} in $\Cat$ 
    if there is a unique arrow $I \aw A$ for every object $A$ in $\Cat$,
    \iii an object $F$ is {\rm final} in $\Cat$
    if there is a unique arrow $A \aw F$ for every object $A$ in $\Cat$,
    \iii an object $O$ is {\rm null} in $\Cat$
    if it is both final and initial.
\ei
\end{defn}

\begin{prop} 
    If $T$ and $T'$ are terminal objects of the same kind, 
    then $T$ is isomorphic to $T'$. 
\end{prop}

\begin{proof}
When $T$ is a terminal object, the axioms imply that $1_T$ 
is the unique arrow of $\Hom(T,T)$.
If $T'$ is terminal of the same kind,
the arrows $T \aw T'$ and $T' \aw T$ must then compose as $1_T$ and $1_{T'}$. 
\end{proof}

\begin{defn}
For any category $\Cat$, its {\rm dual} 
or {\rm opposite} category $\Cat^{op}$ 
has the same objects as $\Cat$ 
and, for every arrow 
$f : A \aw B$ in $\Cat$, 
a reversed arrow $f^{op} : B \aw A$ in $\Cat^{op}$. 
\end{defn}

In the following fundamental examples, 
we give an initial and a terminal object when they exist. 
In many interesting examples, 
the set of morphisms between two objects is  
also an object of the category.
This is not true in general 
and we precise when it is the case\footnote{
    The existence of a {\it hom-object} is a defining property of cartesian 
    categories.
}.

{\bf Examples of Categories:} 
\begin{enumerate}[itemsep=0.6em]
\item A partial order $(X, \geq)$ is a category with a unique arrow 
$x \aw y$ whenever $x \geq y$. 

The identity axiom $x \geq x$ expresses the reflexivity of the order relation,
while the transitivity asking that if $x \geq y$ and $y \geq z$, then $x \geq z$, 
is given by the existence of compositions.

Initial and final elements correspond 
to maximal and minimal elements respectively. 

\item The category $\Set$ whose objects are sets and arrows are functions. 

The set of arrows $\Hom(A, B)$ is itself a set, denoted $B^A$.

The empty set $\vide$ is initial and the point $\bullet = \{ \vide \}$ 
is final in $\Set$.

\item The category $\Alg_\K$ of unital algebras over a field $\K$ whose
    arrows are algebra morphisms. 

    The field $\K$ is both initial and final in $\Alg_\K$, it is a null object.

\item The category $\Top$ whose objects are topological spaces
    and arrows are continuous functions. 
    
   The point $\bullet$ is final in $\Top$.

\item The category ${\bf Types}$ 
whose objects are variable types
and arrows are programs\footnotemark{}. 

The set of arrows $\Hom({\tt a}, {\tt b})$ represents the
programs with input of type ${\tt a}$ and output of type ${\tt b}$.   
It is itself a type, denoted by $({\tt a \aw b})$.

The empty or {\it bottom} type $\perp$ is initial,
while the unit or {\it top} type $\top$ is final. 
An arrow of type ${\tt a} \aw \perp$ represents a program 
which does not terminate. 

\footnotetext{
This example is motivated by functional programming, 
although types also aim to provide
with a constructivist and rigorous ground for mathematical logic.
See for instance the Curry-Howard "proofs as programs" correspondence 
and Martin-Löf's theory of types.
}

\item For every object $X$ of a category $\Cat$, the category 
    $\Cat_X$ above $X$ 
    has arrows $f: A \aw X$ as objects. 
    A morphism $\ph : f \aw g$ in $\Cat_X$ between $f : A \aw X$ 
        and $g : B \aw X$ is an arrow $\ph : A \aw B$ in $\Cat$ 
    such that the following diagram commutes:
    \begin{equation} 
        \begin{tikzcd}
            A \rar{\ph} \drar[swap]{f} & B \dar{g} \\
            & X 
        \end{tikzcd}
    \end{equation}
    There is a similar category $\Cat^X$ below $X$ 
    defined by reversing arrows, which amounts to 
    reading the above definition in $\Cat^{op}$.  

\end{enumerate}

\subsection{Functors}

Just as morphisms describe relations between objects in a category, 
functors describe relations between categories by bringing
every object to an object and every arrow to an arrow. 

\begin{defn}
A covariant functor $T$ from two categories $\Cat$ and $\Cat'$ is defined by:
\bi 
\iii An object $T(A)$ of $\Cat'$ for every object $A$ of $\Cat$
\iii An arrow $T(f) : T(A) \aw T(B)$ in $\Cat'$ 
for every arrow $f : A \aw B$ in $\Cat$.
\ei
satisfying the following axioms:
\bei
    \item $T(1_A) = 1_{T(A)}$,
    \item $T(fg) = T(f) \cdot T(g)$ 
\ee
\end{defn}

{\bf Examples of Functors:} 
\begin{enumerate}[itemsep=0.6em]
    \item When $(X, \geq)$ and $(Y, \geq)$ are partially ordered sets,
        a functor from $X$ to $Y$ is an order-preserving map from $X$ to $Y$.

        This defines the category $\Ord$ of ordered sets with 
        order-preserving map as morphisms.

    \item For every object $X$ of a category $\Cat$, 
        there are canonical functors $\Hom(-,X)$ and $\Hom(X,-)$
        from $\Cat$ to $\Set$. 

        The pull-back of $f : A \aw B$ is the map $f^* : \Hom(B,X) \aw \Hom(A,X)$
        defined by $f^*(u) = u \circ f$ for every $u : B \aw X$.
        $\Hom(-,X)$ is a contravariant functor. 

        The push-out of $g : C \aw D$ is the map $g_* : \Hom(X,C) \aw \Hom(X, D)$
        defined by $g_*(v) = g \circ v$ for every $u : X \aw C$.
        $\Hom(X,-)$ is a covariant functor.

    \item A contravariant functor from $\Top$ to $\Alg$ is defined by
        associating to each topological space $\Om$ the 
        algebra $C(\Om)$ of real continuous functions on $\Om$.

        For every arrow $\ph : \Om \aw \Om'$ in $\Top$, 
        its pullback $\ph^* : C(\Om') \aw C(\Om)$ 
        defined by $\ph^*(f) = f \circ \ph$ is an algebra morphism.

        This is a particular case of the previous example,
        as $C(\Om) = \Hom(\Om, \R)$ in $\Top$.

    \item The endofunctor ${\tt list}$ of ${\bf Types}$ 
        associates to each type ${\tt a}$ the type ${\tt[a]}$ 
        of lists whose elements are of type ${\tt a}$. 

        Any function ${\tt f: a \aw b}$ induces a function 
        ${\tt map(f) : [a] \aw [b]}$
        returning the list of images under ${\tt f}$ 
        of the input list's elements. 
        \begin{equation} \begin{split}
        \tt map(f) &: \tt[x, \; \dots xs] \mapsto [f(x), \;\dots map(f)(xs)]  \\
            &: \tt [\;] \mapsto [\; ] 
        \end{split} \end{equation}

    \item When $\Cat$ is a concrete category, 
        there is a canonical forgetful functor from $\Cat$ to $\Set$.

\end{enumerate}

\subsection{Natural Transformations}

Relations between functors are described by 
natural transformations, also called functor morphisms,
as they allow to view functors as objects of a category.

\begin{defn}
Let $T, T'$ be two functors from $\Cat$ to $\Cat'$. 
    A natural transformation $\Phi$ from $T$ to $T'$ is a 
    collection of morphisms $\Phi(A) : T(A) \aw T(A')$ in 
    $\Cat'$ for all $A$ in $\Cat$, such that the diagram:
\begin{equation} \begin{tikzcd}
    T(A) \rar{\Phi(A)} \dar{T(f)} 
    & T'(A) \dar{T'(f)} \\
    T(B) \rar{\Phi(B)} & T'(B) 
\end{tikzcd}
    \end{equation}
is commutative for all $f : A \aw B$ in $\Cat$.
\end{defn}

Given two categories $\Cat$ and $\Cat'$, 
the functor category $[\Cat, \Cat']$ has functors
$T: \Cat \aw \Cat'$ as objects and natural transformations 
$\Phi : T \aw T'$ as morphisms, where:  
\bi
\iii The identity $1_T: T \aw T$ is defined by 
$1_T(A) = 1_{T(A)}$ for all object $A$ of $\Cat$,
\iii The composition of $\Phi : T \aw T'$ and $\Psi : T' \aw T''$ is 
defined by $(\Psi \circ \Phi)(A) = \Psi(A) \circ \Phi(A)$.
\ei

{\bf Examples:}
[[ $\Hom(-,X)$ and Yoneda ]]
[[ adjunction example ]]

\section{Limits and Colimits}

Universal properties allow for an abstract definition of limits,
unifying some simple constructions such as sums and products of sets
with more elaborate ones, such as inductive and projective limits. 
Lecture notes from H. Cartan \cite{Cartan-64} were a great resource
on the subject. The book by Dwyer and Spalinski 
\cite{Dwyer-Spalinski} should however prove more useful 
for the reader interested in modern categorical constructions. 

\subsection{Definition}

A {\it diagram} of shape ${\bf D}$ in a category $\Cat$ consists
of a functor $C: {\bf D} \aw {\Cat}$ where ${\bf D}$ 
is a small\footnote{
    A category is said {\it small} when the class of its objects 
    actually forms a set.
}
category describing the diagram shape.
It is a collection $C(f) : C_\aa \aw C_\bb$ of arrows in $\Cat$ for 
all $f : \aa \aw \bb$ in ${\bf D}$.

A {\it cone} over $C$ in $\Cat$ is an object $S$ of $\Cat$
and a collection of morphisms $\ph_\aa : S \aw C_\aa$ for all
$\aa \in {\bf D}$ such that the following diagram in $\Cat$ commutes
for every $f : \aa \aw \bb$ in ${\bf D}$:
\begin{equation} 
\begin{tikzcd}
    & S \dlar[swap]{\ph_\aa} \drar{\ph_\bb} & \\
    C_\aa \ar{rr}{C(f)} & & C_\bb  
\end{tikzcd}
\end{equation}
In other words, $(S, \ph)$ extends the functor $C$ to 
the category ${\bf D}_0$ preceding ${\bf D}$ with an initial element.
A morphism between two cones $(S, \ph)$ and $(S', \ph')$ over $C$ is 
a morphism $\psi : S \aw S'$ in $\Cat$ such that 
the following diagram commutes for all $\aa \in {\bf D}$:
\begin{equation} 
\begin{tikzcd}
    S \rar{\psi} \drar[swap]{\ph_\aa} & S' \dar{\ph'_\aa} \\
    & C_\aa
\end{tikzcd}
\end{equation}
A {\it limit} of a diagram $C : {\bf D} \aw \Cat$ 
is a final element $(L, \lambda)$ in the category of cones over $C$. 
When a limit $L$ exists, it is defined up to isomorphisms in $\Cat$ 
by the universal property requiring that
for every cone $(S, \ph)$ over $C$ there be a unique morphism 
$\psi : S \aw L$ factorising $S$ through $L$. 
\begin{equation} 
\begin{tikzcd}
    & S \dar[dashed]{\psi} 
    \ar[ddl, swap, bend right=20, "\ph_\aa"] \ar[ddr, bend left=20, "\ph_\bb"] & \\
    & L \dlar[swap]{\lambda_\aa} \drar{\lambda_\bb} & \\
    C_\aa \ar{rr}{C(f)} & & C_\bb  
\end{tikzcd}
\end{equation}

\begin{defn} 
    A category $\Cat$ is called {\rm complete}
    when every small diagram $C : {\bf D} \aw \Cat$ has a limit.
When it exists, we denote by $\lim_{\bf D} C$ the limit of $C$
defined up to isomorphism.
\end{defn}

A {\it colimit} of a diagram 
$C : {\bf D} \aw \Cat$, 
is reciprocally defined by reversing arrows. 
It is an initial element in 
the category of cones under $C$, made of 
extensions of $C$ to the category ${\bf D}_1$ 
appending a final element to ${\bf D}$. 
When it exists, the universal property satisfied by a colimit $L'$ of $C$
is represented by the diagram: 
\begin{equation} 
\begin{tikzcd}
    C_\aa \ar[rr, "C(f)"] 
    \ar[dr, swap, "\lambda'_\aa"] \ar[ddr, swap, bend right=20, "\ph'_\aa"] 
&   & 
    C_\bb \ar[dl, "\lambda'_\bb"] \ar[ddl, bend left=20, "\ph'_\bb"]
\\
    & L' \ar[d, dashed, "\psi'"]  & \\
    &   S' &
\end{tikzcd}
\end{equation}

\begin{defn}
A category $\Cat$ is called {\rm cocomplete} when 
    every small diagram $C : {\bf D} \aw \Cat$ has a colimit.
    When it exists, we denote by ${\rm colim}_{\bf D}\, C$ the colimit
    of $C$ defined up to isomorphism.
\end{defn}

We give a few examples of limits below, although the following paragraphs
will illustrate much better the universality of limits. 

{\bf Examples:}
\begin{enumerate}[itemsep=0.6em]
    \item When ${\bf D}$ is the empty category, 
        limits and colimits of the empty diagram in $\Cat$
        are initial and final objects 
        of $\Cat$ respectively. 

    \item Any object $A$ of a category $\Cat$ defines a diagram, whose shape 
        ${\bf D}$ is the category with only one object and its identity map.
        The limit and colimit are both represented by $1_A : A \aw A$.

    \item Let $u \in \R^\N$ denote a sequence of real numbers. 
        The induced set map defines 
        a functor between the partial orders 
        $(\Part(\N), \incl)$ and $(\Part(\R), \incl)$ 
        associating to a subset $S \incl \N$ its direct image $u(S) \incl \R$.

        Consider now its restriction $\tilde u$ 
        to subsets of the form $S_n = \{ n, n+1, \dots \}$ for $n \in \N$. 
        The limit of $\tilde u$ is the largest subset $L \incl \R$ such that 
        $L \incl u(S_n)$ for all $n$. It consists of all
        the accumulation points of $u$.

    \item in some abelian category e.g. $\Vect$ [[Kernel and Image]].
\end{enumerate}

\subsection{Sums and Products}

Any two objects $A, A'$ in a category $\Cat$ define a 
diagram of shape a category ${\bf D}$
with two objects and identities as morphisms.
When it exists,
a final cone over $A$ and $A'$ defines 
their 
{\it product} $A \times A'$, satisfying the universal property depicted by:  
\begin{equation} 
\begin{tikzcd} 
    & & A \\
X \ar[urr, bend left=20] \ar[drr, bend right=20] \rar[dashed] & 
A \times A' \urar{\pi} \drar[swap]{\pi'} & \\
& & A' 
\end{tikzcd}
\end{equation}
Their {\it sum} or {\it coproduct} $ A \sqcup A'$ 
is reciprocally defined as an initial cone under $A$ and $A'$.
\begin{equation} 
\begin{tikzcd}
A \drar{j} \ar[drr, bend left=20] & & \\
& A \sqcup A' \rar[dashed] & Y  \\
A' \urar[swap]{j'} \ar[urr, bend right=20] & & 
\end{tikzcd}
\end{equation}

{\bf Examples:}
\begin{enumerate}[itemsep=6px]
\item 
In $\Set$ and $\Top$, the product of $A$ and $A'$ is their cartesian product 
$A \times A'$, while their sum is the disjoint union $A \sqcup A'$. 

\item  In $\Grp$ the product and coproduct of $G$ and $G'$ coincide as $G \times G'$.
In $\Vect$, the product and the sum of $V$ and $V'$ also coincide as $V \oplus V'$. 
This is a general property of abelian categories. 

\item 
In the category ${\bf Com}$ of unital commutative 
algebras, the coproduct of $A$ and $A'$
is their tensor product $A \otimes A'$, with canonical injections 
$1 \otimes -$ and $- \otimes 1$.
\end{enumerate}

\subsection{Pushouts and Pullbacks}

Consider the diagram shape given by ${\bf D} : \aa \aw \bb \wa \aa'$.
The limit of this kind of diagrams, when it exists, 
defines the pullback or fibered product $A \times_B A'$
of $A$ and $A'$ over $B$: 
\begin{equation}
\begin{tikzcd} 
    & & A \drar{v} & \\
X \ar[urr, bend left=20] \ar[drr, bend right=20] \rar[dashed] & 
A \times_B A' \urar{\pi} \drar[swap]{\pi'} & & B\\
    & & A' \urar[swap]{v'} & 
\end{tikzcd}
\end{equation}
The pushout or amalgated sum $A \sqcup_B A'$ of $A$ and $A'$ over $B$, 
when it exists, 
is defined by the dual universal property:
\begin{equation} 
\begin{tikzcd}
    & A \drar{j} \ar[drr, bend left=20] & & \\
    B \urar{u} \drar[swap]{u'} & & A \sqcup_B A' \rar[dashed] & Y  \\
    & A' \urar[swap]{j'} \ar[urr, bend right=20] & & 
\end{tikzcd}
\end{equation}
Note that the morphisms are implicit in the notations 
$A \times_B A'$ and $A \sqcup_B A'$ although the resulting objects depends on them.

{\bf Examples:}
\begin{enumerate}[itemsep=6px]
\item In $\Set$ as in $\Top$, the fibered product of $v: A \aw B$ and $v': A' \aw B$ 
    is defined by: 
    \begin{equation} A \times_B A' = \{ (x,x') \in A \times A' \st v(x) = v'(x') \} \end{equation}
    while the pushout of $u: B \aw A$ and $u': B \aw A'$ is the quotient:
    \begin{equation} A \sqcup_B A' = A \sqcup A' \; / \; \big( u(y) \sim u'(y) \big)_{y \in B} \end{equation}

\item In the category of commutative algebras ${\bf Com}$, 
    the pushout $A \otimes_B A'$ of $u : B \aw A$ and $u': B \aw A'$ 
    is the quotient of $A \otimes A'$ by the equivalence relation 
        generated by the action of all $b \in B$: 
        \begin{equation} 
        \big( a \cdot u(b)\otimes a' \big) \sim 
        \big( a \otimes u'(b) \cdot a' \big) \end{equation}
\end{enumerate}

\pagebreak

\subsection{Sheaves and Cosheaves}

\begin{defn}
A {\rm presheaf} over a topological space $(\Om, {\cal T}_\Om)$ is a functor
$F : {\cal T}_\Om \aw \Set$, associating:
\bi
\iii to each open subset $U \incl \Om$ a set $F(U)$ of sections over $U$,
\iii to each ordered pair $V \incl U$ a restriction map $\rho_{VU} : F(U) \aw F(V)$. 
\ei
A presheaf is thus contravariant\footnote{
    As a subcategory of $\Set$, this is the right choice of arrows 
    on ${\cal T}_\Om$.
}
from $({\cal T}_\Om, \incl)$ to $\Set$,
and covariant 
from ${\cal T}_\Om^{op} = ({\cal T}_\Om, \cont)$ to $\Set$. 
\end{defn}

The fiber of $F$ at $x \in \Om$ is defined as the colimit of $F$ over the 
partial order $({\cal V}_x, \cont)$ of neighborhoods containing $x$ 
and denoted by $F_x$: 
\begin{equation} F_x = \colim_{{\cal V}_x} F \end{equation}
For all $U$ containing $x$, the image $s(x)$ of a section
 $s \in F(U)$ under the canonical map $F(U) \aw F_x$ 
is called the germ of $s$ at $x$.

When ${\cal U}_{\Om'}$ is an open covering of $\Om' \incl \Om$
closed under intersection, the limit of $F$ over $({\cal U}_{\Om'}, \cont)$
is the set of compatible sections on ${\cal U}_{\Om'}$: 
\begin{equation} \lim_{{\cal U}_{\Om'}} F \simeq 
\Big\{\; (s_U) \in \prod_{U \in {\cal U}_{\Om'}} F(U) 
\; \Big| \; 
\forall U, V \in {\cal U}_{\Om'} \;\; (s_U)_{|U\cap V} = (s_V)_{|U \cap V}  
\; \Big\}
\end{equation} 
A presheaf $F$ over $\Om$ is a sheaf when for all such 
covering ${\cal U}_{\Om'}$ of $\Om'$,
the sections of $F$ over $\Om'$ are in one-to-one correspondence
with the compatible sections on ${\cal U}_{\Om'}$. 

\begin{defn}
A {\rm sheaf} $F$ over $\Om$ is a presheaf such that: 
    \begin{equation} F(\Om') \simeq \lim_{{\cal U}_{\Om'}} F \end{equation} 
for every open covering ${\cal U}_{\Om'}$ of $\Om' \incl \Om$ 
closed under intersection.
\end{defn}

Morphisms of sheaves are defined as natural transformations of functors,
and the category of sheaves over $\Om$ is naturally defined 
by inclusion in the functor category $[{\cal T}^{op}_{\Om}, \Set]$. 
When $F$ is a sheaf, it is customary to denote it by $F(\Om)$ 
with a slight abuse\footnote{
$F(U)$ may not be the image of $F(\Om)$ under $\rho_{U\Om}$.
}
of notations.
Note that the sheaf axiom implies that the diagram: 
\begin{equation} 
\begin{tikzcd} 
    F(U \cup V) \rar \dar & F(V) \dar \\
    F(U) \rar & F(U \cap V) 
\end{tikzcd}
\end{equation}
is a pullback square in $\Set$, for all open $U, V \incl \Om$.
When limits exist, sheaves may be defined in any category, 
and one is often mostly interested with sheaves of abelian groups, 
rings, modules, {\it etc.} fibered products and limits coinciding with
those coming from sets.

There is a dual notion of cosheaf, although seemingly less common. 

\begin{defn} 
A {\rm pre-cosheaf} over a topological space $(\Om, {\cal T}_\Om)$ in 
a category with colimits $\Cat$ is a covariant 
    functor $G: {\cal T}_\Om \aw \Cat$ associating: 
\bi
    \iii to each open subset $U \incl \Om$ an object $G(U)$ in $\Cat$
    \iii to each ordered pair $V \incl U$ a morphism $j_{UV} : G(V) \aw G(U)$ in $\Cat$
\ei 
\end{defn}

\begin{defn} 
    A {\rm cosheaf} over $\Om$ is a pre-cosheaf such that: 
    \begin{equation} G(\Om') \simeq \colim_{{\cal U}_{\Om'}} G \end{equation} 
    for every open covering ${\cal U}_{\Om'}$ of 
    $\Om' \incl \Om$ closed under intersection.
\end{defn}

The category of $\Cat$-valued cosheaves on $\Om$ is similarly defined 
as a subcategory of $[{\cal T}_{\Om}, \Cat]$. 
The cosheaf axiom implies that $G(\vide)$ is initial in $\Cat$, 
and that the diagram:
\begin{equation} 
\begin{tikzcd}
    G(U \cap V) \dar \rar & G(V) \dar \\
    G(U) \rar & G(U \cup V) 
\end{tikzcd}
\end{equation}
is a pushout square in $\Cat$ 
for all open $U, V \incl \Om$.

{\bf Examples:}
\begin{enumerate}[itemsep=6px]
\item The space $C(\Om)$ of real continuous functions over 
a topological space $\Om$ 
defines the fundamental example of a sheaf of algebras, with obvious 
restrictions. 

\item Suppose given a numerable set $\Om$ 
with finite sets $E_i$ for all $i \in \Om$,
and let $E_\aa = \prod_{i \in \aa} E_i$ for all $\aa \incl \Om$ 
Then $E_\Om$ defines a sheaf of sets over 
the discrete topological space $\Om$.

\item When $f \in C(\Om, \Om')$ is a continuous map of topological spaces, 
the map $U' \mapsto f^{-1}(U')$ is a cosheaf of sets over $\Om'$. 

\item Given $E_\Om$ as above, let $A_\aa = \R^{E_\aa}$ 
    denote the algebra of real functions on $E_\aa$, for all $\aa \incl \Om$. 
        Then for all $\bb, \bb' \incl \Om$ we have $A_{\bb \cup \bb'} = A_\bb \otimes_{A_{\bb \cap \bb'}} A_{\bb'}$ 
    and $A$ is a cosheaf of algebras over $\Om$.
\end{enumerate}

\section{Differential Structures}

Simplices generalise the usual figures of point, segment, triangle, 
tetrahedron, {\it etc.} They
correspond to elementary objects in topology and geometry, 
as any $n$-dimensional manifold may be triangulated by simplices 
of dimension $n$. 
They also carry the fundamental combinatorial properties 
of differential calculus, which will motivate the much more algebraic definition
of simplicial objects. 

\subsection{Simplicial Complexes}

\newcommand{\vect}{\overrightarrow}

Given an affine space $E$ and $n + 1$ affinely independent points 
$P_0,\dots, P_n$, 
the convex polyhedron generated by those points
is called the $n$-{\rm simplex} of vertices $P_i$:
\begin{equation} S = \Big\{ M \in E \;\Big|\; 
\vect{OM} = \sum_{i=0}^n \lambda_i \, \vect{OP}_i 
\quad{\rm for }\; \lambda_j \geq 0 \;{\rm and}\; 
\sum_i \lambda_i = 1 \Big\}
\end{equation}
Note that for any choice of origin $O \in E$, the barycentric 
coordinates $\lambda_i$ of $M$ are uniquely determined 
independently of $O$. 
Barycentric coordinates identify points of a simplex
with probability measures on its set of vertices. 

\begin{defn} 
    The {\rm topological $n$-simplex} over the set of $n+1$ vertices $\Om$
    is defined by: 
\begin{equation} 
| S_\Om | = \Big\{\; \lambda : \Om \aw \R^+ \;\Big| \;
\sum_{i\in \Om} \lambda_i = 1 \;\Big\} 
\end{equation}
    It is a convex subset of $\R^{\Om}$, and a topological space for 
    the topology induced by $\R^{\Om}$.
\end{defn}

Let $S$ denote the simplex of vertices $\Om$.
A $q$-face $S'$ of $S$ is defined by a set of $q+1$ vertices 
$\Om' \incl \Om$, such that $S' \incl S$ 
consists of the barycentric coordinates 
$\lambda$ that vanish on $\Om - \Om'$.
The interior of $S$ 
for the topology of $\R^\Om$ 
consists of all the non-vanishing barycentric coordinates
$\lambda > 0$, and coincides with the complement of all the proper faces of $S$ 
within $S$.

There is an equivalence $\Om \mapsto |S_\Om|$ 
between the categories of finite sets and topological 
simplices ordered by inclusion,
where $\Part(\Om)$ is identified with the set of faces of $|S_\Om|$.
It is therefore natural to identify a simplex with its set of vertices.
Every $f : \Om \aw \Om'$ induces a continuous map $f_* : |S_\Om| \aw |S_\Om'|$
defined by: 
\begin{equation} (f_* \lambda)_{j} = \sum_{i \in \Omega \st f(i) = j} \lambda_i \end{equation}
and sending every face of 
$|S_\Om|$ to a face of $|S_{\Om'}|$.
Simplices hence define a covariant functor $\Set_{\rm f} \aw \Top$,
which could be extended to the larger category of measurable spaces. 

\begin{defn}
An {\rm abstract simplicial complex} $(\Om, K)$ is 
a finite set of vertices $\Om$ together with a collection of faces
$K \incl \Part(\Om)$ 
made of finite subsets of $\Om$, such that
for all $\aa \in K$, every $\bb \incl \aa$ is also in $K$.
\end{defn}

The $n$-skeleton $K_n$ of a simplicial complex $K$
consists of all its $n$-faces, {\it i.e.} 
faces having exactly $n + 1$ vertices.
The abstract simplex $S_{\Om}$ is the trivial simplicial complex $(\Om, \Part(\Om))$
having all possible faces.
A simplicial complex $(\Om, K)$ is essentially a reunion 
of abstract simplices $S_\aa$, for $\aa$ in $K$.

The topological space $|K|$ associated to a simplicial complex $(\Om, K)$ 
is obtained by gluing the simplices of $K$ along their 
intersecting faces:
\begin{equation} |K| = \colim_{\aa \in K} |S_{\aa}| \end{equation}
The inductive limit, taken over the functor $\aa \mapsto |S_{\aa}|$,
is essentially a reunion in the ambient topological simplex $|S_\Om|$.

\begin{defn}
A {\rm simplicial morphism} $f : (\Om, K) \aw (\Om', K')$ 
is a map of sets $f: \Om \aw \Om'$ 
such that for all face $\aa$ of $K$, its image $f(\aa) \incl \Om'$ is a face of $K'$. 
\end{defn}

Simplical complexes form a category $\KS$.
A simplicial map $f : K \aw K'$
induces for all $\aa \in K$ a continuous function 
$f^\aa_* : |S_{\aa}| \aw |S_{f(\aa)}|$.
These maps extend
to a map of topological spaces 
$f_* : |K| \aw |K'|$ and topological realisation
defines a covariant functor $\KS \aw \Top$.

The definition of a simplicial complex $K$ with vertices in $\Om$ 
could be naturally extended when $\Om$ is numerable 
and more generally, when $\Om$ is a measurable space.

\subsection{Simplicial Objects}

For every $n \in \N$, denote by $[n] = \{ 0, \dots, n\}$
the total order with $n+1$ elements, and by 
$\big( [n], S_n \big)$ the abstract $n$-simplex with ordered vertices.

\begin{defn} 
The {\rm simplicial category} ${\bf \Delta}$ is defined by:
\bi
\iii objects: $[n]$ for any $n$ in $\N$,
\iii morphisms: $[m] \aw [n]$ order-preserving map.
\ei
Equivalently, ${\bf \Delta}$ is the subcategory of 
$\Ord$ with objects $[n]$ for $n \in \N$.
\end{defn}

An ordered $n$-simplex $\sigma$ in 
a simplicial complex $(\Om, K)$ is a simplicial map 
$\sigma: \big([n], S_n \big) \aw (\Om, K)$.
The ordered $n$-simplex 
$\sigma = (\sigma_0, \dots, \sigma_n)$ is said non-degenerate when 
$\sigma_i \neq \sigma_j$ for $i \neq j$ and the 
underlying set map is injective. 

A simplicial complex $(\Om, K)$ then defines
a contravariant functor $\vec{K} : \Simp^{op} \aw \Set$ where: 
\bi
\iii $\vec{K}_n = \vec{K}\big([n] \big)$ is the set of ordered $n$-simplices
in $K$, 
\iii $t^* : \vec{K}_n \aw \vec{K}_m$ is defined for all
$t : [m] \aw [n]$ by the pullback $\sigma \mapsto \sigma \circ t$.
\ei
Denoting by $|\sigma| = \Img(\sigma)$ the image of an 
ordered $n$-simplex, we have 
$|t^* \sigma| \incl |\sigma|$ for all $t: [m] \aw [n]$.

A simplicial map $f: (\Om, K) \aw (\Om', K')$ induces 
a natural transformation $f_* : \vec{K} \aw \vec{K}'$,
defined by the pushforward $\sigma \mapsto f \circ \sigma$. 
The natural transformation $f_*$ is a morphism 
in the category of functors $[\Simp^{op}, \Set]$ 
and the assignment $(\Om, K) \mapsto \vec K$ defines 
a covariant functor $\KS \aw [\Simp^{op}, \Set]$. 
The set of ordered simplices $\vec{K}$ of a simplicial complex 
$(\Om, K)$ is the fundamental example of a simplicial set, and motivates 
the following more general definition of simplicial objects in an arbitrary category.

\begin{defn}
A {\rm simplicial object} in a category $\Cat$ is a 
    functor $X : \Simp^{op} \aw \Cat$.
\end{defn}

Simplicial objects in $\Cat$ form a category $[\Simp^{op}, \Cat]$ 
with natural transformations as morphisms.
Given a simplicial object $X$, we denote $X\big([n] \big)$ by $X_n$ 
and for $t : [m] \aw [n]$ 
we denote $X(t)$ by $t^* : X_n \aw X_m$. 

A category $\Cat$ defines a simplicial set $N(\Cat)$ called its nerve, defined by:
\begin{equation} N_n(\Cat) = \Hom( [n], \Cat  ) \end{equation}
An ordered $n$-simplex $\sigma \in N_n(\Cat)$ 
is a covariant functor $\sigma : \big( [n], \leq \big) \aw \Cat$. 
Equivalently, it is a commutative diagram
of $n+1$ objects $\sigma_0, \dots, \sigma_n$ 
with arrows $\sigma_{ij} : \sigma_i \aw \sigma_j$ for all $i < j$.

Given a simplicial set $X: \Simp^{op} \aw \Set$, 
the group of chains $\Z[X] : \Simp^{op} \aw \Ab$ is the 
simplicial abelian group freely
generated by $X$, with:
\begin{equation} \Z_{n}[X] = \bigoplus_{\sigma \in X_n} \Z \cdot e_{\sigma} \end{equation}
and every map $t: [m] \aw [n]$ inducing 
a group morphism $t^* : \Z_n[X] \aw \Z_m[X]$ 
defined by $t^*(e_\sigma) = e_{t^* \sigma}$.
Chain groups thus define a functor $\Z[\,\cdot\,]$ from
simplicial sets to simplicial abelian groups.

For all $n \in \N$ and $0 \leq i \leq n$,
consider the $i$-th face map $\bord^i : [n-1] \aw [n]$ 
defined as the injection of $[n-1]$ whose image misses the $i$-th vertex in $[n]$.
Note that face maps generate all injective maps of $\Simp$ and
satisfy the following fundamental
commutation relations: 
\begin{equation} \bord^i \circ \bord^j = \bord^{j-1} \circ \bord^i \txt{for} i < j \end{equation} 
where the sums of indices $i + j$ and $(j-1) + i$ have reversed parity.

Every simplicial abelian group $G : \Simp^{op} \aw \Ab$ has a 
canonical boundary operator $\bord : G \aw G$
where for all degree $n$, the map $\bord : G_n \aw G_{n-1}$ is defined by:
\begin{equation} \bord = \sum_{i=0}^n (-1)^i \; (\bord^i)^* \end{equation}
The commutation relations of face maps imply that 
$\bord^2 = \bord \circ \bord = 0$. 

Let $G = \Z \big[ \vec{K}\big]$ be the group of chains 
in a simplicial complex $K$. 
If $\sigma$ is an oriented simplex in $K$, 
then $\bord \sigma$ is the oriented boundary of $\sigma$. 
The fundamental equation $\bord^2 = 0$ reflects the geometric fact 
that the boundary of a boundary is empty.

\begin{defn} A {\rm cosimplicial object} in a category $\Cat$ 
    is a functor $Y : \Simp \aw \Cat$.
\end{defn}

Every cosimplicial abelian group $F : \Simp \aw \Ab$ has a 
canonical coboundary operator $d : F \aw F$, defined by the family 
of maps $d^n : F^n \aw F^{n+1}$ with:
\begin{equation} d^n = \sum_{i = 0}^{n + 1} (-1)^i \; (\bord^i)_* \end{equation}
The operator $d$ is called the differential of $F$
and satisfies $d^2 = d \circ d = 0$.

Given a covering $\cal{U}$ of a topological space $\Om$,
its {\v C}ech nerve is the simplicial set: 
\begin{equation} {\check {\cal U}}_n =
\{ \sigma : [n] \aw {\cal U} \st U_{\sigma} = 
\sigma_0 \cap \dots \cap \sigma_n \neq \vide \}
\end{equation}
For every $\sigma \in {\check {\cal U}}_n$ and $t : [m] \aw [n]$, the associated
simplex $t^* \sigma$ in ${\check {\cal U}}_m$
satisfies $\Img(t^* \sigma) \incl \Img(\sigma)$ in $\cal U$ 
so that the intersection $U_{t^* \sigma}$ contains $U_{\sigma}$ in $\Om$.
When $F$ is a sheaf of abelian groups over $\Om$, 
it defines a cosimplicial abelian group $\check F(\cal{U}) : \Simp \aw \Ab$ 
where:  
\begin{equation} \check F^n({\cal U}) = \bigoplus_{\sigma \in {\check {\cal U}}_n }
F(U_\sigma) \end{equation}
For $t : [m] \aw [n]$ the map $t_* : \check F^m({\cal U}) \aw \check F^n({\cal U})$ 
is defined for every $f \in \check F^m({\cal U})$ by:
\begin{equation} (t_*f)_{\sigma} = (f_{t^* \sigma})_{|U_{\sigma}} \end{equation}
and $\check{F}({\cal U})$ 
is called the group of {\v C}ech cochains of $F$ in ${\cal U}$.

\subsection{Homology}

\begin{defn}
A {\rm differential group} $(G, \bord)$ 
is an abelian group $G$ together with an endomorphism 
$\bord : G \aw G$ satisfying $\bord^2 = 0$. 
The morphism $\bord$ is called the {\rm boundary operator}
of $G$. 
\end{defn}

Given a simplicial complex $(\Om, K)$, a fundamental example is given 
by the group of chains $\big( \Z\big[ \vec{K} \big], \bord \big)$.
When $(G, \bord)$ is any differential group,
its boundary operator defines the two following subgroups:
\bi 
\iii a {\it cycle} is an element of $Z(G) = \Ker(\bord)$,
\iii a {\it boundary} is an element of $B(G) = \Img(\bord)$,
\ei
The rule $\bord^2 = 0$ implies that every boundary is a cycle 
and $B(G) \incl Z(G)$. 

\begin{defn} 
The {\rm homology group} of $(G, \bord)$ is 
    the quotient group $H(G) = \Ker(\bord) / \Img(\bord)$.
\end{defn}

More generally, $x, x' \in G$ are 
said homologous when there exists $y \in G$ such that $x' = x + \bord y$,
we then write $x \sim x'$ and 
denote by $[x]$ the class of $x \in G$ for this equivalence relation.
Homology groups consist of the equivalence classes of cycles.

When $(\Om, K)$ is a simplicial complex, the homology of its 
group of chains is denoted by $H(K;\Z)$.

\begin{defn} 
A {\rm morphism of differential groups} $f : (G, \bord) \aw (G', \bord')$ 
is a map of abelian groups
$f : G \aw G'$ such that the following diagram is commutative:
\begin{equation} 
\begin{tikzcd}  
G \dar{\bord} \rar{f} & G' \dar{\bord'} \\
G \rar{f} & G'
\end{tikzcd}
\end{equation}
In particular, $f$ sends $Z(G)$ in $Z(G')$ and $B(G)$ in $B(G')$.
\end{defn}

Differential groups form a category which we denote by $\Ab_\bord$. 
Given a map of simplicial sets $f : X \aw X'$,
the group morphism $f_* : \Z[X] \aw \Z[X']$ commutes with boundary operators 
and chain groups define a covariant functor from 
simplicial sets to $\Ab_\bord$.

The following proposition expresses that 
homology defines a functor $H : \Ab_{\bord} \aw \Ab$,
and we get in particular by left composition with the group of chains, 
a functor $\KS \aw \Ab$.
Functoriality is a fundamental property of homology, as 
it was introduced to yield algebraic invariants of
topological spaces, the homology groups of two 
homeomorphic spaces being isomorphic. 

\begin{prop} 
    A map of differential groups $f : (G, \bord) \aw (G', \bord')$ 
    induces a morphism in homology 
    denoted by $[f] : H(G) \aw H(G')$. 
\end{prop}

\begin{proof}
    A morphism $f$ sends $Z(G)$ to $Z(G')$ and $B(G)$ to $B(G')$, 
    hence $f : Z(G) \aw Z(G')/B(G')$ factors through $Z(G) \aw Z(G)/B(G)$.
\end{proof}

\begin{defn} 
    Two maps $f,f' : (G, \bord) \aw (G', \bord')$ are {\rm homotopic}
    when there exists a map $h : G \aw G'$ of abelian groups such 
    that $f - f' = \bord' h + h \bord$.
\end{defn}

We write $f \sim f'$ when $f$ and $f'$ are homotopic. 
When $h$ is a homotopy from $f$ to $f'$, 
the sum of the two outer paths coincides 
with the inner arrow of the following diagram:
\begin{equation} 
\begin{tikzcd}[column sep=large]
    \, & G' \dar{\bord'} \\
    G \rar["{f-f'}" description] \urar{h} \dar{\bord} & G' \\
    G \urar[swap]{h} & \, 
\end{tikzcd}
\end{equation}
The homotopy relationship vanishes in homology. 

\begin{prop} 
When $f, f' : (G, \bord) \aw (G', \bord')$ are homotopic, 
their induced maps in homology $[f], [f'] : H(G) \aw H(G')$ coincide. 
\end{prop}

\begin{proof} 
    Let $h : G \aw G'$ denote a homotopy between $f$ and $f'$, 
    so that $f - f'  = \bord' h + h \bord$. \\
    For all $z \in Z(G)$ we have $f(z) - f'(z) = \bord' h(z) \in B(G)$,
    so that $[f(z)] = [f'(z)]$ in $H(G')$.
\end{proof}

When $(G', \bord)$ is a subgroup of $(G, \bord)$ 
with $\bord (G') \incl G'$, the boundary operator of $G$
factors to the quotient onto $G/G'$ 
where it induces a boundary $\bord'$.
The relative homology 
of the pair $(G,G')$ is defined as 
$H(G, G') = H(G / G')$
More precisely, let us denote by:
\bi
\iii $Z(G,G') = \bord^{-1}(G')$ the set of relative cycles,
\iii $B(G,G') = G' + \bord G$ the set of relative boundaries.
\ei 
Then $H(G,G')$ is the quotient of $Z(G,G')$ by $B(G,G')$. 
Noting that $\bord$ sends $Z(G,G')$ to $Z(G')$ and
$B(G,G')$ to $B(G')$, there is a 
canonical morphism in homology $[\bord] : H(G,G') \aw H(G')$.

The injection $f : G' \aw G$ and the projection $g : G \aw G/G'$ 
induce maps in homology:
\begin{equation} \begin{tikzcd}
    H(G') \rar{[f]} & H(G) \rar{[g]} & H(G,G') 
\end{tikzcd} \end{equation} 

\begin{prop}
   Given a subgroup $(G', \bord)$ of $(G, \bord)$, 
    the homology sequence of the pair $(G,G')$ is exact:
    \begin{equation} \begin{tikzcd} 
         & H(G) \drar{[g]} & \\
        H(G') \urar{[f]} & & H(G,G') \ar[ll, swap, "{[\bord]}"] 
    \end{tikzcd}
    \end{equation}
\end{prop}

\begin{proof} Prove that:
    \bi 
    \iii $\Img([\bord]) = \Ker([f])$,
    \iii $\Img([f]) = \Ker([g])$,
    \iii $\Img([g]) = \Ker([\bord])$.
    \ei
\end{proof}

When $(G, \bord)$ and $(G', \bord')$ are graded differential groups, 
the endomorphisms from $G$ to $G'$ form a graded
differential complex $\big( L(G,G'), \div \big)$ with:
\begin{equation} L_n(G,G') = \prod_k L(G_k, G_{k+n}) \end{equation}
where the boundary of $f \in L_n(G,G')$ is defined by:
\begin{equation} \div(f) = \bord' f - (-1)^n f \bord \end{equation}
In particuar, a $0$-cycle is a morphism of differential groups:
\begin{equation} \div(f) = 0 \quad \eqvl \quad \bord' f = f \bord \end{equation}
While a $0$-boundary is homotopic to zero: 
\begin{equation} \div(h) = \bord' h + h \bord \end{equation}

\chapter{Statistical Systems}

    This chapter introduces the theoretical framework for a local approach to statistics, 
where the term {\it statistical system} 
may be understood in a double sense. 
In physics, this commonly refers to a pair $(E_\Om, H_\Om)$ 
where $E_\Om = \prod_{i \in \Om} E_i$
is the configuration space of the system 
and $H_\Om: E_\Om \aw \R$ is the hamiltonian or energy function, inducing
the statistics as a function of temperature\footnote{
    The probability of observing a configuration $x_\Om \in E_\Om$ 
    is proportional to the Gibbs density $\exp(-H_\Om(x_\Om) / k_B T)$, 
    where $k_B$ denotes the Boltzmann constant. 
    and $T$ is the equilibrium temperature.
    Computing the normalisation factor $Z_\Om$ requires to integrate 
    the Gibbs density over the whole configuration space $E_\Om$, 
    which is typically intractable. 
}. 
Our localisation procedure on a hypergraph\footnote{
    A {\it hypergraph} $X \incl \Part(\Om)$ is just a collection of subsets of $\Om$.
    A graph is a particular case of hypergraph 
    whose elements consist of points (vertices) and pairs (edges). 
    The choice of $X$ reflects a splitting of $\Om$ into 
    intersecting chunks: 
    the coarser the covering, the more precise the local approximations. 
} $X \incl \Part(\Om)$
also leads to an 
{\it inductive system} $(A_\aa)_{\aa \in X}$ of local algebras, 
where $X$ should only be chosen coarse enough 
for the hamiltonian to decompose
as a sum of local interactions potentials $(h_\aa)_{\aa \in X}$.

With an emphasis on functoriality, 
section 1 presents natural constructions 
for global statistics. 
Our approach will focus on the algebra
of observables $A_\Om$,
formed by functions on the configuration space $E_\Om$.  
Probability measures shall then be recovered by duality, 
as positive and normalised linear forms on $A_\Om$. 
This line of thought, common in the field of operator algebras, 
somehow differs from the usual probabilistic definitions. 
It has the considerable advantage of unifying classical statistics 
with quantum states. 

The complex $A_\bullet(X)$ of local observable fields, defined in section 2, 
will serve to parameterise\footnote{
    It is actually the differential sequence $A_0(X) \wa A_1(X) \wa \dots $ which 
    will yield a projective resolution of the target subspace of $A_\Om$, 
    isomorphic to the first homology group $A_0(X) / \delta A_1(X)$ 
    of the whole complex $A_\bullet(X)$. 
    The global hamiltonian $H_\Om \in  A_\Om$ 
    hence actually defines a homology class 
    $[h]  = h + \delta A_1(X) \incl A_0(X)$ of
    interaction potentials (see theorem \ref{H0}).
} the low-dimensional subspace of $A_\Om$ in which the global 
hamiltonian lies. 
Its construction essentially lifts the functor 
$A:X^{op} \aw \Alg$ of local observables 
to a functor on the nerve $\check A : N(X)^{op} \aw \Alg$,
associating to any ordered chain $\aa \cont \dots \cont \cc$ a copy of $A_\cc$.
The simplicial structure of the nerve will provide $A_\bullet(X)$ with a 
boundary operator $\div$ and whose action 
$A_1(X) \aw A_0(X)$, discrete analog of a divergence, 
will describe the dynamic of message-passing algorithms as diffusion equations.
Section 2 describes such constructions in their generality\footnote{    
    Thanks to investigations by D. Bennequin, 
    we became aware that the construction of 
    a complex by the same lifting to the nerve 
    (in an abstract setting) was already considered by Grothendieck and Verdier 
    in SGA-4-V \cite{SGA-4-V, Moerdijk}. 
}. 

Specialising to the 
setting where each $A_\aa = \R^{E_\aa}$ is the algebra of functions
on the cartesian product $E_\aa = \prod_{i \in \aa} E_i$, 
we compute the homology of $A_\bullet(X)$ in section 3. 
Its acyclicity shall come as a consequence of the
{\it interaction decomposition} theorem \ref{int-decomposition}. 
This fundamental result for statistics is recalled and proved
through harmonic analysis, following an original proof by Mat{\v u}s, 
before relating homology classes of potentials 
$[h] \incl A_0(X)$ with global hamiltonians $H_\Om \in A_\Om$ 
in theorem \ref{H0}.

\newpage

\section{Global Statistics}

The main purpose of this short and informal section is to 
introduce the fundamental structures 
involved with statistics, along with their notations:
\begin{center}
{\renewcommand{\arraystretch}{1.6} 
    \begin{tabular}{| C{2.2cm} | C{2.2cm} | C{2.2cm} | C{2.2cm} |}

        \hline
        Microstates & Observables & Measures  & States \\
        \hline
        $E$         & $A = C(E)$  & $A^*$     & $\Delta \incl A^*$ \\
        \hline
        $\Top$      & $\Alg$      & $\Vect$   & ${\bf Conv}$  \\
        \hline
    \end{tabular}
}
\end{center}
In this picture, most columns are related by functors. 
In particular, the space $\Delta$ of statistical states
can be functorially defined from a set of microstates $E$ 
or from a C*-algebra of observables $A$. 

A fundamental component of statistical physics is the Gibbs state 
map $\rho : A \aw \Delta$ defined by:
\begin{equation} \rho(H) = \big[ \e^{-H} \big] \end{equation}
where $H$ computes the energy of the system and the bracket denotes normalisation. 
Both theoretical and computational problems with the normalisation factor 
$ Z(H) = \int_E \e^{-H}$ arise 
when $E$ gets large. It involves a computation 
of exponential complexity in the dimension of $E$, while 
the study of phase transitions requires to let the number of atoms go to infinity.

One may thus be lead to give up global observations, 
and decide that only small enough 
regions of the global system may be simultaneously observed. 
This approach underlies the present work and will rely heavily on functoriality.
What follows could then be thought of as a description of 
local models for statistics, which one may join consistently 
to cover larger systems. 
This localisation procedure will still efficiently describe
collective phenomena, when performed on a covering which is 
coarse enough compared to the range of interactions. 

We also hope that the following general discussion 
may give perspective on  possible extensions of the present work 
to the continuous and quantum settings. See Meyer \cite{Meyer-86} 
for a good reference on quantum statistics, although our approach was 
much more inspired by Souriau \cite{Souriau-QG}.

\subsection{Microscopic States}

 In classical probability theory, one starts with a 
measurable set $E$ describing all possible outcomes of an experiment. 
Consider for instance a physical system of 
$N$ atoms, labelled by $i, j, \dots$, each of which 
having degrees of freedom in $E_i$. 
A configuration of the full system is 
given by an element of the cartesian product:
\begin{equation} E = \prod_{i=1}^N E_i \end{equation} 
A configuration is also called a microscopic state of the system. 

In what follows, we shall keep the notation $E$ for configuration spaces. 
In most applications covered by this thesis, it is enough 
to view $E$ as an object of the category $\Set_{\rm f}$ of finite sets.
However, some of our constructions may gain generality by considering 
topological
spaces in $\Top$. 

Starting with a set of microscopic states is a classical point of view, 
although somehow artificial and arbitrary. 
It is only valid at high enough temperatures 
as quantum mechanics leads to give up the 
fiction of microscopic states. 
Classical probabilities and quantum states will both be naturally 
described by the states of an algebra of observables.

\subsection{Observables}

In quantum mechanics, one starts with a C*-algebra\footnote{
A C*-algebra is an algebra over $\C$ with 
$(i)$ a continuous and complete norm $|\cdot|$ and
$(ii)$ an antilinear involution $*$ such that 
$|a^* a| = |a|^2$. 
}
 $A$ of observables, 
describing all possible linear combinations of 
measurements that may be performed on a system. 
Classical statistics also fit very nicely in this framework,
by restricting oneself to commutative algebras of observables. 

Given a topological space $E$ describing classical microscopic states,
we let: 
\begin{equation} A = C(E) \end{equation} 
denote the commutative algebra of continuous and bounded real functions over $E$, 
equipped with the infinite norm $|| u ||_\infty = \sup_{x \in E}|u(x)|$.  
A classical observable is just a function of the microscopic states. 
This assignment defines a contravariant functor $C : \Top^{op} \aw \Alg$, 
as any continuous map $\ph : E \aw E'$ has a
pull-back $\ph^* : A' \aw A$ defined by:
\begin{equation} \big( \ph^*u \big)(x') = u\big(\ph(x)\big) \end{equation}
In most of this work, the algebra of observables $A=C(E)$ will be commutative
and given by such a procedure. We however emphasize
that once given the algebra, one may very well forget about the underlying set.

We will be mostly interested in the finite setting where: 
\begin{equation} A = \R^E \end{equation} 
is a finite dimensional vector space, isomorphic to the multiplicative 
Lie group $G = (\R_+^*)^E$ of strictly positive observables 
under the expontential map, i.e. 
could be viewed as the abelian Lie algebra of $G$. 
Restricting to finite configuration spaces will leave aside
most technical difficulties, greater generality is only 
mentioned here for the sake of perspective.

At the quantum level, the prototype of a C*-algebra is given by a
 Von Neumann algebra: 
\begin{equation} A \incl B(\cal{H}) \end{equation}
of bounded operators over a complex Hilbert space $\cal H$, 
with complex adjunction $a^* = \bar{a}^t$ as involution,
although most constructions can be carried on the algebra 
most naturally, without any reference to a particular Hilbert space. 

Every C*-algebra $A$ has a positive cone $A^+$ defined by:
\begin{equation} a \geq 0 \quad\eqvl\quad \exists b \in A \txt{with} b^* b = a \end{equation} 
Any positive element $a \geq 0$ is self-adjoint and satisfies $a^* = a$;
its spectrum is contained in $\R^+$. 
The above would also describe positive functions of $C(E)$.
Positivity will be a fundamental concept when defining the states of 
the algebra.

\subsection{Linear Forms and Measures}

Given an algebra with a continuous norm $A$, 
its topological dual $A^*$ is the vector space of continuous linear forms on $A$.
The topological dual defines a contravariant functor $\Alg^{op} \aw \Vect$ 
as any linear map $T : A \aw A'$ has an 
adjoint map $T^*: A'^* \aw A^*$  defined for all  
$\lambda \in A'^*$ and $a \in A$ by:
\begin{equation} \croc{T^*\lambda}{a} = \croc{\lambda}{T a} \end{equation}
The duality comes from the underlying vector space
and is common enough not to be discussed. 
We only briefly review some classical constructions and notations. 

When $A = C(E)$ is the real algebra 
of continuous and bounded functions over $E$, its 
dual $A^*$ is the space of Borel measures of finite mass on $E$,
equipped with the $L^1$-norm: 
\begin{equation} \croc{\lambda}{f} = \int_{x \in E} f(x) \cdot \lambda(dx) \end{equation}
for all $\lambda \in A^*$ and $f \in A$,  with
$\big| \croc{\lambda}{f} \big| \leq ||f||_\infty \cdot ||\lambda||_1$.

A continuous map $\ph: E \aw E'$ induces a map of algebras
 $\ph^*: A' \aw A$ by pull-back. 
Its adjoint map is the push-forward of measures $\ph_* : A^* \aw A'^*$, 
defined for every $\lambda \in A^*$ and every measurable
subset $S' \incl E'$ by: 
\begin{equation} 
(\ph_* \lambda)(S') = \int_{x \in \ph^{-1}(S')} \lambda(dx) 
\end{equation}
When $E$ is finite, 
the push-forward of a measure $\lambda \in A^*$
is given by its weight on each point $x' \in E'$:
\begin{equation} (\ph_* \lambda)(x') = \sum_{\ph(x) = x'} \lambda(x) \end{equation}
In applications, the set maps we will consider are 
projections of the form $\ph : E_1 \times E_2 \aw E_1$. 
The pushforward of $\ph$ is then called the marginal projection 
on $E_1$, or partial integration along $E_2$.

When $A = B(\cal H)$ is the algebra of bounded operators on a Hilbert space, 
equipped with a continuous trace operator 
$\Tr: A \aw \C$, one may define the hermitian 
scalar product of $a$ and $b$ in $A$ by $\Tr(a^*b)$. 
A linear form $\lambda \in A^*$, that is also continuous for the hermitian 
norm induced, may be represented by an element of $B(\cal H)$ such that:
\begin{equation} \croc{\lambda}{a} = \Tr(\lambda^* a) \end{equation}
This point of view is the most commonly used in quantum statistics.

\subsection{Statistical States}

A state of a unital involutive algebra $A$ is a linear form $\om \in A^*$
satisfying the two following axioms:
\begin{enumerate}[label=(\roman*)]
\item $\om(a^* a) \geq 0$ for all $a \in A$ (positivity)
\item $\om(1) = 1$ (normalisation)
\ee
The states of $A$ form a convex subset of linear forms $\Delta \incl A^*$.

When $A = C(E)$ is commutative and $A^*$ is the space of finite mass measures
on $E$, the above axioms define positive measures of mass 1, {\it i.e.} 
probability densities on $E$ and we have $\Delta = \Prob(E)$.
According to Gelfand's theorem, any commutative C*-algebra is 
isomorphic to a complex algebra $C(E, \C)$ 
of continuous bounded functions over a compact space $E$,
called the spectrum of $A$.

When $A$ is a generic C*-algebra, the Gelfand-Naimark-Segal construction
associates to each state $\om \in \Delta$ a Hilbert space representation 
${\cal H}_\om$ and a unit cyclic vector $\psi_\om \in {\cal H}_\om$ 
such that for all $a \in A$:
\begin{equation} \om(a) = \croc{\psi_\om}{a \cdot \psi_\om} \end{equation}
In quantum mechanics, this expression 
traditionally defines the mean value of a self-adjoint observable $a$
when the system is in the state $\psi_\om \in {\cal H}_\om$. 
For every self-adjoint $a$, the spectral projections of $\psi_\om$ define 
a probability distribution on the spectrum of $a$. 
This may be viewed as a consequence of Gelfand's theorem,
as the commutative C*-algebra generated by $a$ and $a^*$ is
isomorphic to $C({\rm Sp}\;a)$.

When $A \incl B({\cal H})$ is already represented on a Hilbert space 
and is equipped with a trace operator, 
operators of $B({\cal H})$ are mapped to linear forms. 
Any positive operator $\rho \in B({\cal H})_+$ such that 
$\Tr(\rho) = 1$ then
defines a state of $A$ by letting for all $a \in A$: 
\begin{equation} \rho(a) = \Tr(\rho a) \end{equation}
This picture can lead to confusion with the previous one, 
as $\cal H$ is not the GNS representation of $\rho$. 
The operator $\rho$ may however be viewed as a vector of 
the Hilbert space ${\cal H} \otimes {\cal H}^*$, 
of which the GNS representation ${\cal H}_\rho$ is a subspace.
In statistical quantum mechanics, $\rho$ is called the 
{\it density matrix}.

\section{Systems} 

This section introduces the differential and module structures on which relies 
the present work. 
The theory will be treated abstractly to keep as much generality as possible,
and deals with what one may call systems of algebraic structures, 
{\it i.e.} a particular type of functors. 

One should still read the theory with the contents of
the previous section in mind, to which it aims to be applied. 
The main idea is to localise the previous structures from a global 
set of variables $\Om = \{ i, j,k, \dots \}$ to a covering 
of $\Om$ by smaller regions $X = \{ \aa, \bb,\cc, \dots \} \incl \Part(\Om)$.
This leads to the definition of local configuration spaces, 
local algebras of observables, {\it etc.} 
related by morphisms every time a region is contained in another. 
Giving $(X, \cont)$ a category structure by agreeing 
that a unique arrow $\aa \aw \bb$ exists whenever $\aa$ contains $\bb$,
we will get functors\footnote{
    Functors with a partial order as source category are often called
    systems in the literature, as they were considered long before the categorical
    language became common use.
}
from $X$ to $\Set$, $\Alg$, {\it etc.}

\begin{center}
    \begin{tabular}{ccc}

\raisebox{-.5\height}{
    \includegraphics[width=0.25\textwidth]{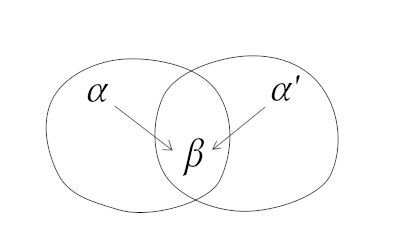}
}
        & 
        \hspace{1cm}
        &
$\bcd 
A_\aa & & A_{\aa'}  \\
& A_\bb \ular{j_{\aa\bb}} \urar[swap]{j_{\aa'\bb}} &
\ecd$

        \end{tabular}
\end{center}

The main result of this section is that we may define a chain complex
$\big( A_\bullet(X) , \div \big)$ of observables. 
Its boundary operator $\div$ will play a crucial role in describing 
the Lagrange multipliers of the cluster variation method and in 
defining transport equations that generalise belief propagation.
This construction was already considered 
by Grothendieck and Verdier 
under the name of {\it canonical projective resolution for presheaves} 
\cite{SGA-4-V, Moerdijk}. 
Their motivations having been more abstract, 
aiming at the unification of all known homology and cohomology theories, 
we believe the present work has the benefit 
of providing a simple and concrete application of this complex. 
The interaction decomposition theorem \ref{int-decomposition}
will also allow us to clarify 
the structure of homology groups in our setting. 

\subsection{Systems of Sets}

Given a numerable set of atoms $\Om$ and a configuration space $E_i$ for all 
$i \in \Om$, let:
\begin{equation} E_\aa = \prod_{i \in \aa} E_i \end{equation}
denote the configuration space of $\aa \incl \Om$. 
There is a projection $\pi^{\bb\aa} : E_\aa \aw E_\bb$ 
for every $\aa \cont \bb$,
forgetting the
state of atoms outside $\bb$. 
This will consist of our fundamental example of a projective system of sets,
for any collection of regions $X \incl \Part(\Om)$. 
\begin{equation} \bcd
E_\aa \drar[swap]{\pi^{\bb\aa}} & & E_{\aa'} \dlar{\pi^{\bb\aa'}} \\
    & E_\bb & 
\ecd \end{equation}
When $X$ covers $\Om$, the system efficiently keeps all the available 
information, as the global configuration space can  
be recovered as the projective limit $E_\Om = \lim_{\aa \in X} E_\aa$.

\begin{defn} 
    A {\rm system of sets} $E$ over a partial order $X$ is 
    a covariant functor $E : X \aw \Set$.
    Denoting by $\ph_* : E_\aa \aw E_\bb$ the map induced
    by $\ph : \aa \aw \bb$ in $X$, a system is said:
    \bi
    \iii {\rm injective} when $\ph_*$ is an injection for all $\ph$.
    \iii {\rm projective} when $\ph_*$ is a surjection for all $\ph$,
    \ei
The functor category $[X, \Set]$ of systems over $X$
has natural transformations $\eta : E \aw E'$ as morphisms. 
\end{defn}

A functor $t : X \aw X'$ induces 
a pull-back functor $t^* : [X', \Set] \aw [X, \Set]$
defined by $t^* E' = E' \circ t$. 
This allows to compare the categories of systems over different partial orders 
$X$ and $X'$ and to define a global category of systems 
of sets over an arbitrary partial order.

\begin{defn}
    We denote by $\{ \Set \}$ the category of systems of sets, with: 
    \bi
    \iii objects $(X,E)$ where $X$ is a partial order
    and $E$ is a system of sets over $X$,
    \iii morphisms $(t, \eta) : (X,E) \aw (X', E')$ 
    where $t: X \aw X'$ is a functor and
    $\eta : E \aw t^* E'$ is a natural transformation.
    \ei
\end{defn}

We introduce the notation $\{ \Set \}$ 
to avoid confusion with the larger category of functors $[-, \Set]$
as the source category $X$ is restricted to partial orders.
As a subcategory of the latter,
it should be thought of in the same way,  
and the partial order hypothesis
will have no influence until the next section.

{\bf Examples.}
\begin{enumerate}[itemsep=3px]
\item 
A single set $E$ is a system over the point 
category $\{ \bullet \}$. 

\item 
The restriction of a system over $X$
to a subcategory $Y \incl X$ is naturally mapped into the original system. 
This provides with a trivial example of morphism in $\{\Set\}$.

\item Given an equivalence relation $\sim$ in $X$, any 
system $E$ over $X$ induces a system $\bar E$ over the quotient 
space $\bar X = (X / \sim)$ defined by: 
\begin{equation} \bar E_{[\aa]} = \bigsqcup_{\aa' \sim \aa} E_{\aa'} \end{equation}
Denoting by $p: X \aw \bar X$ the quotient map, 
the natural transformation from $E$ to $p^*\bar E$ is canonically
defined by inclusion of $E_\aa$ in the disjoint union $\bar E_{[\aa]}$.
\end{enumerate}

\subsection{Systems of Abelian Groups}

The category $\{ \Ab \}$ of abelian group systems is 
defined by restricting to functors $G : X \aw \Ab$.
To such a system $G$, we shall associate
chain and cochain complexes 
denoted by 
$\big(G_\bullet(X), \div \big)$ and $\big(G^\bullet(X), d\big)$
respectively. 
Their construction extends $G$ to a functor on the nerve\footnote{
The nerve of a category is defined in paragraph 1.3.2.
}
of $X$ 
and in the following we denote by $\bar \aa \in N_p(X)$ a $p$-chain 
$\aa_0 \aw \dots \aw \aa_p$ in $X$.

Consider the contravariant functor $\hat G : N(X)^{op} \aw \Ab$ defined by 
$\hat G_{\bar \aa} = G_{\aa_0}$ for all $\bar \aa \in N_p(X)$. 
For every subchain $\bar \bb$ of $\bar \aa$, the map 
$\hat G_{\bar \aa} \aw \hat G_{\bar \bb}$ is induced by 
$G_{\aa_0} \aw G_{\bb_0}$ as $\aa_0 \aw \bb_0$.
The simplicial set structure of $N(X)$ thus makes 
$\hat G$ a simplicial abelian group, and $\hat G$ defines 
a chain complex $G_\bullet(X)$ equipped with 
a boundary operator $\div : G_{n+1}(X) \aw G_n(X)$, where:
\begin{equation} G_n(X) = \prod_{\bar \aa \in N_n(X)} \hat G_{\bar \aa} \end{equation}

Reciprocally, a covariant functor
$\check G : N(X) \aw \Ab$ is defined by letting 
$\check G_{\bar \aa} = G_{\aa_p}$ 
for $\bar \aa \in N_p(X)$.
When $\bar \bb$ is a subchain of degree $k \leq p$ of $\bar \aa$, 
we have a map $\check G_{\bar \bb} \aw \check G_{\bar \aa}$ 
as $\bb_k \aw \aa_p$.
Hence $\check G$ is a cosimplicial abelian group
and defines a cochain complex $G^\bullet(X)$ 
with a differential $d: G^n(X) \aw G^{n+1}(X)$, 
where: 
\begin{equation} G^n(X) = \prod_{\bar \aa \in N_n(X)} \check G_{\bar \aa} \end{equation}

Dual constructions of course arise when $G: X^{op} \aw \Ab$ is a cosystem
over $X$. In this case we still let $\hat G_{\bar \aa} = G_{\aa_0}$ 
and $\check G_{\bar \aa} = G_{\aa_p}$ to define functors 
$\hat G : N(X) \aw \Ab$ and $\check G : N(X)^{op} \aw \Ab$. 
Applications will involve both covariant and contravariant functors of abelian 
groups but their extension to the nerve 
will mostly be done through $\check G$. 
The following table might be useful:
\begin{center}
{\renewcommand{\arraystretch}{1.6} 
    \begin{tabular}{| C{3cm} | C{2cm} | C{2cm} |}
        \hline
        *   & $\big( G_\bullet(X), \div \big)$ 
            & $\big( G^\bullet(X), d \big)$ \\
        \hline
        $G: X \aw \Ab$  & $\hat G$  & $\check G$ \\
        \hline
        $G : X^{op} \aw \Ab$ & $\check G$ & $\hat G$ \\
        \hline
    \end{tabular}
}
\end{center}

{\bf Fundamental Example.}

\begin{enumerate}[itemsep=3px]
    \item When $E : X \aw \Set$ is a system of sets over $(X, \cont)$, 
    it defines a cosystem of algebras $A : X^{op} \aw \Alg$ 
    by letting $A_\aa = \R^{E_\aa}$ for all $\aa \in X$. 
    In the chain complex $(A_{\bullet}(X), \div)$, 
    a 1-chain $\ph$ is defined by a collection of local observables 
    $\ph_{\aa\bb} \in \R^{E_\bb}$ and its boundary $\div \ph$ given by: 
    \begin{equation} \label{div} 
        \div \ph_\bb(x_\bb) = \sum_{\aa' \contst \bb} \ph_{\aa' \bb}(x_\bb) 
        - \sum_{\bb \contst \cc'} j_{\bb\cc'}(\ph_{\bb \cc'})(x_{\bb}) 
    \end{equation} 
        where $j_{\bb\cc}(\ph_{\bb\cc}) \in \R^{E_\bb}$  
        denotes the pullback of $\ph_{\bb\cc} \in \R^{E_\cc}$ 
        under the map $x_\bb \in E_\bb \mapsto x_\cc \in  E_\cc$. 
        For $\ph \in A_{n+1}(X)$, its boundary $\delta \ph \in A_n(X)$ 
        is similarly given by: 
\begin{equation}
    \ba{lllcl} \tag{2.20.$n$} \label{boundary} 
        \delta \ph_{\bb_0 \dots \bb_n}(x_{\bb_n})
        & = &  &\disp \sum_{\aa_0 \contst \bb_0} 
            &  \disp \ph_{\aa_0 \bb_0 \dots \bb_n}(x_{\bb_n}) \\
        & \disp + \sum_{k = 1}^{n}  
            & (-1)^k 
            & \disp \sum_{\bb_{k - 1} \contst \aa_{k} \contst \bb_{k}} 
            & \ph_{\bb_0 \dots \bb_{k - 1} \aa_k \bb_k \dots \bb_n}(x_{\bb_n}) \\
        & + & (-1)^{n+1} 
            & \disp \sum_{\bb_n \contst \aa_{n+1}} 
            & \ph_{\bb_0 \dots \bb_n \aa_{n+1}}(x_{\aa_{n+1}})   
    \ea 
\end{equation} 
    where $x_{\aa_{n+1}}$ denotes\footnote{
        When $E_\aa = \prod_{i \in \aa} E_i$, 
        the forgetful map $x_\aa \mapsto x_\bb$ will simply 
        be suggested by lowering the subscript $\aa$ to $\bb \incl \aa$. 
        Functions on $E_\aa$ that do not depend on 
        $x_{\aa \setminus \bb}$ belong to the subalgebra 
        $\R^{E_\bb} \incl \R^{E_\aa}$. 
        The inclusion map $j_{\bb\aa} : \R^{E_\bb} \aw \R^{E_\aa}$ 
        is a restriction of the identity and 
        will therefore be kept implicit in our notation.
    } the image of $x_{\bb_n}$ 
    for every $\bb_n \cont \aa_{n+1}$.

\item When $E : X \aw \Set$ and $A = \R^E$, duality defines 
    a system of vector spaces $A^*: X \aw \Vect$. 
    In the cochain\footnote{
        Although it is a cochain complex, 
        we write degrees as indices for $A_\bullet^*(X)$ 
        as it is the dual vector space of $A_\bullet(X)$. 
    }
    complex $(A^*_\bullet(X), d)$, 
    a 0-cochain $q$ is defined by a collection of linear forms
    $q_\aa \in L(\R^{E_\aa}, \R)$ and 
    its differential $dq$ given by: 
    \begin{equation}  \label{d} 
        dq_{\aa\bb} = q_\bb - \Sigma^{\bb\aa}(q_\aa)
    \end{equation}
        where $\Sigma^{\bb\aa}(q_\aa) \in L(\R^{E_\bb}, \R)$ 
    denotes the pushforward of $q_\aa \in L(\R^{E_\aa}, \R)$ 
    by the map $E_\aa \aw E_\bb$.
    For $\psi \in A^*_n(X)$, its differential $d\psi \in A_{n+1}(X)$ 
    is similarly given by, in coordinates:
    \begin{equation} \tag{2.21.$n$} \label{differential}
    \ba{llll}
    d \psi_{\aa_0 \dots \aa_{n+1}}(x_{\aa_{n+1}})
    & = & 
        & \psi_{\aa_1 \dots \aa_{n+1}}(x_{\aa_{n+1}})  \\
    & \disp + \sum_{k = 1}^{n}
        & (-1)^k 
        & \psi_{\aa_0 \dots \aa_{k - 1} \aa_{k+1} \dots  \aa_{n+1}}(x_{\aa_{n+1}}) \\
    & + & (-1)^{n+1} 
            & \disp \sum_{x'_{\aa_n} \mapsto x_{\aa_{n+1}}}
            \psi_{\aa_0 \dots \aa_{n}}(x'_{\aa_n})   
    \ea 
    \end{equation}
    the last sum running over preimages $x'_{\aa_n} \in E_{\aa_n}$ 
    of $x_{\aa_{n+1}} \in E_{\aa_{n+1}}$, both sets assumed finite\footnote{
        Otherwise replace with an integral over the preimage of $\{x_{\aa_{n+1}}\}$ 
        in $E_{\aa_n}$ under the right integrability assumptions. 
    }. 

\end{enumerate}

The difference of incoming and departing fluxes 
defining $\div \ph$ in (\ref{div}) 
recalls the classical {\it divergence} 
operator of differential geometry.  
Acting on a smooth vector field  
$\vec \ph : \R^3 \aw \R^3$, 
the divergence is intrinsically related to the
fundamental {\it Gauss formula}: 
\begin{equation} \label{gauss-geometry}
\int_V {\rm div}(\vec \ph) \, dv = \int_{\dr V} \vec \ph \cdot \vec n \,ds 
\end{equation}
Integrating ${\rm div}(\vec \ph)$ on 
a volume $V \incl \R^3$ yields the outbound\footnotemark{} flux of $\vec \ph$ 
through the boundary of $V$. 
\footnotetext{
    With the sign convention from physics. 
    In Hodge theory where $d^* = - {\rm div}$, an 
    {\it inbound} flux is measured instead.
} 

Proposition \ref{gauss} plays a 
crucial role in understanding the topological structure of
message-passing algorithms. 
To serve as counterparts for the integration\footnotemark{} 
supports of (\ref{gauss-geometry}), let us denote by:
\bi
\iii $\LL^\aa = \{ \bb \in X \st \aa \aw \bb \}$ the cone below $\aa$ in $X$,
\iii $d \LL^{\aa} = \{ \aa'\bb' \in N_1(X) \st 
    \aa' \not\in \LL^\aa, \:\bb' \in \LL^\aa \}$ 
    the coboundary of $\LL^\aa$.
\ei 
\footnotetext{
The sum over $\LL^\aa$ is the {\it zeta transform} 
of chapter 3. Proposition \ref{gauss} hence provides with a beautiful 
illustration of the analogy between 
Möbius inversion formulas and the fundamental theorem of calculus
already noted by Rota in \cite{Rota-64}.
}

\begin{prop}[Gauss Formula]  
    \label{gauss}
    Given a cosystem $G : X^{op} \aw \Ab$ of abelian groups, 
    let $(G_\bullet(X), \div)$ denote the associated chain complex. 
    For every $\ph \in G_1(X)$, we have: 
    \begin{equation} \sum_{\bb' \in \LL^\aa} (\div\ph)_{\bb'} = 
    \sum_{\aa'\bb' \in d \LL^\aa} \ph_{\aa'\bb'} \end{equation}
\end{prop}

\begin{proof}
In the sum of $\div \ph$ over $\LL^\aa$, 
each term $\ph_{\bb'\cc'}$ is counted twice with opposite signs
if $\bb' \in \LL^\aa$. 
\end{proof}

    \vspace{-0.1cm}
\begin{figure}[H]
    \sbox0{\includegraphics[width=0.65\textwidth]{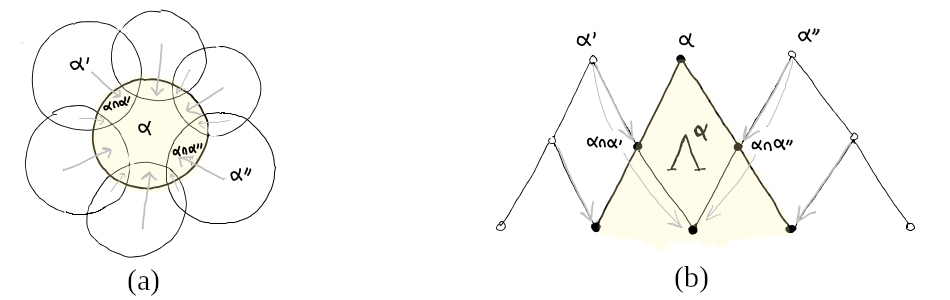}
    }
\begin{center}
    \begin{minipage}{0.8\textwidth}
\centering
\usebox0
    \vspace{-0.2cm}
        \caption{
            $X \incl \Part(\Om)$ illustrated as 
            (a) a Venn diagram and (b) a semi-lattice. 
            The colored region represents $\LL^\aa = \{ \bb' \incl \aa\}$ 
            while arrows represent some of the flux contributions to 
            the Gauss formula 
            $\sum_{\bb' \in \LL^\aa} \delta \ph_{\bb'} 
            = \sum_{\aa'\bb' \in d\LL^\aa} \ph_{\aa'\bb'}$.
        }
\end{minipage}
\end{center}
    \vspace{-0.5cm}
\end{figure}


\subsection{Systems of Rings and Modules}

The category $\{\Ring \}$ of ring systems is similarly defined 
by restricting to functors $R : X \aw \Ring$.
To avoid confusion in applications to come, 
it will be more convenient to consider cosystems over $X$, 
{\it i.e.} functors $R : X^{op} \aw \Ring$. 
This subsection, mostly inspired by Kodaira \cite{Kodaira-49},
explores the different products and module structures
one may generalise from the usual theory with scalar coefficients. 

Given a cosystem of rings $R$ over $X$, we give a natural 
ring structure $\big( R^\bullet(X), + , \wedge \big)$ 
to its associated cochain complex, by defining 
the exterior\footnote{
    With scalar coefficients, $\wedge$ is called the Alexander product
    in Kodaira \cite{Kodaira-49}. 
}
product of a $p$-field $\ph$ with a $k$-field $\psi$ 
as the $(p+k)$-field:
\begin{equation} (\ph \wedge \psi)_{\aa \dots \bb \dots \cc}
= \ph_{\aa \dots \bb} \cdot \big[ \psi_{\bb \dots \cc}\big]_\aa \end{equation}
where $[\psi_{\bb \dots \cc}]_\aa$ denotes the image of 
$\psi_{\bb \dots \cc} \in R_\bb$ in $R_\aa$. 
The exterior product is compatible with the differential and
we have the graded Leibniz rule: 
\begin{equation} d(\ph \wedge \psi) = d \ph \wedge \psi 
+ (-1)^{|\ph|} \; \ph \wedge d \psi \end{equation}
denoting by $|\ph|$ the degree of $\ph$,
and $\big( R^\bullet(X), +, \wedge, d \big)$ is a differential ring.

There is a natural notion of module over a ring system
which makes any ring system a module over itself. 
The dual notion of comodule will also be of interest to us,
they both give a differential module structure to one of the
associated complexes.

\begin{defn}
We call module over $R: X^{op} \aw \Ring$ 
any cosystem of abelian groups $M : X^{op} \aw \Ab$ such that:
\bi
\iii $M_\aa$ is an $R_\aa$-module for all $\aa \in X$,
\iii $[r_\bb]_\aa \cdot [m_\bb]_\aa = [r_\bb \cdot m_\bb]_\aa$ 
for all $\aa \aw \bb$ in $X$.
\ei
where $[m_\bb]_\aa$ denotes the image of $m_\bb$ in $M_\aa$. 
\end{defn}

When $M$ is a module over $R$, the cochain complex $M^\bullet(X)$ 
inherits a module structure over $R^\bullet(X)$
for the action of $\ph \in R^p(X)$ on $m \in M^k(X)$ extending the 
exterior product:
\begin{equation} (\ph \times m)_{\aa \dots \bb \dots \cc}
= \ph_{\aa \dots \bb} \cdot \big[ m_{\bb \dots \cc}\big]_\aa \end{equation}
and the differential $\nabla$ on $M^\bullet(X)$ then 
satisfies the graded module Leibniz rule:
\begin{equation} \nabla(\ph \times m) = d \ph \times m 
+ (-1)^{|\ph|} \; \ph \times \nabla m \end{equation}
so that $\big(M^\bullet(X), \nabla \big)$ is a differential 
module over $\big( R^\bullet(X), d \big)$. 

\begin{defn}
We call comodule over $R: X^{op} \aw \Ring$ 
any system of abelian groups $M : X \aw \Ab$ such that:
\bi
\iii $M_\aa$ is an $R_\aa$-module for all $\aa \in X$,
\iii $r_\bb \cdot [m_\aa]_\bb = [r_\bb \cdot m_\aa]_\bb$ 
for all $\aa \aw \bb$ in $X$.
\ei
where $[m_\aa]_\bb$ denotes the image of $m_\aa$ in $M_\bb$ 
and $\cdot$ the action 
of any $R_\bb$ such that $\aa \aw \bb$ on $M_\aa$.
\end{defn}

When $M$ is a comodule over $R$, the action
of $\ph \in R^k(X)$ on $m \in M_n(X)$ is the chain of $M_{n-k}(X)$
given by the interior product:
\begin{equation} (\ph \idot m)_{\bb \dots \cc} = \sum_{\aa' \dots \bb} 
\ph_{\aa' \dots \bb} \cdot \big[ m_{\aa' \dots \bb \dots \cc} \big]_\bb \end{equation}
and by $(\ph \idot m) = 0$ when $n < k$. 
Note that $\idot$ actually defines a right action as 
$(\ph \wedge \psi) \idot m = \psi \idot (\ph \idot m)$.
The boundary of $M_\bullet(X)$ satisfies the graded module Leibniz rule:
\begin{equation} \div(\ph \idot \psi) = d\ph \idot \psi + (-1)^{|\ph|}\; \ph \idot \div \psi \end{equation}
and $\big(M_\bullet(X), \div \big)$ is a differential module 
over $\big( R^\bullet(X), d \big)$.

{\bf Examples.}
\begin{enumerate}[itemsep=3px]

    \item Consider the constant ring system $\Z$ over $X$,  
        and for every $\aa \in X$, the $1$-cochain 
        $\lambda^\aa \in \Z^1(X)$ defined by: 
        \begin{equation} \lambda^\aa_{\aa'\bb'} = 
        \left\{ \begin{array}{c}
            1  \txt{if} \aa' = \aa  \\
            0  \txt{if} \aa' \neq \aa
        \end{array} \right. 
        \end{equation}
        An easy computation shows that $d\lambda^\bb \in \Z^2(X)$ may be
        written for all $\bb \in X$ as: 
        \begin{equation} d \lambda^\bb = \sum_{\aa \aw \bb} \lambda^\aa \wedge \lambda^\bb \end{equation}

    \item Any cosystem of algebras $A : X^{op} \aw \Alg$ is a 
        comodule over $\Z$. For every $\bb \in X$ 
        we then denote by $i_\bb$ the degree $-1$ 
        endomorphism of $A_\bullet(X)$ defined by 
        $i_{\bb}(\ph) = \lambda^\bb \idot \ph$.
        For every $\bar \cc \in N(X)$, 
        if $\bb \not\aw \cc_0$ we have $i_{\bb}(\ph)_{\bar \cc} = 0$ and  
        otherwise:
        \begin{equation} i_{\bb}(\ph)_{\bar \cc} = \ph_{\bb \bar \cc} \end{equation} 
        The previous example and  the graded module Leibniz rule give
        the following formula:
        \begin{equation} \div(i_\bb \ph) = i_\bb \Big( \sum_{\aa \aw \bb} i_\aa \ph \Big) 
        - i_\bb(\div \ph) \end{equation}
        which will be very useful in the proof of proposition 
        \ref{higher-zeta-cocycle}.
\end{enumerate}

\section{Local Statistics}

Localising the statistical structures of section 1, 
this section restricts the previous constructions 
to a projective system of sets given by
cartesian products $E_\aa = \prod_{i \in \aa} E_i$ of atomic configuration spaces, 
and focuses on the inductive system $A_\aa = \R^{E_\aa}$ 
of algebras induced above it.

A fundamental consequence of this hypothesis is the 
so-called {\it interaction decomposition theorem}. 
It asserts that each algebra of observables splits 
as a direct sum $A_\aa = \bigoplus_\bb Z_\bb$, where 
the interaction vector spaces $Z_\bb$ may be chosen consistently over $X$ 
and play the role of independent generators\footnote{
    This fact motivates the {\it free sheaves} terminology of \cite{ExtrafineSheaves},
    where a slightly more general setting was considered with Bennequin and Sergeant.
    The present hypotheses are however enough to cover all the applications to 
    statistics we shall be interested in.
}.
This well-known yet subtle result may be given a significant number of proofs 
and calls for greater generality,
but the simple setting we consider allows for a beautiful proof 
via harmonic analysis\footnote{
    The theorem name, {\it interaction decomposition}, 
    was suggested by the regretted Franti{\v s}ek Mat{\v u}s,
    whom we had the chance to meet in Marseille in 2017. 
    We follow his proof given \cite{Matus-88} for its elegance, 
    see also \cite{Kellerer-64, Speed-79, ExtrafineSheaves}.
}.

The main result of this section is the acyclicity\footnote{
    Acyclicity means that higher homology groups vanish, i.e. 
    $\Ker(\delta_n) = \Img(\delta_{n+1}) \incl A_n(X)$ for every $n > 1$. 
    The sequence $A_0(X) \wa A_1(X) \wa \dots$ hence provides 
    with a {\it projective resolution} of the first 
    homology group $\Ker(\delta_0) / \Img(\delta_1) = A_0(X) / \delta A_1(X)$, 
    in which the global hamiltonian lies, thanks to the isomorphism 
    of theorem \ref{H0}.
}
of the 
complex of local observables $A_\bullet(X)$. 
We show that the homology class
of $h \in A_0(X)$ is completely determined by its global sum
$H_\Om = \sum_\aa h_\aa$ in $A_\Om$. 
In a physical terminology, one would say
that two potentials $h$ and $h'$ are homologous if and only 
if they define the same global hamiltonian $H_\Om$.

\subsection{Statistical System}

From now on, we consider a finite set of atoms $\Om$ with finite sets
of microstates $E_i$ for all $i \in \Om$.
The configuration space $E_\aa$ of any region $\aa \incl \Om$ is 
defined as: 
\begin{equation} E_\aa = \prod_{i \in \aa} E_i \end{equation} 
For every $\aa \cont \bb$, we denote the canonical projection 
by $\pi^{\bb\aa} : E_\aa \aw E_\bb$.  
Alternatively, one could say that 
 $E : \big( \Part(\Om),  \cont \big) \aw \Set$  
defines a sheaf of finite sets over the finite topological space $\Om$. 

Given these sets of microstates, local algebras of observables are defined by 
$A_\aa = \R^{E_\aa}$ and for every $\aa \cont\bb$, 
the canonical injection $j_{\aa\bb} : A_\bb \aw A_\aa$ is the pull-back 
of $\pi^{\bb\aa}$ sending a function on $E_\bb$ to its cylindrical 
extension on $E_\aa$.
The multiplicative structure of $E$ is carried
to $A$ by:
\begin{equation} A_\aa = \bigotimes_{i \in \aa} A_i \end{equation}
as $A_\aa$ is linearly generated by the Dirac masses $(\delta_{x_\aa})$ 
on $E_\aa$, which may be written as pure tensors 
of the form $\delta_{x_i} \otimes \dots \otimes \delta_{x_j}$, 
and the extension map $j_{\aa\bb}$ may then be viewed as the
tensor multiplication $1_{\aa \setminus \bb} \otimes -$ 
with the unit of $A_{\aa \setminus \bb}$.
Better suited generators $(\chi_{k_\aa})$ will be introduced in the next section.

In the following, we suppose chosen a covering $X \incl \Part(\Om)$ closed under
intersection, {\it i.e.} such that: 
\begin{equation} \aa \in X \txt{and} \bb \in X \quad\impq\quad
\aa \cap \bb \in X \end{equation}
We then restrict $E$ to a system $E : X \aw \Set$ 
and $A$ to a cosystem $A : X^{op} \aw \Alg$ over $X$.
The closure hypothesis is fundamental for 
the interaction decomposition theorem to hold, it will 
also be useful for the generalised combinatorics we propose in chapter 3.

\begin{defn} \label{complexes}
    We denote by:
\bi 
\iii $\big( A_\bullet(X), \div \big)$ 
    the chain complex of {\rm local observables}, 
\iii $\big( A^*_\bullet(X), d \big)$ 
    the cochain complex of {\rm local measures}, 
\iii $\Delta_\bullet(X) \incl A^*_\bullet(X)$ the convex subspace 
    of {\rm local probabilities}.
\ei 
\end{defn}

This localisation procedure will lead us to represent the global hamiltonian 
$H_\Om \in A_\Om$ by a homology class 
of interaction potentials $(h_\aa) \in A_0(X)$ satisfying:
\begin{equation}  H_\Om = \sum_{\aa \in X} h_\aa \end{equation} 
The global Gibbs state $p_\Om \in \Delta_\Om$ would also be ideally  
represented by its local marginals $(p_\aa) \in \Delta_0(X)$ 
or an approximation of the latter. 
Marginals of $p_\Om$ are said consistent  
as $p_\bb$ is the marginal of $p_\aa$ for every $\bb \incl \aa$,
and the following more general notion of 
cohomology class in $\Delta_0(X)$ 
will substitute for the global probabilities of $\Delta_\Om$. 

\begin{defn} \label{Gamma}
    The convex subspace $\Gamma(X)$ of {\rm consistent local probabilities} 
    is defined by: 
    \begin{equation} \Gamma(X) = \Delta_0(X) \cap \Ker(d) \end{equation} 
    We also denote by $\Gammint(X)$ the space of consistent positive 
    local probabilities. 
\end{defn}

The image of $\Delta_\Om$ in $\Delta_0(X)$ is in general 
a strict convex polytope of $\Gamma(X)$, 
as although any consistent $q \in \Gamma(X)$ may 
always be extended to a global measure $q_\Om \in A^*_\Om$, 
the non-negativity of $q_\Om$ is not insured. 
This was already noticed by Vorob'ev who first characterised the
simplicial complexes $X$ having the property that any 
consistent family of probabilities in $\Gamma(X)$ may be extended\footnotemark{} 
to $\Delta_\Om$. 
They essentially coincide with the {\it \retractable} hypergraphs  
on which we show message-passing algorithms to be exact in chapter 6. 

\footnotetext{
    See also \cite{Abramsky-2011} for a sheaf theoretic approach of this problem, 
    and relations to the notion of contextuality. 
    The complex we consider is related to a barycentric subdivision to the 
    complex of \cite{Abramsky-2011}, as the categorical nerve 
    subdivides the {\v C}ech nerve of a covering, 
    and the isomorphism of homology groups is proved in \cite{ExtrafineSheaves}. 
}

\subsection{Interaction Decomposition}

For every $\aa \cont \bb$ in $X$, the algebra $A_\bb$ 
is naturally embedded in $A_\aa$. Therefore $A_\aa$ contains 
observables that may be split as a sum of observables 
on strict subregions of $\aa$.
We call them {\it boundary} observables and write:
\begin{equation} B_\aa = \sum_{\aa \contst \bb} A_\bb \end{equation} 
We say that $Z_\aa$ is an {\it interaction subspace} of $A_\aa$ 
if it is a supplement of $B_\aa$. One may then write: 
\begin{equation} A_\aa = Z_\aa \oplus 
\Big( \: \sum_{\aa \contst \bb} A_\bb  \: \Big) \end{equation}
and continue this procedure inductively, which is the content
of the interaction decomposition theorem.

\begin{thm} \label{int-decomposition}
Given an interaction subspace $Z_\aa \incl A_\aa$ for every $\aa \in X$,
one has for all $\aa$: 
\begin{equation} A_\aa = \bigoplus_{\aa \cont \bb} Z_\bb \end{equation}
\end{thm}

It will be useful to consider the following rewording of the theorem. 
In this point of view, interaction subspaces define 
a cosystem of vector spaces $Z: X^{op} \aw \Vect$ with trivial maps, 
embedded in $A$, and we denote by $Z_0(X)$ their direct sum viewed
as a subspace of $A_0(X)$.

\begin{cor} \label{conservation}
Given a choice of interaction subspaces, we have 
a projection  $P : A_0(X) \aw Z_0(X)$ given by:
\begin{equation} P(u)_\bb = \sum_{\aa \cont \bb} P^{\bb\aa}(u_\aa) \end{equation}
where $P^{\bb\aa} : A_\aa \aw Z_\bb$ is a projection 
of $A_\aa$ onto $Z_\bb$, vanishing on $Z_\cc$ for every $\cc \neq \bb$. 
\end{cor}

The theorem asserts that a direct sum decomposition holds for {\it any} choice 
of interaction subspaces. As supplements of boundary observables, 
they are not defined intrinsically, 
although unless explicit mention of the opposite, 
$Z_\aa$ will from now on be assumed orthogonal to $B_\aa$ for the 
canonical scalar product of $\R^{E_\aa}$. 
This choice will play a particular
role in describing the high temperature limit, see for instance 
\ref{T_0 Z} and \ref{transversality-0}.  

\begin{defn}
The {\rm canonical interaction subspaces} $Z_\aa \incl A_\aa$ 
are defined as orthogonal supplements of boundary observables 
for the canonical scalar product of $A_\aa \simeq \R^{E_\aa}$: 
\begin{equation}
Z_\aa = \Big( \sum_{\bb \inclst \aa} \Img \big( j_{\aa\bb} \big) \Big)^\perp 
\quad\eqvl\quad
Z_\aa = \bigcap_{\bb \inclst \aa} \Ker \big( \Sigma^{\bb\aa} \big) 
\end{equation}
\end{defn}

Instead of proving \ref{int-decomposition} in its generality\footnote{
    [[appendix: proof by action of permutations]] 
},
inspired by Mat{\v u}s, we show how 
harmonic analysis gives an enlightening perspective on 
the canonical interaction decomposition. 
This seems to be the most natural construction of the 
projections $P^{\bb\aa}$ and is the one we implemented in javascript. 

Chose an ordering of $E_j$ for all $j \in \Om$, 
so that each $E_\aa$ may be identified with the finite torus:
\begin{equation} E_\aa \simeq \prod_{j \in \aa} \frac{\Z} { N_j \Z} \end{equation} 
and consider for every $\aa \in X$ the complexified algebra of observables
$\tilde A_\aa = \C^{E_\aa}$.
It is a fundamental result of abelian group theory that the spectrum 
$\hat E_\aa = \Hom(E_\aa, \C^*)$ of $E_\aa$ 
defines an orthormal basis of $\tilde A_\aa$ for its canonical 
hermitian product, 
where the so-called characters of $\hat E_\aa$ are plane waves generating 
a discrete Fourier transform on $\tilde A_\aa$.

By duality, there is a natural injection $\hat E_\bb \aw \hat E_\aa$ 
by pull-back of the projection $\pi^{\bb\aa} : E_\aa \aw E_\bb$.
Consider the subset $\hat F_\aa \incl \hat E_\aa$ of characters defined by: 
\begin{equation} \hat F_\aa = \bigcup_{\bb \inclst \aa} \hat E_\bb \end{equation}
It is easily seen that $\hat F_\aa$ is an orthonormal basis 
of the subspace $\tilde B_\aa$ of complex boundary observables.
Its complement $\hat G_\aa$ then provides with an orthonormal basis of 
the interaction subspace $\tilde Z_\aa = \tilde B_\aa^{\perp}$.
The interaction decomposition theorem here amounts to the 
observation that $\hat E_\aa$ is recovered as the disjoint union:
\begin{equation} \hat E_\aa = \bigsqcup_{\bb \incl \aa} \hat G_\bb \end{equation}
A boundary observable in $\tilde A_\aa$ is spanned 
by plane waves with wave vectors tangent to some $E_\bb$ with $\bb \inclst \aa$, 
and the spectral support of an interaction 
observable is in the complement of such wave vectors. 

\begin{figure}[H]
    \sbox0{\includegraphics[width=0.45\textwidth]{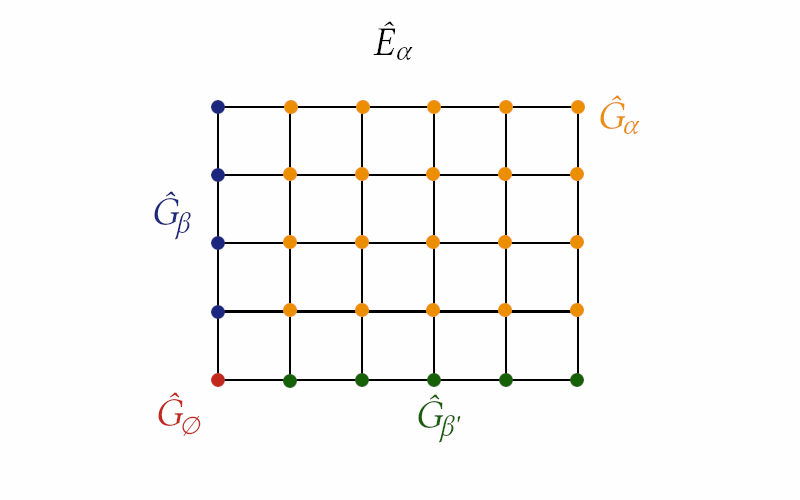}}
\begin{center}
\begin{minipage}{\wd0+0.1\textwidth}
\centering
\usebox0
\caption{
    Spectral decomposition $\hat E_\aa = \bigsqcup_{\bb \incl \aa} \hat G_\bb$. 
}
\end{minipage}
\end{center}
\end{figure}

More concretely, define for every $k_\aa \in E_\aa$
the character $\chi_{k_\aa} \in \hat E_\aa$ by:
\begin{equation} \chi_{k_\aa} (x_\aa) = \e^{i \croc{k_\aa}{x_\aa}} 
\txt{where} \croc{k_\aa}{x_\aa} = \sum_{j \in \aa} \frac{k_j x_j}{2 \pi N_j}
\end{equation}
The character $\chi_{k_\bb} \in \hat E_\bb$ is extended
to $\hat E_\aa$ by letting $k_j = 0$ for $j \notin \bb$. 
Identifying $\hat E_\aa$ with $E_\aa$, 
the discrete Fourier transform defines a unitary
transformation of $\tilde A_\aa$.
For all $u_\aa \in \tilde A_\aa$, write:
\begin{equation} u_\aa = \sum_{k_\aa \in E_\aa} \hat u_\aa(k_\aa) \cdot \chi_{k_\aa} \end{equation} 
The orthogonal projection $P^{\bb\aa} : \tilde A_\aa \aw \tilde Z_\bb$ is then 
simply given by selecting only the modes of $G_\bb \incl E_\bb$, 
complement of $\bigcup_{\bb \inclst \cc} E_\cc$: 
\begin{equation} P^{\bb\aa}(u_\aa) = \sum_{k_\bb \in G_\bb} 
\hat u_\aa(k_\bb) \cdot \chi_{k_\bb} \end{equation}

The real Fourier transform on $A_\aa$ would obviously give
a similar explicit construction of real interaction subspaces.
This construction appealed to us for both its great
conceptual simplicity and ease of implementation. 
It however strongly relies on the multiplicative
structure of the underlying system of sets.

The following proposition gives a more abstract characterisation 
of interaction subspaces, as other approaches\footnote{
    See \cite{ExtrafineSheaves} where the decomposition is carried in 
    a much more general setting. 
} 
would be necessary for general inductive systems of vector spaces.
Note that it is not true in general that the inductive limit
splits as a direct sum. 

\begin{prop}
    The inductive limit of $A$ is isomorphic to 
    the direct sum of interaction subspaces: 
    \begin{equation} Z_0(X) \simeq \colim_{\aa \in X} A_\aa \end{equation} 
\end{prop}

\begin{proof}
    Let $(V, f)$ denote a cone\footnote{
        See section 1.2 for the categorical definition of limits.
    }
    below $A$ defined by a collection of consistent maps $f_\bb : A_\bb \aw V$.
    Identify $(V,f)$ with a map $f : A_0(X) \aw V$ and suppose given
    $\tilde f : Z_0(X) \aw V$ such that $\tilde f \circ P = f$. 
    Then for every $v_\aa \in Z_\aa \incl A_0(X)$, we have 
    $\tilde f(v_\aa) = f(v_\aa)$. This defines 
    a unique map $\tilde f : Z_0(X) \aw V$ factorising $f$ through $P$,
    and $Z_0(X)$ satisfies the universal property of the colimit.
\end{proof}

The inclusions of each $A_\aa$ into the global algebra $A_\Om$ induce
a map $\zeta_\Om : A_0(X) \aw A_\Om$ on the direct sum, given by:
\begin{equation} \zeta_\Om : h \longmapsto \sum_{\aa \in X} h_\aa \end{equation}
We denote by $B_\Om \incl A_\Om$ its image. 
It follows from the previous proposition that we have a universal map
from $Z_0(X)$ to $B_\Om$, factorising $\zeta_\Om$ through $P$: 
\begin{equation} \bcd 
A_0(X) \rar{P} \drar[swap]{\zeta_\Om}
    &    Z_0(X) \dar[dashed]        \\
    &   B_\Om \incl A_\Om
\ecd \end{equation}
Chosing a global interaction subspace to write $A_\Om = Z_\Om \oplus B_\Om$,
the interaction decomposition theorem asserts that 
$A_\Om = Z_\Om \oplus Z_0(X)$ so that $Z_0(X)$ and $B_\Om$ are isomorphic.
This allows for two different representations of the inductive limit, 
a local one, and a global one. 

\begin{prop}
    The two following maps are equivalent representations of $\colim A$: 
    \begin{enumerate}[label=(\roman*)]
        \item $P : A_0(X) \law Z_0(X)$
        \item $\zeta_\Om: A_0(X) \law B_\Om$ 
    \end{enumerate}
\end{prop}

\subsection{Homology and Cohomology}

The global hamiltonian $H_\Om$ of a physical system is typically given as a
sum of local interactions:
\begin{equation} H_\Om = \sum_{\aa \in X} h_\aa \end{equation}
and we will see that $H_\Om = \zeta_\Om(h)$ 
represents a unique homology
class in $A_0(X)$. 
The interaction decomposition theorem
is essential to compute the homology of $A_\bullet(X)$. 
We shall first characterise the first homology of $A_0(X)$ and 
the first cohomology of $A^*_0(X)$, both isomorphic to $Z_0(X)$, 
before proving the acyclicity of the whole complex $A_\bullet(X)$. 

\begin{thm} \label{H0-P}
    The interaction projection $P : A_0(X) \aw Z_0(X)$ 
    induces an isomorphism in the first homology group of observable 
    fields: 
    \begin{equation} \frac{A_0(X) } {\div A_1(X)} \simeq Z_0(X) \end{equation}
\end{thm}

\begin{proof}
    A Gauss formula on the cone over $\bb \in X$ first ensures 
    that for all $\ph \in A_1(X)$, we have: 
    \begin{equation} P(\div \ph)_\bb = \sum_{\aa' \aw \bb} P^{\bb\aa'}(\div_{\aa'} \ph) 
    = \sum_{\aa' \aw \bb} \sum_{\bb' \not\aw \bb} 
    P^{\bb\aa'}( \ph_{\aa'\bb'} ) = 0 \end{equation}
    as $P^{\bb\aa'}(A_{\bb'})$ is non-zero if and only if $\bb'$ contains $\bb$.
    Hence $P$ vanishes on boundaries and we denote the induced quotient map by $[P]$.
    Reciprocally, given $u \in A_0(X)$ we define $\ph \in A_1(X)$ by
    $\ph_{\aa\bb} = P^{\bb\aa}(u_\aa)$ and consider its boundary: 
    \begin{equation} \div_\bb \ph =
    \sum_{\aa' \aw \bb} \ph_{\aa'\bb} - \sum_{\bb \aw \cc'} \ph_{\bb\cc'} 
    = P(u)_\bb - u_\bb 
    \end{equation}
    When $P(u) = 0$, the above gives $u = - \div \ph$ 
    so that $\Ker(P) = \div A_1(X)$. Hence $[P]$ is injective 
    and induces an isomorphism between $H_0(X) = A_0(X) / \div A_1(X)$ 
    and its image $Z_0(X)$.
\end{proof}

It is now a simple consequence of the previous subsection that 
total energy, viewed as a global observable of $A_\Om$, 
is a maximal homological invariant of $A_0(X)$:

\begin{cor} \label{H0}
    Two observable fields $h, h' \in A_0(X)$ are homologous
    if and only if: 
    \begin{equation} \sum_{\aa \in X} h_\aa = \sum_{\aa \in X} h'_\aa \end{equation}
    and $\zeta_\Om : A_0(X) \aw B_\Om$ induces an isomorphism in homology.
\end{cor}

Theorem \ref{H0-P} implies that 
the first cohomology of $A^*_0(X)$ is also isomorphic to $Z_0(X)$ by duality. 
It is however very instructive to construct the isomorphism explicity, 
as this representation of consistent measures already involves 
a fundamental automorphism $\zeta$ of $A_0(X)$, the {\it zeta transform},
main object of the next chapter. We shall later rewrite theorem \ref{Ker-d} as:  
\begin{equation} \label{zeta-Z0}
    \Ker(d) \:\simeq\: {\color{rouge} \zeta'} \cdot Z_0(X) 
\end{equation} 
{\color{rouge} Erratum:} The action of {\color{rouge} $\zeta'$}
is defined by (\ref{eq:Ker-d}) below. Although it requires to rescale 
the injection of $A_\bb$ into $A_\aa$ by volumic factors, 
it still fits under the definition of zeta transforms given by (\ref{zeta-system}). 
Without the volumic terms, one has 
$\Ker(\nabla) = \zeta \cdot Z_0(X)$ as per (\ref{T_0 Z}), 
where $\nabla$ is also an adjoint of $\div$, but where 
partial integrations have been replaced by conditional expectations 
for the uniform measure.

\begin{thm} \label{Ker-d} 
A collection of measures $(q_\aa) \in A^*_0(X)$ is consistent 
if and only if there exists a collection of interaction potentials 
$(u_\aa) \in Z_0(X)$ such that for all $\aa \in X$: 
    \begin{equation} \label{eq:Ker-d}
        q_\aa = \sum_{\aa \cont \bb} \:
        {\color{rouge} \frac{|E_\bb|}{|E_\aa|}}
        \croc{u_\bb}{-}_\aa 
    \end{equation}
where $\croc{-}{-}_\aa$ denotes the canonical scalar product 
    of $\R^{E_\aa}$.
\end{thm}

\begin{proof} 
    \newcommand{\Err}[2]{{\color{rouge} \frac{|E_{#1}|}{|E_{#2}|}}\:}
    As $\Sigma^{\bb\aa}$ is the orthogonal projection 
    onto $\R^{E_\bb} \incl \R^{E_\aa}$ for the canonical scalar product 
    of $\R^{E_\aa}$, 
    we have $\Sigma^{\bb\aa}(Z_\cc) = Z_\cc$ 
    for every $\cc \incl \bb$ 
    and $\Sigma^{\bb\aa}(Z_\cc) = 0$ otherwise. 
Hence for every $q \in A^*_0(X)$ of the form (\ref{eq:Ker-d}) one has: 
\begin{equation} \Sigma^{\bb\aa}(q_\aa) 
    = \sum_{\aa \cont \cc}  \Sigma^{\bb\aa}\bigg( \Err{\cc}{\aa} u_\cc \bigg) 
= \sum_{\bb \cont \cc} \Err{\cc}{\bb} u_\cc = q_\bb \end{equation}
    Reciprocally, given a consistent $q \in A^*_0(X)$ one may recover 
    $u \in Z_0(X)$ by Möbius inversion, see \ref{moebius-0}. 
    It however suffices here to conclude by dimension using theorem \ref{H0-P}.
\end{proof} 

\newcommand{\pz}{{\bf z}}
\newcommand{\pb}{{\bf b}} 

Let us now prove the acyclicity of $A_\bullet(X)$ by constructing 
an explicit retraction of $A_\bullet(X)$ to $Z_0(X)$. 
For every $\aa \in X$, we denote by $\pz_\aa$ and $\pb_\aa$ 
the coherent projectors onto $Z_\aa$ and $B_\aa$, so that:
\begin{equation} \id_{A_\aa} = \pz_\aa \oplus \pb_\aa \end{equation}
They induce projectors $\pz$ and $\pb$ 
onto the subspaces $Z_\bullet(X)$ and $B_\bullet(X)$ of $A_\bullet(X)$. 

Denoting by $d$ the adjoint of $\div$ for the canonical 
metric of $A_\bullet(X)$, consider the homotopy 
$\eta = \pz \circ d \circ \pb$. 
Explicitly, 
$\eta : A_p(X) \aw A_{p+1}(X)$ acts on $\ph \in A_p(X)$ by: 
\begin{equation} \eta(\ph)_{\bar \aa \bb} = (-1)^{p+1} \, \pz_\bb(\ph_{\bar \aa}) \end{equation}

\begin{prop}
    The map $\eta = \pz \, d \, \pb$ defines a homotopy
    between the identity of $A_\bullet(X)$ and the homogeneous extension
    of the interaction projection $P : A_\bullet(X) \aw Z_0(X)$: 
    \begin{equation} \eta \div + \div \eta = 1 - P \end{equation}
\end{prop}

\begin{proof} 
    First, notice that $B_\bullet(X)$ is stable under $\div$ so that
    $\pz \div \pb = 0$ and $\div$ splits in the triangular form:
    \begin{equation} \div = \pb \div \pb + \pb \div \pz + \pz \div \pz \end{equation}
    As $\eta =\pz d \pb$, we then have: 
    \begin{equation} \div \eta + \eta \div =
    (\pb \div  \pz) \eta  + \eta (\pb \div \pz) 
    + \pz (\div \eta  + \eta \div) \pb 
    \end{equation}

    As maps between $Z_\bullet(X)$ and $B_\bullet(X)$, 
    the interaction theorem implies that $\pb \div \pz$
    inverts $\eta$ on non-zero degrees. More precisely, 
    for every $p \geq 1$ and $\ph \in A_p(X)$ we have:
    \begin{equation} ( \pb \div \pz )(\ph)_{\bar \bb} = (-1)^{p}
    \sum_{\bar \bb \contst \cc'} \pz_\cc(\ph_{\bar \bb \cc'}) \end{equation}
    Denoting by $\pz_0$ the homogeneous projection onto $Z_0(X)$
    induced by $\pz$, we may write: 
    \begin{equation} (\pb \div \pz) \eta = \pb \txt{and} 
    \eta (\pb \div \pz) = \pz - \pz_0 \end{equation}
    so that $(\pb \div \pz) \eta + \eta (\pb \div \pz) = 1 - \pz_0$ 
    and it only remains to show that 
    $\pz(\div \eta + \eta \div) \pb = \pz_0 - P$. 

    On the zero degree, 
    we have for all $u \in A_0(X)$ and $\bb \in X$:
    \begin{equation} (\pz \div \eta)(u)_\bb = - \sum_{\aa' \contst \bb} \pz_\bb(u_{\aa'}) \end{equation}
    which is precisely $\pz_0(u)_\bb - P(u)_\bb$. 

    On higher degrees, we have for every $\ph \in A_p(X)$ and 
    $\bar \bb \in N_p(X)$: 
    \begin{equation} (\pz \div \eta)(\ph)_{\bar \bb} 
    = (-1)^{p+1} \sum_{k = 0}^p (-1)^k \sum_{\dr_k \bar \aa' = \bar \bb}
    \pz_{\aa'_{p+1}}(\ph_{\aa'_0 \dots \aa'_p})
    \end{equation}
    The last summand $k = p+1$ of  $\div \eta (\ph)_{\bar \bb}$ is
    valued in $B_{\bar \bb}$ and truncated by $\pz_{\bb_p}$
    which enforces $\aa'_{p+1} = \bb_p$.
    On the other hand, we have:
    \begin{equation} (\eta \div \pb)(\ph)_{\bar \bb} = 
    \pz_{\bb_p}\bigg( (-1)^p 
    \sum_{k = 0}^p (-1)^k \sum_{\div_k \bar\bb' = \bb_0 \dots \bb_{p-1}}
    \pb_{\bb'_p}(\ph_{\bar \bb'}) \bigg) 
    \end{equation}
Whenever $\bb_p \not \inclst \bb'_p$, 
we have $\pz_{\bb_p} \circ \pb_{\bb'_p} = 0$.
We may hence assume that $\bb_p \inclst \bb'_p$ and 
let $\bar \aa' = \bar \bb'\bb_p$ so that for every $k \leq p$ we have:
    \begin{equation} \div_k \bar \aa' = (\div_k \bar \bb') \bb_p =  \bar \bb \end{equation}
Comparing with the above, we see that 
$\pz \div \eta + \eta \div \pb$ vanishes on non-zero degrees, 
while it is equal to $\pz \div \eta = \pz_0 - P$ otherwise,
which finishes to show that:
\begin{equation} \div \eta + \eta \div = 1 - \pz_0 + \pz_0 - P = 1 - P \end{equation}
\end{proof}

\begin{thm} 
    The homology groups of $A_\bullet(X)$ are given by: 
    \begin{equation} H_0(X) \simeq Z_0(X) \txt{and}  H_n(X) = 0 \txt{for} n \geq 1 \end{equation}
\end{thm}

\chapter{Combinatorics}

In this chapter, we review the classical theory of Dirichlet
convolution and Möbius inversion.  They give algebraic foundations  
to many combinatorial problems, starting with the so-called 
inclusion-exclusion principles, generalising the usual formula
$|E_i \cup E_j| = |E_i| + |E_j| - |E_i \cap E_j|$.
The following structures however
really emerged from considerations of number theory.

\begin{center}
\includegraphics[width=0.3\textwidth]{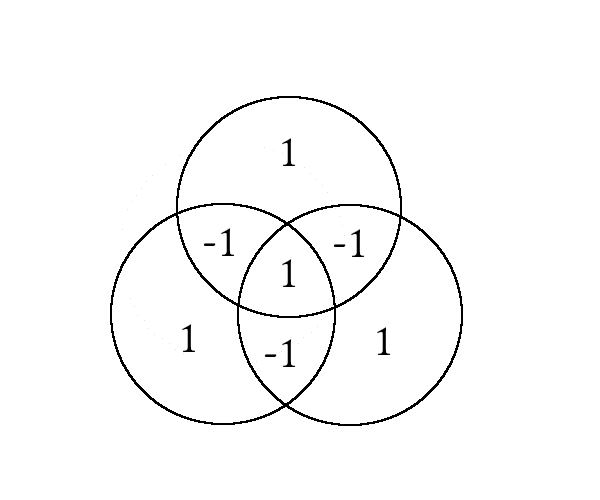}
\end{center}

These methods give 
insight on the conceptual distinction between 
intensive and extensive quantities, stemming from the definition 
internal hamiltonian of a region $\aa \incl \Om$ as a sum of local interactions: 
\begin{equation} H_\aa = \sum_{\aa \cont \bb} h_\bb \end{equation}
The extensivity of $H$ may be thought of as a consequence of the 
quick decay of $h$ on large regions\footnote{
    This remark shall be made more precise in the next chapter, where
    the extensivity of entropy {\it a priori} justifies
    the Bethe approximation scheme. 
}. 
Möbius inversion 
will show the relation between $H$ and $h$ to be bijective, 
a fundamental correspondence which will be written as:
\begin{equation} H = \zeta \cdot h \quad\eqvl\quad h = \mu \cdot H \end{equation}

In section 1, we start by reviewing 
the classical theory of incidence algebra. 
In section 2 we investigate some interesting properties and 
exhibit a locally cohomological character of the zeta transform 
$h \mapsto \zeta \cdot h$. 
In the last and exploratory section, we propose an extension of $\zeta$ and $\mu$ 
as two reciprocal endomorphisms on the whole complex of 
observable fields $A_\bullet(X)$. This extension of $\zeta$ to higher degrees 
is done consistently with the local cocyle properties satisfied on the zero degree,
while the degree one Möbius transform will allow for various 
enhancements of the belief propagation algorithm.

\section{Dirichlet Convolution}

The convolution product on the set 
of arithmetic functions $f : \N^* \aw \C$ was first introduced
by Dirichlet in 1837 to prove his theorem on
arithmetic progressions. Relying on the 
partial order structure of $\N^*$ induced by divisibility,
it is given by:
\begin{equation} (f * g)(n) = \sum_{d | n} f(d) \, g(n/d) \end{equation}
A meromorphic function $\hat f$ may be associated to every arithmetic function 
$f$ of subpolynomial growth 
by analytic extension of its Dirichlet generating series, defined for
all $s \in \C$ such that ${\rm Re}(s)$ is large enough by:
\begin{equation} \hat f (s) = \sum_{n \in \N^*} \frac{f(n)}{n^s} \end{equation}
The assignment $f \mapsto \hat f$ 
interlaces Dirichlet convolution with the usual product on $\C$
and we have for every arithmetic functions $f$ and $g$:
\begin{equation} \widehat{f * g} = \hat f \cdot \hat g \end{equation}
The Riemann zeta function $\hat \zeta$ extends the 
generating series of the arithmetic function $\zeta$,
defined by $\zeta(n) = 1$ for all $n \in \N^*$. 
The classical Möbius inversion formula states 
that $\zeta$ has an inverse for $*$ given by $\mu(n) = (-1)^k$ 
if $n$ is the product of $k$ distinct primes, and zero otherwise.
This in turn yields the coefficients of the generating series of
$\hat \mu = 1/ \hat \zeta$.

Dirichlet convolution and Möbius inversion formulas were then studied in 
more general contexts, notably by Fréchet, and a fundamental and seminal reference 
on the subject is the general treatment of Möbius functions 
given by Rota in $\cite{Rota-64}$.
Dirichlet convolution may not only be defined on a locally finite partial
order but on any locally finite category. We however restrict to the former
and refer the interested reader to $\cite{Leinster-08}$ for greater
generality.

\subsection{Incidence Ring}

Given a locally finite\footnote{
    X is locally finite if the set of non-degenerate 
    chains from any element to an other is finite. 
}
partially ordered set $X$ 
and a unital ring $R$, 
let us denote by $R_1(X)$ the $R$-module 
of $R$-valued functions on the set $ N_1 X$ of $1$-chains in $X$: 
\begin{equation}  R_1(X) = 
\big\{\, \ph : X \times X \aw R \st \aa \not\aw \bb \impq \ph_{\aa\bb} = 0 \,\big\}
= R^{ N_1 X} \end{equation}
It is equipped with the Dirichlet convolution product $*$ defined 
for every $\ph, \psi \in  R_1(X)$ by:
\begin{equation} (\ph * \psi)_{\aa\cc} = \sum_{\aa \aw \bb' \aw \cc}
\ph_{\aa\bb'} \cdot \psi_{\bb'\cc} \end{equation}
This product is associative. It has a unit ${\bf 1}$, which is
the Kronecker symbol\footnote{
    We prevent from using the usual Kronecker symbol
    notation $\delta$ to avoid confusion
    with codifferentials.
} 
defined by ${\bf 1}_{\aa\aa} = 1$ 
and ${\bf 1}_{\aa\bb} = 0$ if $\aa \neq \bb$. 

The unital ring  $\big(  R_1(X) , + , * \big)$ 
is called the incidence ring of $X$, or incidence algebra 
when $R$ is a field.
Combinatorics are however mostly interested 
with relative integers and it is remarkable that Möbius inversion 
can be carried out on $\Z$.

\subsection{Incidence Modules}

Let $M$ denote a module on $R$, and denote by $ M_1(X)$ 
the space of $M$-valued functions on $ N_1 X $. 
Then $ M_1(X)$ is a module on the incidence ring $ R_1(X)$ 
for the left action defined for $m \in  M_1(X)$ by: 
\begin{equation} (\ph \cdot  m)_{\aa\cc} = \sum_{\aa \aw \bb' \aw \cc}
\ph_{\aa\bb'} \cdot m_{\bb'\cc} \end{equation}
when $M$ is a left $R$-module, and symetrically, for the right action:
\begin{equation} (m \cdot \psi)_{\aa\cc} = \sum_{\aa \aw \bb' \aw \cc}
m_{\aa\bb'} \cdot \psi_{\bb'\cc} \end{equation}
when $M$ is a right $R$-module.

Let us now denote by $M_0(X)$ the space of $M$-valued functions on $X$. 
It is also equipped with a left action of $ R_1(X)$ defined for 
every $n \in M_0(X)$ by:
\begin{equation} (\ph \cdot n)_\aa = \sum_{\aa \aw \bb'} \ph_{\aa\bb'} \cdot n_{\bb'} \end{equation} 
when $M$ is a left module, and otherwise with the right action:
\begin{equation} (n \cdot \psi)_\bb = \sum_{\aa' \aw \bb} n_{\aa'} \cdot \psi_{\aa'\bb} \end{equation}

Following Rota, one may relate these two actions and 
append to $X$ initial and terminal elements $0$ and $1$ 
when necessary, defining a possibly larger poset $\bar X$.
To $m \in  M_1(X)$ associate $n \in M_0(X)$ by $n_\aa = m_{\aa 1}$. 
This interlaces the two left actions of $ R_1(\bar X)$ on 
when $M$ is a left module. 
Similarly, letting $n_\bb = m_{0 \bb}$ interlaces the two right actions. 

Notice that when $M = R = \R$, the vector spaces $ R_1(X)$ and $\R_0(X)$ 
have a canonical scalar product for which the left and right actions define 
adjoint endomorphisms.

\subsection{Systems} 

Given a cosystem $R: X^{op} \aw \Ring$, 
we may similarly equip the space of upper $1$-fields 
$\R^1(X)$ with a convolution product. 
There is a ring morphism $R_\aa \wa R_\bb$
for every $\aa \aw \bb$ in $X$, 
and for $r_\bb \in R_\bb$ we denote by $[r_\bb]_\aa$ its image in $R_\aa$.
Then for $\ph, \psi \in  R_1 (X)$, let:
\begin{equation} (\ph * \psi)_{\aa\cc} = \sum_{\aa \aw \bb' \aw \cc}  
\ph_{\aa \bb'} \cdot [\psi_{\bb' \cc}]_\aa 
\end{equation}
By duality, we may similarly define the convolution product $*$ on 
the space of lower fields $R_1(X)$  
when $R: X \aw \Ring$ is a system of rings.

When $M : X^{op} \aw \Ring$ is a module over the cosystem $R$, 
we may extend the left action of $ R^1(X)$
on the space of upper $1$-fields $ M^1(X)$. 
We have a morphism $M_\aa \wa M_\bb$ for every $\aa \aw \bb$ in $X$, 
and for every $m_\bb \in M_\bb$, we still denote by $[m_\bb]_\aa$ its
image in $M_\aa$. The action is then given by:
\begin{equation} (\ph \cdot \psi)_{\aa\cc} = \sum_{\aa \aw \bb' \aw \cc}  
\ph_{\aa \bb'} \cdot [m_{\bb' \cc}]_\aa 
\end{equation}
And a left action of $ R^1(X)$ can similarly be defined on 
$M_0(X)$ by:
\begin{equation} (\ph \cdot m)_\aa = \sum_{\aa \aw \bb'} 
\ph_{\aa\bb'} \cdot  [m_{\bb'}]_\aa \end{equation}
The action of $ \Z_1(X)$ on the 
module $A_0(X)$ of observable fields,
falls into this case. It is the main example we shall be interested in,
before considering functors  
$M$ assigning to each $M_\aa$ the space of functionals on $A_0(\LL^\aa)$
or on the set $\Sigma_\aa$ of states of $A_\aa$.

If $R$ is a fixed ring, 
then $ R^1(X) =  R_1(X)$ has a natural action on 
the space of lower fields $M_1(X)$:
\begin{equation} (\ph \cdot m)_{\aa\cc} = \sum_{\aa \aw \bb'} 
\ph_{\aa\bb'} \cdot m_{\bb' \cc} \end{equation} 
Although functoriality of $M$ does not appear, this case is worth mentioning
as it covers for instance
the action of $ \Z_1(X)$ on $A_1(X)$. 

Dual constructions take place when $M$ is a covariant functor. 
In particular, this is the case when $M = A^*$ is the module of linear 
forms on observables,
and the right action of $ \Z_1(X)$ on $A^*_0(X)$ is the adjoint 
of its left action on $A_0(X)$.

\section{Properties of the Zeta Transform}

The zeta function $\zeta \in  \Z_1(X)$ 
is defined by $\zeta_{\aa\bb} = 1$ for all $\aa \aw \bb$ in $X$. 
Given a cosystem of abelian groups $M : X^{op} \aw \Ab$, 
the left action of $\zeta$ on $M_0(X)$ is given by:
\begin{equation} \label{zeta-system}
(\zeta \cdot f)_\aa = \sum_{\aa \aw \bb} [f_{\bb}]_\aa 
\end{equation}
This section exposes elementary properties
of the endomorphism thus defined on $M_0(X)$. 
The first one, bijectivity, is given by the classical Möbius inversion 
formulas. The next one, locality, is also significant however obvious it is. 
As an easy consequence of the previous chapter, 
we then introduce a local cocycle property satisfied by the zeta transform.

\subsection{Möbius Inversion} 

The fundamental theorem of Möbius inversion
states that $\zeta$ has an inverse $\mu \in  \Z_1(X)$ for Dirichlet convolution,
it naturally implies the bijectivity of the zeta transform on any module.

\begin{thm}[Möbius - Rota] \label{moebius-0}
    The inverse of $\zeta \in \Z_1(X)$ is 
    the Möbius function $\mu \in \Z_1(X)$: 
    \begin{equation} \mu = \sum_{k \geq 0} (-1)^k (\zeta - {\bf 1})^k \end{equation}
\end{thm}
    
\begin{proof}
The $n$-th power of $\zeta$ for $*$ counts the number of $n$-chains 
between any two elements $\aa, \bb \in X$:
\begin{equation} \zeta^n_{\aa\bb} = \sum_{\aa = \aa_0 \aw \dots \aw \aa_n = \bb} 1 \end{equation} 
and $(\zeta - {\bf 1})^n_{\aa\bb}$ similarly 
counts the number of non-degenerate $n$-chains 
from $\aa$ to $\bb$. Because $X$ is locally finite, there exists
$N$ such that $(\zeta - {\bf 1})^N_{\aa\bb} = 0$
and $\zeta = {\bf 1} + (\zeta - {\bf 1})$ is invertible.
\end{proof}

Note that more practical identities 
arise from the relations $\mu * \zeta = \zeta * \mu = {\bf 1}$,
allowing to compute the values of $\mu$ inductively. 
Starting from $\mu_{\aa\aa} = 1$, 
one may for instance iterate over $\cc$ under $\aa$:
\begin{equation} \mu_{\aa\cc} = 1 - \sum_{\aa \geq \bb' > \cc} \mu_{\aa\bb'} \end{equation} 
The Möbius function $\mu$ is closely related to diverse inclusion-exclusion
principles. 
Consider for example the problem of finding a collection of integers 
$(c_\aa) \in \Z_0(X)$ such that for all $\bb \in X$: 
\begin{equation} \sum_{\aa' \aw \bb} c_{\aa'} = 1 \end{equation}
One may compute these coefficients inductively starting from maximal cells.
Note that this expression involves the right action of $\zeta$,
and as by Möbius inversion 
$c \cdot \zeta = 1$ is equivalent to $c = 1 \cdot \mu$ 
we have: 
\begin{equation} c_\bb = \sum_{\aa' \aw \bb} \mu_{\aa' \bb} \end{equation}
We call $c = 1 \cdot \mu$ 
and $\bar c = \mu \cdot 1$  
the right and left Möbius numbers. 

Given $X \incl \tilde X$, let us say that $X$ 
is {\it full} if for all $\aa \in X$, every $\bb \in \tilde X$ 
such that $\bb \incl \aa$ is also in $X$. 

\begin{prop} 
    If $X$ is full in $\tilde X$, 
    then the restriction of the Möbius function of $\tilde X$ 
    to $X$ is the Möbius function of $X$. 
\end{prop}

\begin{prop} \label{c-mu} 
If $\tilde X = \{ \Om \} \sqcup X$ completes $X$ with an initial element $\Om$
and $\tilde \mu \in \Z_1(\tilde X)$ 
denotes the Möbius function on $\tilde X$, then 
for every $\bb \in X$ we have $c_\bb = - \tilde \mu_{\Om \bb}$. 
\end{prop}

\begin{proof} 
The Möbius inversion formula in $\tilde X$ between $\Om$ and $\bb \in X$ 
    gives:
\begin{equation} \sum_{\Om \aw \aa \aw \bb} \tilde \mu_{\aa\bb} =
    \tilde \mu_{\Om\bb} + \sum_{\aa \aw \bb} \mu_{\aa\bb} = 
    \tilde \mu_{\Om \bb} + c_\bb = 0 \end{equation}
\end{proof}

{\bf Examples.} 
\begin{enumerate}[itemsep=3px]
\item Let $X = \N$ denote the standard total order of integers.
     The Möbius function is given by
    $\mu_{\aa\aa} = 1$ and $\mu_{\aa\bb} = -1$ if $\bb = \aa + 1$, zero otherwise.
We recover the classical finite differences formula:
\begin{equation} U_\aa = \sum_{\bb = 0}^\aa u_\bb 
\quad \eqvl \quad 
u_\aa = U_{\aa} - U_{\aa - 1} \end{equation}
which is the discrete version of the fundamental theorem of calculus.

\item Let $X = \Part(\Om)$ for a finite set $\Om$ and denote by $|\aa|$ the 
    cardinal of $\aa \incl \Om$. The Möbius function is then given by: 
    \begin{equation} \mu_{\aa\bb} = (-1)^{|\aa| - |\bb|} \end{equation}
    Suppose given now for each $i \in \Om$ 
    a measurable subset ${\cal A}_i$ of some probability space, relate to:
    \begin{equation} \P(\cup_{i\in \aa} {\cal A}_i) = 
    \sum_{\bb \incl \aa} (-1)^{|\bb|} \P(\cap_{j \in \bb} {\cal A}_j) \end{equation}

\item Total interaction\footnote{
        See the interaction decomposition theorem \ref{int-decomposition} 
        for the definition of $P^{\bb\aa}$. 
    }$P^{\bb\aa}$ and marginal projections $\Sigma^{\bb\aa}$. 
    \begin{equation} \Sigma^{\bb\aa} = \sum_{\bb \aw \cc'} P^{\cc'\aa} \quad\eqvl\quad 
        P^{\bb\aa} = \sum_{\bb \aw \cc'} \mu_{\bb\cc'} \cdot \Sigma^{\cc'\aa} \end{equation}
    This is a property of the canonical metric and the marginal projections,
        beware that in general:
        \begin{equation} \E^{\bb\aa} \neq \sum_{\bb \aw \cc'} P^{\cc'\aa} \end{equation}
        as $\bigoplus_{\cc \not\in \LL^\bb} Z_\cc$ might not be the orthogonal 
        to $A_\bb$ in $A_\aa$. 
\end{enumerate}

\subsection{Locality}

For every $\aa \in X$, we denote by 
$\LL^\aa$ the cone under $\aa$ consisting of 
those $\bb \in X$ such that $\aa \aw \bb$. 
As subset of $X$, there is a natural restriction map: 
\begin{equation} {\rm r}_\aa : M_0(X) \law M_0(\LL^\aa) \end{equation}
For every $f \in M_0(X)$
the value of $(\zeta \cdot f)_\aa$ only depends on the values
$f_\bb$ for $\aa \aw \bb$.  
In other words, $\zeta$ commutes with the restriction to $\LL^\aa$
and we have ${\rm r}_\aa \circ \zeta = \zeta \circ {\rm r}_\aa$.
\begin{equation} \bcd 
M_0(X) \rar{{\rm r}_\aa} \dar[swap]{\zeta} & M_0(\LL^\aa) \dar{\zeta} \\
M_0(X) \rar[swap]{{\rm r}_\aa} & M_0(\LL^\aa) 
\ecd \end{equation}
This locality property comes from the form of the action
of $\Z_1(X)$ on $M_0(X)$ and is absolutely not specific to $\zeta$. 
Such local endomorphisms preserve the sheaf structure\footnote{
    For the restriction maps $M_0(Y) \aw M_0(Z)$ induced by every 
    $Z \incl Y \incl X$. 
}
of $M_0(X)$.

Denoting by $i_\aa : M_0(X) \aw M_\aa$ the evaluation on $\aa$, 
a consequence of the locality of $\zeta$ is that 
we may factorise $i_\aa \circ \zeta$ through ${\rm r}_\aa$:
\begin{equation} \bcd
M_0(X) \rar{{\rm r}_\aa} \drar[swap]{i_\aa \zeta} 
& M_0(\LL^\aa) \dar[dashed]{\zeta_\aa} \\
    &   M_\aa
\ecd \end{equation}
We denote by $\zeta_\aa$ the factorised map,  
and will now show that $\zeta_\aa$ induces a map in homology.

\subsection{Local Cocycle Property}

In this paragraph, we relate the Gauss formula \ref{gauss} 
to properties of the zeta transform, 
fortifying its analogy with integral calculus underlined by Rota. 
Recall that the coboundary of $\LL^\aa$ 
was defined as the set of $1$-chains 
$d\LL^\aa = 
        \big\{ \aa'\bb' \in N_1(X)  
        \: \big| \: 
        \bb' \in \LL^\aa 
        {\rm\;and\;} \aa' \not\in \LL^\aa \big\} 
$ 
so that for every $\ph \in M_1(X)$:

\begin{center}
    \vspace{-0.2cm}
    \begin{tabular}{ccc}

    $ \disp \sum_{\bb' \in \LL^\aa} \div_{\bb'} \ph = 
    \sum_{\aa'\bb' \in d\LL^\aa} \ph_{\aa'\bb'} 
    $

    & 
    \hspace{0.5cm}
    &

    \raisebox{-.5\height}{
    \includegraphics[width=0.25\textwidth]{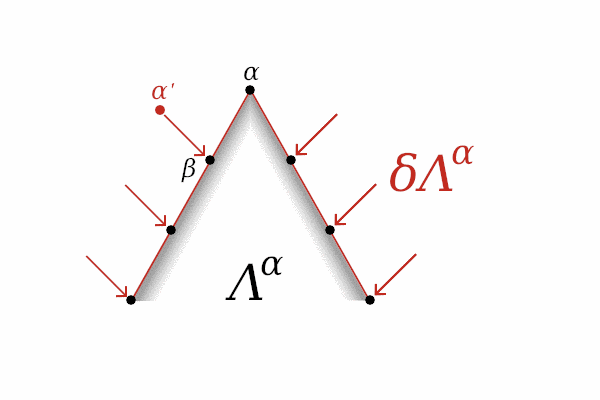}
    }
    \end{tabular}
\end{center}

In chapter 5, we shall rely on this Gauss formula to relate belief propagation
to a transport equation. 
Its right hand side may also be related to the zeta transform, as expressed by 
the following equivalent proposition, 
while extending the zeta transform to $M_1(X)$ will allow 
for the even more natural formula (\ref{zeta-div-zeta-Om})
of the next section. 

\begin{prop} \label{zeta-cocycle}
    For every $\ph \in M_1(X)$, we have:
    \begin{equation}
    \zeta(\div \ph)_\aa = \sum_{\aa' \not\in \LL^\aa} 
    \zeta(i_{\aa'} \ph)_\aa 
    \end{equation}
    In particular, $i_{\aa}\zeta$ vanishes on $\div M_1(\LL^\aa)$.
\end{prop}

It follows from proposition \ref{zeta-cocycle}
property that $\zeta_\aa$ factorises through 
the canonical projection of $M_0(\LL^\aa)$ onto its quotient 
by $\div M_1(\LL^\aa)$:
\begin{equation} \bcd
M_0(\LL^\aa) \rar \ar[dr, swap, "\zeta_\aa"] 
& \disp \frac{M_0(\LL^\aa)}{\div M_1(\LL^\aa)} \dar[dashed] \\
      & M_\aa
\ecd \end{equation}
and in particular, $\zeta_\aa$ induces a map in homology 
$[\zeta_\aa] : H_0(\LL^\aa) \law M_\aa$.
In the case of the complex of local observables $A_\bullet(X)$   
defined in \ref{complexes}, 
it follows from
theorem \ref{H0} that $[\zeta_\aa]$ actually induces an isomorphism between the 
homology classes of $A_0(\LL^\aa)$ and its image $B_\aa \incl A_\aa$.
This fact is a consequence of the interaction decomposition theorem, 
and in general, $[\zeta_\aa]$ may fail to be injective.

\section{Higher Degree Combinatorics}

In this section, we extend the zeta transform to a full endomorphism of 
the chain complex $M_\bullet(X)$ satisfying 
higher degree analogs of the previous section properties. 
Our initial motivation was to search for a higher degree diffusion equation
that would generalise belief propagation. This program appeared only feasible for 
the linearised algorithm, although it soon became clear to us that 
the degree 1 zeta transform was already hidden in the degree 0 algorithm 
with Dirichlet boundary conditions, and that modifying the degree 0 algorithm 
to perform a degree 1 Möbius inversion on its messages seemed 
deeply natural to eliminate their redundancies.

In proving the invertibility of our extension of $\zeta$,  
we will need to assume that $X \incl \Part(\Om)$ is closed 
under intersection\footnote{ 
    The same hypothesis was necessary for the interaction decomposition
    to hold, section 2.3.2.
}
and forms a semi-lattice. 
For every $\bb \in X$ and $\aa_0 \dots \aa_n \in N_n(X)$, 
we denote by $\bb \cap (\aa_0 \dots \aa_n)$ the possibly degenerate
chain $(\bb \cap \aa_0) \dots (\bb \cap \aa_n)$, 
while it will be implicit in our notations 
that fields vanish on non-ordered sequences.

\subsection{Extended Zeta Transform} 

For every $\aa \in X$ and $\bb \incl \aa$, we denote 
by $\LL^\aa_\bb$ the complement of $\LL^\bb$ in $\LL^\aa$.
When $\ph \in M_1(X)$ is of degree one, we shall define 
$\zeta(\ph)_{\aa\bb}$ 
as the flux passing from $\LL^\aa$ to $\LL^\bb$:
\begin{equation} \zeta(\ph)_{\aa\bb} = \sum_{\bb' \in \LL^\aa_\bb} \: \sum_{\cc' \in \LL^\bb} 
\ph_{\bb'\cc'} \end{equation}
More generally, a chain $\aa_0 \dots \aa_n$ in $N_n(X)$ yields 
a sequence of cones $\LL^{\aa_0} \cont \dots \cont \LL^{\aa_n}$ 
and we simply define $\zeta$ by summing in the interspaces.

\begin{defn} 
    We extend the zeta transform to
    $\zeta : M_\bullet(X) \aw M_\bullet(X)$ 
    by letting for every $\ph \in M_n(X)$: 
\begin{equation} \zeta(\ph)_{\aa_0 \dots \aa_n} 
= \sum_{\bb_0 \in \LL^{\aa_0}_{\aa_1}} \:
\sum_{\bb_1 \in \LL^{\aa_1}_{\aa_2}}
\dots 
\sum_{\bb_n \in \LL^{\aa_n}}
\ph_{\bb_0 \dots \bb_n}
\end{equation}
\end{defn}

From the above formula, it clearly appears that 
the definition of $\zeta$ involves an inductive 
extension to higher degrees and the following proposition will 
greatly ease computations.

\begin{prop} 
    For every $n \geq 1$, the action of $\zeta$ on $M_n(X)$ is related 
    to that on $M_{n-1}(X)$ by:
\begin{equation} 
\zeta(\ph)_{\aa_0 \dots \aa_n} = \sum_{\bb_0 \in \LL^{\aa_0}_{\aa_1}}
\zeta(i_{\bb_0} \ph)_{\aa_1 \dots \aa_n} 
\end{equation} 
\end{prop}

\subsection{Locality and Colocality} 

For any $\aa \in X$, the zeta transform still
commutes with the restriction to $\LL^\aa$:
\begin{equation} \bcd 
M_\bullet(X) \rar{{\rm r}_\aa} \dar[swap]{\zeta} & M_\bullet(\LL^\aa) \dar{\zeta} \\
M_\bullet(X) \rar[swap]{{\rm r}_\aa} & M_\bullet(\LL^\aa) 
\ecd \end{equation}
The locality of $\zeta$ is better precised by the following proposition,
that makes use of the semi-lattice structure of $X$ and will soon come to use 
in extending Möbius inversion. 

\begin{prop} \label{zeta-cap}
    For every $\bb \in X$ and $\aa_0 \dots \aa_n \in N_n(X)$, we have: 
    \begin{equation} \zeta({\rm r}_\bb \ph)_{\aa_0 \dots \aa_n} 
    = \zeta(\ph)_{\bb \cap(\aa_0 \dots \aa_n)} \end{equation}
\end{prop}

\begin{proof}
    As ${\rm r}_\bb \ph$ is supported inside 
    $\LL^\bb$, the sums over $\LL^{\aa_i}_{\aa_{i+1}}$ 
    may be restricted to 
    $\LL^\bb \cap \LL^{\aa_i}_{\aa_{i+1}}
= \LL^{\bb \cap \aa_i}_{\bb \cap \aa_{i+1}}$.
\end{proof}

Locality also implies that for every $n \geq 1$, composing
the partial evaluation $i_{\aa_0} : M_n(X) \aw M_{n-1}(\LL^{\aa_0})$ 
with $\zeta$ factorises through ${\rm r}_{\aa_0}$:
\begin{equation} \bcd
M_n(X) \rar{{\rm r}_{\aa_0}} \drar[swap]{i_{\aa_0} \zeta} 
    & M_n(\LL^{\aa_0}) \dar[dashed] \\
    & M_{n-1}(\LL^{\aa_0})  
\ecd \end{equation}
A more subtle observation on the support of $\zeta$ is given by 
the following dual property, 
expressing the independence of $i_{\aa_1} i_{\aa_0} \zeta$ with 
respect to fields supported inside $\LL^{\aa_1}$. 

\begin{prop} \label{zeta-colocality}
    For every $\aa_0 \dots \aa_n \in N_n(X)$ and $\ph \in M_n(\LL^{\aa_1})$,
    we have $\zeta(\ph)_{\aa_0 \dots \aa_n} = 0$.
\end{prop}

The sum defining $\zeta(\ph)_{\aa_0 \dots \aa_n}$ 
indeed only 
involves evaluations in the first argument of $\ph$ inside 
$\LL^{\aa_0}_{\aa_1} = \LL^{\aa_0} \setminus \LL^{\aa_1}$ and 
does not depend on the restriction of $\ph$ to $\LL^{\aa_1}$.
The kernel of $i_{\aa_1}i_{\aa_0}\zeta$ hence contains 
$M_n(\LL^{\aa_1})$ and we have the following factorisation:
\begin{equation} \bcd
M_n(X) \rar \drar[swap]{i_{\aa_1} i_{\aa_0} \zeta} 
    & \disp \frac {M_n(\LL^{\aa_0})}{M_n(\LL^{\aa_1})}  \dar[dashed] \\
    & M_{n-2}(\LL^{\aa_1})  
\ecd \end{equation}

\subsection{Extended Möbius Transform}

The definition of $\mu$ will make use of the semi-lattice structure of $X$.
For every unordered sequence 
$\bb_0, \dots, \bb_n \in X$, 
let us denote by $[\bb_0 \dots \bb_n]_\cap$
the $n$-chain $\bb_0(\bb_0 \cap \bb_1) \dots (\bb_0 \cap \dots \cap \bb_n)$.

\begin{defn}
    We extend the Möbius transform to $\mu: M_\bullet(X) \aw M_\bullet(X)$ 
    by letting for all $\Phi \in M_n(X)$: 
\begin{equation} \mu(\Phi)_{\aa_0 \dots \aa_n} = 
\sum_{\bb_n \in \LL^{\aa_n}} \mu_{\aa_n \bb_n}
\dots 
\sum_{\bb_1 \in \LL^{\aa_1}_{\bb_2}} \mu_{\aa_1 \bb_1} 
\sum_{\bb_0 \in \LL^{\aa_0}_{\bb_1}} \mu_{\aa_0 \bb_0} 
    \cdot \Phi_{[\bb_0 \dots \bb_n]_\cap} 
\end{equation}
\end{defn}

In contrast with $\zeta$, the supports of the sums defining $\mu$
do depend on the summation variables $\bb_i$.
We may still give an inductive construction of $\mu$, 
allowing to conveniently prove reciprocity with $\zeta$.
For every $\aa_0 \in X$ and $n \geq 1$, let us introduce the map
$\nu_{\aa_0} : M_n(X) \aw M_{n-1}(X)$ defined by:
\begin{equation} \label{nu}
\nu_{\aa_0}(\Phi)_{\aa_1 \dots \aa_n} = 
\sum_{\bb_0 \in \LL^{\aa_0}_{\aa_1}}
    \mu_{\aa_0 \bb_0} \cdot \Phi_{\bb_0 \cap (\aa_0\dots \aa_n)} 
\end{equation}
and extend this map to $\nu_{\aa_0}: M_0(X) \aw M_{\aa_0}$ 
by letting:
\begin{equation} \label{nu0} 
\nu_{\aa_0}(\Phi) = \sum_{\bb_0 \in \LL^{\aa_0}}
\mu_{\aa_0 \bb_0} \cdot \Phi_{\bb_0} 
\end{equation}
The action of $\mu \in M_n(X)$ then deduces from that on $M_{n-1}(X)$
by the relation 
$i_{\aa_0} \circ \mu = \mu \circ \nu_{\aa_0}$ as expressed by the following
proposition.

\begin{prop} \label{moebius-rec}
For every $\aa_0 \dots \aa_n \in N_n(X)$ and $\Phi \in M_n(X)$, we have:
\begin{equation} \mu(\Phi)_{\aa_0 \dots \aa_n} = 
\nu_{\aa_n} \dots \nu_{\aa_0} (\Phi) \end{equation}
\end{prop}

\begin{proof} 
In the last sum over $\bb_0 \in \LL^{\aa_0}_{\bb_1}$
defining $\mu(\Phi)_{\aa_0 \dots \aa_n}$,
we may recognise $\nu_{\aa_0}(\Phi)_{[\bb_1 \dots \bb_n]_\cap}$ 
\end{proof}

With this characterisation of $\mu$, we may now 
generalise the Möbius inversion formula to $M_\bullet(X)$.
Having two reciprocal endomorphisms $\zeta$ and $\mu$ acting on all degrees 
will allow to conjugate non-homogeneous operators and consider\footnote{
    See chapter 5.
}
for instance 
$\div^\zeta = \zeta \circ \div \circ \mu$ or 
$\nabla^\mu = \mu \circ \nabla \circ \zeta$.

\begin{thm} \label{moebius}
The Möbius transform is the inverse of the zeta transform. 
\end{thm}

The proof\footnotemark{} of the inversion theorem will come as an easy consequence
of the following lemma:
\footnotetext{
    [[appendix: add explicit computation on $A_1(X)$.]]
}

\begin{lemma}
    For every $\aa_0$ in $X$, we have:
    \begin{equation} 
    \nu_{\aa_0} \circ \zeta = \zeta \circ i_{\aa_0}
    \end{equation}
    The above extends to $M_0(X) \aw M_{\aa_0}$ 
    by agreeing that $\zeta \circ i_{\aa_0} = i_{\aa_0}$, 
\end{lemma}

\begin{proof}
Injecting the recurrence relation defining $\zeta$, 
    we may write $\nu_{\aa_0}\zeta(\ph)_{\aa_0 \dots \aa_n}$ as: 
\begin{equation} 
\sum_{\bb_0 \in \LL^{\aa_0}_{\aa_1}} \mu_{\aa_0 \bb_0}
    \cdot \zeta(\ph)_{\bb_0 \cap (\aa_0 \dots \aa_n)} 
        = \sum_{\bb_0 \in \LL^{\aa_0}_{\aa_1}} \mu_{\aa_0 \bb_0}
        \sum_{\cc_0 \in \LL^{\bb_0} \cap \LL^{\aa_0}_{\aa_1} } 
        \zeta(i_{\cc_0} \ph)_{\bb_0 \cap(\aa_1 \dots \aa_n)} \\
\end{equation}
    As $i_{\cc_0} \ph$ is supported in $\LL^{\cc_0}$, we have
    $\zeta(i_{\cc_0} \ph)_{\bb_0 \cap (\aa_1 \dots \aa_n)} 
    = \zeta(i_{\cc_0} \ph)_{\cc_0 \cap (\aa_1 \dots \aa_n)}$
    in virtue of prop.\ref{zeta-cap}.
    Both $\bb_0$ and $\cc_0$ running over $\LL^{\aa_0}_{\aa_1}$
    with the only additional condition that $\bb_0 \aw \cc_0$,
    we have:
    \begin{equation} \nu_{\aa_0}\zeta(\ph)_{\aa_0 \dots \aa_n} 
    = \sum_{\cc_0 \in \LL^{\aa_0}_{\aa_1} }
    \bigg(\: \sum_{\aa_0 \aw \bb_0 \aw \cc_0} \mu_{\aa_0 \bb_0} \:\bigg) \:
    \zeta(i_{\cc_0} \ph)_{\cc_0 \cap (\aa_1 \dots \aa_n)} 
    \end{equation}
    Recognising the classical Möbius inversion formula 
    $(\mu * \zeta)_{\aa_0 \cc_0} = {\bf 1}_{\aa_0 \cc_0}$, we get:
    \begin{equation} \nu_{\aa_0}\zeta(\ph)_{\aa_0 \dots \aa_n} = 
    \zeta(i_{\aa_0} \ph)_{\aa_1 \dots \aa_n} \end{equation}
    When $\ph$ is of degree zero the identity reduces to 
    the $\nu_{\aa_0} \zeta(\ph) = \mu \zeta(\ph)_{\aa_0} = \ph_{\aa_0}$.
\end{proof}

\begin{proof}[Proof of theorem \ref{moebius}]
    For every $\aa_0 \dots \aa_n$ in $N_n(X)$, we have: 
    \begin{equation} \begin{split} 
        i_{\aa_n} \dots i_{\aa_0} \circ \mu \circ \zeta &= 
        \nu_{\aa_n} \dots \nu_{\aa_0} \circ \zeta \\
        &= \zeta \circ i_{\aa_n} \dots i_{\aa_0} \\
        &= i_{\aa_n} \dots i_{\aa_0} 
    \end{split} 
    \end{equation}
\end{proof}

\subsection{Local Cocycle Property}

The following formula is the higher degree analog of 
proposition \ref{zeta-cocycle}.

\begin{prop} \label{higher-zeta-cocycle} 
    For every $\aa_0 \dots \aa_n \in N_n(X)$ and $\psi \in M_{n+1}(X)$ we have: 
\begin{equation} \zeta( \div \psi )_{\aa_0 \dots \aa_n}
    = \sum_{\aa'_0 \not \in \LL^{\aa_0}} 
\zeta(i_{\aa'_0} \psi)_{\aa_0 \dots \aa_n}
\end{equation}
    In particular, if $\psi \in M_n(\LL^{\aa_0})$, 
    then $i_{\aa_0} \zeta(\div \psi) = 0$.
\end{prop}

\begin{proof}
We will prove the relation by induction on the degree $n$. Let us first
recall the module Leibniz rule for the interior product $\idot$, 
in section 2.2 example 2:
\begin{equation} \div \circ i_\bb = i_\bb \circ \Big( \sum_{\aa \in X} i_\aa \Big) 
    - i_\bb \circ \div \end{equation}
For every $\psi \in M_{n+1}(X)$ and by construction of $\zeta$,
we may thus rewrite $\zeta(\div \psi)_{\aa_0 \dots \aa_n}$ as:
    \begin{equation} \begin{split}
\sum_{\bb_0 \in \LL^{\aa_0}_{\aa_1} }
\zeta(i_{\bb_0} \div \psi)_{\aa_1 \dots \aa_n}
&= \sum_{\bb_0 \in \LL^{\aa_0}_{\aa_1} }
   \sum_{\aa'_0 \in X}  
    \zeta(i_{\bb_0} i_{\aa'_0} \psi)_{\aa_1 \dots \aa_n} \\
&- \sum_{\bb_0 \in \LL^{\aa_0}_{\aa_1} }
    \zeta(\div i_{\bb_0} \psi)_{\aa_1 \dots \aa_n}
\end{split} \end{equation}

We may now use the proposition on degree $n-1$ to rewrite the summand 
    of the second term. Because $i_{\bb_0} \psi$ is supported inside 
    $\LL^{\bb_0} \incl \LL^{\aa_0}$, 
    observing that 
    $(X \setminus \LL^{\aa_1}) \cap \LL^{\bb_0}
    = \LL^{\aa_0}_{\aa_1} \cap \LL^{\bb_0}$ 
    we get:
    \begin{equation} \zeta(\div i_{\bb_0} \psi)_{\aa_1 \dots \aa_n} 
    = \sum_{{\aa'_1} \in \LL^{\aa_0}_{\aa_1}} 
    \zeta(i_{\aa'_1} i_{\bb_0} \psi)_{\aa_1 \dots \aa_n} 
    = \zeta(i_{\bb_0} \psi)_{\aa_0 \dots \aa_n} \end{equation}
Swapping the first two sums in the previous expression
    of $\zeta(\div \psi)_{\aa_0 \dots \aa_n}$ we are left with the difference:
    \begin{equation} 
\zeta(\div \psi)_{\aa_0 \dots \aa_n}  
= \sum_{\aa'_0 \in X} 
    \zeta(i_{\aa'_0} \psi)_{\aa_0 \dots \aa_n} \\
- \sum_{\bb_0 \in \LL^{\aa_0}_{\aa_1} }
    \zeta(i_{\bb_0} \psi)_{\aa_0 \dots \aa_n} 
\end{equation}
In virtue of prop.\ref{zeta-colocality}, 
    we have $\zeta(i_{\bb_0} \psi)_{\aa_0 \dots \aa_n} = 0$ 
    for every $\bb_0 \in \LL^{\aa_1}$. 
    We get the desired formula on degree $n$ by 
    rewriting the second sum over $\bb_0 \in \LL^{\aa_0}$.
\end{proof}

For every $\aa_0 \in X$, the local cocyle property similarly 
implies that we have the factorisation: 
\begin{equation} \bcd 
M_n(X) \rar \drar[swap]{i_{\aa_0}\zeta} 
& \disp \frac {M_n(\LL^{\aa_0})} {\div M_{n+1}(\LL^{\aa_0}) } \dar[dashed] 
\\
    & M_{n-1}(\LL^{\aa_0}) 
\ecd \end{equation}
This will come as an essential feature of $\zeta$ when defining 
higher degree transport equations generalising belief propagation.
Letting a field $\ph \in M_n(X)$ evolve up to a boundary term 
$\div \psi$, the partial evaluations $i_{\aa_0} \Phi$ of 
its zeta transform $\Phi = \zeta(\ph)$ 
shall only depend on the 
generalised messages $\psi$ coming from the outside of $\LL^{\aa_0}$.

Note that when $X$ contains a maximal element $\Om$, proposition 
\ref{higher-zeta-cocycle} may be rewritten as: 
\begin{equation}  \label{zeta-div-zeta-Om} 
    \zeta(\div \psi)_{\bar \aa} = \zeta(\psi)_{\Om \bar \aa}  
\end{equation} 
In general, formula (\ref{zeta-div-zeta-Om}) 
could serve as a natural definition for the notation $\zeta(\psi)_{\Om\bar\aa}$. 
One could also define $\tilde X = \{ \Om \} \sqcup X$ 
by prepending $X$ with an initial element 
and extending the module system $M$ by $M_\Om = \colim_{\aa} M_\aa$. 
This point of view will be very useful in understanding the canonical diffusion 
flux in chapter 5 and proving proposition \ref{mu-gauss}. 

\begin{thm}
    We have the commutation relation: 
    \begin{equation} \tilde \zeta \circ \div = i_\Om \circ \tilde \zeta \end{equation}
    so that $\tilde \zeta : M_\bullet(X) \aw M_\bullet(\tilde X)$ defines 
    a morphism of chain complexes between $\big( M_\bullet(X) , \div \big)$ 
    and $\big( M_\bullet(\tilde X),  i_\Om \big)$ on positive degrees,
    while $M_\bullet(\tilde X)$ is naturally extended by $M_{-1}(\tilde X) = M_\Om$.
\end{thm} 

Note that $M_0(X)$ is then naturally isomorphic to the 0-cycles of $M_0(\tilde X)$.
The extended complex $M_\bullet(\tilde X)$ is acyclic as 
$i_\Om \ph = 0$ implies $\ph = i_\Om (e_\Om \ph)$, although 
this require $M_\bullet(X)$ to be acyclic.

\chapter{Energy and Information}

    Many laws of nature remarkably take the form 
of variational principles. Since the lagrangian formulation 
of mechanics they have become a fundamental constituent 
of most physical theories, from thermodynamics to modern quantum field theory. 
As likelihood optimisation problems, they are also now a central 
occupation in data science and artificial intelligence, 
leading one to wonder if such variational principles should not help in 
understanding the self-organisation of biological systems. 
Set aside the theoretical beauty of variational principles 
remains the challenge of designing efficient computations of their solutions,
while local and parallel optimisation algorithms in return 
provide with good abstract models\footnotemark{} for neuronal interactions. 
\footnotetext{
    See for instance \cite{Friston-Parr-2017} 
    for models of active inference 
    and neuronal message-passing 
    relying on belief propagation. 
}

\begin{center}
    \includegraphics[width=0.25\textwidth]{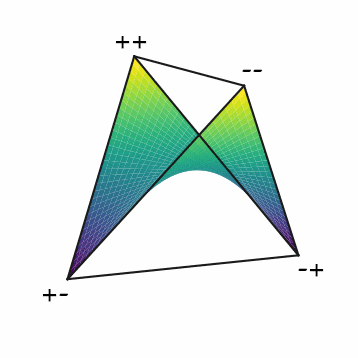}
\end{center}

Entropy generates such variational principles and could be seen as the main 
object of this chapter. 
Legendre duality has long been a classical technique in thermodynamics, 
transforming variational principles into differently constrained ones. 
Considering Legendre duality as a smooth correspondence between 
observables and statistical states 
will motivate the formal study of free energy, 
viewed as a functional of the hamiltonian, given in section 1.
The functorial properties satisfied by effective energy, 
a conditional form of free energy,
will be essential to understand the structure of 
belief propagation in chapter 5.
Giving a fine description of the structure of interactions 
and correlations, we relate mutual informations 
to a combinatorial localisation of entropy in section 2, 
which will be fundamental in understanding the content of Bethe's
local approximation of entropy. 
The main result of this chapter concludes section 3:
we give a homological description of the solutions of 
Kikuchi's cluster variation method, 
which estimates the marginals 
of a global probabilistic model by attempting to solve
a variational principle on a local approximation of entropy.

\section{Free and Effective Energies}

In this section, we first review some fundamental properties of free energy, 
seen as a functional on observables, although it is more customarily defined
as a function of the inverse temperature $\theta =\frac{1} {kT}$ once given 
the hamiltonian $H$ governing the system. 
Free energy is then defined from the partition function 
$Z(\theta) = \sum \e^{- \theta H}$ by letting 
$F(\theta) = -\frac 1 {\theta}\ln Z(\theta)$.
It is well known that both functionals encode many physical properties of the system:
the partition function may be seen as the Laplace transform of the 
spectral density of $H$ and its derivatives 
yield the moments of $H$ in the Gibbs state $[\e^{-\theta H}]$.

Temperature simply acts as an energy scaling and we
here argue that given a region $\aa \in X$, 
free energy is best seen as a smooth function 
$\Fh^\aa \in C^\infty(A_\aa)$ of the hamiltonian itself.
More generally, for every subregion
 $\bb \incl \aa$, partial integration along the fibers of 
$E_\aa \aw E_\bb$ in the sum defining $\Fh^{\aa}$ will yield 
a smooth map $ \Fh^{\bb\aa} \in C^\infty(A_\aa, A_\bb)$ 
which we call effective energy.
We finally show that our definition of free energy simply generates 
Gibbs propability densities
through its differential $\Fh^\aa_* : A_\aa \aw A^*_\aa$,
while the differential of effective energies 
will yield conditional expectations in the Gibbs state.

\subsection{Free Energy}

\begin{defn} For every $\aa \in X$, we call free energy the functional 
$\Fh^\aa \in C^\infty(A_\aa)$ defined by:
\begin{equation} \Fh^\aa(H_\aa) = - \ln \Big( \sum_{E_\aa} \e^{- H_\aa} \Big) \end{equation}
\end{defn}

This definition agrees with the free energy of an isolated system 
governed by the hamiltonian $H_\aa$ 
at inverse temperature $\theta = 1$. 
One may reintroduce the temperature dependency by setting:
\begin{equation} \Fh_\theta^\aa(H_\aa) = \theta^{-1} \Fh^\aa(\theta H_\aa) \end{equation}
A fundamental property of free energy is additivity 
along non-interacting subsystems, which may be thought of
as a weaker form of extensivity. 

\begin{prop}[Additivity] \label{Fh-add}
If $\aa$ is the disjoint union of $\bb_1, \dots, \bb_n$ 
and $H_\aa$ splits as $\sum_i H_{\bb_i}$, then:
\begin{equation} \Fh^\aa \Big( \sum_i H_{\bb_i} \Big) = \sum_i \Fh^{\bb_i}(H_{\bb_i}) \end{equation}
\end{prop}

\begin{proof}
    When $\aa = \sqcup_i \bb_i$ we have:
    $ -\ln \sum_{E_\aa} \prod_{i=1}^n \e^{-H_{\bb_i}} 
    = -\ln \prod_{i=1}^n \sum_{E_{\bb_i}} \e^{-H_{\bb_i}}$
\end{proof}

The additivity of $\Fh^\aa$ along constants may be seen as a particular case
of the previous proposition, 
identifying $\R$ with $A_\vide$. 
For every $H_\aa \in A_\aa$ and $\lambda \in \R$,
we have: 
\begin{equation} \Fh^{\aa}(H_\aa + \lambda) = \Fh^\aa(H_\aa) + \lambda \end{equation}
One should however note that $\Fh^\aa(0) \neq 0$ and free energy 
contains an entropic contribution $ -\ln |E_\aa|$. 
It is hence important for the previous proposition that the $\bb_i$'s 
actually cover $\aa$, 
and one should beware that $\Fh^\aa(H_{\bb_i})$ and $\Fh^{\bb_i}(H_\aa)$ 
differ by the entropic term $\Fh^\aa(0) - \Fh^{\bb_i}(0)$.
One should thus define a reduced free energy functional $\tilde \Fh^\aa$ 
by subtracting the entropic term 
to get $\tilde \Fh^\aa \circ j_{\aa\bb} = \tilde\Fh^\bb$.

The free energy of $H_\aa$ is by definition
the additive constant we subtract from $H_\aa$ to renormalise
the Gibbs density $\e^{-H_\aa}$:
\begin{equation} \big[ \e^{-H_\aa} \big] = \e^{-H_\aa + \Fh^\aa(H_\aa)}  \end{equation}
The smooth hypersurface $\{ \Fh^\aa = 0 \}$ given by the
image of $H_\aa \mapsto H_\aa - \Fh^\aa(H_\aa)$ 
is hence diffeomorphic to the space
$\Del_\aa$ of non-vanishing probability densities.
The Gibbs state is however even more naturally recovered through
the differential of $\Fh^\aa$.

\begin{prop}[Gibbs Expectations]
    The differential $d\Fh^\aa_\theta : A_\aa \aw A^*_\aa$ 
    of free energy is given by: 
    \begin{equation} d\Fh^\aa_\theta (H_\aa) = [\e^{-\theta H_\aa}] \end{equation}
We shall denote by $\E^\aa_\theta \in \Omega^1(A_\aa)$
the differential of $\Fh^\aa_\theta$ viewed as a $1$-form over $A_\aa$.
\end{prop}

\begin{proof}
    Using $de^x = \e^x dx$ and $d \ln(y) = \frac{dy}{y}$ 
    we have for every perturbation $f_\aa$ of $H_\aa \in A_\aa$:
    \begin{equation} \Fh^\aa(H_\aa + f_\aa) = \Fh^\aa(H_\aa) + 
    \frac {\sum f_\aa \e^{-H_\aa}}{\sum \e^{-H_\aa}} + o\,(f_\aa) \end{equation}
    The linear term is 
    precisely the expectation $\E^\aa[f_\aa]$ relative to the Gibbs state
    $[\e^{-H_\aa}]$.
    The formula at generic temperatures easily follows from the chain rule 
    applied to $H_\aa \mapsto \theta H_\aa$.
\end{proof}

The previous proposition generalises the usual 
properties expected from thermodynamic potentials, which
yield thermodynamic variables through their derivatives. 
Fixing the hamiltonian $H_\aa$ and 
considering free energy as a sole function of inverse temperature, 
we would recover internal energy as: 
\begin{equation} {\mathscr U}_\aa(\theta) = \E^\aa_\theta[H_\aa] = 
\frac{d (\theta \Fh^\aa_\theta)}{d \theta} \end{equation}

\begin{prop}[Integral Form]
The free energy of $H_\aa \in A_\aa$ is given by the integral formula: 
    \begin{equation} \Fh^\aa(H_\aa) = \Fh^\aa(0) +  
    \int_{\theta = 0}^1 \E^\aa_\theta[H_\aa] \:d\theta \end{equation}
where $\E^\aa_\theta$ denotes expectation in 
    the Gibbs state $[\e^{-\theta H_\aa}]$. 
\end{prop}

\begin{proof}
    This is the fundamental theorem of calculus applied along 
    the path $\theta \mapsto \theta H_\aa$.
\end{proof}

Another important property of free energy is concavity,  
as it will allow us to view Shannon entropy as the Legendre transform
of free energy in section 3.

\begin{prop}[Concavity]
For every $U_\aa, V_\aa \in A_\aa$ and $t \in [0,1]$, we have: 
    \begin{equation} \Fh^\aa \big( t\, U_\aa + (1-t)\, V_\aa \big) 
    \; \geq \; 
    t\, \Fh^\aa(U_\aa) + (1-t)\, \Fh^\aa(V_\aa) \end{equation}
\end{prop}

\begin{proof} Both $x \mapsto \e^{-x}$ and $y \mapsto - \ln(y)$ are convex
    while the latter is order reversing. 
\end{proof}

A consequence of the concavity of free energy is the negative signature of 
its second differential, explicitly given by the following proposition:

\begin{prop}[Fisher Metric] 
The second differential of free energy $D^2 \Fh^\aa : A_\aa \aw S^2 A^*_\aa$ 
yields the opposite covariance of observables with respect to the Gibbs state:
    \begin{equation} - D^2\Fh^\aa(f_\aa, g_\aa) = \E^\aa[f_\aa \cdot g_\aa] 
    - \E^\aa[f_\aa] \cdot \E^\aa[g_\aa] 
    \end{equation}
This is also the Fisher information 
metric for the exponential parametrisation of $\Del_\aa$ 
by $A_\aa$.
\end{prop}

\begin{proof}
    $D^2 \Fh^\aa(f_\aa, g_\aa)$ is
    the variation of $\E^\aa[f_\aa]$ along a perturbation $g_\aa$ 
    of the hamiltonian and:
    \begin{equation} - \frac{\dr}{\dr g_\aa} \bigg( 
    \frac {\sum f_\aa \e^{-H_\aa - g_\aa}} {\sum \e^{- H_\aa - g_\aa}} 
    \bigg)
    = \frac 
    {\sum f_\aa\, g_\aa \e^{-H_\aa}} 
    {\sum \e^{-H_\aa}} 
    - \frac{
        \big( \sum f_\aa \e^{-H_\aa} \big) \big( \sum g_\aa \e^{-H_\aa} \big) 
    }{\big( \sum \e^{-H_\aa} \big)^2 }
    \end{equation} 
\end{proof}

Note that Gibbs states induce a non-degenerate metric 
$\E^\aa[f_\aa \cdot g_\aa]$, 
its image under the projection
$f_\aa \mapsto f_\aa - \E^\aa[f_\aa]$ 
being the covariance bilinear form. 
The exponential map $\expm_\aa$ is a diffeomorphism from 
the level hypersurface $\{ \Fh^\aa = 0 \} \incl A_\aa$ to the manifold 
$\Del_\aa$ of non-degenerate probability densities, 
while the previous projection is  precisely the orthogonal projection 
onto the tangent space of $\{ \Fh^\aa = 0 \}$.

\subsection{Effective Energy}

Partial integration defines a linear map $\Sigma^{\bb\aa} : A_\aa \aw A_\bb$,
associating to $f_\aa$ the observable on $\bb$:
\begin{equation} \Sigma^{\bb\aa} f_\aa \;:\; x_\bb \longmapsto 
\sum_{y \in E_{\aa \setminus \bb}} f_\aa(x_\bb, y) \end{equation}
Partial integration is however not an algebra morphism and should 
rather be seen as an interlacing of the natural map
$A^*_\aa \aw A^*_\bb$ with identifications of each algebra
with its dual vector space.
Effective energy, as defined below, may then be thought of analogously:
\begin{equation} \bcd
A_\aa \rar{\Sigma^{\bb\aa}} & A_\bb \dar{\mlog_\bb} \\
A_\aa \uar{\expm_\aa} \rar{\Fh^{\bb\aa}} & A_\bb 
\ecd \end{equation}

\begin{defn}
For all $\aa \aw \bb$ in $X$, we call effective energy 
the map $\Fh^{\bb\aa} \in C^{\infty}(A_\aa, A_\bb)$ defined by:
\begin{equation} \Fh^{\bb\aa}(H_\aa) = 
-\ln \Big( \Sigma^{\bb\aa}\big( \e^{-H_\aa} \big) \Big) \end{equation}
\end{defn}

Note that the free energy $\Fh^\aa$ is recovered 
as the effective energy of vacuum $\Fh^{\varnothing \aa}$.
More generally, $\Fh^{\bb\aa}$ may be thought of as
the conditional free energy of $\aa$ given the microstate of $\bb$.  
This conditioning is functorial, as expressed by this first 
remarkable property.

\begin{prop}[Functoriality] 
    For every $\aa \aw \bb \aw \cc$ in $X$, we have in $C^\infty(A_\aa, A_\cc)$:
\begin{equation} \Fh^{\cc \bb} \circ \Fh^{\bb \aa} = \Fh^{\cc \aa} \end{equation}
Effective energies hence define a system $\Fh : X \aw \Mfd$ 
of differentiable manifolds over $X$.
\end{prop}

\begin{proof}
    We have $\Fh^{\cc\bb} \circ \Fh^{\bb\aa}(H_\aa) 
    = - \ln \big( \Sigma^{\cc\bb} \circ \Sigma^{\bb\aa} (\e^{-H_\aa}) \big)$ 
    while $\Sigma^{\cc\bb} \circ \Sigma^{\bb\aa} = \Sigma^{\cc\aa}$.
\end{proof}

The following weaker commutative diagram 
expresses that effective energy describes Gibbs states 
marginalisation at the level of hamiltonians. 
Consistency of the Gibbs states will hence transpose to a 
notion of effective consistency on hamiltonians, requiring 
that $H_\bb = \Fh^{\bb\aa}(H_\aa)$ for every $\aa \aw \bb$ in $X$. 

\begin{equation} \bcd 
A_\aa \rar{\Fh^{\bb\aa}} \dar[swap]{\gibbs_\aa} & A_\bb \dar{\gibbs_\bb} \\
\Del_\aa \rar{\Sigma^{\bb\aa}} & \Del_\bb 
\ecd \end{equation}

\begin{prop}[Marginalisation] \label{eff-marginal}
    For every $\aa \aw \bb$ in $X$ we have:
    \begin{equation} \Sigma^{\bb\aa} \big[ \e^{-H_\aa} \big] 
    = \big[ \e^{-\Fh^{\bb\aa}(H_\aa)} \big] \end{equation}
\end{prop}

The following proposition expresses that information on the microstates 
outside the support of a hamiltonian does not affect
effective energy. 
More precisely, given a hamiltonian $H_\aa \in A_\aa$ and $\bb \incl \aa$, 
extending $H_\aa$ to $\aa' = \aa \sqcup \cc$ and letting 
$\bb' = \bb \sqcup \cc$, we have the commutative diagram:
\begin{equation} \bcd A_{\aa'} \rar{\Fh^{\bb'\aa'}} & A_{\bb'} \\
A_{\aa} \uar{j_{\aa'\aa}} \rar{\Fh^{\bb\aa}} & A_{\bb} \uar[swap]{j_{\bb'\bb}} 
\ecd \end{equation}

\begin{prop}\label{eff-indep}
    If $\bb \incl \aa$ and $H_\aa \in A_\aa$, then for every $\bb' \incl \aa'$ 
    such that $\aa \setminus \bb = \aa' \setminus \bb'$:
    \begin{equation} \Fh^{\bb'\aa'}(H_\aa) = \Fh^{\bb\aa}(H_\aa) \end{equation}
\end{prop} 

\begin{proof} $\Sigma^{\bb'\aa'}$ and $\Sigma^{\bb\aa}$ involve a sum over
    the same variables and $\e^{-H_\aa}$ only depends on those in $\aa$.
\end{proof}

Just like conditional independence generalises the notion of independence, 
additivity of free energy along uninteracting subsystems generalises to
the following conditional form:

\begin{prop}[Pairwise Conditional Additivity]\label{eff-add}
    If $\aa = \bb \cup \bb'$ and $\bb \cap \bb' = \cc$, we have:
    \begin{equation} \Fh^{\cc\aa}(H_\bb + H_{\bb'}) = 
    \Fh^{\cc\bb}(H_\bb) + \Fh^{\cc\bb'}(H_{\bb'}) 
    \end{equation}
\end{prop}

\begin{proof}
    We have $- \ln \big( \Sigma^{\cc\aa}(\e^{-H_\bb}\e^{-H_{\bb'}}) \big)
    = -\ln \big( 
    \Sigma^{\cc\bb}(\e^{-H_\bb} )\: \Sigma^{\cc\bb'}(\e^{-H_\bb'})
    \big)$ as $\Sigma^{\cc\aa}$ 
    involves a sum over the variables in the disjoint union
    $(\bb \setminus \cc) \sqcup (\bb'\setminus \cc)$, 
    while $\e^{-H_\bb}$ is independent of the variables in $\bb' \setminus \cc$ 
    and reciprocally.
\end{proof}

In particular, $\Fh^{\bb\aa}$ is additive along $A_\bb$ so that
for every $H_\aa \in A_\aa$ and $H_\bb \in A_\bb$ we have:
\begin{equation} \Fh^{\bb\aa}(H_\aa + H_\bb) = \Fh^{\bb\aa}(H_\aa) + H_\bb \end{equation}
As before, one should beware that $\Fh^{\bb\aa}(0) \neq 0$ as effective energy
contains the entropic term $-\ln |E_{\aa \setminus \bb} |$. 
The reduced effective energy $\tilde\Fh^{\bb\aa}$ defined by 
subtracting this term to $\Fh^{\bb\aa}$ is 
a smooth section of $j_{\aa\bb}$.
Considering larger coverings, we may now rewrite conditional additivity 
as follows:

\begin{prop}[Conditional Additivity] If $\aa$ is the union of 
    $\bb_1, \dots, \bb_n$ and $H_\aa$ splits as $\sum_i H_{\bb_i}$, 
    then denoting by $\cc_i$ the intersection of $\bb_i$ with 
    $\bigcup_{j \neq i} \bb_j$ and by $\cc$ the reunion 
    $\bigcup_i \cc_i$ we have: 
\begin{equation} \Fh^{\cc\aa} \Big( \sum_i H_{\bb_i} \Big)
= \sum_i \Fh^{\cc_i\bb_i}(H_{\bb_i}) \end{equation}
\end{prop}

\begin{proof}
Reasoning by induction on $n$, let $\tilde \bb_1 = \bb_1 \cup \cc$ and  
$\tilde \bb_2 = \bigcup_{j =2}^n \bb_j \cup \cc$ with hamiltonians 
    $\tilde H_{\tilde \bb_1} = H_{\bb_1}$ and 
    $\tilde H_{\tilde \bb_2} = \sum_{i > 1}^n H_{\bb_i}$.
    Pairwise 
    conditional additivity gives $\Fh^{\cc\aa}(\sum_i H_{\bb_i}) = 
    \Fh^{\cc\tilde \bb_1}(\tilde H_{\tilde \bb_1}) + 
    \Fh^{\cc \tilde\bb_2}(\tilde H_{\tilde \bb_2})$ as 
    $\tilde \bb_1 \cap \tilde \bb_2 = \cc$, 
    while $\Fh^{\cc\tilde \bb_1}(H_{\bb_1})$ and 
    $\Fh^{\cc_1 \bb_1}(H_{\bb_1})$ coincide as conditional free energies 
    of $\tilde \bb_1 \setminus \cc = \bb_1 \setminus \cc_1$. 
    The induction hypothesis applied to 
    $\Fh^{\cc\tilde \bb_2}(\tilde H_{\tilde \bb_2})$ then terminates the proof.
\end{proof}

Effective energy is still a concave functional in the sense below, and
its point-wise Legendre transform above $x_\bb \in E_\bb$ will generate 
the conditional Shannon entropy $S_\aa(p_\aa | x_\bb)$ defined in section 2.

\begin{prop}[Concavity]
For every $U_\aa, V_\aa \in A_\aa$ and $t \in [0,1]$, we have: 
    \begin{equation} \Fh^{\bb\aa} \big( t\, U_\aa + (1-t)\, V_\aa \big) 
    \; \geq \; 
    t\, \Fh^{\bb\aa}(U_\aa) + (1-t)\, \Fh^{\bb\aa}(V_\aa) \end{equation}
the inequality being understood in the partial order of functions on $E_\bb$.
\end{prop}

While the Gibbs state was recovered through the differential of 
free energy, the differential of effective energy carries
the effects of Gibbs state conditioning.

\begin{prop}[Conditional Expectations] \label{Eba}
    The differential $d\Fh^{\bb\aa} : A_\aa \aw A^*_\aa \otimes A_\bb$ of 
    effective energy is the conditional expectation given the microstate 
    of $\bb$ with respect to the Gibbs state on $\aa$:
    \begin{equation} d\Fh^{\bb\aa}(H_\aa) \; : \; 
    \left\{ \ba{ccc} 
    A_\aa & \law & A_\bb \\
    f_\aa & \longmapsto & \E^\aa[\, f_\aa \st \bb \,]
    \ea \right. 
    \end{equation}
    We denote by $\E^{\bb\aa} \in \Omega^1(A_\aa, A_\bb)$ 
    the differential of $\Fh^{\bb\aa}$ viewed as an $A_\bb$-valued $1$-form
    over $A_\aa$.
\end{prop}

\begin{proof}
    Using again $de^x = \e^x dx$ and $d \ln(y) = \frac{dy}{y}$ 
    we now have for every perturbation $f_\aa$ of $H_\aa$:
    \begin{equation} \Fh^{\bb\aa}(H_\aa + f_\aa) = \Fh^{\bb\aa}(H_\aa) + \frac
    { \Sigma^{\bb\aa}(f_\aa \e^{-H_\aa}) } 
    { \Sigma^{\bb\aa}(\e^{-H_\aa}) }
    + o\,(f_\aa) \end{equation}
    Rescaling the fraction by the normalisation factor 
    $\Sigma^{\vide\aa}(\e^{-H_\aa})$, denoting by $p_\aa$ the Gibbs state
    on $\aa$ and by $p_\bb$ its marginal on $\bb$, the linear term rewrites as:
    \begin{equation} \Sigma^{\bb\aa} \Big( f_\aa \cdot \frac {p_\aa} {p_\bb} \Big) = 
    \E_{p_\aa}[f_\aa \st \bb] \end{equation}
    As according to the Bayes rule, $p_\aa / p_\bb$ is 
    the conditional probability on $\aa$ given the microstate of $\bb$.
\end{proof}

Conditional expectation has a simple geometrical characterisation
which is worth recalling. 
For the riemannian metric on $A_\aa$ induced
by the Gibbs state:
\begin{equation} \scal{f_\aa}{g_\aa} = \E^\aa[f_\aa \cdot g_\aa] \end{equation}
$\E^{\bb\aa}$ is just the orthogonal projection of $A_\aa$
onto the subspace $A_\bb$ of observables depending only on the 
state of $\bb$. 
One might hope for a consistent orthogonal splitting of $A_\aa$ as\footnote{
    See the interaction decomposition theorem, section 2.3.2.
} $\bigoplus Z_\bb$, defining each interaction subspace $Z_\bb$ by orthogonality
with $B_\bb$. However, as soon as $\bb$ and $\bb'$ 
interact with each other, correlations will imply that $Z_\bb$ is no longer
orthogonal to $Z_{\bb'}$ for the metric on $A_\aa$.

\begin{prop}[Integral Form] 
The effecive energy of $H_\aa \in A_\aa$ on $\bb$ is given by the integral:
    \begin{equation} \Fh^{\bb\aa}(H_\aa) = \Fh^{\bb\aa}(0) +  
    \int_{\theta = 0}^1 \E^{\bb\aa}_\theta[H_\aa] \:d\theta \end{equation}
    where $\E^{\bb\aa}_\theta$ denotes conditional expectation 
    with respect to the Gibbs state $[\e^{-\theta H_\aa}]$. 
\end{prop}

\begin{proof}
    This is again the fundamental theorem of calculus applied along the 
    path $\theta \mapsto \theta H_\aa$.
\end{proof}

\section{Information Quantities}

This section introduces the Shannon entropy functional,
also called Shannon information. 
We then relate mutual information quantities with a combinatorial 
decomposition of entropy into summands, 
on which will rely Bethe-Kikuchi local approximations 
of entropy and the cluster variation method. 

\subsection{Shannon Entropy}

Given a finite set $E_\aa$, we denote by $\State_\aa$ the space
of probability measures on $E_\aa$. 
Shannon entropy is the concave functional $S_\aa : \State_\aa \aw \R$ defined by:
\begin{equation} S_\aa(p_\aa) = -\sum_{E_\aa} p_\aa \ln(p_\aa) \end{equation}
It reaches its global maximum 
$S_\aa([1_\aa]) = \ln |E_\aa| $
on the uniform measure, which may be called the Boltzmann entropy of $E_\aa$.
Let us briefly recall some fundamental properties of entropy. 
Plenty of good resources already exist on the subject, 
see for instance \cite{Mezard-Montanari} or any reference on information 
geometry\footnote{
    Other references to specific properties of entropy 
    will be given inline, 
    such as \cite{Bennequin-Baudot-2, Vigneaux-phd, Leinster-08}.
}. 

Given a surjection $\pi^{\bb\aa} : E_\aa \aw E_\bb$,
and $p_\aa \in \Delta_\aa$, 
let us denote by $p_\bb = \Sigma^{\bb\aa}(p_\aa)$ its marginal distribution
on $E_\bb$ and by $p_{\aa|x_\bb}$ its conditional probability 
distribution given any $x_\bb \in E_\bb$, supported by the fiber 
of $x_\bb$.
The conditional entropy of $p_\aa$ given $E_\bb$ is defined as: 
\begin{equation} 
S_\aa(p_\aa \st p_\bb) = \sum_{x_\bb \in E_\bb} 
p_\bb(x_\bb)  \cdot S_\aa(p_{\aa|x_\bb}) 
    = - \sum_{x_\bb \in E_\bb} 
    \sum_{x' \in E_{\aa \setminus \bb}} 
    p_\aa(x_\bb, x')
    \ln \frac {p_\aa(x_\bb, x')} {p_\bb(x_\bb)}
\end{equation} 
When $E_\aa = E_\bb \times E_\cc$ and 
$p_\aa = p_\bb \otimes p_\cc$
then an easy computation shows that 
$S_\aa(p_\aa \st p_\bb) = S_\cc(p_\cc)$.

The fundamental property of entropy is the chain rule:
\begin{equation} S_\aa(p_\aa) = S_\bb(p_\bb) + 
S_\aa(p_\aa \st p_\bb) \end{equation}
Shannon showed that entropy is essentially
characterised by this functional equation\footnote{
    This functional equation was given a natural interpretation 
    in the context of monads and operads by Leinster \cite{Leinster-08}, 
    and also retrieved as a cocycle equation by 
    Bennequin and Vigneaux, see \cite{Vigneaux-phd}.
},
up to a multiplicative constant and under an additional monotonicity condition
with respect to the cardinal of $E_\aa$. 
When $p_\aa =
 \otimes_{\bb'}  p_{\bb'}$
is a tensor product of independent probabilities on 
$E_\aa = \prod_{\bb'} E_{\bb'}$, we have in particular:
\begin{equation} S_\aa(p_\aa) = \sum_{\bb'} S_{\bb'}(p_{\bb'}) \end{equation}
expressing that entropy is additive along independent systems\footnote{
    This should be related to the additivity property of free energy 
    \ref{Fh-add}.
}. 

\subsection{Möbius Inversion} 

For every $\aa \in X$, let us denote by 
$\stF_\aa = C^{\infty}(\State_\aa)$ the
vector space of smooth functionals on $\State_\aa$. 
For every $\aa \aw \bb$, the pullback of 
the marginal projection $\State_\aa \aw \State_\bb$ defines 
an injection $\stF_\aa \wa \stF_\bb$, 
so that $\stF$ is a contravariant functor of vector spaces on $X$.
In particular, $\tilde \Z_1(X)$ acts on $\stF_0(X)$ and for every collection
of functionals $({\mathscr L}_\aa)$ we have the equivalence:
\begin{equation} {\mathscr L}_\aa = \sum_{\aa \aw \bb'} [ {\mathpzc l}_{\bb'} ]_\aa
\quad\eqvl\quad {\mathpzc l}_\aa = 
\sum_{\aa \aw \bb'} \mu_{\aa \bb'} \cdot [ {\mathscr L}_{\bb'} ]_\aa
 \end{equation}
We call ${\mathpzc l} \in \stF_0(X)$ the combinatorial localisation 
of $\mathscr{L}$.

A particular case of functionals 
is given by the expectation values of observables. 
Given $H \in A_0(X)$, let ${\mathpzc U} \in \stF_0(X)$ 
be defined for every $\aa \in X$ by:
\begin{equation} {\mathpzc U}_\aa (p_\aa) = \croc{p_\aa}{H_\aa} = \E_{p_\aa}[H_\aa] \end{equation}
The natural map from $A_0(X)$ to $\stF_0(X)$
is a morphism of $\tilde \Z_1(X)$-modules.
When $H = \zeta \cdot h$, the localisation of 
${\mathpzc U}$ as $\zeta \cdot {\mathpzc u}$ 
defines another field of functionals ${\mathpzc u}  \in \stF_0(X)$ where:
\begin{equation} {\mathpzc u}_\aa(p_\aa) = \croc{p_\aa}{h_\aa} = \E_{p_\aa}[h_\aa] \end{equation}
The internal energy of a system ${\mathpzc U}_\Om$ is the expectation value 
of the global hamiltonian $H_\Om$, typically given as a sum
$\sum_\aa h_\aa$ of local interactions with $h_\Om = 0$. 
This remark will prove the Bethe approximation scheme to be exact on internal energy,
as $\mathpzc{U}_\Om = \sum_\aa \mathpzc{u}_\aa$ with $\mathpzc{u}_\Om = 0$.

\subsection{Mutual Informations}

\newcommand{\X}{{\cal X}} 

Suppose given $n$ random variables $\X_1, \dots \X_n$. 
In this paragraph, we introduce the various amounts of mutual information,
following Hu Kuo Ting in \cite{Hu-62}. 
First, denote by $S(\X_i \cup \X_j)$ the entropy of their
joint law $(\X_i, \X_j)$.
If the variables were pairwise independent, 
the chain rule of entropy would give for all distinct $i,j$:
\begin{equation} S(\X_i \cup \X_j) = S(\X_i) + S(\X_j) \end{equation}
But more generally, denoting by $S(\X_i | \X_j)$ the
conditional entropy of $(\X_i, \X_j)$ given 
$\X_j$, we have: 
\begin{equation} S(\X_i \cup \X_j) = S(\X_i | \X_j) + S(\X_j) \end{equation}
We may define a quantity $S(\X_i \cap \X_j)$ as
the difference $S(\X_i) - S(\X_i | \X_j)$.
These relations naturally remind of those satisfied by 
an additive function on sets,
where independent variables correspond
to disjoint subsets, 
and conditioning describes set difference. 
Note however that in the following theorem, 
the order of logical operations matters 
(see Bennequin-Baudot \cite{Bennequin-Baudot-2} for examples). 

\begin{thm}[Hu Kuo Ting]
Given random variables $\X_1, \dots, \X_n$ and their 
probability distributions, there exists 
sets $A_1, \dots, A_n$ and an additive real function $\ph$ on 
the algebra generated by those sets, such that for all 
    operation $Q$ generated by $\cup, \cap$ and $-$ 
    which 1) forms collections of unions 
    2) takes successive intersections of these unions and 3) 
    subtracts one of them, one has:
\begin{equation} S(Q(\X_1, \dots, \X_n)) = \ph(Q(A_1, \dots, A_n)) \end{equation}
\end{thm} 

All quantities obtained this way are generically called
amounts of information by Hu Kuo Ting. 
Of particular importance are the mutual informations, 
appearing in the right hand side of:
\begin{equation} S( \X_{i_1} \cup \dots \cup \X_{i_k} ) 
= \sum_{1 \leq p \leq k} S(\X_{i_p}) - \sum_{1 \leq p < q \leq k} 
S(\X_{i_p} \cap \X_{i_q}) + \dots 
+ (-1)^{k + 1} \; S(\X_{i_1} \cap \dots \cap \X_{i_k})  
\end{equation}
We will denote the mutual information of $\X_{i_1}, \dots \X_{i_k}$ by:
\begin{equation} I(\X_{i_1}, \dots, \X_{i_k}) = S(\X_{i_1} \cap \dots \cap \X_{i_k}) \end{equation} 
The following theorems will best express the significance
of mutual informations.
Although very similar in appearance,
the second theorem is a 
much more recent result than the first, 
proved in \cite{Bennequin-Baudot-2}.

\begin{thm}[Hu]
$\X_1, \dots \X_n$ form a Markov chain if and only if
for all $1 \leq i_1 < \dots < i_k \leq n$:
\begin{equation} I(\X_{i_1}, \dots \X_{i_k}) = I(\X_{i_1}, \X_{i_k}) \end{equation}
\end{thm}

\begin{thm}[Bennequin]\label{info-indep}
$\X_1, \dots, \X_n$ are independent if and only if
for all $1 \leq i_1, \dots, i_k \leq n$:
\begin{equation} I(\X_{i_1}, \dots, \X_{i_k}) = 0\end{equation}
\end{thm}

Let us give an explicit definition of mutual informations
as functionals of the probability distributions,
by relating Hu Kuo Ting's construction to a
Möbius inversion on entropy functionals. 
Denoting by $(\X_i)_{i \in \Om}$ a finite set of random variables,
and by $p_\aa$ the joint law
of $(\X_i)_{i \in \aa}$ for every $\aa \incl \Om$,
each joint entropy may be expressed as:
\begin{equation} S_\aa(p_\aa) = \sum_{\aa \cont \bb'} {\mathpzc s}_{\bb'}(p_{\bb'})  \end{equation}
Comparing with Hu Kuo Ting's formula, 
the mutual information
$I_\bb(p_\bb)$ is then given by $(-1)^{|\bb| + 1}\, \mathpzc{s}_\bb$.
The entropy summands given by Möbius inversion on $\Part(\Om)$ satisfy:
\begin{equation} \mathpzc{s}_\aa(p_\aa) = \sum_{\aa \cont \bb'} (-1)^{|\aa| - |\bb'|} \;
S_{\bb'}(p_{\bb'}) \end{equation}
while mutual informations are given by:
\begin{equation} I_\aa(p_\aa) = \sum_{\aa \cont \bb'} (-1)^{|\bb'| + 1} \;
S_{\bb'}(p_{\bb'}) \end{equation}
Fixing $\aa \in X$, this expression may also be seen as a Möbius inversion on the
opposite
partial order spanned by each $\X_i$ for $i \in \aa$, 
viewed as a maximal element.

\subsection{Bethe Entropy}

As first recognised by Morita in \cite{Morita-57}, 
Bethe's approximation of entropy is essentially a truncation of the previous
Möbius inversion procedure.  
Consider a general covering 
$X \incl \Part(\Om)$ not containing $\Om$,
and let $\bar X = \{ \Om \} \cup X$. 
Through  Möbius inversion on $\bar X$, 
entropy can still be exactly localised by
$\mathpzc{s} \in \stF_0(\bar X)$ with: 
\begin{equation} \mathpzc{s}_\aa = \sum_{\aa \aw \bb'} \mu_{\aa\bb'} \cdot [S_{\bb'}]_\aa \end{equation}
In particular, the global entropy 
$S_\Om$ is recovered as: 
\begin{equation} S_\Om = \mathpzc{s}_{\Om} + \sum_{\aa \in X} [\mathpzc{s}_{\aa}]_\Om  \end{equation}
The Bethe approximation of entropy $\Sb_\Om = S_\Om - \mathpzc{s}_\Om$,
according to the previous paragraph, 
may be seen to only neglect all mutual informations 
of the form $I_\om(p_\om)$ for $\om \in \Part(\Om)$ not contained in any $\aa \in X$.
Intuitively, $\mathpzc{s}_\Om$
is expected\footnote{
 Schlijper proved this procedure to converge to the true entropy 
 per lattice point for the Ising 2D model in \cite{Schlijper-83}
}
to be small when large
enough regions are taken in $X$ to cover $\Om$, by extensivity of entropy.

The approximate entropy $\Sb_\Om(p_\Om)$ only depends on the marginal 
distributions $(p_\aa)_{\aa \in X}$ of $p_\Om$, 
so that $\Sb_\Om$ factors through the canonical map $\State_\Om \aw \Gamma(X)$.
We call Bethe entropy the functional $\Sb$ defined on $\State_0(X)$ by:
\begin{equation} \Sb = \sum_{\aa \in X} \mathpzc{s}_\aa = \sum_{\bb \in X} c_\bb \cdot S_\bb \end{equation}
Each term $\mathpzc{s}_\aa$ acting on $\Delta_\aa$, 
so that the restriction of $\Sb$ to the image of $\State_\Om$ in 
$\State_0(X)$ coincides with $\Sb_\Om$. 
Only the restriction of Bethe entropy to $\Gamma(X)$ shall be relevant, 
it is however important to remark that $\Gamma(X)$ is in general larger than 
the image of $\Delta_\Om$.

\section{Variational Principles}

Thermodynamics characterise the equilibrium state of a 
system through variational principles expressed in terms of 
macroscopic variables such as pressure, volume, temperature, {\it etc.} 
The statistical counterparts of such variational principles lead to
a variety of functionals, expressed in terms of the hamiltonian 
$H_\Om$ and the statistical state $p_\Om$ of the system, 
while macroscopic variables such as pressure and temperature
arise as Lagrange multipliers 
associated to volume and energy constraints.  

In this section, we suppose given the hamiltonian 
$H_\Om = \sum_\aa h_\aa$ as a sum of interactions over $X$.
We briefly recall two classical variational principles 
characterising the Gibbs equilibrium state $p_\Om$. 
The first one maximises entropy under a mean energy constraint, 
while the second one second performs a Legendre transform to free energy. 
The latter we shall locally approximate by Kikuchi's cluster variation method,
before giving the announced homological characterisation \ref{cvm} 
of its solutions. 
This result will justify the use of message-passing algorithms in the next chapter, 
which iterate over heat fluxes 
corresponding to Lagrange multipliers for the consistency constraints.

\subsection{Maximal Entropy}

\newcommand{\Uint}{{\mathpzc U}}

The simplest variational principle characterising 
the Gibbs distribution 
states that once the internal or mean energy $\Uint_\Om$ of the 
system is fixed, 
equilibrium is reached at the maximally entropic probability density.
Internal energy is the smooth functional 
$\Uint_\Om : \Delta_\Om \aw \R$ giving the expectation value 
of the hamiltonian:
\begin{equation} \Uint_\Om(q_\Om) = \croc{q_\Om}{H_\Om} = \E_{p_\Om}[H_\Om] \end{equation}
Let us denote by $\lambda_\infty = \inf(H_\Om)$ the minimal energy and 
by $\lambda_0 = \Uint_\Om([1])$ the mean energy for the uniform 
distribution on $E_\Om$.

\begin{thm}[Maximal Entropy Principle]
    Given $\lambda \in \:] \lambda_\infty, \lambda_0 [$ 
    the constrained variational problem:
    \begin{equation} S_\Om(p_\Om) = \max_{ \big\{ \Uint_\Om = \lambda \big\} }  S_\Om \end{equation}
    has a unique solution given by $p_\Om = \big[ \e^{- \theta H_\Om} \big]$ 
    for some inverse temperature $\theta \in \R_+^*$.
\end{thm}

\begin{proof} 
The cotangent space of $\Del_\Om \incl A^*_\Om$
is isomorphic to the quotient of $A_\Om$ by additive constants, 
    generating the normalisation constraint $\croc{p_\Om}{1} = 1$. 
The inverse temperature then arises as a Lagrange multiplier 
for the energy constraint 
$\Uint_\Om = \lambda$, as for every positive $p_\Om \in \Del_\Om$: 
\begin{equation} \der{S_\Om}{p_\Om} = \theta \cdot \der{\Uint_\Om}{p_\Om} \mod \R 
\quad\eqvl\quad
- \ln(p_\Om) = \theta H_\Om \mod \R \end{equation}
Reciprocally, consider the smooth path $\theta \mapsto [\e^{- \theta H_\Om}]$ 
    described by the Gibbs states in $\Del_\Om$ for $\theta \in \:]0, + \infty[$. 
    When the temperature goes to zero and $\theta \to \infty$ 
    the Gibbs distribution goes to the uniform distribution on the minimisers 
    of $H_\Om$ and $\Uint_\Om \to \lambda_{\infty}$. 
    When the temperature goes to infinity and $\theta \to 0$, the Gibbs distribution 
    goes to the uniform measure and $\Uint_\Om \to \lambda_0$. 
\end{proof}

Gibbs states as a function of a generalised inverse temperature 
$\theta \in \R$ are one-parameter subgroups of $\Del_\Om$ for
its multiplicative structure, while $\theta \in \R_+$ restricts to semi-groups,
as pictured by the figure below.

At the macroscopic level, the form of the maximal entropy principle suggests that
the equilibrium entropy $S(\Uint)$ be naturally defined 
as a function of internal energy $\Uint$, while the temperature $T$ measures
a kind of inverse entropic susceptibility:
\begin{equation} \frac 1 {kT} = \theta = \der{S}{\Uint} \end{equation} 
The Legendre transform essentially consists in a change of variables, 
defining an equivalent functional
parametrised by the derivative of the original one. We would here recover
the so-called free entropy\footnote{
    Free entropy is also called the Massieu potential, as introduced in his 1869 
    note \cite{Massieu-69}.
}:  
\begin{equation} \Psi(\theta) = S(\Uint) - \theta \Uint \end{equation}
The equilibrium free energy recovered as 
$F(\theta) = - \Psi(\theta)/\theta$ being more commonly used. 
It is remarkable that the statistical free energy $\Fh^\Om(H_\Om)$ 
may be defined 
from the Shannon entropy $S_\Om(p_\Om)$ in a perfectly similar manner.

\subsection{Thermal Equilibrium}

Free energy variational principles describe a system interacting 
with a thermostat, exchanging arbitrary amounts of energy without modifying
the temperature of the latter. Such variational principles are very 
natural as they for instance describe interaction with the atmosphere. 
Let us call {\it variational free energy} the smooth 
bifunctional $\Fg_\Om : \Delta_\Om \times H_\Om \aw \R$ given by: 
\begin{equation} \Fg_\Om(p_\Om, H_\Om) = \croc{p_\Om}{H_\Om} - S_\Om(p_\Om) \end{equation}
This functional generates the Legendre transform 
of $S_\Om$ by minimisation of $\Fg_\Om(\,-\,,H_\Om)$ 
and does yield the equilibrium free energy $\Fh^\Om$ 
at inverse temperature\footnote{
    One may reintroduce temperature dependency by letting
    $\Fg_\Om^\theta(p_\Om, H_\Om) = \theta^{-1} \Fg_\Om(p_\Om, \theta H_\Om)$, 
    recovering the classical formula 
    $\Fg_\Om^\theta = \Uint_\Om - \theta^{-1} S_\Om$.
    At fixed temperatures one may however always assume $\theta = 1$ up to a choice 
    of units. 
}
$\theta = 1$.

\begin{thm}[Minimal Free Energy Principle] 
    For every $H_\Om \in A_\Om$, the variational problem:
    \begin{equation} \Fg_\Om(p_\Om, H_\Om) = \min_{\Delta_\Om} 
    \Fg_\Om(\,-\,, H_\Om) \end{equation}
    has a unique solution given by the Gibbs state $[\e^{- H_\Om}]$ 
    reaching the equilibrium free energy $\Fh^\Om(H_\Om)$.
\end{thm}

\begin{proof} 
    Describing the cotangent space of $\Del_\Om$ by the quotient $A_\Om/\R$, 
    a critical $p_\Om \in \Del_\Om$ satisfies:
    \begin{equation} \der{\Fg_\Om}{p_\Om} = H_\Om + \ln(p_\Om) = 0 \mod \R \end{equation} 
    Reciprocally, concavity of entropy implies that $p_\Om = [\e^{-H_\Om}]$ 
    indeed minimises $\Fg_\Om(\,-\,, H_\Om)$.
\end{proof}

The theorem really states that free energy 
is the Legendre transform of Shannon entropy. 
Although not bijective, the Legendre duality between hamiltonians 
and statistical states is fundamental:
\begin{equation} p_\Om = [\e^{-H_\Om}] \quad\eqvl\quad  H_\Om = -\ln(p_\Om) \mod \R \end{equation}
and defines a surjective abelian group morphism $A_\Om \aw \Del_\Om$. 
The computability of $H_\Om$ however does not at all imply that of 
$p_\Om$, as normalising the Gibbs density would require to compute 
the equilibrium free energy $\Fh^\Om(H_\Om)$, involving an integral 
over $E_\Om$ of exponential complexity in the cardinal of $\Om$. 

\subsection{Cluster Variation Method}


Introduced by Kikuchi in \cite{Kikuchi-51}, the cluster variation method
seeks to approximate the marginals $(p_\aa) \in \Gamma(X)$ 
of the global Gibbs state $p_\Om$
by a consistent collection of
local probabilities\footnote{
    See \ref{Gamma} for the definition of $\Gamma(X) \incl A^*_0(X)$, 
    space of consistent local probabilities. 
} $(q_\aa) \in \Gamma(X)$, 
obtained through a variational principle on 
a local approximation of free energy. 

\begin{defn} {\rm Bethe free energy} is the smooth bifunctional 
    $\Fb : \Delta_0(X) \times A_0(X) \aw \R$ defined by: 
    \begin{equation} \Fb(p, H) = \sum_{\bb  \in X} c_\bb
    \Big( \croc{p_\bb}{H_\bb} - S_\bb(p_\bb) \Big) 
    \end{equation}
\end{defn}

Because of the Möbius numbers $c_\bb$ appearing in its definition, 
the Bethe free energy $\Fb$ is no longer convex in general, 
and $\Fb(\,-\,, H)$ may have a great multiplicity\footnotemark{}
of critical points inside the space $\Gammint(X)$ of consistent 
positive densities. 
\footnotetext{
    For numerical studies see \cite{Weiss-97,Murphy-Weiss-99,Knoll-2017},
    A first mathematical proof of multiplicity is given by Bennequin
    in \cite{Bennequin-IEM}.
}
We provide with a rigorous characterisation of 
the critical points of $\Fb(\,-\,,H)$ constrained to 
$\Gammint(X)$,  
by showing that they bear a homological relationship with the reference
hamiltonian field $H$. 

\begin{thm} \label{cvm}
    A positive and consistent statistical field 
    $p \in \Gammint(X)$ is critical for the constrained 
    Bethe free energy $\Fb(\,-\,,H)_{|\Gamma(X)}$ 
    if and only if there 
    exists $\ph \in A_1(X)$ such that:
    \begin{equation} - \ln(p) \simeq H + \zeta \cdot \div \ph \mod \R_0(X) \end{equation}
\end{thm}

The proof of theorem \ref{cvm} shall conclude this chapter, 
the flux term $\div \ph$ essentially appearing as Lagrange multipliers 
associated to the consistency constraint $dp = 0$. 
The correspondence theorem \ref{correspondence}
between stationary states of belief propagation 
and critical points of $\Fb$ shall come as an easy consequence of \ref{cvm}, 
once message-passing algorithms have been related with transport equations.
A crucial combinatorial argument in proving \ref{cvm} is 
contained in the following proposition:

\begin{prop} \label{zeta-bord} 
    For every $H \in A_0(X)$ and $q \in A^*_0(X)$ such that $dq = 0$ we have:
    \begin{equation} \croc{q}{\mu \cdot H} = \croc{q}{c\, H} \end{equation}
    In particular $c\, H \in \Img(\div)$ if and only if 
    $\mu \cdot H \in \Img(\div)$, 
    or equivalently if $H \in \zeta \cdot \Img(\div)$. 
\end{prop}

\begin{proof}
    Using $q_\bb = \Sigma^{\bb\aa}(q_\aa)$ 
    and $c_\bb = \sum_{\aa \aw \bb} \mu_{\aa\bb}$ we have:
\begin{equation} \croc{q}{\mu \cdot H} 
    = \sum_{\aa \aw \bb} \mu_{\aa\bb} \, \croc{\Sigma^{\bb\aa}(q_\aa)}{H_\bb} 
    = \sum_{\bb}c_\bb \, \croc{q_\bb}{H_\bb}
    = \croc{q}{c\, H}
\end{equation}
    In particular $\mu \cdot H \perp \Ker(d) \eqvl c\, H \perp \Ker(d)$, 
    while $\Img(\div) = \Ker(d)^\perp$ as $d$ is the adjoint of $\div$.
\end{proof}

Homological invariance of $\Fb$ with respect to interaction potentials also
follows from \ref{zeta-bord}. 
This interesting property sheds light on the form of the 
Lagrange multipliers appearing in theorem \ref{cvm} 

\begin{prop} \label{hom-invariance}
    For every consistent $p \in \Gamma(X)$ 
    and every $U = H + \zeta \cdot \div \ph$ in $A_0(X)$ we have: 
    \begin{equation} \Fb(p, U) = \Fb(p, H) \end{equation}
\end{prop}

\begin{proof}
    By orthogonality of 
    $\Ker(d)$ with $\Img(\div)$ and using the lemma, 
    the internal energy terms satisfy:
    \begin{equation} \croc{p}{c\, U} = \croc{p}{\mu \cdot U} = 
    \croc{p}{\mu \cdot H} + \croc{p}{\div \ph} 
    = \croc{p}{c \, H} \end{equation}
\end{proof}

Recall that homologous interaction potentials $u = h + \div \ph$ 
define the same global hamiltonian as $h$: 
\begin{equation} H_\Om = \sum_{\aa \in X} h_\aa = \sum_{\bb \in X} c_\bb H_\bb \end{equation}
and by \ref{H0}, one furthermore has  $\sum_\aa u_\aa = H_\Om$ 
if and only if $u$ is homologous to $h$.
Hence a subtle difference between the homological invariance of \ref{hom-invariance} 
and the following proposition lies 
in the assumption that $p \in \Gamma(X)$ is given by the marginals 
of a global probability distribution\footnotemark{} $p_\Om \in \Delta_\Om$. 
In that case, it is well-known that the Bethe approximation 
yields an exact measure of internal energy. 
\footnotetext{
    The image of $\Delta_\Om(X)$ in general forms a strict convex polytope 
    of $\Gamma(X)$. Its boundaries are the image of the positivity constraints 
    on the global density, see \cite{Vorobev-62} and \cite{Abramsky-2011}. 
}

\begin{prop} \label{Fb-energy}
    Given $p_\Om \in \Delta_\Om$ of image $p$ in $\Gamma(X)$, 
    for any local hamiltonians $H_\aa \in A_0(X)$ of
    global hamiltonian $H_\Om \in A_\Om$ we have:
    \begin{equation} \Fb(p, H) = \croc{p_\Om}{H_\Om} - \check S(p) \end{equation}
\end{prop}

\begin{proof}When $p_\bb = \Sigma^{\bb\Om}(p_\Om)$ for every $\bb \in X$, 
    by definition of $\Fb$ and $H_\Om$ we have:
    \begin{equation} \Fb(p, H)
    = \sum_{\bb \in X} \croc{p_\Om}{c_\bb H_\bb} - \check{S}(p) 
    = \croc{p_\Om}{H_\Om} - \check{S}(p) 
    \end{equation}
\end{proof}

The Bethe approximation also counts  
degrees of freedom properly, as measured by the 
maximum entropy reaches in the high temperature limit. 
The following proposition further justifies its soundness, 
it is proved in a slightly different form in \cite{Yedidia-2005}.

\begin{prop} 
    When $X$ is a $\cap$-closed covering\footnotemark{} of $\Om$, 
    the high temperature limit $\Fb([1], H)$ of the 
    Bethe free energy coincides with that of the true 
    free energy $\Fg([1_\Om], H_\Om) = \langle H_\Om \rangle - \ln |E_\Om|$. 
    \footnotetext{
        The $\cap$-closure of $X$ 
        is in fact equivalent to the <<region-graph>> condition
        of \cite{Yedidia-2005}.
        Note that $X$ was assumed closed under $\cap$ since section 2.3, 
        as this property was necessary for the interaction decomposition to hold. 
    }
\end{prop}

\begin{proof}
    According to \ref{Fb-energy}, we only need to show that 
    $\check{S}([1]) = \ln |E_\Om| = \sum_i \ln |E_i|$. 
    The assumption on $X$ implies that
    for every $i \in \Om$, there exists a minimal $\bb_i \in X$ containing $i$ 
    so that:
    \begin{equation} \sum_{\aa \in X} c_\aa \ln |E_\aa|  
    = \sum_{\aa \in X} c_\aa \sum_{i \in \aa} \ln |E_i| 
    = \sum_{i \in \Om} \ln |E_i| \sum_{\aa \cont \bb_i} c_\aa 
    = \sum_{i \in \Om} \ln |E_i| 
    \end{equation}
\end{proof}

Before proceeding to the proof of theorem \ref{cvm}, 
we finally introduce the slightly finer characterisation 
of critical points of $\Fb$ given by \ref{cvm2}. 
It will be especially useful in proving the correspondence 
theorem \ref{correspondence} in the next chapter. 

\begin{defn} \label{div'}
    Denote by $\div'$ the truncation of the boundary $\div$ 
    to $X \setminus \{ \vide \}$: 
    \begin{equation} \div'_\aa \ph =  \div_\aa \ph \txt{if} \aa \neq \vide 
    \txt{and} \div'_\vide \ph = 0 \end{equation}
\end{defn} 

\begin{thm} \label{cvm2}
    Assuming $H_\vide = 0$, a consistent statistical state $p \in \Gammint(X)$ 
    is critical for the constrained Bethe free energy 
    $\Fb(\,-\,,H)_{|\Gamma(X)}$ if and only if there exists 
    $\ph \in A_1(X)$ such that: 
    \begin{equation} - \ln(p) = H + \zeta \cdot \div' \ph \end{equation} 
\end{thm}

\begin{proof}[Proof of theorem \ref{cvm}]
    To account for the normalisation constraints $\croc{p_\aa}{1} = 1$,
    we may describe the cotangent space at $p$ of $\Del_0(X) \incl A^*_0(X)$ 
    as the quotient $A_0(X) / \R_0(X)$ 
    and write:
    \begin{equation} \frac {\dr \Fb} {\dr p} \simeq 
    \sum_{\bb \in X} c_\bb \big( H_\bb + \ln (p_\bb) \big) \mod \R_0(X) \end{equation} 
    The flux term comes as a collection of Lagrange multipliers
    for the consistency constraints,
    a consistent $p \in \Gammint(X)$ being critical if and only if 
    the differential of $\Fb(\,-\,,H)$ 
    vanishes on $\Ker(d) = \Img(\div)^\perp$ or:
    \begin{equation} c\big( H + \ln(p)\big) \in \Img(\div) + \R_0(X) \end{equation}
    Proposition \ref{zeta-bord} is crucial\footnote{
        As noticed by D. Bennequin, 
        the original proof given in \cite{Yedidia-2005} is problematic
        when there exists $\bb$ such that $c_\bb = 0$.
    }
    to state that the above is equivalent to:
    \begin{equation} H + \ln(p) \in \zeta \cdot \Img(\div)  + \R_0(X) \end{equation}
\end{proof}

\begin{proof}[Proof of theorem \ref{cvm2}]
    First note that $\zeta \cdot \div' \ph \simeq 
    \zeta \cdot \div \ph \mod \R_0(X)$,
    so that we only need to reduce the Lagrange multipliers of \ref{cvm} 
    to the form given by \ref{cvm2}. 
    Assume there exists $\psi \in A_1(X)$ and $\lambda \in \R_0(X)$ 
    such that:
    \begin{equation} - \ln(p) = H + \zeta \cdot \div \psi + \lambda \end{equation}
    Define $\ph \in A_1(X)$ by letting $\ph_{\aa\bb} = \psi_{\aa\bb}$ 
    for non-empty $\bb$, and otherwise letting: 
    \begin{equation} \ph_{\aa \vide} = \psi_{\aa\vide} - 
    \sum_{\aa \cont \bb} \mu_{\aa\bb}\: \lambda_\bb  
    \end{equation}
    Then $\zeta(\div \ph)_\aa = \zeta(\div \psi)_\aa + \lambda_\aa$ 
    for all non-empty $\aa$, while $- \ln(p_\vide) = 0$.
\end{proof}

\chapter{Message-passing and Diffusion}

    In this chapter, we consider dynamics over the space of statistical 
fields that enlarge a class of bayesian inference algorithms 
known as {\it message-passing algorithms}. 
The generalised belief propagation algorithm introduced 
by Yedidia, Freeman and Weiss in \cite{Yedidia-2001} is the most interesting 
of these processes and reviewed in section 1.
For a general study of belief propagation on networks and their relations 
to statistical physics, see \cite{Mezard-Montanari}. 

\begin{center}
\includegraphics[width=0.6\textwidth]{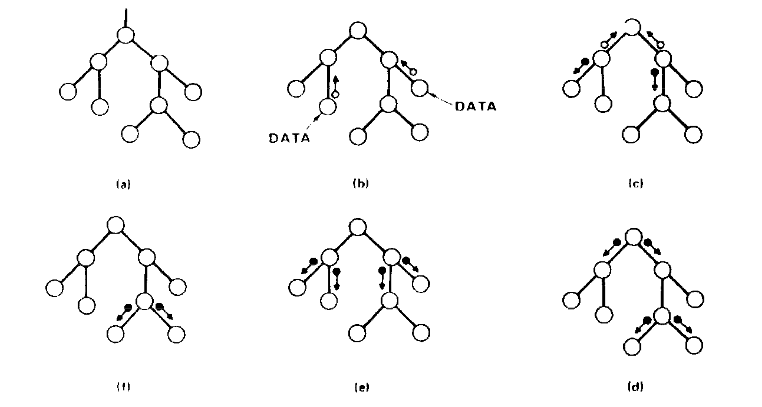}
\end{center}

The main contribution of this thesis is to reveal the homological 
character of message-passing algorithms. 
Among the deep consequences of this new view 
stands a collection of conservation laws.  
Section 2 shows how, at the level of energies, 
their discrete dynamics approximate 
a continuous-time transport equation of the form $\dot u = \div \Phi(u)$.
Recovering belief propagation through a naïve Euler scheme
of time step 1, the present approach yields new algorithms at smaller time
scales. Shorter time steps seem highly advisable to avoid unstable behaviour. 

Finally, building on the higher degree combinatorics of chapter 3, 
we propose a combinatorial enhancement of the algorithm
and its continuous version,  
through a degree-one Möbius inversion on the energy flux. 
We show in section 3 how this canonical diffusion
eliminates redundancies and natively allows 
to enforce Dirichlet boundary conditions, a fundamental component of learning.

\section{Belief Propagation} 

Belief propagation was first introduced as
a decoding algorithm by Gallager in his 1960 PhD thesis 
on low-density parity-check codes \cite{Gallager-63}.
In his setting, a collection of code bits and parity-check bits,
transmitted through a noisy communication channel, 
exchange messages to update the belief on their true value
until parity-check consistency is achieved. 
The algorithm was later rediscovered by Pearl in 1982 \cite{Pearl-82}
to perform exact inference on bayesian trees, then generalised
by Yedidia {\it et al.} in 2001 \cite{Yedidia-2001} 
on more general coverings $X \incl \Part(\Om)$. 

This introductory section first defines
the algorithm in its general form, 
before reviewing some of its remarkable properties,
a starting point of this work
having been the theorem of \cite{Yedidia-2005}
establishing a correspondence between fixed points 
of the algorithm and critical points of a Bethe free energy. 

\subsection{Algorithm}

Given $X \incl \Part(\Om)$, generalised belief propagation 
assumes priors are described by a collection $(f_\aa)$
of strictly positive observables indexed by $\aa \in X$.
The dynamic takes place on a set of messages $(m_{\aa\bb})$
according to the update rule:
\begin{equation} m_{\aa\bb} \wa m_{\aa\bb} \cdot 
\Sigma^{\bb\aa} \left( 
\frac{
    \prod\limits_{\bb' \in \Lambda^\aa \setminus \Lambda^\bb} f_{\bb'} 
    \; \times
    \prod\limits_{\aa'\bb' \in d\LL^{\aa} \setminus d\LL^{\bb}} m_{\aa'\bb'}
}{
    \prod\limits_{\bb'\cc' \in d\LL^{\bb} 
    \setminus d\LL^{\aa}} m_{\bb'\cc'} 
}
\right)
\end{equation}
where $\Lambda^\aa \incl X$ still stands for
the cone of subregions of $\aa$ and $d\LL^{\aa}$ for its 
coboundary\footnote{
    $d\LL^\aa$ is the set of ordered pairs $\aa' \cont \bb'$ such that 
    $\aa' \not\in \Lambda^\aa$ and $\bb' \in \Lambda^\aa$
    (see 3.2.3). Our formula stays as close as possible 
    as that of \cite{Yedidia-2005} 
    but we already write the products over set differences of cone coboundaries. 
}. 
Denoting by $G_\aa \incl A_\aa$ the multiplicative group 
of strictly positive observables for every $\aa \in X$, 
the algorithm thus iterates over $m \in G_1(X)$ given $f \in G_0(X)$.
This crude yet classical formula 
justifies the alternative name of {\it sum-product algorithm}.

Priors and messages serve to define a collection 
of beliefs $(q_\aa)$ according to the formula:
\begin{equation}\label{beliefs}
q_\aa = \Bigg[\:
\prod_{\bb' \in \Lambda^\aa} f_{\bb'} 
\; \times 
\prod_{\aa'\bb' \in d\LL^{\aa}} m_{\aa'\bb'} 
\: \Bigg] 
\end{equation}
the normalisation bracket projecting 
$G_\aa$ onto the subspace $\Del_\aa$ of non-vanishing probability 
densities. 
The dynamic of the beliefs is bound to that of the 
messages, now following the much nicer update rule:
\begin{equation}\label{messages}
m_{\aa\bb} \wa m_{\aa\bb} \cdot \frac {\Sigma^{\bb\aa}(q_\aa)} {q_\bb}
\end{equation}
Stationary states of the algorithm hence correspond 
to consistent\footnote{ 
The beliefs $(q_\aa)$ are consistent when $q_\bb$ is the marginal of $q_\aa$ 
whenever $\bb$ is contained in $\aa$.
} 
beliefs $(q_\aa)$ 
searched within a particular subspace defined 
by the priors $(f_\aa)$. We shall relate this subspace 
to a homology class of $A_0(X)$ in section 2. 
Another result not yet stated to our knowledge, 
is that stationarity of beliefs implies stationarity of messages. 
This will justify to focus on the dynamic over $q \in \Del_0(X)$ following
the update rule: 
\begin{equation} \label{belief-update}
q_\aa \wa \Bigg[ \:
q_\aa \; \times \prod_{\aa'\bb' \in d\LL^\aa} 
\frac{\Sigma^{\bb'\aa'}(q_{\aa'})}{q_{\bb'}}
\: \Bigg]
\end{equation}
We will denote by $\BP : \Del_0(X) \aw \Del_0(X)$ the associated
smooth map, the algorithm reading $q \wa \BP(q)$. 
It will only be a matter of preference whether to keep track of the messages
$m \in G_1(X)$ over time, 
the relevance of the final messages being questionable as the
assignment $(f, m) \mapsto q$ is not injective. 

\subsection{Properties}

Belief propagation was initially considered on graphs, 
which we describe by coverings $X \incl \Part(\Om)$
containing the vertex $\{ i \}$ for all $i \in \Om$,
and otherwise consisting only of edges of the form $\aa = \{ i, j \}$. 
In general, the underlying graph of Gallager's algorithm contains loops, but 
Pearl's algorithm, describing hierarchical data structures, 
was however restricted to acyclic graphs.  
{\it Retractable hypergraphs} $X \incl \Part(\Om)$ shall be defined in 
\ref{retractable}
to generalise the acyclic property of trees, 
the following theorem will then extend Pearl's result. 
On such retractable systems, the fundamental property justifying 
interest in belief propagation is 
the exact, finite-time and parallelised algorithms it provides 
to compute the probabilistic model's marginals. 

\begin{thm}
    Assume $X \incl \Part(\Om)$ is retractable\footnote{
        See definition \ref{retractable}.
    }.
    Given positive priors $(f_\aa)$ and initial messages $(m_{\aa\bb})$, 
    the beliefs $(q_\aa)$ iterated through $\BP$ converge 
to the exact marginals of the Markov field: 
\begin{equation} p_\Om = \Bigg[\: \prod_{\aa \in X} f_\aa \:\Bigg] \end{equation}
Moreover, convergence is reached in finite time less
or equal to the diameter of $X$.
\end{thm}

When $X$ is a graph with loops, in general when $X$ is an unretractable
hypergraph, belief propagation still performs 
approximate bayesian inference surprisingly well. 
The following theorem of \cite{Yedidia-2005} beautifully bridges 
bayesian learning with statistical physics and greatly motivated the
present work. It expresses that stationary states of belief propagation 
actually compute the critical points of a Bethe free energy approximation, 
which according to Kikuchi's cluster variation method, should give 
a good estimate of the marginals of the Markov field $p_\Om = [\e^{-H_\Om}]$.

\begin{thm} 
    For any $X \incl \Part(\Om)$, the fixed points of $\BP$ 
    with priors $(f_\aa)$ are in one-to-one correspondence 
    with the critical points of $\Fb_H$ with respect to the local
    hamiltonians: 
    \begin{equation} H_\aa = \sum_{\bb' \in \LL^\aa} - \ln f_{\bb'} \end{equation}
\end{thm}

It was recognised by D. Bennequin that the proof
given by Yedidia {\it et al.} is problematic when there exists $\aa \in X$ 
such that the Möbius number $c_\aa$ vanishes, 
and we shall correct their proof in section 2.

In the general setting, the uniqueness of equilibria is not maintained.
Although convergence of the algorithm remains an open question, 
the existence of at least one fixed point is insured.

\begin{thm} \label{existence}
    For any $X \incl \Part(\Om)$ and set of priors, 
    there exists a fixed point of $\BP$.
\end{thm}

\begin{proof}
The topology of $\Delta_0(X)$ is that of a product of spheres and the 
proof given in \cite{Bennequin-IEM} relies on Brouwer's fixed point theorem,
after having shown that $\BP : \Del_0(X) \aw \Del_0(X)$ keeps away
from the boundary of $\Delta_0(X)$.
\end{proof}

\section{Statistical Diffusions}

We now introduce an ordinary differential equation 
of the form $\dot u = \div \Phi(u)$ on interaction potentials. 
This transport equation is very reminiscent of the heat equation 
and similar diffusions\footnotemark{}. 
Trying to understand the geometric nature of message-passing 
formulas and unveiling their connections with algebraic topology was the first
motivation of this thesis: 
we relate belief propagation to a 
coarse integrator of this diffusion flow. 
\footnotetext{
Apart from the non-linearity of the algorithm, 
the main difference lies in the central role of
the zeta transform defining local hamiltonians 
from interaction potentials. See section 6.3.
}

This homological picture completes 
the correspondence between the stationary states of belief propagation 
and the critical points of Bethe free energy $\Fb(\,-\,,H)$, 
as described by theorem \ref{cvm}. 
Both lie at the intersection of a homology class 
of interaction potentials\footnote{
    Technically $U \in H + \zeta \cdot \Img(\div)$ should be understood up to 
    additive constants. We shall discuss how to best deal with 
    normalisation constraints later on. 
} with the space of consistent beliefs,
through the non-linear correspondence summarised in table \ref{table-beliefs}:
\begin{equation} [\e^{-U}] \in \Gammint(X) \txt{and} U \in  H + \zeta \cdot \Img(\div) \end{equation}
Behind this correspondence is an interesting form of duality, exchanging constraints
with degrees of freedom. The free energy variational problem looks for constrained
$q \in \Gammint(X)$ such that variations are generated by Lagrange multipliers, 
while the algorithm iterates over $\ph \in A_1(X)$ until beliefs 
eventually reach a consistent equilibrium state. 
\begin{table}[H]
\renewcommand{\arraystretch}{1.6} 
    \vspace{0.3cm}
\begin{center}
    \begin{tabular}{| C{4cm} | C{4cm} |}
        \hline
        interaction potentials & $u = h + \div \ph$ \\
        \hline
        local hamiltonians & $U = \zeta \cdot u$ \\
        \hline
        beliefs & $q = \big[\e^{-U} \big]$ \\
        \hline
    \end{tabular}
    \caption{\label{table-beliefs}
        Interaction potentials, local hamiltonians, 
        and local Gibbs states -- or beliefs. 
    }
    \vspace{-0.3cm}
\end{center}
\end{table}
Decomposing the dynamic into elementary operations,
the vector field we introduce is of the form 
$\Tspt = \div \circ (-\DF) \circ \zeta$ on the vector space 
$A_0(X)$ of interaction potentials:
\begin{equation} \label{bp-diagram}
\bcd 
A_0(X) \rar{\zeta} & A_0(X) \dlar{-\DF} \\
A_1(X) \uar{\div} 
\ecd
\end{equation}
The definition of the non-linear map 
$\DF$ is the object of the first subsection. 
Preparing for the study of other flux functionals, 
we then characterise a family of transport equations 
of the form $\dot u = \div \Phi(u)$ sharing the common property 
of yielding critical points of Bethe free energy at equilibrium.
Specialising to the case where $\Phi = -\DF \circ \zeta$, we show 
how explicit Euler schemes on $A_0(X)$ generalise the discrete dynamic 
of $\BP$ on $\Del_0(X)$ to arbitrary time scales. 

\subsection{Effective Energy Gradient and Consistency} 

\bgroup
\renewcommand{\Z}{{\cal Z}}

The effective energy gradient $\DF \in C^\infty\big( A_0(X), A_1(X)\big)$ 
will be essential in defining the currents carrying energy from 
one region to another. 
Cancellation of these currents will then define equilibrium as 
a collection of effective hamiltonians 
$U \in A_0(X)$ whose effective energies are consistent with one another.

\begin{defn} 
We call effective energy gradient 
the smooth map $\DF : A_0(X) \aw A_1(X)$ defined by: 
\begin{equation} \DF(H)_{\aa\bb} = H_\bb - \Fh^{\bb\aa}(H_\aa) \end{equation}
    where $\Fh^{\bb\aa}(H_\aa) = - \ln \Sigma^{\bb\aa}(\e^{-H_\aa})$ 
    is the effective energy as defined in 4.1.2. 
\end{defn}

The zero locus of $\DF$ 
is naturally diffeomorphic to the space of consistent positive measures 
by \ref{eff-marginal}. 
There will be two different ways to account for the normalisation constraints.

\begin{defn} \label{CU}
    We say that a collection of local hamiltonians $U \in A_0(X)$ is:
    \bi
    \iii {\rm consistent} if $\DF(U) = 0$, 
    \iii {\rm projectively consistent} if $\DF(U) \in \R_1(X)$.
    \ei
    We denote by $\CU(X) \incl \CU'(X)$ the spaces of consistent 
    and projectively consistent local hamiltonians.
\end{defn}

Note for instance that $0 \in A_0(X)$ is only projectively consistent as:
\begin{equation} \DF(0)_{\aa\bb} = \ln |E_\bb| - \ln |E_\aa| \end{equation}
There is however a unique consistent $U = \ln |E| \in \R_0(X)$ 
such that $U_\vide = 0$, accounting for the entropic 
contributions in the high temperature limit.

\begin{prop} \label{consistent-U}
Consistent probability densities may be parametrised by local hamiltonians 
via:
\bei
\item 
$\CU'(X)$ is the inverse image of $\Gammint(X)$ 
under $U \mapsto [\e^{-U}]$ and:
\begin{equation} \Gammint(X) \simeq \CU'(X) / \R_0(X) \end{equation} 
\item 
$\{ U \in \CU(X) \st U_\vide = \lambda \}$ 
is diffeomorphic to $\Gammint(X)$ under $U \mapsto \e^{-(U -\lambda)}$ 
for all $\lambda \in \R$ 
and: 
    \begin{equation} \Gammint(X) \simeq \CU(X) / \R \end{equation}
\ee
\end{prop}

\begin{proof} 
    (i) Letting $p = [\e^{-U}]$, for all $\aa \cont \bb$ 
    one has $p_\bb = \Sigma^{\bb\aa}(p_\aa)$ 
    if and only if $U_\bb = \Fh^{\bb\aa}(U_\aa) \mod \R$. 
    As $[\e^{-U}] = [\e^{-U'}]$ if and only if 
    $U' \in U + \R_0(X)$, the quotient $\CU'(X) / \R_0(X)$ 
    is diffeomorphic to $\Gammint(X)$. \\[0.4em]
    (ii)
    $\DF(U) = 0$ implies consistency and 
    $\Sigma^{\vide \aa}(\e^{-U_\aa}) = \e^{-U_\vide} = \e^{- \lambda}$ 
    holds a single normalisation factor. Reciprocally 
    when $p \in \Gammint(X)$, letting $U = - \ln(p) + \lambda$
    we have $\DF(U) = 0$ with $U_\vide = -\ln(1) + \lambda = \lambda$. 
\end{proof}

The manifold of consistent local hamiltonians $\CU(X)$ 
has a natural riemannian structure. 
Given any consistent positive 
$p \in \Gammint(X)$, consider the inner product on $A_0(X)$ defined by:
\begin{equation} \label{metric-p} 
\croc{f}{g}_p = \sum_{\aa \in X} \E_{p_\aa}[ f_\aa \cdot g_\aa ] 
\end{equation}
As $\CU'(X)$ is mapped onto $\Gammint(X)$,
a metric of the form (\ref{metric-p}) is 
naturally associated to any $U \in \CU'(X)$. 
The tangent space $T_U \CU(X)$ 
is moreover described by the cohomology 
of a differential $\nabla$, 
adjoint of $\div$ for the metric induced by $[\e^{-U}]$, 
a characterisation which will be particularly useful in section 6.3.


\begin{defn} \label{nabla}
    We denote\footnotemark{} 
    by $\nabla = \DF_*$ the linearised effective energy gradient: 
    \begin{equation} \nabla(H)_{\aa\bb} = H_\bb - \E^{\bb\aa}[H_\aa] \end{equation}
    viewed as the smooth map $\nabla : \CU'(X) \aw \Hom\big(A_0(X), A_1(X)\big)$, 
    letting $\E^{\bb\aa} = d \Fh^{\bb\aa}$ as per \ref{Eba}. 
    \footnotetext{
        To avoid burdening notations, we will often leave the dependency 
        of $\nabla$ in $p \in \Gammint(X)$ or $U \in \CU'(X)$ implicit. 
    }
\end{defn}

With this notation, we have $T \CU(X) = \Ker(\nabla)$ 
as a subbundle of $TA_0(X)$ restricted above $\CU(X)$.
As the following suggests, 
$\nabla$ naturally extends to a differential 
acting on all degrees of $A_\bullet(X)$. 

\begin{prop} 
    Given $p = [\e^{-H}] \in \Gamma(X)$, 
    we have $\nabla = \div^*$ for the metric induced by $p$. 
\end{prop} 

\begin{proof} 
    This is the consequence of the fact 
    that the conditional expectation $\E^{\bb\aa} : A_\aa \aw A_\bb$  
    is adjoint to the canonical extension $j_{\aa\bb} : A_\bb \aw A_\aa$  
    for the metrics induced by $p_\aa$ and $p_\bb$: 
    \begin{equation} \sum_{E_\bb} g_\bb \cdot p_\bb \cdot \E^{\bb\aa}[f_\aa] 
    = \sum_{E_\aa} j_{\aa\bb}(g_\bb) \cdot p_\aa \cdot f_\aa \end{equation}
    Standard computations\footnote{
        The proof of $\delta = d^*$ is then perfectly analogous 
        to the case of scalar coefficients, 
        see Kodaira \cite{Kodaira-49} for instance. 
    } then show 
    that $\croc{\nabla f}{\ph}_p = \croc{f}{\div \ph}_p$ for all 
    $f \in A_0(X)$ and $\ph \in A_1(X)$.
\end{proof}

Acting on local hamiltonians, the effective energy gradient $\DF$ 
will be generally precomposed by $\zeta$. 
Writing the hamiltonian $H = \zeta \cdot h$ as a sum of local interactions, 
one may think of $\DF(H)_{\aa\bb}$ as the effective contribution 
of $\LL^\aa \setminus \LL^\bb$ to the energy of $\LL^\bb$:
\begin{equation} \DF(\zeta \cdot h)_{\aa\bb} = 
\Fh^{\bb\aa}\bigg( 
\sum_{\bb' \in \LL^\aa \setminus \LL^\bb} h_{\bb'} 
\bigg) 
\end{equation}
To complete the picture on consistency, 
let us finally parametrise 
$\Gammint(X)$ by interaction potentials. 

\begin{defn} \label{Z} 
We denote by: 
\bi 
\iii 
    $\Z(X) = \mu \cdot \CU(X)$ 
the manifold of {\rm consistent interaction potentials}, 
\iii 
    $\Z'(X) = \mu \cdot \CU'(X)$ 
the manifold of 
{\rm projectively consistent interaction potentials}. 
\ei
\end{defn}

\begin{prop} \label{consistent-u}
Consistent probability densities may be parametrised by interaction potentials
via:
\bei
\item 
$\Z'(X)$ is the inverse image of $\Gammint(X)$ 
under $u \mapsto [\e^{-\zeta \cdot u}]$ and:
\begin{equation} \Gammint(X) \simeq \Z'(X) / \R_0(X) \end{equation} 
\item 
$\{ u \in \Z(X) \st u_\vide = \lambda \}$ 
    is diffeomorphic to $\Gammint(X)$ under $u \mapsto \e^{\lambda -\zeta \cdot u}$ 
for all $\lambda \in \R$ 
and: 
    \begin{equation} \Gammint(X) \simeq \Z(X) / \R \end{equation}
\ee
\end{prop} 

\begin{proof}
    This follows from \ref{consistent-U} as 
    $\zeta \cdot \R_0(X) = \R_0(X)$, 
    while for $u = \mu \cdot U$ one has $u_\vide = U_\vide$.
\end{proof}

The origin $\bar 0 \in \Z(X)$, 
associated to the uniform state $p = [1]$,
describes the high temperature limit where all variables are independent. 
The following characterisation 
of $T_{\bar 0}\Z(X)$ should be seen as a consequence of this 
independency\footnote{
    The conditional expectation $\E^{\bb\aa}$ 
    is always adjoint to the inclusion $j_{\aa\bb}$
    for the metric induced by the probability density, 
    as orthogonal projection of $A_\aa$ onto $A_\bb$. 
    The effect of interactions between $\bb$ and $\bb'$ is however
    that $A_\bb$ is no longer orthogonal to $A_{\bb'}$ in $A_\aa$. 
    Although one could define interaction subspaces as 
    $Z_\aa = \cap_{\bb \inclst \aa} \Ker(\E^{\bb\aa})$ 
    they would hence not satisfy $\E^{\bb\aa}(Z_{\bb'}) = \{ 0 \}$. 
}.

\begin{prop} \label{T_0 Z}
Letting $\bar 0 = \mu \cdot \ln |E| \in \Z(X)$, one has: 
    \begin{equation} T_{\bar 0} \Z(X) = Z_0(X) \end{equation} 
where $Z_0(X) \incl A_0(X)$ is the image 
of the canonical interaction decomposition defined in 2.3.2. 
\end{prop}

\begin{proof}
This is equivalent to $\Ker(d) = \zeta' \cdot Z_0(X)$ as proved in \ref{Ker-d}. 
Conditional expectations $\E^{\bb\aa}$ w.r.t. uniform probabilities are proportional 
to the partial integrations $\Sigma^{\bb\aa}$ w.r.t. counting measures 
as $\E^{\bb\aa} = \frac 1 {|E_{\aa \setminus \bb}|} \Sigma^{\bb\aa}$. 
The volumic factors are what we need 
for $j_{\aa\bb}$ to be a {\it section} of $\E^{\bb\aa}$, 
meaning $\E^{\bb\aa} \circ j_{\aa\bb} = {\rm id}_{A_\bb}$, 
and the combinatorial argument of \ref{Ker-d} to hold.
\end{proof}

\subsection{Diffusions and Correspondence Theorems}

Given a smooth {\it flux functional} $\Phi : A_0(X) \aw A_1(X)$, 
we consider the transport equation:  
\begin{equation} \label{tspt}
\dot u = \div \Phi(u) 
\end{equation}
Let us mention two main differences of the present
approach with usual message-passing algorithms:
\be 
\item (\ref{tspt}) is an ordinary differential equation, and 
    $\dot u = \frac{du}{dt}$ represents the derivative of $u: \R \aw A_0(X)$
    with respect to a continuous time variable\footnotemark{}. 
\item (\ref{tspt}) is a dynamic on the vector space of interaction potentials.  
A dynamic is induced on the multiplicative group  of positive beliefs 
by letting $q = [\e^{- \zeta \cdot u}]$ but we do not take the usual point 
of view of a dynamic over messages.
\ee 
\footnotetext{ 
    The possibility to differentiate w.r.t. continuous time 
    was probably occluded by the common preference for the multiplicative point of 
    view of <<beliefs>>. 
    Ironically, there is absolutely no originality in 
    the additive point of view of <<energies>> as Gallagher's electronic apparatus 
    was already logarithmic -- additions are simpler than multiplications for 
    humans and machines all the same. 
}
Difference 1 consists in a significant improvement for the stability of 
such algorithms, belief propagation being recovered through the 
unreasonably coarse finite difference approximation $\dot u(t) \simeq u(t+1) - u(t)$. 
Difference 2 makes the homological character of message-passing more apparent. 
Note that a dynamic on $\ph \in A_1(X)$ may be recovered as 
$\dot \ph = \Phi(h + \div \ph)$ for some 
initial interaction potentials $h \in A_0(X)$. 

In this paragraph, we show the announced correspondence 
between stationary points of belief propagation 
and critical points of Bethe free energy. 
Although this correspondence is made more coherent by 
viewing belief propagation as a dynamic over beliefs rather than on messages,
this approach required us to prove that stationarity 
of beliefs implied stationarity of messages. 
This will come as a
consequence of the {\it faithfulness} of the flux functional 
$\Phi = - \DF \circ \zeta$, given by proposition \ref{faithful}.

More generally, the following definitions help characterise the
class of flux functionals admissible for (\ref{tspt}) 
to seek critical points of a Bethe free energy. 
Morally, consistency of $\Phi$ will imply that critical points 
are stationary under $\dot u = \div \Phi(u)$, while the 
stronger faithfulness property is necessary for the converse to hold.
The existence and design of 
an {\it optimal} flux functional satisfying these 
properties is a natural question then to be raised.

\begin{defn} 
    A smooth flux functional $\Phi : A_0(X) \aw A_1(X)$ will be said:
    \bi
    \iii {\rm consistent} if $u \in \Z(X) \impq \Phi(u) = 0$
    \iii {\rm faithful} if $u \in \Z(X) \eqvl \div \Phi(u) = 0$
    \iii {\rm projectively faithful} if $u \in \Z'(X) \eqvl \div \Phi(u) \in \R_0(X)$
    \iii {\rm locally faithful} if the restriction of $\Phi$ to a neighbourhood 
    of $\Z(X)$ is faithful. 
    \ei
\end{defn}

Suppose now given 
reference interaction potentials $h \in A_0(X)$ such that $h_\vide = 0$, 
and let $H = \zeta \cdot h$. 
The correspondence is essentially a rephrasing of 
theorem \ref{cvm2}, characterising critical points of $\Fb(\,-\,, H)$, 
all difficulties being kept hidden behind the faithfulness assumption. 
To account for normalisation\footnotemark{}, 
we still denote by $\div'$ the truncation of $\div$ to $X \setminus \{ \vide \}$
as defined in \ref{div'}. 
\footnotetext{
    Although belief propagation has to normalise each step, 
    scaling factors being otherwise unstable on graphs with loops, 
    first experiments suggest that this normalisation procedure might be
    superfluous for finer integrators of (\ref{tspt}). 
}

\begin{thm} \label{correspondence} 
    Assume $\Phi$ is faithful. 
    Then for all $u \in A_0(X)$, the following are equivalent:
    \bei
    \item There exists $\ph$ such that 
    $u = h + \div' \ph$ is stationary for $\dot u = \div' \Phi(u)$. 
    \item The beliefs $\e^{- \zeta \cdot u}$ are critical for $\Fb(\,-\,, H)$
    constrained to $\Gammint(X)$. 
    \ee
\end{thm}

\begin{lemma} 
    For all $\ph \in A_1(X)$ if $\div' \ph = 0$ then $\div \ph = 0$. 
\end{lemma} 

\begin{proof}
    From the global Gauss formula 
    $\sum_\bb \div_\bb \ph = 0$, if $\div_\aa \ph = 0$ 
    for all $\aa \neq \vide$ then $\div_\vide \ph = 0$, see \ref{gauss}.
\end{proof}

\begin{proof}[Proof of theorem \ref{correspondence}] 
    Recall that $q \in \Gammint(X)$ is critical for
    $\Fb(\,-\,, H)$ constrained to $\Gamma(X)$  
    if and only if there exists Lagrange multipliers $\ph \in A_1(X)$ such that 
    $- \ln(q) = H + \zeta \cdot \div' \ph$ by theorem \ref{cvm2}.
    \bi
    \iii
    Assume $q$ is critical and let $\zeta \cdot u = - \ln(q)$. 
    Then there exists $\ph$ such that $u = h + \div' \ph$ 
    while $\DF(\zeta \cdot u) = 0$ by consistency. 
    We have $\div' \Phi(u) = \div\Phi(u) = 0$ by faithfulness of $\Phi$ and
    $u$ is stationary. 
    \iii
    Reciprocally, assume $u = h + \div' \ph$ is stationary.
    According to the lemma, $\div' \Phi(u) = 0$ implies 
    $\div \Phi(u) = 0$ so that 
    $\DF(\zeta \cdot u) = 0$ by faithfulness of $\Phi$. 
    Letting $U = \zeta \cdot u$, it follows that 
    $\DF(U) = 0$ with $U_\vide = 0$ hence $q = \e^{-U} \in \Gammint(X)$ 
    is consistent and critical for the constrained free energy. 
    \ei
\end{proof} 

\begin{thm} \label{proj-correspondence} 
    Assume $\Phi$ is projectively faithful. 
    For all $u \in A_0(X)$, the following are equivalent:
    \bei
    \item There exists $\ph$ such that 
    $u = h + \div' \ph$ is projectively stationary for $\dot u = \div' \Phi(u)$. 
    \item The beliefs $[\e^{- \zeta \cdot u}]$ are critical for $\Fb(\,-\,, H)$
    constrained to $\Gammint(X)$. 
    \ee 
\end{thm}

\begin{proof}
{\color{white} blank}\\[-0.5cm]

    \bi
    \iii If $q = [\e^{- \zeta \cdot u}] \in \Gammint(X)$ is critical 
    for $\Fb(-, H)_{|\Gamma(X)}$, 
    there exists $\ph$ such that $u = h + \div' \ph$ by \ref{cvm2}.
    Consistency of $q$ implies $u \in \Z'(X)$ by \ref{consistent-u} 
    and $\div \Phi(u) \in \R_0(X)$ by projective faithfulness of $\Phi$. 

    \iii Reciprocally, assume $u = h + \div'\ph$ is projectively stationary. 
    As $\div \Phi(u) \in \R_0(X)$ implies $u \in \Z'(X)$ by projective faithfulness, 
    $q = [\e^{- \zeta \cdot u}]$ is consistent and critical for the 
    constrained free energy. 
    \ei
\end{proof} 

Let us now show that the flux functional 
inducing belief propagation 
is faithful. 
The algorithm shall be recovered as a coarse integrator 
of the differential equation:
\begin{equation}
    \ba{ll|l} \label{transport}
    \dot u = \div \ph &\quad\txt{where}\;\;\quad 
    & \ph = - \DF(U) \\[0.4em]
    & & U = \zeta \cdot u 
\ea
\end{equation}
Transport takes place at the level of effective interaction potentials $u$,
an energy flux $\ph$ balancing the effective hamiltonians $U = \zeta \cdot u$ 
until they reach effective consistency. 
The evolution of 
$u$ being restricted to a single homology class, 
let us emphasise that 
total energy is conserved along any integral curve of (\ref{transport}). 
Hence for any reference interaction potentials $h$ homologous to $u(0)$, 
one has: 
\begin{equation} \label{tspt-conservation}
\sum_{\aa} u_\aa(t) = \sum_{\aa} h_\aa 
\end{equation}
This is a direct consequence of \ref{conservation} and 
obviously holds along integral curves of (\ref{tspt}) as well.

\begin{defn} 
We call: 
\bi
    \iii {\rm standard diffusion flux} the functional 
    $\Phi: A_0(X) \aw A_1(X)$ defined by $\Phi = - \DF \circ \zeta$, 
    \iii {\rm standard diffusion} the vector field
    $\Tspt$ on $A_0(X)$ defined by $\Tspt = \div \Phi$. 
\ei
\end{defn}

\begin{prop} \label{faithful} 
    The standard diffusion flux is faithful.
\end{prop}

\begin{proof}
The proposition 
claims that $\div \DF(H) = 0$ implies $\DF(H) = 0$ for any local hamiltonians $H$.
Denoting by $d$ the adjoint of $\div$ 
for the canonical metric of $A_0(X)$, we have  
    for every $v, H \in A_0(X)$:
    \begin{equation} \croc{v}{\div \DF(H)} = \croc{d(v)}{\DF(H)} \end{equation}
    Assuming 
    that $\div \DF(H) = 0$ and letting $v = \e^{-H}$ 
    the above integration by parts formula gives:
\begin{equation} \croc{d(\e^{-H})}{\DF(H)} = \sum_{\aa\bb \in N_1(X)} 
    \Big\langle
        \e^{-H_\bb} - \Sigma^{\bb\aa}(\e^{-H_\aa}) 
    \:\Big| \:
        H_\bb + \ln \Sigma^{\bb\aa}(\e^{-H_\aa}) 
    \Big\rangle_{E_\bb}
    = 0
\end{equation}
    It then follows by 
    monotonicity of $y \mapsto - \ln(y)$ that the differences
    $v_\bb - \Sigma^{\bb\aa}(v_\aa)$ and 
    $H_\bb - \Fh^{\bb\aa}(H_\aa)$ have opposite signs. As
    $v_\bb = \Sigma^{\bb\aa}(v_\aa)$ is equivalent to 
    $H_\bb = \Fh^{\bb\aa}(H_\aa)$,
    we do have $\DF(H) = 0$.
\end{proof}

In the next paragraph, we 
shall enforce normalisation constraints through a Möbius inversion
on the flux of $\Phi$ outbound to $\vide$. 
We introduce the flux functional $\Phi'$ 
defined by $\Phi'(u)_{\aa\bb} = \Phi(u)_{\aa\bb}$ when $\bb \neq \vide$ 
and otherwise by:
\begin{equation} \label{flux-bpn}
    \Phi'(u)_{\aa\vide} = 
    \sum_{\aa \aw \bb'} \mu_{\aa\bb} \cdot \Phi(u)_{\bb\vide} 
\end{equation}
Let us show that $\Phi'$ is projectively faithful. 

\begin{prop} \label{bpn-faithful-lemma}
    For every $H \in A_0(X)$ we have the equivalence: 
    \begin{equation} \div \DF(H) \in \R_0(X) \quad\eqvl\quad 
    \exists \lambda \in \R_0(X) \txt{s.t.} \DF(H - \lambda) = 0 
    \end{equation}
\end{prop} 

\begin{proof} 
    Letting $\Delta = \div d$ denote the laplacian on $\R_0(X)$, 
    we have $\Img(\Delta) = \Img(\div) + \Img(d)$ 
    by Hodge decomposition.
    By additivity of effective energy along constants, 
    one has  
    $\DF(H - \lambda) = \DF(H) - d \lambda$ 
    for all $\lambda \in \R_0(X)$ so that: 
    \begin{equation} \div \DF(H - \lambda) = \div \DF(H) - \Delta(\lambda) \end{equation}
    If $\div \DF(H) = \lambda' \in \R_0(X)$, 
    then $\lambda' \in \Img(\div) \incl \Img(\Delta)$ 
    and there exists  $\lambda \in \R_0(X)$ such that 
    $\Delta(\lambda) = \lambda'$. 
    It follows that $\div \DF(H - \lambda) = 0$.
    The faithfulness property \ref{faithful} 
    of $\DF \circ \zeta$ then implies $\DF(H - \lambda) = 0$.  
\end{proof}

\begin{prop} 
    The flux $\Phi'$ defined by (\ref{flux-bpn}) is 
    projectively faithful.
\end{prop}

\begin{proof} 
    From proposition \ref{bpn-faithful-lemma}, 
    projective stationarity 
    $\div \Phi (u) \simeq 
    \div \Phi'(u) \simeq 0 \mod \R_0(X)$ implies the
    projective consistency of $u$ as 
    $\DF(\zeta \cdot u) = \DF(\lambda) \in \R_1(X)$ 
    for some $\lambda \in \R_0(X)$. 
\end{proof}

\subsection{Euler Schemes and Belief Propagation}

The ordinary differential equation (\ref{transport}) defines 
a continuous-time transport $\dot u = \div \Phi(u)$ of energy at the level of 
interaction potentials. 
Its flow may be approximated by common methods of 
numerical integration, 
but belief propagation is related to the 
time-step-1 
explicit Euler scheme $u \wa u + \div \Phi(u)$, 
which approximates the flow of the vector field $\Tspt = \div \Phi$ by:
\begin{equation} \e^{n \Tspt} \simeq (1 + \Tspt)^n \end{equation} 
    We warn against the use of such a coarse scheme\footnotemark{} as
$1$ may actually be met by the integrator's Lipschitz bound. 
A straightforward refinement is to reduce the time step to 
$0 < \lambda < 1$, yielding new and better-behaved belief propagation 
algorithms, while higher order integrators could also prove useful.
\footnotetext{ 
    Consider for instance the real ODE $\dot y = - a y$ and 
    the behaviour of $\e^{-a n \tau} \simeq (1 - a \tau)^n$ 
    for different values of $a \tau$. 
}

Defining effective hamiltonians as 
$U(t) = H + \zeta \cdot \div \ph(t)$, it is
the differential equation $\dot \ph = - \DF(U)$ 
that is more precisely related to the multiplicative algorithm\footnotemark{} 
of equations (\ref{beliefs}) and (\ref{messages}). 
The crucial ingredient in recovering (\ref{beliefs}) 
is the Gauss formula of proposition \ref{gauss}, which gives on every cone $\LL^\aa$: 
\begin{equation} \label{bp-gauss}
U_\aa(t) = H_\aa + \sum_{\aa'\bb' \in d \LL^\aa} 
\ph_{\aa' \bb'}(t) 
\end{equation}
Equation (\ref{bp-gauss}) is, up to additive constants, the 
logarithm of (\ref{beliefs}) defining beliefs from messages. 
Approximating the evolution $\dot \ph = - \DF(U)$ 
by the finite difference iteration $\ph \wa \ph - \DF(U)$ then yields:
the logarithm of the message update rule (\ref{messages}): 
\begin{equation}
    \ph_{\aa\bb} \wa 
    \ph_{\aa\bb}  
    - \ln \bigg( \frac { \Sigma^{\bb\aa} \e^{-U_\aa} } 
    {\e^{-U_\bb}} \bigg) 
\end{equation}
\footnotetext{
    Letting $U = - \ln(q)$ and $\ph = - \ln(m)$.
    Reference hamiltonians are related to priors 
    by $H = \zeta \cdot h$ with $h = -\ln(f)$.
}

Faithfulness of $\Phi = - \DF \circ \zeta$ justifies our choice 
to drop messages out of memory
and focus on the transport equation $\dot u = \div \Phi(u)$ instead.
Accounting for normalisation could be done by projecting the 
evolution of $u$ onto $A_0(X) / \R_0(X)$. 
As a different point of view, 
we show that splitting the flux $\Phi'$ of $(\ref{flux-bpn})$ 
as $\Phi' = \Phi'_{int} + \Phi'_{out}$, where $\Phi'_{out}$ gathers all the 
flux terms $(\Phi'_{\aa\vide})$ directed to $\vide$,  
allows to naturally enforce normalisation at each step. 
This only requires to replace $\div$ by its truncation $\div'$ 
to $X \setminus \{ \vide \}$, as defined in \ref{div'}, 
and prepares for the more general 
boundary conditions considered in the next section. 

\begin{prop} 
    Under the correspondence $q = \e^{- \zeta \cdot u}$ 
    belief propagation is a splitting scheme 
    for the transport equation $\dot u = \div' \Phi'(u)$
    associated to the decomposition $\Phi' = \Phi'_{int} + \Phi'_{out}$, 
    each term being integrated through an explicit Euler scheme of time step 1.
\end{prop}

\begin{proof} 
    Consider the evolution $\dot U = \X_{int}(U) + \X_{out}(U)$ 
    induced on the effective hamiltonians $U = \zeta \cdot u$,
    where $\X_{int}(U) = \zeta \cdot \div \Phi'_{int}(u)$
    and $\X_{out}(U) = \zeta \cdot \div' \Phi'_{out}(u)$. 
    Applying Gauss formulas on $\LL^\aa \setminus \{ \vide \}$, 
    one may view $\X_{int}$ and $\X_{out}$
    bound into and out of $\LL^\aa \setminus \{ \vide \}$ respectively.  
    \begin{equation} 
    \ba{lll}

    \disp \X_{int}(U)_\aa
    = \sum_{\aa'\bb' \in d\LL^\aa} \Phi'_{int}(u)_{\aa'\bb'}
    & \txt{with} & 
    \Phi'_{int}(u)_{\aa'\bb'} = \Fh^{\bb'\aa'}(U_{\aa'} - U_{\bb'}) \\
    
    \disp \X_{out}(U)_\aa  
    = - \sum_{\bb' \in \LL^\aa_\vide}
    \Phi'_{out}(u)_{\bb' \vide} 
    & \txt{with} & 
    \disp \Phi'_{out}(u)_{\bb'\vide} = 
    \sum_{\bb' \aw \cc'} \mu_{\bb'\cc'} \, \Fh^{\cc'}(U_{\cc'}) 

    \ea 
    \end{equation}
    Möbius inversion formulas 
    yield $\X_{out}(U)_\aa = - \Fh^\aa(U_\aa)$. Consider then the splitting scheme:
    \begin{equation} U(n + 1) \simeq 
    \big( 1 + \X_{out} \big) \circ \big( 1 + \X_{int} \big)
    \big( U(n) \big) \end{equation}
    The first step $U \wa U + \X_{int}(U)$ is, up to constants 
    the logarithm of the belief update rule (\ref{belief-update}).
    The second step $U \wa U + \X_{out}(U)$ corresponds 
    to normalising\footnote{
        The total outbound flux 
        $\div_\vide \Phi'_{out} = \sum_\aa c_\aa \, \Fh^\aa(U_\aa)$ 
        is a Bethe approximation of the total free energy $\Fh^\Om(U_\Om)$.
    }
    the belief 
    $q_\aa = [\e^{-U_\aa}] = \e^{-U_{\aa} + \Fh^\aa(U_\aa)}$. 
\end{proof}

Approximating the flow of $\Tspt$ by $\e^{n\Tspt} \simeq (1 + \Tspt)^n$ 
may lead to serious convergence and stability issues in regimes 
where the norm of $\Tspt$ and its Lipschitz bound are not strictly smaller than $1$. 
In simple examples with cycles,  
$\Tspt_*$ may
for instance presents periodic eigenvalues of the form $\e^{i 2\pi / n}$. 

A straightforward and recommendable improvement of $\BP$ is to 
reduce the time scale of the explicit Euler scheme 
and approximate the flow by  
$\e^{n \cdot \lambda  \Tspt} \simeq (1 + \lambda \Tspt )^n$.
Iterating the non-linear operator $(1 + \lambda \Tspt)$ corresponds 
to updating messages according to:
\begin{equation} \label{messages-l}
m_{\aa\bb} \wa m_{\aa\bb} \cdot 
\bigg( \frac{\Sigma^{\bb\aa}(q_\aa)}{q_\bb} \bigg)^\lambda 
\end{equation}
It is to be expected that 
reasonably small values of $\lambda$ around 0.5 may already change
the algorithm's behaviour dramatically. 
In some regimes, belief propagation has been reported to converge poorly after 
undergoing a kind of <<phase transition>>. Changing the time scale 
may probably overcome this limitation.

\begin{defn} 
    For every $\lambda > 0$, we call 
    {\rm belief propagation of time scale} $\lambda$ 
    the algorithm iterating over a collection $(q_\aa) \in \Del_0(X)$
    of beliefs according to the update rule: 
\begin{equation} 
q_\aa \wa \Bigg[ \:
    q_\aa \; \times \prod_{\aa'\bb' \in d\LL^\aa} 
    \bigg(
\frac{\Sigma^{\bb'\aa'}(q_{\aa'})}{q_{\bb'}}
    \bigg)^{\lambda}
\: \Bigg]
\end{equation}
We denote by $\BP_\lambda : \Del_0(X) \aw \Del_0(X)$ the smooth 
map inducing the above dynamic.
\end{defn}

The following identity is a straightforward consequence of the  
homological character of BP and conservation of the total energy. 
It was already known in particular cases, but not stated 
as a general fact to our knowledge. 

\begin{prop}[Conservation] \label{bp-conservation}
Let $q \in \Del_0(X)^\N$ denote a sequence of beliefs iterated 
from $\BP_\lambda$ for some $\lambda > 0$. Then the quantity: 
    \begin{equation} q_\Om(t) = \prod_{\aa \in X} q_\aa(t)^{c_\aa} \end{equation} 
        remains constant in $G_\Om = (\R_+^*)^{E_\Om}$ up to a scaling factor. 
\end{prop}

\egroup

\section{Canonical Diffusion} 

As expressed by theorem \ref{correspondence}, stationary states of 
any transport equation derived from a faithful flux functional 
solve the problem of finding consistent pseudo-marginals 
critical for a Bethe free energy. 
Therefore it remains a practical and theoretical open question 
whether $\Phi = - \DF \circ \zeta$ is optimal among faithful fluxes, 
and if not whether a better-behaved flux may be designed. 
Completing diagram \ref{bp-diagram} with a Möbius inversion 
on degree one, 
we introduce a homological vector field $\tau = \div \phi$ 
associated to the flux $\phi = \mu \circ (- \DF) \circ \zeta$ :
\begin{equation} \label{bpc-diagram}
\bcd 
A_0(X) \rar{\zeta} & A_0(X) \dar{-\DF} \\
A_1(X) \uar{\div} & A_1(X) \lar{\mu} 
\ecd 
\end{equation}
The symmetry of diagram \ref{bpc-diagram} is of great appeal. 
It involves conjugation of operators on $A_\bullet(X)$ by the extended transforms
$\zeta$ and $\mu$, which shall have interesting 
cohomological consequences on the linearised dynamic.
We claim that $\phi$ behaves better than the standard flux $\Phi$ 
for three main reasons:
\begin{enumerate}[label=(\roman*),itemsep=0pt]
    \item
the flux bound into $\LL^\aa$ is a Bethe approximation of the total effective 
energy of $X \setminus \LL^\aa$. 
\item
the flux from $\LL^\aa$ to a subcone $\LL^\bb$ is the effective energy 
of $\LL^\aa \setminus \LL^\bb$. 
\item the algorithm $u \wa u + \div \ph(u)$ restricted
    to a cone $\LL^\aa \incl X$ converges 
    in one step. 
\end{enumerate}
The standard flux $\Phi$ brings redundancies and 
(i) expresses that effective contributions of neighbouring regions
are properly counted. In addition (ii) will allow for a natural enforcement
of Dirichlet boundary conditions,  
the outbound flux taking care of reaching consistency with the boundary. 
Such boundary conditions, fixing the state of an exterior subset of variables, 
are a fundamental constituent of learning. 

We prepare this section by introducing the differential calculus 
we shall use in presence of boundary, which relies 
on a boundary operator $\divint$ truncating $\div$ to interior variables. 
Investigating some fundamental properties of the canonical flux, 
we then show that $\phi = \mu \circ (- \DF) \circ \zeta$
satisfies a local faithfulness condition 
and prove proposition \ref{mu-gauss} supporting claims (i) and (ii) above,
before going through some of the algorithms that generalise belief propagation
by approximate integration of $\dot u = \div \phi(u)$.

\subsection{Calculus with Boundary}

In contrast with differential geometry, there is no intrinsic notion of boundary 
on $X \incl \Part(\Om)$ and 
deciding which variables belong to the boundary and which do not is either
arbitrary or dictated by experience\footnote{
    In practice, the set of fixed boundary variables of a neural network 
    depends on the training or testing context.
}. 
Assuming $\bord \Om \incl \Om$ describes a set of variables 
whose state is given by the exterior, the following
compatibility condition on $X$ 
will define the boundary $\bord X$ at the level of regions.  

\begin{defn}[Boundary] 
    Given a subset of variables $\bord \Om \incl \Om$, let us denote by: 
    \bi
    \iii $\bord \aa = \aa \cap \bord \Om$ the boundary of a region $\aa \incl \Om$,
    \iii $\bord X = \{ \bord \aa \st \aa \in X \}$ the 
    boundary of $X \incl \Part(\Om)$. 
    \ei
    We say that a covering $X \incl \Part(\Om)$ 
    is {\rm adapted to the boundary} $\bord \Om$
    whenever $\bord X \incl X$.
\end{defn}

As a consequence, 
the boundary $\bord \LL^\aa$ of the cone 
$\LL^\aa \incl X$ is the cone $\LL^{\bord \aa} \incl \bord X$ 
for every $\aa \in X$. 
In other words, every $\aa$ has only one maximal subregion\footnote{ 
    This contrasts with the topological setting, as one may for instance 
    take the boundary of the 2-simplex to be the triangle 
    formed by its three edges.
    As there exists locally consistent 
    pseudo-marginals on the triangle
    that do not have a consistent global extension, 
    such a notion of boundary would be problematic when trying to enforce globally 
    inconsistent Dirichlet boundary conditions.
}
$\bord \aa$ belonging to $\bord X$, with $\bord \aa = \vide$ when 
$\aa$ does not contain any exterior variable. 
We finally define $\aa$ to be {\it interior} 
whenever it is not contained in $\bord X$,
equivalently, when $\aa$ contains at least one variable of $\Om \setminus \bord \Om$:

\begin{defn}
If $X \incl \Part(\Om)$ is adapted to $\bord \Om$, 
we denote by $\Xint = X \setminus \bord X$ the {\rm interior} of $X$.
\end{defn}

To the splitting 
$X = \Xint \sqcup \bord X$, we associate the direct sum decomposition 
$A_0(X) = A_0(\Xint) \oplus A_0(\bord X)$ 
and write $u = u_{|\Xint} + u_{|\bord X}$ for every $u \in A_0(X)$. 
When enforcing Dirichlet boundary conditions, we shall restrict the evolution 
of the interaction potentials $u$ to the interior of $X$,  
their trace $u_{|\bord X}$ on the boundary describing input data
or exterior stimuli. The evolution of $u_{|\Xint}$ shall bear 
a homological character through the following definition.
Although a general harmonic theory with boundary 
may be developed on the whole complex $A_\bullet(X)$, 
we solely focus on degrees zero and one for now.

\begin{defn} We call {\rm interior divergence} 
    the map $\divint: A_1(X) \aw A_0(\Xint)$  
    truncating $\div$ to $\Xint$, defined by
    $\divint(\ph) = \div(\ph)_{|\Xint}$ 
    for every $\ph \in A_1(X)$.
\end{defn}

The following proposition is the analog of the integration by parts formula 
in differential geometry.  
Given a submanifold $V \incl \R^3$ with boundary $\bord V$, 
one has for every scalar field $u$ and vector field $\vec \ph$:
\begin{equation} \label{ipp-geom}
    \int_{V} \vec{\rm grad}(u) \cdot \vec{\ph}\; dv
= - \int_{V} u \, {\rm div}(\vec{\ph}) \; dv
+ \int_{\bord V} u \,({\vec{\ph}} \cdot \vec{n}) \; ds
\end{equation}
denoting by $\vec n$ the outbound unit normal vector on $\bord V$. 
The formal adjunction of $\vec{\rm grad}$ with $-{\rm div}$ is 
tweaked by a boundary term representing the integral of the 
outbound flux 
of $\ph$ against $u$.

\begin{prop}[Integration by parts] \label{ipp}
    Let $\nabla$ denote the adjoint of 
    $\div$ for a given metric. 
    Then for every $u \in A_0(X)$ and every $\ph \in A_1(X)$ we have:
\begin{equation} \croc{\nabla u}{\ph} = \croc{u}{\divint \ph} + b(u, \ph) \end{equation}
    where $b(u, \ph) = \croc{u \st {\bf 1}_{\bord X}}{\div \ph}$ 
    denotes the scalar product of the restrictions of $u$ and $\div \ph$ 
    to $\bord X$. 
\end{prop}

\begin{proof}
    In absence of boundary we have the classical integration by parts 
    formula $\croc{\nabla u}{\ph} = \croc{u}{\div \ph}$ expressing
    the adjunction of $\nabla$ with $\div$. 
    The above simply consists of introducing a 
    splitting of $\croc{u}{\div \ph}$ 
    as $\croc{u\st {\bf 1}_{\Xint} + {\bf 1}_{\bord X}}{\div \ph}$
    in presence of a boundary. Note that the boundary term: 
    \begin{equation} b(u, \ph) \;=\; \sum_{\bb \in \bord X} 
    \Big\langle u_\bb \;\Big|\; 
    \sum_{\aa \aw \bb} \ph_{\aa\bb} - \sum_{\bb \aw \cc} \ph_{\bb\cc} 
    \Big\rangle
    \end{equation}
    does present a formal analogy with \ref{ipp-geom} despite 
    the unusual <<thickness>> of the boundary $\bord X$. 
    The analogy becomes clearer if one assumes 
    $\nabla u = 0$ on $\bord X$ as the above reduces to:
    \begin{equation} b(u, \ph) \;=\; \sum_{\bb \in \bord X} 
    \Big\langle u_\bb \;\Big|\; \sum_{\aa \in \Xint} \ph_{\aa\bb}
    \Big\rangle 
    \end{equation} 
    by duality of $\nabla$ with $\div$ on $A_\bullet(\bord X)$, 
    and represents the integral of the outbound flux of $\ph$ against $u$.
\end{proof}

As proposition \ref{ipp} suggests, the differential calculus 
of $\divint$ shall only differ from that of $\div$ by the appearance
of boundary terms representing energy fluxes leaving $\Xint$ through $\bord X$. 
It will be useful to treate those seperately
and we decompose $A_1(X)$ as the direct sum 
$A^{int}_1(X) \oplus A^{out}_1(X)$ according to: 

\begin{defn} 
    For every $\ph \in A_1(X)$, we introduce the splitting 
    $\ph = \ph^{int} + \ph^{out}$ defined by: 
\bi
    \iii
    $\ph^{int}_{\aa\bb} = \ph_{\aa\bb}$ for every $\bb \in \Xint$,
    \iii 
    $\ph^{out}_{\aa\bb} = \ph_{\aa\bb}$ for every $\bb \in \bord X$.
\ei
    We respectively call $\ph^{int}$ and $\ph^{out}$ 
    the {\rm interior} and {\rm outbound}
components of $\ph$.
\end{defn}

Note that 
$\div : A_1(X) \aw A_0(X)$ then 
takes a block-triangular form as illustrated by the diagram:
\begin{equation}\label{div-diagram}
\bcd
    A^{int}_1(X) \dar & \oplus &
    A^{out}_1(X) \ar[dll] \dar[dashed] \\
    A_0(\Xint) & \oplus & A_0(\bord X) 
\ecd
\end{equation}
Full-line arrows represent the components of $\divint$
while the dotted arrow represents the truncation $\div - \divint$.
In particular, we have $\div(\ph^{int}) = \divint(\ph^{int})$ while 
and the boundary term $b(u, \ph) = b(u, \ph^{out})$ of prop. \ref{ipp} 
depends only on $\ph^{out}$.

\begin{prop}[Gauss formula on $\LL^\aa \setminus \bord \LL^\aa$]
    \label{gauss-intercone}
For every $\phi \in A_1(X)$ we have: 
\begin{equation}
    \zeta(\divint \phi)_\aa = 
    \tilde \zeta(\phi^{int})_{\Om \aw \aa} 
    - \zeta(\phi^{out})_{\aa \aw \bord\aa}
\end{equation}
    where $\tilde \zeta(\phi^{int})_{\Om \aw \aa} = 
    \disp \sum_{\aa'\bb' \in d\LL^\aa}
    \phi^{int}_{\aa'\bb'}$
denotes the total flux bound from $\Xint$ to the interior of $\LL^\aa$. 
\end{prop} 

\begin{proof}
    By definition of $\divint$ and of the action of $\zeta$ on $A_0(X)$ we have: 
    \begin{equation} \zeta(\divint \phi)_\aa 
    \; = \sum_{\bb' \in \LL^\aa} \divint_{\bb'}\phi 
    \; = \sum_{\bb' \in \LL^\aa \setminus \bord \LL^\aa} \div_{\bb'} \phi 
    \; = \sum_{\bb' \in \LL^\aa_{\bord \aa}} \Big(
    \sum_{\aa' \aw \bb'} \phi_{\aa'\bb'} - \sum_{\bb' \aw \cc'} \phi_{\bb'\cc'}
    \Big) \end{equation}
    Terms of the form  $\phi_{\aa'\bb'}$ with $\aa' \in \LL^\aa_{\bord \aa}$ 
    cancel out the $\phi_{\bb'\cc'}$ with $\cc' \in \LL^\aa_{\bord \aa}$
    so that we may write: 
    \begin{equation} \zeta(\divint \phi)_\aa 
    \; = \sum_{\aa' \in \LL^\Om_\aa} \: \sum_{\bb' \in \LL^\aa} \phi^{int}_{\aa'\bb'}
    - \sum_{\bb' \in \LL^\aa_{\bord \aa}} \: \sum_{\cc' \in \LL^{\bord \aa}} 
    \phi^{out}_{\bb'\cc'} \end{equation}
    Recognising zeta transforms on degree one, 
    the outbound flux reads $\zeta(\phi^{out})_{\aa \aw \bord \aa}$ and 
    the inbound flux reads $\tilde \zeta(\phi^{int})_{\Om \aw \aa}$ 
    where $\phi^{int}$ is extended by zero to $\tilde X = \{ \Om \} \cup X$
    and $\tilde \zeta$ acts on $A_1(\tilde X)$.
\end{proof}

The zeta transforms in degree one appearing in proposition \ref{gauss-intercone} 
retrospectively provide with a strong motivation for the higher degree
combinatorics of chapter 3. They will also justify performing 
a Möbius inversion on the standard flux functional $\Phi = - \DF \circ \zeta$.

\subsection{Canonical Flux} 

We now introduce the homological vector field $\tau = \div \phi$ which 
we claim to define a canonical diffusion on interaction potentials. 
Because the evolution of effective hamiltonians $U = \zeta \cdot u$ 
integrates the energy flux on cones, 
Gauss formulas applied $\phi = \mu \cdot \Phi$ will shed light
on the necessity of performing a Möbius inversion in degree one on the 
standard diffusion flux\footnote{
    We believe this combinatorial correction to be a significant improvement
    of the GBP algorithm of \cite{Yedidia-2005}, however for the standard BP 
    algorithm on graphs, one has $\phi = \mu \cdot \Phi = \Phi \mod \R_1(X)$ 
    and the correction is useless.
}.

\begin{defn}\label{diffusion}
    We call: 
    \bi 
    \iii {\rm canonical diffusion flux} 
    the smooth functional $\phi = \mu \circ (-\DF) \circ \zeta$ defined from $A_0(X)$
    to $A_1(X)$, 
    \iii {\rm canonical diffusion} the 
    smooth vector field $\tau = \div \phi$ defined on $A_0(X)$.
    \ei 
\end{defn} 

We shall write $\phi = -\DF^\mu$ to emphasise on the conjugation 
of $-\DF$ by the Möbius transform. 
More generally, the following notations  
will be useful in switching from one point of view to the other. 

\begin{defn} 
    For every smooth map $T : A_\bullet(X) \aw A_\bullet(X)$ we denote by: 
    \bi 
    \iii $T^\zeta = \zeta \circ T \circ \mu$ the $\zeta$-conjugate of $T$,
    \iii $T^\mu = \mu \circ T \circ \zeta$ the $\mu$-conjugate of $T$. 
    \ei
\end{defn} 

The extensions of $\zeta$ and $\mu$ to all degrees 
really allow for two equivalent point 
of views on $A_\bullet(X)$, of which the associated conjugations 
seem to be an essential feature. 
Hence the two equivalent differential equations for 
the canonical diffusion $\tau$ on interaction potentials 
and its conjugate vector field
$\tau^\zeta$ inducing the evolution of effective hamiltonians: 
\begin{equation} \label{mu-Tspt}
{\renewcommand{\arraystretch}{1.6} 
    \begin{tabular}{C{3.5cm} C{1cm}  C{3.5cm}}
$
\disp \left\{ 
\ba{l} 
    \dot u = \div \phi \\
    \phi = - \DF^\mu(u) 
\ea
\right. 
$
        & $\eqvl$ &
$
\disp \left\{ 
\ba{l} 
    \dot U = \div^\zeta (\Phi) \\
    \Phi = - \DF(U) 
\ea
\right.
$
    \end{tabular}
}
\end{equation}
The following theorem best illustrates a first effect of Möbius 
inversion on the effective energy flux. 
This correction is necessary for 
the total flux bound into $\LL^\aa$ to correctly approximate 
the global effective energy of $\Om \setminus \LL^\aa$.
Theorem \ref{Tau} also implies one-step convergence to 
local hamiltonians of $U \wa U - \div^\zeta \DF(U)$ whenever 
the underlying hypergraph contains a maximal cell\footnotemark{}. 
\footnotetext{
    $X$ may for instance contain $\Om$, but the more practical consequence 
    holds for restrictions of the algorithm to cones $\LL^\aa \incl X$,
    for instance when 
    updating units asynchronously and independently of one another, see 
    section 6.1.
}
algorithm $u \wa u - \div \DF^\mu(u)$ restricted to $A_0(\LL^\aa)$ 
below any cell $\aa \in X$. 

\begin{thm} \label{Tau}
The evolution of effective hamiltonians under 
$\dot U = - \div^\zeta \DF(U)$ reads:
\begin{equation} \label{tau-zeta} 
\dot U_\aa = 
    \check \Fh^\Om(U \st \aa) - U_\aa 
\end{equation}
where 
$\check \Fh^\Om(U \st \aa) = \sum_{\om} c_\om \, \Fh^\om(U_\om \st \om \cap \aa)$
denotes the Bethe approximation of 
$\Fh^\Om(U_\Om \st \aa)$ 
when $X$ does not contain $\Om$, and is equal to the latter otherwise. 
\end{thm}

Substituting the expression of the effective energy gradient 
$\DF(U)$ for $\Phi$, the theorem will come as a direct consequence of:

\begin{prop} \label{mu-gauss}
If $X \incl \Part(\Om)$ does not contain $\Om$,
then for every $\Phi \in A_1(X)$ we have:
\begin{equation} 
    \div^\zeta(\Phi)_\aa =
    \check \Phi_{\Om \aw \aa} 
\end{equation}
where 
$\check \Phi_{\Om \aw \aa} = 
\disp \sum_{\om \not\in \LL^\aa} c_\om \, \Phi_{\om \aw \aa \cap \om}$
    denotes the Bethe approximation of an expected flux $\Phi_{\Om \aw \aa}$.
\end{prop}

\begin{proof} 
Letting 
$\phi = \mu \cdot \Phi$ we have 
$\div^\zeta(\Phi) = \zeta(\delta \phi)$. 
Using notations of prop. \ref{gauss-intercone}, in absence of boundary
the Gauss formula on $\LL^\aa$ reads, following (\ref{zeta-div-zeta-Om}): 
\begin{equation} \div^\zeta(\Phi)_\aa 
= \tilde \zeta( \mu \cdot \Phi)_{\Om \aw \aa} 
\end{equation}
Consider the Möbius transform $\tilde \phi = \tilde \mu \cdot \Phi$
of the natural extension of $\Phi$ to $\tilde X = \{ \Om \} \sqcup X$ by zero.
Locality of $\tilde \mu$ implies 
that $\tilde \phi$ and $\phi$ coincide on $X$ 
and only differ on terms of the form $\Om \aw \bb$ 
where $\phi$ vanishes. 
By Möbius inversion on $\tilde X$ we have
    $\tilde \zeta(\tilde \phi)_{\Om \aw \aa} = 
     \Phi_{\Om \aw \aa} = 0$
so that:
\begin{equation} \tilde \zeta(\phi)_{\Om \aw \aa}  = 
 - \tilde \zeta \big( \tilde \phi - \phi\big)_{\Om \aw \aa} 
= - \zeta \big( i_\Om(\tilde \phi) \big)_\aa \end{equation}
From the inductive construction
of $\mu$ given by
proposition \ref{moebius-rec} we have 
$i_\Om( \tilde \phi) = \mu \big( \tilde \nu_\Om(\Phi) \big)$
which is equivalent to 
$\zeta\big( i_\Om (\tilde \phi) \big) = \tilde \nu_\Om (\Phi)$
    by Möbius inversion on $A_0(X)$. 
    Substituting the identity 
    $c_\bb = - \tilde \mu_{\Om \bb}$ given by proposition \ref{c-mu} 
    into equation \ref{nu} defining $\tilde \nu_{\Om}$, we finally get:
    \begin{equation} - \tilde \nu_\Om(\Phi)_\aa =
    - \sum_{\bb \in \LL^\Om_\aa} \tilde\mu_{\Om \aa} \, \Phi_{\bb \cap (\Om \aw \aa)} 
    = \sum_{\bb \not\in \LL^\aa} c_\bb \, \Phi_{\bb \aw \bb \cap \aa}
    = \check \Phi_{\Om \aw \aa} \end{equation}
    which gives the desired expression for 
    $\tilde \zeta(\phi)_{\Om \aw \aa} = \tilde \zeta( \mu \cdot \Phi)_{\Om \aw \aa}$.

    Note that if $\Om \in X$, Möbius inversion 
    in $A_1(X)$ would have simply given $\div^\zeta(\Phi)_\aa = \Phi_{\Om \aw \aa}$.
\end{proof} 

Another reason for performing Möbius inversion on the energy flux 
comes with the enforcement of Dirichlet boundary conditions on $A_0(\bord X)$. 
The outbound flux from $\LL^\aa$ to its boundary $\bord \LL^\aa = \LL^{\bord\aa}$
now coincides with the effective energy of $\LL^\aa \setminus \bord \LL^\aa$, 
as expressed by the following theorem. 
Its effect will be to ensure consistency of $U \in A_0(X)$ with the prescribed 
values on $\bord X$.

\begin{thm}
The evolution of effective hamiltonians under 
$\dot U = -\divint^\zeta \DF(U)$ reads: 
\begin{equation} 
    \dot U_\aa =  \sum_{\om \cap \aa \not\in \bord X} 
    c_\om \: \Fh^\Om \big( U_\om - U_{\om \cap \aa} \,\big|\, \om \cap \aa \big) 
    \: - \: 
    \Fh^\aa \big( U_\aa - U_{\bord \aa} \big| \,\bord\, \aa \big) 
    \end{equation}
\end{thm}

The following counterpart of 
the Gauss formula on $\LL^\aa \setminus \bord \LL^\aa$ 
will again prove the theorem. 
It comes as an easy consequence of prop. \ref{mu-gauss} 
and the lemma below. 

\begin{prop} \label{mu-gauss-bord}
If $X$ does not have a maximal element, then for every $\Phi \in A_1(X)$ we have:
\begin{equation} 
    \divint^\zeta(\Phi)_\aa =
    \check \Phi^{int}_{\Om \aw \aa} 
    - \Phi^{out}_{\aa \aw \bord \aa} 
\end{equation}
where 
    $\check \Phi^{int}_{\Om \aw \aa} = 
    \disp \sum_{\om \not\in \LL^\aa} c_\om \, \Phi^{int}_{\om \aw \aa \cap \om}$
    denotes the Bethe approximation of an expected flux $\Phi^{int}_{\Om \aw \aa}$.
\end{prop}

\begin{lemma}\label{phi-int} 
    For every $\phi = \mu \cdot \Phi \in A_1(X)$, 
    we have $\phi^{out} = \mu \cdot \Phi^{out}$ 
    and $\phi^{int} = \mu \cdot \Phi^{int}$. 
\end{lemma}

\begin{proof}[Proof of proposition \ref{mu-gauss-bord}]
    Letting $\phi = \mu \cdot \Phi$, 
    in accordance with lemma \ref{phi-int} the Gauss formula 
    with boundary \ref{gauss-intercone}
    gives the following expression
    for $\zeta(\divint \phi)_\aa = \divint^\zeta(\Phi)_\aa$:
    \begin{equation}
    \tilde \zeta(\phi^{int})_{\Om \aw \aa} 
    - \zeta(\phi^{out})_{\aa \aw \bord \aa} 
    = \tilde \zeta(\mu \cdot \Phi^{int})_{\Om \aw \aa} 
    - \zeta( \mu \cdot \Phi^{out})_{\aa \aw \bord \aa}
\end{equation}
The outbound flux reads 
$\Phi^{out}_{\aa \aw \bord \aa}$ by Möbius inversion on $A_1(X)$, 
while it follows from proposition \ref{mu-gauss}
    applied to $\div^\zeta(\Phi^{int})_\aa$ 
that the inbound flux is 
the Bethe approximation $\check \Phi^{int}_{\Om \aw \aa}$.
\end{proof}

\begin{proof}[Proof of lemma \ref{phi-int}] 
For all $\aa \in X$ and $\bb \in \bord X$, we have by definition of the action
    of $\mu$ on $A_1(X)$: 
    \begin{equation} \phi^{out}_{\aa\bb} = \sum_{\bb' \in \LL^\bb} \mu_{\bb \bb'} 
    \sum_{\aa' \in \LL^\aa_{\bb'}} \mu_{\aa\aa'} \cdot \Phi_{\aa' \aw \aa' \cap \bb'}
    = \big(\mu \cdot \Phi^{out}\big)_{\aa\bb} \end{equation}
as $\bb \in \bord X$ implies $\aa' \cap \bb' \in \bord X$ for all 
$\bb' \incl \bb$ and $\phi^{out}$ only depends on $\Phi^{out}$. 
We may then conclude from 
$\phi^{out} = \mu \cdot \Phi^{out}$ that
$\phi^{int} = \phi - \phi^{out}$ 
coincides with $\mu \cdot \Phi^{int} = \mu \cdot (\Phi - \Phi^{out})$
by linearity of $\mu$.
\end{proof} 

Note that the properness of $\phi$ follows from that of $\Phi$ 
by invertibility of $\mu$, however
we were only able to prove a local faithfulness
property and the global faithfulness of $\phi$ 
will remain an open question. 

\begin{defn} 
    A flux functional $\phi: A_0(X) \aw A_1(X)$ will 
    be said {\rm locally faithful} if 
    there exists an open neighbourhood ${\cal V}$ of 
    $\{ \DF \circ \zeta = 0 \} \incl A_0(X)$ 
    such that for all $u \in \cal V$: 
    \begin{equation} \div \phi(u) = 0  \quad \eqvl \quad  \DF(\zeta \cdot u) = 0 \end{equation}
\end{defn}

\begin{prop} \label{loc-faithful}
    The flux functional $\phi = \mu \circ (- \DF) \circ \zeta$ 
    is locally faithful. 
\end{prop}

\begin{proof}
    Let $u \in A_0(X)$ 
    such that $\DF(\zeta \cdot u) = 0$ 
    denote a field of consistent interaction potentials.
    Writing $\Phi = - \DF \circ \zeta$ as before, 
    for every $V = \zeta \cdot v \in A_0(X)$ we have 
    according to propositions \ref{mu-gauss} and \ref{Eba}:
    \begin{equation} \label{phi-lin}
    \zeta \big( \div \phi(u + v) \big)_\aa 
    = \check \Phi(u + v)_{\Om \aw \aa} 
    = \sum_{\om \notin \LL^\aa} c_\om \, 
    \E^\om \big[\, V_\om - V_{\om \cap \aa} \,\big|\, \om \cap \aa \,\big]  
    + o\,(v)
    \end{equation}
    where conditional expectations are taken for the consistent statistical field
    $p = [\e^{-U}]$ with $U = \zeta \cdot u$. 
    Now note that although $p \in \Gammint(X)$ may not derive from a global 
    probability density $p_\Om \in \Delta_\Om$, there does exist a global 
    density $q_\Om \in A^*_\Om$ such that $p_\aa = \Sigma^{\aa\Om}(q_\Om)$ 
    for all $\aa \in X$ by acyclicity of $A^*_\bullet(X)$. 
    We may thus define global <<conditional expectation>> maps by 
    letting for all $\om \aw \bb$ in $X$: 
    \begin{equation} \E^{\Om}[\,V_\om \st \bb \,] 
    = \frac {\Sigma^{\bb\Om}(q_\Om \cdot V_\om)} {p_\bb} \end{equation}
    such that $\E^{\Om}[\,V_\om \st \bb \,]$ coincides 
    with $\E^{\om}[\,V_\om \st \bb \,]$ as a consequence 
    of $\Sigma^{\om \Om}(q_\Om) = p_\om$.
    Observing that $V_{\om \cap \aa} = \zeta(v_{|\LL^\aa})_\om$ 
    and letting $V_\Om = \tilde \zeta(v)_\Om$, 
    the linearised right hand side of \ref{phi-lin} now reads:
    \begin{equation} 
    \E^\Om \Big[\: \sum_{\om \in X} c_\om \: \big( V_\om - V_{\om \cap \aa} \big)
    \: \Big| \: \aa \: \Big] 
    \; = \; 
    \E^\Om \big[\, V_\Om - V_\aa \, \big| \, \aa \, \big] \end{equation}
    Up to second order terms in $v$, we hence have
    $\div \phi(u + v) \simeq 0$ if and only if 
    $V_\aa \simeq \E^\Om[V_\Om \,|\, \aa ]$ for all $\aa$, 
    which implies $V_\bb \simeq \E^\aa[V_\aa \,|\, \bb]$ for all $\aa \aw \bb$. 
    Hence $\div \phi$ satisfies the linearised faithfulness condition:
    \begin{equation} \div \phi(u + v) \simeq 0 \quad \eqvl \quad
    \DF( U + V ) \simeq 0 \end{equation} 
    and the tangent spaces of 
    $\{ \div \phi = 0 \}$ and $\{ \DF \circ \zeta = 0 \}$ at $u$ coincide. 
    If one could show $\{ \div \phi = 0 \}$ to be connected, 
    the global faithfulness of $\phi$ would follow.
\end{proof}

\chapter{Geometry of Equilibria}

\bgroup
\renewcommand{\Z}{{\cal Z}}
\newcommand{\Shf}{{\mathscr F}}

In this chapter, we study the influence of the geometry
of $X \incl \Part(\Om)$ on the stationary points of message-passing algorithms, 
as described by the intersection of the manifold $\Z(X)$ of consistent
interaction potentials with homology classes of the form $[h] = h + \div A_1(X)$. 

    \vspace{-0.1cm}
\begin{figure}[H]
    \sbox0{\includegraphics[width=0.45\textwidth]{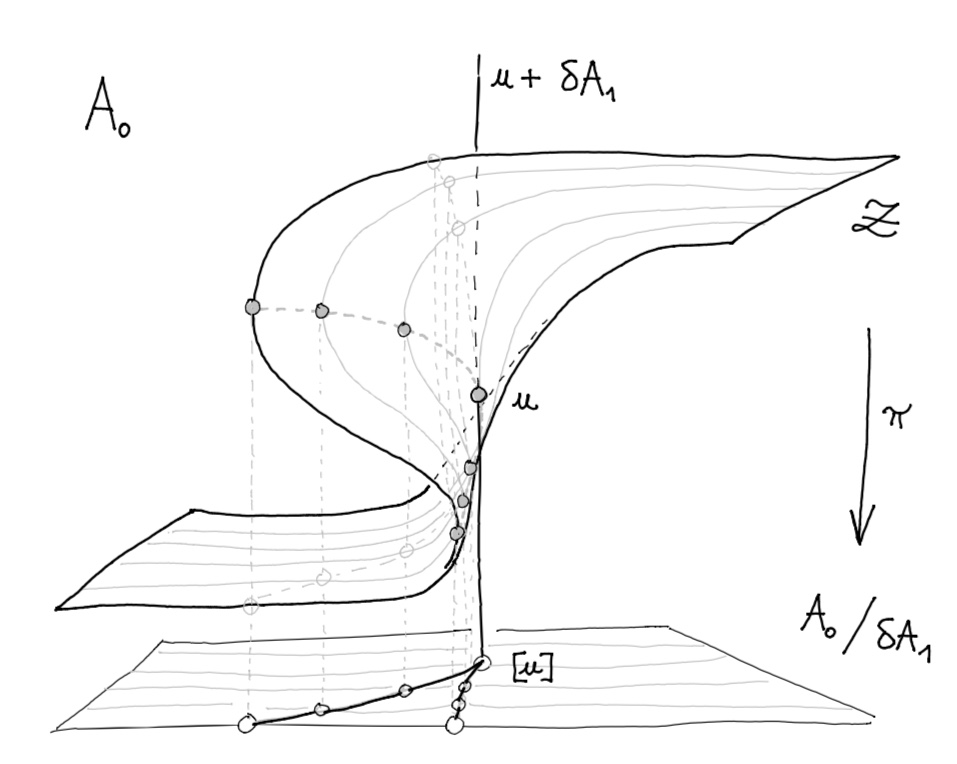}
    }
\begin{center}
    \begin{minipage}{0.8\textwidth}
\centering
\usebox0
    \vspace{-0.2cm}
        \caption{
            A cuspidal singularity of the projection 
            $\Z(X) \aw A_0(X) / \div A_1(X)$. 
            Vertical lines represent classes of homologous potentials, 
            some of which intersect the stationary surface more than once.
        }
\end{minipage}
\end{center}
    \vspace{-0.5cm}
\end{figure}

We first extend a well-known uniqueness theorem on trees to a wider class 
of hypergraphs which we call {\it retractable}. 
This constructive procedure moreover demonstrates the finite-time 
convergence of a message-passing scheme on retractable hypergraphs. 

We introduce a canonical map $T : A_0(X)/\delta A_1(X) \aw \Z(X)$ 
sending any potential $h \in A_0(X)$ to the true 
effective potential $h^* \in \Z(X)$ one seeks to approximates, 
as induced by the global hamiltonian. 
We show that any $h \in \Z(X)$ is fixed by $T$ on retractable hypergraphs, 
so that the unique message-passing equilibrium coincides 
yields the true marginals of the global probability distribution. 
We then provide with the simplest example of a graph 
with loops such that $T$ induces a non-trivial dynamical system on $\Z(X)$,
as the true potential $h^*$ to approximate 
lie outside of the homology class $[h] = h + \delta A_1(X)$ 
message-passing explores.

As numerical studies on graphs have already shown, 
the number of stationary points grows quickly with the number of loops. 
We study the appearance of multiple equilibria in 
$\Z(X) \cap [h]$ through singularities of the 
quotient map $\Z(X) \aw A_0(X)/\delta A_1(X)$, 
which we relate to the spectral properties 
of a pseudo-laplacian operator $L$ describing the linearised diffusion flow,
and finally provide with what we believe to be the first explicit examples 
of bifurcations on graphs. 

\newpage

\section{Uniqueness and Retractability}

In this section, we prove the uniqueness of a stationary state 
on a class of {\it retractable} hypergraphs which generalise trees. 
These hypergraphs also appear in the pseudo-marginal extension problem, 
to which Vorob'ev gave a criterion of solvability in \cite{Vorobev-62}, 
and are thus related to what he called {\it regular} complexes. 

Precising first the sheaf structure of the manifold $\Z(X)$ 
of consistent potentials, we then exhibit an {\it extensibility} 
property of $\Z(X)$, on which the construction of the unique equilibrium 
by successive extension crucially relies. 

Note that this section carries out proofs on a generic sheaf denoted by $\Shf(X)$, 
mainly purposed to represent $\Z(X)$ or one of its tangent fibers. 
Generic notations were also intended so as to leave some space for
an extension of the results to higher-degree analogs of $\Z(X)$, 
e.g. cocycles of $\Ker(\nabla^\mu)$.

\subsection{Hypergraph Geometry and Sheaves}

\begin{defn} 
    A hypergraph $X$ will be called a {\rm simple extension}\footnote{
        In other words, there is a maximal cell $\aa$ and a strict subcell
        $\bb \inclst \aa$ such that every $\cc \not \incl \aa$ intersects 
        $\aa$ inside $\bb$. 
        When $X$ is a graph, this implies that $\aa$ is a {\it terminal edge} 
        i.e. an edge with at most one vertex linked to another edge. 
    } of $X' \incl X$ when: 
    \begin{equation} X = \LL^\aa \sqcup_{\LL^\bb} X' \end{equation}
or equivalenlty, when
$X = \LL^\aa \cup X'$ and $X \cap X' = \LL^\bb$ 
for some $\aa \in X$ and $\bb \in X'$. 
\end{defn}

\begin{defn} 
    A hypergraph $X$ is called a {\rm normal extension}\footnote{
        We try to remain close to the vocabulary of Vorob'ev,
        who called
        a sequence of simple extensions $X = X_n \cont \dots \cont X_0$ 
        a {\it normal series} of $X$. 
        See \cite{Vorobev-62} for his interesting characterisations 
        of retractable (called {\it regular}) simplicial complexes. 
} of $X'$ if there exists a 
sequence: 
\begin{equation} X = X_n \cont \dots \cont X_0 = X' \end{equation} 
of simple extensions from each $X_i$ to $X_{i+1}$.
\end{defn}

\begin{defn} \label{retractable}
    A hypergraph $X$ is called {\rm retractable} if it is a normal 
    extension of $\{ \vide \}$.
\end{defn}

    \vspace{-0.1cm}
\begin{figure}[H]
    \sbox0{\includegraphics[width=0.5\textwidth]{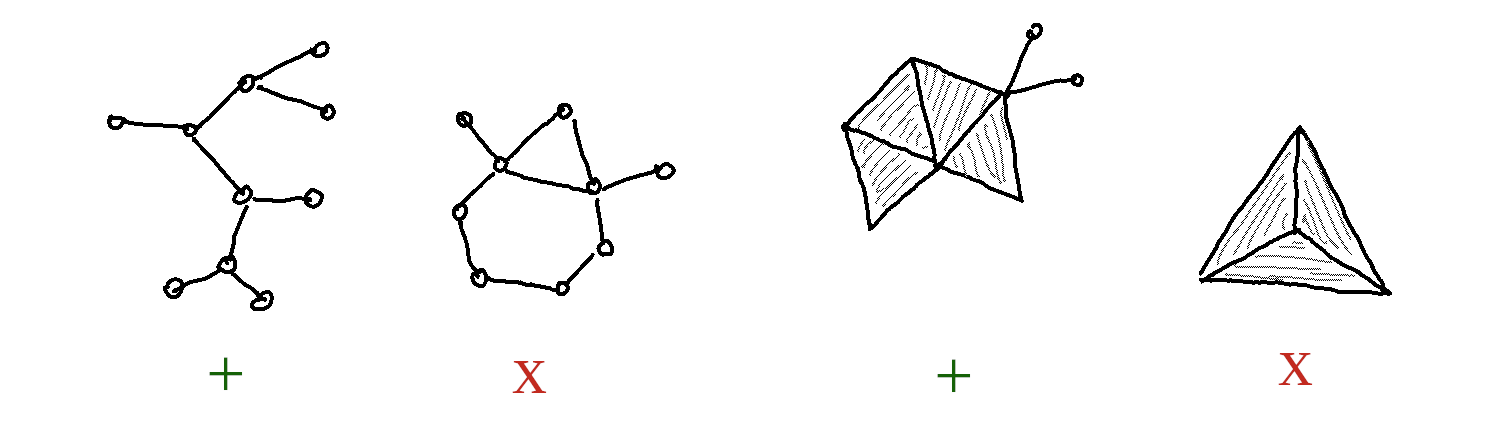}}
\begin{center}
    \begin{minipage}{0.6\textwidth}
\centering
\usebox0
    \vspace{-0.2cm}
    \caption{
        retractable hypergraphs 
        (+) and other ones (x).
    }
\end{minipage}
\end{center}
    \vspace{-0.5cm}
\end{figure}

Note that a graph  $X$ is a normal extension of $X'$ if and only if 
it retracts onto $X'$ in the usual sense.  
In particular, a graph $X$ is retractable if and only if it is a tree, 
{\it i.e.} an acyclic graph. 

\begin{defn} 
    The {\rm Alexandrov topology} of a hypergraph $X \incl \Part(\Om)$ is 
    the topology generated by the basis of open neighbourhoods 
    $(\LL^\aa)_{\aa \in X}$. 
\end{defn}

\begin{prop}
$Y \incl X$ is open for the Alexandrov 
topology if and only if 
$\aa \in Y$ implies $\LL^\aa \incl Y$.
\end{prop}

\begin{prop} \label{zeta-alexandrov}
    Given a sub-hypergraph $Y \incl X$, the
    restriction map ${\rm r}_{YX} : A(X) \aw A(Y)$ 
    commutes with the zeta transforms of $A(X)$ and $A(Y)$ 
    if and only if $Y$ is open in $X$. 
    \begin{equation} \bcd 
    A(X) \dar[swap]{{\rm r}_{YX}} \rar{\zeta_X} & A(X) \dar{{\rm r}_{YX}} \\
    A(Y) \rar{\zeta_Y} & A(Y) 
    \ecd \end{equation}
\end{prop}

\begin{prop} 
    The set ${\cal Z}(X)$ of consistent interaction potentials 
    forms a subsheaf of $A_0(X)$ for the Alexandrov topology.
\end{prop}

\begin{proof} 
For every $u \in \Z(X)$ and $Y \incl X$, 
one has $u_{|Y} \in \Z(Y)$ as 
from 
$\DF_Y( \zeta_Y \cdot u_{|Y}) = \DF_X(\zeta_X \cdot u)_{|Y} = 0$
by proposition \ref{zeta-alexandrov}.  
The sheaf colimit property is satisfied 
by construction as $\Z(X) = \colim_\aa \Z(\LL^\aa)$. 
\end{proof}

\subsection{Cocyclic and Extensible Sheaves}

\begin{defn}
A subsheaf $\Shf(X)$ of $A(X)$ will be said {\rm cocyclic} if: 
\begin{equation} \forall u \in A(X) \quad \exists! v = u + \div \ph \in \Shf(X) \end{equation}
Equivalently, $\Shf(X)$ is cocyclic if 
it is a section of the fibration $A(X) \aw A(X) \big/ \div A(X)$. 
\end{defn}

Letting $\nabla = \div^*$ for a given metric, 
the space $\Ker(\nabla) \incl A(X)$ of cocycles defines,
by orthogonality of $\Ker(\nabla)$ and $\Img(\div)$, 
the fundamental example of a globally cocyclic subsheaf of vector spaces.
This property is in turn equivalent to the uniqueness of equilibrium
of the heat equation $\dot u = \div(\nabla u)$ 
inside the homology class of any $h \in A_0(X)$.

\begin{defn}
The subsheaf $\Shf(X)$ is said {\rm locally cocyclic}
if $\Shf(\LL^\aa)$ is cocyclic for all $\aa \in X$.
\end{defn}

\begin{defn} 
Given a locally cocyclic subsheaf 
$\Shf(X) \incl A(X)$, we denote by $T_\aa : A(X) \aw A(X)$ the map extending 
$T_\aa : A(\LL^\aa) \aw \Shf(\LL^\aa)$ identities i.e. 
    \begin{equation} \tag{6.4.bis}
        T_\aa(u)_\bb = \left\{ \ba{ll}
        u_\bb + \delta_\bb \ph & \txt{if} \bb \incl \aa  \\
        u_\bb & \txt{otherwise} 
        \ea \right. 
    \end{equation}
\end{defn}

When $\Shf(X) = \Z(X)$ is the manifold of consistent potentials, 
theorem \ref{Tau} will imply that the map 
$u \mapsto T_\aa(u)$ 
is realised by one-step iteration 
of the $\mu\BP_1$ diffusion algorithm 
with messages restricted\footnote{
    It is common and often advised to update regions asynchronously:
    for instance draw a random $\aa$, compute messages below $\aa$, 
    update beliefs accordingly \cite{Yedidia-2005}. 
    Although argued to perform better, asynchronous 
    variations keep the same stationarity conditions and homological 
    structure of the synchronous algorithm. 
} to $\LL^\aa$
i.e. $\ph \in A_1(\LL^\aa)$ is given by $- \DF^\mu(u)_{|\LL^\aa}$,
as precised by the following proposition. 

\begin{prop} \label{Z-loc-cocyclic}
The subsheaf ${\cal Z}(X)$ of consistent interaction potentials is locally cocyclic.
\end{prop}

\begin{proof}
Given $\aa \in X$ and $u \in A(\LL^\aa)$, 
let $U = \zeta(u)$ and let $V_\bb = \Fh^{\bb\aa}(U_\aa)$ 
for every $\bb \incl \aa$. 
By consistency of $V$, the 
associated interaction potentials $v$ are in ${\cal Z}(\LL^\aa)$. 
One also has $v - u \in \div A_1(\LL^\aa)$ as defining for instance 
$\ph \in A_1(\LL^\aa)$ by: 
\begin{equation} \ph_{\aa\bb} = \sum_{\bb \aw \cc} 
\mu_{\bb\cc} \; \Fh^{\cc\aa}(U_\aa - U_\cc) \end{equation} 
and $\ph_{\bb\cc} = 0$ for all $\bb \neq \aa$ 
yields $U + \zeta \cdot \div \ph = V$. 
Note the correspondence with theorem \ref{Tau}. 

Furthermore, assume $v = u + \div A_1(\LL^\aa)$ 
with $u, v \in {\cal Z}(\LL^\aa)$.
Denoting by 
$U, V \in A(\LL^\aa)$ the associated 
local hamiltonians, 
    the Gauss formula \ref{gauss}
implies $U_\aa = V_\aa$ while consistency 
gives $V_\bb = U_\bb = \Fh^{\bb\aa}(U_\aa)$ for every $\bb \in \LL^\aa$. 
Hence $u$ and $v$ coincide on $\LL^\aa$ by Möbius inversion. 
\end{proof}

Belief propagation derives from the flux functional 
$\Phi(u) = \DF( \zeta \cdot u)$ and we shall see 
that the sheaf $\mathcal{Z}(X) = \{ \DF \circ \zeta  = 0 \}$ 
of manifolds in $A_0(X)$ is in general only locally cocyclic. 
To prove the uniqueness of its equilibria 
when the hypergraph $X$ is retractable, 
we have to show that the locally cocyclic sheaf ${\cal Z}(X)$ 
is then {\it globally} cocyclic. 
Our proof is constructive and relies on another property of ${\cal Z}(X)$, 
related to the extension of sections $u \in {\cal Z}(X')$ 
along normal extensions $X \cont X'$, under appropriate homological constraints. 

\begin{defn} Given a subsheaf $\Shf(X)$ of $A(X)$ and $X' \incl X$, 
    we introduce the subsets of $A(X)$: 
\bi
    \iii 
    $\Shf(X, X') = 
    \big\{ u \in A(X) \:\big|\: 
    u_{|\LL^\aa} \in \Shf(\LL^\aa) \;{\rm for \; all}\; 
    \aa \in X \setminus X' \big\}$
    \iii 
    $\Shf(X|X') = 
    \big\{ u \in \Shf(X,X') + A(X') \;\big|\; u_{|X'} \in \Shf(X') \big\}$
\ei
\end{defn} 

Note that according to this definition, 
$\Shf(X,X) = A(X)$ while $\Shf(X|X) = \Shf(X)$.

\begin{defn} 
    A locally cocylic sheaf $\Shf(X)$ of $A(X)$ is said {\rm extensible} 
    if for every $\aa \aw \bb$ in $X$ and
    every $u \in \Shf(\LL^\aa|\LL^\bb)$, its unique homologous representative 
    $v \in \Shf(\LL^\aa)$ extends $u_{|\LL^\bb}$.
\end{defn}

$\Shf(X)$ is extensible if and only if the following commutative diagram 
is commutative for all $\aa \cont \bb$:
\begin{equation} \bcd 
\Shf(\LL^\aa | \LL^\bb) \drar \rar{T_\aa} 
& \Shf(\LL^\aa) \dar \\
    & \Shf(\LL^\bb) 
\ecd \end{equation}

The subsheaf $\Shf(X|X')$ is by definition the inverse image 
of $\Shf(X') \incl A(X')$ under the restriction map 
$\Shf(X,X') + A(X') \aw A(X')$. 
In the following commutative diagram, 
where all horizontal arrows are inclusions and all 
descending arrows are restrictions to $X'$, 
the central square is hence a pull-back square. 
\begin{equation} \bcd 
\Shf(X) \rar \drar 
    & \Shf(X|X') \rar \dar 
    & \Shf(X,X') + A(X') \dar \rar 
    & A(X) \dlar \\
& \Shf(X') \rar & A(X') & 
\ecd \end{equation}
As we shall see, whenever $X$ is a normal extension of $X'$ and 
$\Shf(X)$ is extensible, 
there exists reciprocal projections $f: A(X) \aw \Shf(X, X')$
and $g: \Shf(X|X') \aw \Shf(X)$ preserving the homological constraints.

\begin{prop}
The subsheaf ${\cal Z}(X)$ of consistent interaction potentials
is extensible. 
\end{prop}

\begin{proof}
Given $\aa \cont \bb$ in $X$ and $u \in \Z(\LL^\aa|\LL^\bb)$, 
let us write $u = h + b$ 
with $h \in \Z(\LL^\aa)$ and $b \in A(\LL^\bb)$ such 
that $u_{|\LL^\bb} \in \Z(\LL^\bb)$.
Letting $H = \zeta(h)$ and $B = \zeta(b)$ and 
following the proof of \ref{Z-loc-cocyclic}, 
the unique homologous $v \in \Z(\LL^\aa)$ is defined by Möbius inversion of 
the local hamiltonians $(V_\cc)$ defined  for every $\cc \in \LL^\aa$ by: 
\begin{equation} V_\cc = \Fh^{\cc\aa}(H_\aa + B_\bb) \end{equation} 
For all $\bb \incl \aa$, we have by consistency of $H$ on $\LL^\aa$ and by
consistency of $U = H + B$ on $\LL^\bb$: 
\begin{equation} V_\cc = \Fh^{\cc\bb} \Big( \Fh^{\bb\aa}(H_\aa) + B_\bb \Big)  
= \Fh^{\cc\bb}(H_\bb + B_\bb) = H_\cc + B_\cc  = U_\cc\end{equation}
It follows that $v_{|\LL^\bb} = u_{|\LL^\bb}$ 
by \ref{zeta-alexandrov}, 
implying commutativity of restrictions with Möbius inversions.
\end{proof}

\subsection{Forward and Backward Passes}

\begin{thm} 
    Assume $X$ is a normal extension of $X'$ and $\Shf(X) \incl A(X)$ 
    is an extensible subsheaf. 
    Then every $u \in \Shf(X|X')$
    admits a unique homologous representative $v \in \Shf(X)$, 
    which extends $u_{|X'}$.
\end{thm} 

\begin{thm} \label{cocyclic}
    When $X$ is retractable, 
    any extensible subsheaf $\Shf(X) \incl A(X)$ is globally cocyclic.
\end{thm}

\begin{prop} \label{double-pass}
Given a normal sequence $X = X_n \cont \dots \cont X_0 = X'$
with $X_j = \LL^{\aa_j} \sqcup_{\LL^{\bb_j}} X_{j-1}$ for all 
$n \geq j \geq 1$ and an extensible subsheaf $\Shf(X) \incl A(X)$,
    the {\rm double pass}:
    \begin{equation} f = T_{\aa_n} \circ \dots \circ T_{\aa_1} \circ \dots \circ T_{\aa_n} \end{equation}
    maps $A(X)$ to the subspace $\Shf(X,X')$ formed by those $u$ such that 
    $u_{|\LL^\aa} \in \Shf(\LL^\aa)$ for all $\aa \in X \setminus X'$. 
\end{prop}

\begin{proof} 
Reasoning by induction on the length of the normal sequence, 
first consider $X_1 = \LL^{\aa_1} \sqcup_{\LL^{\bb_1}} X'$. 
Given $u \in A(X_1)$, let $v = T_{\aa_1}(u)$ 
so that $v_{|\LL^{\aa_1}} \in \Shf(\LL^{\aa_1})$ 
by definition of $T_{\aa_1}$, and $v \in \Shf(X_1,X')$.  
Assume now that 
$f_{j-1} = T_{\aa_{j-1}}  \dots  T_{\aa_1} 
 \dots  T_{\aa_{j-1}}$ 
does induce a map from $A(X_{j-1})$ to $\Shf(X_{j-1}, X')$ 
for some $1 < j \leq n$, and let $u \in A(X_j)$. 
\bi
\iii {\it forward pass:}
letting $\tilde u = T_{\aa_j}(u)$, 
we have $\tilde u_{|\LL^{\aa_j}} \in \Shf(\LL^{\aa_j})$. 

\iii{\it induction:}
letting $\tilde v = f_{j-1}(\tilde u)$, 
we have $\tilde v_{|X_{j-1}} \in \Shf(X_{j-1}, X')$. 
In particular $\tilde v_{|\LL^{\bb_j}} \in \Shf(\LL^{\bb_j})$ 
while $\tilde u_{|\LL^{\aa_j}} \in \Shf(\LL^{\aa_j})$
implies $\tilde v_{|\LL^{\aa_j}} \in \Shf(\LL^{\aa_j}| \LL^{\bb_j})$. 
\iii{\it backward pass:}
letting $v = T_{\aa_j}(\tilde v)$
we have $v_{|\LL^{\aa_j}} \in \Shf(\LL^{\aa_j})$ 
extending $\tilde v_{|\LL^{\bb_j}}$ by extensibility. 
\ei
Hence $T_{\aa_j} \dots T_{\aa_1} \dots T_{\aa_j}(u)$ 
lies in $\Shf(X_j , X')$ for all $1 \leq j \leq n$ 
and $f$ maps $A(X)$ to $\Shf(X,X')$. 
\end{proof}

\begin{prop} \label{backward-pass}
Given a normal sequence $X = X_n \cont \dots \cont X_0 = X'$
with $X_j = \LL^{\aa_j} \sqcup_{\LL^{\bb_j}} X_{j-1}$ for all 
$n \geq j \geq 1$ and an extensible subsheaf $\Shf(X) \incl A(X)$,
    the {\rm backward pass}:
    \begin{equation} g = T_{\aa_n} \circ \dots \circ T_{\aa_1} \end{equation}
    maps $\Shf(X|X')$ to $\Shf(X)$.
\end{prop}

\begin{proof} 
Assume first that $u \in \Shf(X_1|X')$, which is equivalent to 
$u_{|\LL^{\aa_1}} \in \Shf(\LL^{\aa_1}|\LL^{\bb_1})$ 
and $u_{|X'} \in \Shf(X')$. 
The extensibility property ensures that 
$v_{|\LL^{\aa_1}}$ extends $u_{|\LL^{\bb_1}}$ so that 
$v$ extends $u_{|X'}$ and $v \in \Shf(X_1)$.
Assuming now that $g_{j-1} = T_{\aa_{j-1}}  \dots  T_{\aa_1}$ 
maps $\Shf(X_{j-1}|X')$ to $\Shf(X_{j-1})$ for some $1 < j \leq n$,
let $u \in \Shf(X_j | X')$.
\bi
    \iii {\it induction:} 
    letting $\tilde v = g_{j-1}(u)$ we have 
    $\tilde v_{|X_{j-1}} \in \Shf(X_{j-1})$. 
    In particular $\tilde v_{|\LL^{\bb_j}} \in \Shf(\LL^{\bb_j})$,
    so that $\tilde v_{|\LL^{\aa_j}} \in u_{|\LL^{\aa_j}} + A(\LL^{\bb_j})$ 
    and $u_{|\LL^{\aa_j}} \in \Shf(\LL^{\aa_j}) + A(\LL^{\bb_j})$
    implies $\tilde v_{|\LL^{\aa_j}} \in \Shf(\LL^{\aa_j}| \LL^{\bb_j})$. 
    \iii {\it backward pass:} 
    letting $v = T_{\aa_j}(\tilde v)$ we have 
    $v_{|\LL^{\aa_j}} \in \Shf(\LL^{\aa_j})$ extending $\tilde v_{|\LL^{\bb_j}}$
    by extensibility. 
\ei
    Hence $T_{\aa_j} \dots T_{\aa_1}(u)$ lies in $\Shf(X_j)$ 
    for all $1 \leq j \leq n$ and $g$ maps $\Shf(X|X')$ to $\Shf(X)$.
\end{proof}

\begin{prop}
Given a normal extension $X \cont X'$ and an extensible subsheaf 
$\Shf(X) \incl A(X)$, the double pass and backward pass maps: 
    \begin{equation}
        f : A(X) \aw \Shf(X,X') \quad\txt{and}\quad 
    g : \Shf(X|X') \aw \Shf(X) \end{equation} 
are uniquely defined, under the constraint that the image of $u$ 
lies in $u + \div A(X, X')$. 
\end{prop}

\begin{proof}
Reasoning by induction on the length of a normal sequence,
first assume that $X = \LL^\aa \sqcup_{\LL^\bb} X'$.
If $u, v \in \Shf(X , X')$ with $v = u + \div \ph$ 
and $\ph \in A(\LL^\aa)$, then $v_{|\LL^\aa} = u_{|\LL^\aa}$
by cocyclicity of $\Shf(\LL^\aa)$ and $\div \ph = 0$. 
Assume now that $X = \LL^\aa \sqcup_{\LL^\bb} X_1$ 
with $X_1 \cont X'$, and 
that for all $u, v \in \Shf(X_1,X')$, if 
$v$ lies in $u + \div A(X_1, X')$ then $v = u$. 
Given $u, v \in \Shf(X,X')$ with $v = u + \div \ph$ 
and $\ph \in A(X, X')$, let us write 
$\div \ph = \div \ph^+ + \div(\ph_{|\LL^\bb}) + \div \ph^-$ 
with $\ph^+ \in A(\LL^\aa)$ and $\ph^- \in A(X_1, X')$.  
\bi
    \iii {\it no forward pass:} 
    letting $\tilde v = T_\bb(u + \div \ph^-)$, we have 
    $\tilde v_{|\LL^\aa} \in \Shf(\LL^\aa) + \div A(\LL^\bb)$ and 
    $u_{|\LL^\bb} \in \Shf(\LL^\bb)$  so that 
    $\tilde v_{|\LL^\aa} \in \Shf(\LL^\aa | \LL^\bb)$. 
    Extensibility implies that 
    $v = T_\aa( u + \div \ph^-) = T_\aa(\tilde v)$ 
    extends $\tilde v_{|\LL^\bb}$.
    \iii {\it induction:} 
    as $v_{|X_1} = \tilde v_{|X_1}$ lies in $u_{|X_1} + \div A(X_1, X')$ 
    we have $v_{|X_1} = u_{|X_1}$ by induction hypothesis. 
    \iii {\it no backward pass:} 
    in particular $\div(\ph)_{|X_1} = 0$ 
    so that $v_{|\LL^\aa} = u_{|\LL^\aa} + \div \ph^+$
    lies in $u_{|\LL^\aa} + \div A(\LL^\aa)$, 
    and $v_{|\LL^\aa} = u_{|\LL^\aa}$ by local cocyclicity. 
\ei
    Hence $v = u$ for every $u, v \in \Shf(X,X')$ 
    such that $v \in u + \div A(X,X')$.
    It follows that the double pass map $f : A(X) \aw \Shf(X, X')$ associated 
    to any given normal sequence is independent of the sequence. 
    Uniqueness of the backward pass map $g : \Shf(X|X') \aw \Shf(X)$ 
    is also obtained by applying the same uniqueness argument to 
    any $u,v \in \Shf(X) \incl \Shf(X,X')$ such that $v \in u + \div A(X,X')$.
\end{proof}

\section{Relations with Global Sections} 

In this section, we show that the unique equilibrium on 
retractable hypergraphs coincides with the true Gibbs state marginals one seeks 
to approximate. 

In general, we introduce a {\it global pass} map $T : A_0(X) \aw \Z(X)$ 
defined by the effective energies of the global hamiltonian and their associated 
interaction potentials and investigate some of its homological properties. 
As $T(h)$ may not be homologous to $h$ in general, 
we signal that the true Gibbs state marginals may not be accessible from 
message-passing schemes. 
After showing that $T$ coincides with the double pass maps of the previous section
on retractable hypergraphs, 
we give an explicit example where $T$ induces a non-trivial iteration over
the manifold of equilibria. 

\subsection{Global Pass}

\begin{defn} 
    We call {\rm global completion} of 
    $X \incl \Part(\Om)$ the hypergraph $\tilde X = X \cup \{ \Om \}$.
\end{defn}

\begin{defn}
Given a global hamiltonian $H_\Om = \disp \sum_{\aa \in X} h_\aa$ in $A_\Om$,
let us define: 
\bi 
\iii the {\rm true effective hamiltonians} by
$H^*_\aa = \Fh^{\aa\Om}(H_\Om)$ for all $\aa \in \tilde X$,
\iii the {\rm true effective potentials} by 
$h^* = \mu \cdot H^*$,
\iii the {\rm true Gibbs state marginals} by
$p = [\e^{-H^*}]$, 
\ei
each of which being consistent, 
as $\DF(H^*) = 0$ implies $h^* \in {\cal Z}(\tilde X)$ and $p \in \Gamma(\tilde X)$.
\end{defn}

\begin{defn} 
We call {\rm global pass} 
$T : A_0(X) \aw A_0(X)$ the smooth map 
associating to potentials $h \in A_0(X)$ the true effective 
potentials $h^* \in {\cal Z}(X)$ deriving from the total energy 
$H_\Om = \sum_\aa h_\aa$.
\end{defn}

The global pass resorts to the global hamiltonian and 
should be thought of as an {\rm extrinsic} action on 
$A_0(X)$, uncomputable in practice.
The following diagram best illustrates the procedure:
\begin{equation} \bcd 
    & A_\Om \drar{\mu \:\circ\: \Fh^{-\Om}} & \\
A_0(X) \urar{\zeta_\Om} \ar[rr,dashed, "T"] & & A_0(X) 
\ecd \end{equation}
It is good to view $T$ as a map from 
$A_0(X)$ to itself as we shall see its iterations may not be trivial.
Its image is however contained in ${\cal Z}(X)$ by definition, 
while it induces a map in homology: 
\begin{equation} \bar T : {\rm H}_0(X; A) \law {\cal Z}(X) \end{equation} 
Whether this map is a diffeomorphism or not is a natural and important question, 
closely related to the uniqueness or multiplicity of message-passing equilibria. 
We shall see its answer to depend on the geometry of the underlying hypergraph.

\begin{prop} \label{Th=h}
    Assume $\tilde h \in {\cal Z}(\tilde X)$ 
    extends $h \in {\cal Z}(X)$. 
    Then $\tilde h_\Om = 0$ implies $T(h) = h$. 
\end{prop}

\begin{proof} 
    Letting $H_\Om = \sum_\aa h_\aa$ and 
    $\tilde H = \tilde \zeta \cdot \tilde h$,
    then $\tilde h_\Om = 0$ implies $\tilde H_\Om = H_\Om$. 
    If furthermore $\tilde h \in {\cal Z}(\tilde X)$, 
    then consistency implies that $\tilde H$ coincides
    with the true effective hamiltonians $H^*$ and $T(h) = h$. 
\end{proof}

\begin{prop} \label{Th*=h*}
    Given $h \in {\cal Z}(X)$ 
    let $\tilde h^* = \tilde T(h) \in {\cal Z}(\tilde X)$. 
    Then $T(h) = h$ implies $\tilde h^*_\Om = 0$.
\end{prop}

\begin{proof}
An immediate consequence of the following proposition.
\end{proof} 

{\bf Remarks.} Due to the apparent reciprocity between propositions 
\ref{Th=h} and \ref{Th*=h*}, one may be tempted to close an implication loop
and it seems important to mention that:
\bi
\iii
$(T(h^*) = h^* )\not\impq (\tilde h^*_\Om = 0)$ \\[0.2em]
Although $T(h^*) = T(h)$ gives
$\Fh^{\aa\Om}(H_\Om) = \Fh^{\aa\Om}(H^*_\Om)$ for all $\aa \in X$,
this may not imply $H^*_\Om = H_\Om$. 
We do have a global interaction decomposition
$A_\Om = Z_\Om \oplus B_\Om$,
however $H_\Om \in B_\Om$ does not imply $\e^{-H_\Om} \in B_\Om$ in general.
Hence $\Sigma^{\aa\Om}(\e^{-H_\Om}) = \Sigma^{\aa\Om}(\e^{-H^*_\Om})$
for all $\aa \in X$ does not imply $\tilde h^*_\Om = 0$.
\iii 
$(\tilde h^*_\Om = 0) \not\impq (T(h) = h)$ \\[0.2em]
Even when  $h \in {\cal Z}(X)$ and its extension by zero $\tilde h$
lies in ${\cal Z}(\tilde X)$, 
assuming $h^*$ homologous to $h$ in $A_0(X)$ 
does not imply $h^* = h$,  
as there may exist multiple $u = h + \div \ph \in {\cal Z}(X)$. 
\ei

\begin{prop} \label{h*-global}
Given $h \in {\cal Z}(X)$, let $\tilde h^* = \tilde T(h)$
and let $h^* = \tilde h^*_{|X}$.
\begin{equation} \tilde h^*_\Om = \sum_{\aa \in X} h_\aa - \sum_{\aa \in X} h^*_\aa \end{equation}
\end{prop}

\begin{proof} 
Letting $H_\Om = \sum_\aa h_\aa$ and $\tilde H^* = 
\tilde \zeta \cdot \tilde h^*$, we have 
$\tilde H^*_\Om = H_\Om$ by definition 
of the true effective hamiltonians. 
Isolating the contribution of $\Om$ to the total energy
then gives $H_\Om = \tilde h^*_\Om + \sum_\aa h^*_\aa$. 
\end{proof}

\subsection{Retractable Hypergraphs}

\begin{thm}\label{global} 
Assume $X \incl \Part(\Om)$ is retractable. 
For every collection of interaction potentials $h \in A_0(X)$ 
the unique $u = h + \div \ph \in {\cal Z}(X)$ 
coincides with the true effective potentials $h^* = T(h)$. 
\end{thm}

Equivalently, the theorem asserts that for any $h \in A_0(X)$,
extending the unique homologous potentials $u = h + \div \ph \in {\cal Z}(X)$ 
given by theorem \ref{cocyclic} 
to $\tilde X = X \sqcup \{ \Om \}$ does yield  
a section $\tilde u \in {\cal Z}(\tilde X)$. 
It must then coincide with the unique 
$\tilde h^*  = \tilde h + \div \tilde \ph \in {\cal Z}(\tilde X)$ 
defined from $\tilde h \in A_0(\tilde X)$ extending $h$, 
and of which the global pass $h^* \in {\cal Z}(X)$ is the restriction: 
\begin{equation} \bcd 
A_0(\tilde X) \rar{\tilde T} & {\cal Z}(\tilde X) \\
A_0(X) \uar \rar{f} & {\cal Z}(X) \uar 
\ecd \end{equation}
The existence of the right ascending arrow hence implies 
commutativity in the above diagram, and the proof of the theorem 
will rely on the following proposition.

\begin{prop} \label{retract}
Given $\Om = \aa \sqcup_{\bb} \aa'$,
assume $X = \LL^\aa \sqcup_{\LL^\bb} \LL^{\aa'}$ and let 
$\tilde X = X \sqcup \{ \Om \}$. 
Then: 
\begin{equation} {\cal Z}(\tilde X) \cap A_0(X) \simeq {\cal Z}(X) \end{equation}
    Equivalently, if $\tilde U = \tilde \zeta \cdot u \in A_0(\tilde X)$ 
    with $u \in A_0(X)$, then $\DF(\tilde U) = 0$ if and only if 
    $\DF(\tilde U)_{|X} = 0$. 
\end{prop}

\begin{proof}
Given any $\tilde U \in A_0(\tilde X)$, first observe that
$\DF(\tilde U) = 0$ is equivalent to $U_\cc = \Fh^{\cc\Om}(U_\Om)$
for all $\cc \in X$. 
The assumption 
$X = \LL^\aa \cup \LL^{\aa'}$ implies that $\Fh^{\cc\Om}$ may be factorised 
either as $\Fh^{\cc\aa} \circ \Fh^{\aa\Om}$ or as
$\Fh^{\cc\aa'} \circ \Fh^{\aa'\Om}$ 
depending on whether $\cc \in \LL^\aa$ or $\cc \in \LL^{\aa'}$,
so that $\DF(\tilde U)$ is actually equivalent to:
\begin{equation} \DF(\tilde U)_{\Om \aw \aa}  = 0 \txt{and}
\DF(\tilde U)_{\Om \aw \aa'} = 0 \txt{and}
\DF(\tilde U)_{|X} = 0 \end{equation}
Assume now that $\tilde U = \tilde \zeta \cdot u$ with $u \in A_0(X)$.
Observing that  $\Om \setminus \aa = \aa' \setminus \bb$,
applying proposition \ref{eff-indep} to 
the effective energy of 
$\tilde U_\Om - \tilde U_\aa = \tilde U_{\aa'} - \tilde U_{\bb}$
which lies in $A_{\aa'}$ gives:
\begin{equation} \DF(\tilde U)_{\Om \aw \aa}  
= \: \Fh^{\aa\Om}\bigg( \sum_{\cc \in X \setminus \LL^\aa} u_{\cc} \bigg)  
= \:\Fh^{\bb\aa'}\bigg( 
    \sum_{\cc \in \LL^{\aa'} \setminus \LL^{\bb}} u_{\cc} \bigg) 
= \DF(\tilde U)_{\aa' \aw \bb}
\end{equation}
Interverting $\aa$ and $\aa'$ similarly 
yields $\DF(\tilde U)_{\Om \aw \aa'} = \DF(\tilde U)_{\aa \aw \bb}$.
For all $\tilde U \in {\tilde \zeta} \cdot A_0(X)$ 
we hence have $\DF(\tilde U) = 0$ if and only if $\DF(\tilde U)_{|X} = 0$,
and it follows that the image of ${\cal Z}(X)$ in $A_0(\tilde X)$ by the natural 
inclusion $A_0(X) \incl A_0(\tilde X)$ coincides with the intersection 
of ${\cal Z}(\tilde X)$ with $A_0(X)$.
\end{proof}

\begin{proof}[Proof of theorem \ref{global}]
Reasoning by induction on the length of $X$, the statement is tautological 
when $X$ already contains $\Om$. 
Given a normal extension 
$X = \LL^\aa \sqcup_{\LL^\bb} X'$ with $X' \incl \Part(\Om')$, 
assume now that for every $u' \in {\cal Z}(X')$
its extension $\tilde u'$ to $\tilde X' = X' \cup \{ \Om' \}$ lies 
in ${\cal Z}(\tilde X')$. Consider the hypergraphs: 
\begin{equation} Y = X \cup \{ \Om' \} \quad\txt{and}\quad \tilde Y = Y \cup \{ \Om \} \end{equation}
Given any $u \in {\cal Z}(X)$, let us show that its extension $\tilde u$ 
to $\tilde X$ by zero lies in ${\cal Z}(\tilde X)$, 
by applying the previous proposition in $\tilde Y$, 
global completion of $Y = \LL^\aa \sqcup_{\LL^\bb} \LL^{\Om'}$
with $\LL^{\Om'} = \tilde X'$. 
\bi
    \iii 
    $u' = u_{|X'} \in {\cal Z}(X')$ extends to 
    $\tilde u' \in {\cal Z}(\tilde X')$ by induction hypothesis,
    \iii 
    $v \in {\cal Z}(Y)$ defined from 
    $u_{|\LL^\aa} \in {\cal Z}(\LL^\aa)$ 
    and $\tilde u' \in {\cal Z}(\LL^{\Om'})$ extends 
    to $\tilde v \in {\cal Z}(\tilde Y)$ 
    by proposition \ref{retract},
    \iii 
    $\tilde v \in {\cal Z}(\tilde Y)$ restricts to 
    $\tilde u \in {\cal Z}(\tilde X)$ 
    as $\DF(\tilde V) = 0$ implies $\DF(\tilde V)_{|\tilde X} = \DF(\tilde U) = 0$,
   where $\tilde V$ and $\tilde U$ denote 
    their zeta transforms in $A_0(\tilde Y)$ and 
    $A_0(\tilde X)$ respectively.
\ei
The induction shows that for every $h \in A_0(X)$, 
its unique homologous representative $u = f(h) \in {\cal Z}(X)$
defined by double pass extends to $\tilde u \in {\cal Z}(\tilde X)$.  
The global pass depending only on the total energy,
we have $T(u) = T(h) = h^*$, 
while $\tilde u \in {\cal Z}(\tilde X)$
with $\tilde u_\Om = 0$ implies $T(u) = u$ 
according to \ref{Th=h}.
\end{proof}

\subsection{Loop Effective Potential} 

The true effective potentials $h^*$ are not homologous to $h$ in general. 

\begin{figure}[H] 
    \sbox0{\includegraphics[width=0.2\textwidth]{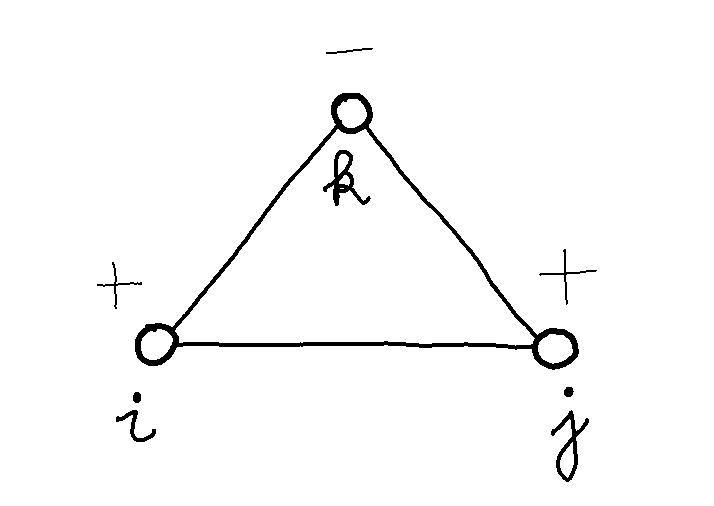}}
\begin{center}
\begin{minipage}
    {0.8\textwidth}
\centering
\usebox0
\caption{   \label{fig-triangle}
    A triangle of binary variables.
}
\end{minipage}
\end{center}
    \vspace{-0.2cm} 
\end{figure}

The simplest example of this phenomenon is given by the the Ising model
on the triangular graph with binary variables $x_i = \pm 1$ 
on each vertex, as depicted in figure \ref{fig-triangle}. 
Consider a uniform coupling constant $w$ and no external field for simplicity,
so that interaction potentials are given by: 
\begin{equation} h_{ij}(x_i, x_j) = - w\, x_i\, x_j \quad\txt{and}\quad h_i(x_i) = 0 \end{equation} 
Letting $H = \zeta \cdot h$ denote the local hamiltonians 
and $q = [\e^{-H}]$, the local Gibbs states $q \in \Gamma(X)$ 
are consistent for any value of $w$.
Considering that 
$A_i = \R \,|+\rangle \oplus \R \,|-\rangle$ and $A_{ij} = A_i \otimes A_j$, 
the following matrix representations\footnote{
    Note that the ring structure of $A_{ij}$ is induced by 
    the element-wise or {\it Hadamard} product, while  
    the matrix-vector product of $q_{ij}$ with $f_j \in A_j$ 
    yields the conditional expectation $\E^{ij}[\,f_j\,|\,i\,] \in A_i$. 
} 
will be convenient:
\begin{equation} 
q_{ij} = \frac 1 {4 \cosh(w)}
\left[ \begin{array}{cc} 
\e^{w} & \e^{-w} \\
\e^{-w} & \e^{w} 
\end{array} \right] 
\quad\txt{and}\quad 
q_i = \frac 1 2
\left[ \begin{array}{c} 
1 \\ 1
\end{array} \right] 
\end{equation}
Writing $q = {\cal Q}(w)$, this defines 
a one-parameter family ${\cal Q} : \R \aw \Gamma(X)$ 
of consistent statistical fields.

Consider now the true effective hamiltonians $H^* \in A_0(X)$ defined 
from the total energy $H_\Om = \sum h_{ij}$. 
In addition to the direct coupling $h_{ij}$, 
the true effective hamiltonian $H^*_{ij}$ 
accounts for the interaction of $i$ and $j$ through $k$,
expressed by the effective energy of $h_{ik} + h_{kj}$.
\begin{equation} 
\big( H^*_{ij} - H_{ij} \big) (x_i, x_j)  
= -\ln \sum_{x_k \in \{ \pm 1\} } \e^{- w \, x_k( x_i + x_j)}
\end{equation}
Introducing the real function $f: y \mapsto \ln \big( \e^y + \e^{-y} \big)$,
which satisfies $f(-y) = f(y)$, 
we have: 
\begin{equation} 
H^*_{ij} - H_{ij} = -
\left[ \begin{array}{cc} 
f(2w) & f(0) \\
f(0) & f(2w) 
\end{array} \right] 
\end{equation}
Letting then $g(w) = \frac 1 2 \big( f(2w) - f(0) \big)$
we may write up to the constant term $\frac 1 2 \big( f(2w) + f(0) \big)$:
\begin{equation} 
H^*_{ij} - H_{ij} \;\simeq \;
g(w) \;
\left[ \begin{array}{rr} 
-1 & 1 \\
1 & -1
\end{array} \right] 
\mod \R
\end{equation}
It follows that $H^*_{ij} - H_{ij}$ is a non-zero element 
of $Z_{ij} \oplus \R$ for every $w \neq 0$ and does not lie 
in $A_i + A_j$.  
This in particular 
implies that $H^* - H$ cannot belong to $\zeta \cdot \delta A_1(X)$,
equivalently that the true effective potentials $h^*$ are not homologous to $h$ 
in $A_0(X)$.
A loop effective potential $\tilde h^*_{\Om} = \tilde h^*_{ijk}$
 measures this obstruction, 
it may be computed using proposition \ref{h*-global}. 
Up to a constant free energy term, $\tilde h^*_{ijk}$ 
is the sum of the contributions $H_{ij} - H^*_{ij}$ of each edge.

The true local Gibbs states 
$p \in \Gamma(X)$ remain in the-one parameter family as $p = {\cal Q}(w')$ 
is given by:
\begin{equation} 
p_{ij} = \frac 1 {4 \cosh(w')}
\left[ \begin{array}{cc} 
\e^{w'} & \e^{-w'} \\
\e^{-w'} & \e^{w'} 
\end{array} \right] 
\quad\txt{with}\quad w' = w + g(w)
\end{equation}
Iterating the global pass $T : A_0(X) \aw A_0(X)$ hence defines 
sequences of couplings $w \in \R^\N$ by the recurrence relation:
\begin{equation} w_{n + 1} = w_n + g(w_n) \end{equation} 
and as $g(w) = \frac 1 2 \ln \cosh(2 w)$ behaves like
a soft absolute value, with $|g'| < 1$, one quickly sees that:
\begin{equation} 
\lim_{ n } w_n = 
\left\{ \begin{array}{cl} 
    0       &   {\rm if}\; w \leq 0 \\
    +\infty &   {\rm otherwise}
\end{array} \right. 
\end{equation}
This should be interpreted by observing that the effective contribution 
of $h_{ik} + h_{kj}$ is always ferromagnetic, whatever the sign of $w$,
 as the path $(i,k,j)$ from $i$ to $j$ is of even length.

A symmetric behaviour should be expected for the Ising model on the square graph, 
while in this case a finer description should also keep track of order 3 interaction
potentials $h^*_{ijk}$, by embedding the graph inside the 3-simplex. 
Finally, although allowing 
weights $w_{ij}$ to differ from edge to edge would induce a more 
interesting dynamic on a higher dimensional space of parameters, 
manual computations may quickly become cumbersome and call for numerical simulations.

\section{Bifurcations and Singularities}

In an effort to witness the emergence of multiple equilibria, 
this section characterises bifurcation states 
as consistent interaction potentials $u$ whose tangent fiber $T_u \Z(X)$ 
is not transverse to $\div A_1(X)$. 
These bifurcations correspond to singularities of the homological projection 
$\Z(X) \aw {\rm H}(X)$, which we relate to spectral properties 
of a {\it twisted laplacian}\footnotemark{}
operator $L = \div \nabla^\mu$ describing the 
linearised diffusion flow, 
we finally exhibit explicit examples of such singularities 
on simple graphs with two loops. 

\footnotetext{
    Although 
    the present problem does share many similiarities with classical harmonic theory,
    let us already emphasise that, in addition to the 
    non-linearity of $\div \DF^\mu$,
    the conjugation by combinatorial transforms
    in its linearised flow
    breaks the adjunction of the 
    boundary $\div$ with the differential 
    $\nabla^\mu = \mu \circ \nabla \circ \zeta$.
}

\subsection{The Transversality Problem}

At the level of local hamiltonians, 
the tangent fibers of the consistent manifold $\CU = \zeta \cdot \Z$
bear a convenient probabilistic and cohomological description, 
involving the linearised effective energy gradient $\nabla = \DF_*$ 
defined in \ref{nabla}: 

\begin{prop} Let $U \in  \CU$ and $p = [\e^{ - U}]$. 
    Then $T_U \CU = \Ker(\nabla_p)$ is defined by the equations: 
    \begin{equation} \nabla_p(V)_{\aa\bb} = V_\bb - 
    \E^\aa_{p_{\aa}}\big[\, V_\aa \,\big|\, \bb \,\big] = 0 \end{equation}
\end{prop}

Let us denote by $\bar 0 = \mu \cdot \ln [1]$ the origin of $\Z$, 
projection of $0 \in \Z'$ onto $\Z$ along $\R_0$. 

\begin{thm} \label{transversality-0}
    The tangent space $Z_0 = T_{\bar 0} {\cal Z}$ is globally cocyclic:
    \begin{equation} A_0 = Z_0 \oplus \div A_1 \end{equation}
    The homological projection $P : A_0 \aw Z_0$ is moreover explicitly 
    given by interaction decomposition.  
\end{thm}

\begin{proof} 
Letting $\ph_{\bb\cc} = {\bf z}_\cc(u_\bb)$ we have by interaction decomposition:
\begin{equation} \label{AtoZ0}
(u + \div \ph)_\bb = \sum_{\aa' \cont \bb} {\bf z}_\bb(u_{\aa'}) 
\end{equation}
Hence $v = u + \div \ph$ 
defines the unique collection of interaction potentials 
in $T_{\bar 0}{\cal Z}$ homologous to $u \in A_0$. 
The map $P : u \mapsto v$ coincides with the interaction projection of
\ref{conservation}.
\end{proof}

When the underlying hypergraph is retractable, theorem \ref{cocyclic}
implies that $T_u\Z$ is cocyclic for all $u \in \Z$.  
Note that even in this case, we do not have an explicit 
and convenient formula such as (\ref{AtoZ0}) 
to compute the projection of $A_0$ onto $T_u\Z$ 
along the fibers of $\div A_1$. 
However a linearised message passing scheme will produce the 
desired result in finite time, 
for instance by double pass as in proposition \ref{double-pass}. 

In general, the fibers of $T\Z$ are still related to the interaction 
decomposition through the following proposition. 
The obtained formula shows how moving $u$ along $\Z$ rotates $Z_0$:
but while doing so, it may very well happen that $T_u \Z$ meets $\div A_1$. 
One should hence not expect to derive from its apparent simplicity 
a homologous projection of $A_0$ onto $T\Z$.

\begin{prop} 
Let $p = [\e^{- \zeta \cdot u}]$ denote local Gibbs states. 
Then: 
\begin{equation} T_u {\cal Z} = (\mu \circ p^{-1} \circ \zeta) \cdot Z_0 \end{equation}
where $p^{-1}$ denotes the region-wise multiplication operator
$v_\aa \mapsto v_\aa / p_\aa$.
\end{prop}

\begin{proof}
By theorem \ref{Ker-d}, in $A_0$ we have $\Ker(d) = \zeta \cdot Z_0$. 
The change of  metric is simply reflected by conjugating 
with the diagonal multiplication operator $p = [\e^{- \zeta \cdot u}]$:
    \begin{equation} \nabla = p^{-1} \cdot d \cdot p \end{equation}
Hence $\Ker(\nabla) = p^{-1} \cdot \Ker(d)$, while 
    by definition \ref{nabla} of $\nabla$ we have 
$T_u {\cal Z} = \Ker(\nabla \circ \zeta)
= \mu \cdot \Ker(\nabla)$.
\end{proof}

\subsection{Twisted Laplacian and Singularities} 

\renewcommand{\Z}{{\cal Z}}
\newcommand{\Lap}{{\cal L}} 
\newcommand{\lap}{L} 
\renewcommand{\S}{{\cal S}}

Let us denote by $\Lap: A_0 \aw T A_0$ the non-linear diffusion operator 
defined in \ref{diffusion}:
\begin{equation} \Lap = \div \DF^{\mu} 
\quad\txt{where}\quad \DF^{\mu} = \mu \circ \DF \circ \zeta \end{equation}
As a vector field, $\Lap$ fixes $\Z$.
The restriction of its tangent map $\Lap_* : TA_0 \aw T(TA_0)$ hence defines 
a linear endomorphism of the tangent bundle above $\Z$, 
describing the linearised action of $\Lap$ in an infinitesimal neighbourhood of $\Z$.

\begin{defn}
    We call {\rm twisted laplacian} the linearised flow $L = \div \nabla^\mu$ 
    induced by $\Lap_*$ above $\Z$:
    \begin{equation} L : T_\Z A_0 \law T_\Z A_0 \end{equation}
\end{defn}

\begin{prop} \label{ker-L}
    For every $v \in T_\Z A_0$, we have $L(v) = 0$ if and only if $v \in T\Z$, 
    i.e.:
    \begin{equation} 
    \Ker(L) = T \Z 
    \end{equation} 
\end{prop}

\begin{proof} 
The local faithfulness property \ref{loc-faithful} of $\phi = \DF^\mu$ 
implies that $\Lap = 0$ is a local equation of $\Z$. 
\end{proof}

Consider now the sub-bundle $\B \incl T_\Z A_0$
spanned by $\div A_1$. 
By definition of $L = \div \nabla^\mu$, we have: 
\begin{equation} 
\Img(L) \incl \B 
\end{equation} 
Hence $L$ could be factorised 
by a map from $T^\perp \Z$ 
to $\B$ having same dimension, as $\nabla = \div^*$ 
implies $\Ker(\nabla) = \Img(\div)^\perp$ 
and $T\Z = \mu \cdot \B^\perp$. 
In constrast with the harmonic case, 
$T\Z$ and $\B$ may however intersect and 
fail to span $T A_0$.

\begin{defn} We denote by $L' : \B \aw \B$ the restriction 
of $L$ to boundaries:
\begin{equation} L' \in \Hom(\B, \B) 
\quad\eqvl\quad
L' \in C^\infty \big(\Z, \,\End(\div A_1)\big) 
\end{equation}
\end{defn}

\begin{prop} 
    The kernel of $L'$ is 
    the intersection $\Ker(\nabla^\mu) \cap \Img(\div)$, equivalently: 
    \begin{equation} \Ker(L') = T \Z \, \cap \, \B \end{equation} 
\end{prop}

\begin{proof} 
    This is a direct consequence of proposition \ref{ker-L} 
    and the definition of $L'$.
\end{proof}

When the underlying hypergraph is retractable, the cocyclicity of $\Z$ 
implies $T_\Z A_0 = T\Z \oplus \B$ so that $L'$ is invertible and 
defines an isomorphism of $\B$ onto itself. 
The double pass \ref{double-pass} suggests\footnotemark{} that $L$
has only positive eigenvalues, the flow of $-L$ hence retracting
$T_\Z A_0$ onto $T\Z$. 
\footnotetext{
    The double pass is nothing but a sequential version of $L$ 
    where maximal cells are updated in a specific order.
}

In general, following Thom's notations \cite{Thom-56} 
we denote by $\S_0 \incl \Z$ the subset where $L'$ is invertible, 
and by $\S_k \incl \Z$ the subset where $L'$ is of corank $k$. 
Then $\Z$ has the structure of a stratified space, 
with respect to the disjoint union: 
\begin{equation} \Z = \bigsqcup_{k \geq 0} \S_k \end{equation}
where $\bar \S_k$ contains $\S_{k+1}$ for all $k \geq 0$, 
and $\bar \S_1$ consists of the {\it singularities} of $L'$, points 
where $L'$ is not invertible, and $T\Z \cap \B$ contains 
directions of {\it bifurcation}. 

\begin{thm}
    The singular set $\bar \S_1$ of $L'$ is nowhere dense in $\Z$. 
\end{thm}

\begin{proof}
Denote by $M \incl \Ker(d)$ the open subset of $\Ker(d) \incl A^*_0$ 
formed by consistent positive densities. 
$M$ is analytically diffeomorphic to $\Z$,
as $q \mapsto - \ln(q)$ is an analytic parametrisation\footnotemark{}
of $\zeta \cdot \Z$, 
and we denote by $f : M \law \Z$ the analytical diffeomorphism then 
obtained by Möbius inversion. 
\footnotetext{
    The coordinate $H(x)$ is given by composing the linear evaluation map
    $q \mapsto q(x)$ with the analytic map $y \mapsto - \ln(y)$. 
}

Letting $N = A_0/ \div A_1$ denote the homology of $A_0$, 
consider a projection $\pi : A_0 \law N$ as given for instance by
(\ref{AtoZ0}). 
The composed projection $g = \pi \circ f$ and its tangent map:
\begin{equation} g_* : TM \law TN \end{equation}
are analytic, while $f_*$ defines a diffeomorphism of 
$\Ker(g_*)$ onto $T\Z \cap \B$ by construction. 
By Sard's theorem, the singular set of $g$ where $\det(g_*) = 0$
has an image of measure zero in $N$. 
In particular, $\det(g_*)$ cannot identically vanish,
as \ref{transversality-0} implies that 
    $g(M)$ contains an open neighbourhood\footnotemark{} in $N$
by transversality of $T_{\bar 0}\Z$ and $\div A_1$.
\footnotetext{
    In fact $g$ is surjective, as implied by 
    the existence theorem \ref{existence} of belief propagation equilibria.
}

    The map $\det(g_*) : M \aw \R$ is analytic and admits 
    an extension $\det_\C(g_*) : M_\C \aw \C$ to $M_\C \incl A^*_0 \otimes \C$.
    If $\det(g_*)$ vanished on an open set in $M$,
    then $\det_\C(g_*)$ would also vanish on an open subset in $M_\C$,
    hence everywhere by analytic continuation.  
    As this would violate Sard's theorem, it follows that 
    the regular set where $g_*$ is invertible is dense in $M$,
    and $T\Z \cap \B = \{ 0 \}$ almost everywhere. 
\end{proof}

\subsection{The Case of Graphs}

\bgroup 
\newcommand{\Edge}{{\rm E}}
\newcommand{\Loop}{{\rm L}}
\newcommand{\Chain}{{\rm C}}
\newcommand{\M}{{\rm M}}

We now specialise to graphs as the equations characterising bifurcations 
in $T\Z \cap \B$ there
show a remarkable resemblance with the Kirchhoff rule 
for the conservation of electric currents in a circuit. 
The remaining major difference is that our currents $\ph \in A_1$ 
are function-valued, while constant scalars cannot contribute to 
the creation of bifurcations, as the second part of the following 
proposition implies. 
These conservation laws also give a very simple explanation
to the absence of bifurcations on trees.

\begin{prop} 
If $v = \div \ph$ is a bifurcation 
in $T_u \Z \cap \B$, then $\ph \in A_1$ is solution of:
\begin{equation} \label{kirchhoff}
\ph_{jk \aw k} = \E^{jk} \Big[ \sum_{i \neq k} \ph_{ij \aw j} 
\:\Big|\: k \:\Big] \mod \R 
\end{equation}
for every edge-vertex pair $jk \aw k$, conditional expectations 
being taken with respect to $[\e^{- \zeta \cdot u}] \in \Gamma$. 
Moreover, $\E^{i}[v_i] = \E^{ij}[v_{ij}] = 0$ for all $i, j$, 
so that $v \neq 0$ implies that $\ph$ cannot belong to $\R_1$. 
\end{prop}

\begin{proof} 
    If $v = \div \ph \in \Ker(\nabla^\mu)$, 
    the variation of hamiltonians $V = \zeta \cdot v$ satisfies 
    $\nabla(V)_{jk \aw k} = 0$ where:
    \begin{equation} 
    V_k = \sum_{i'} \ph_{i'k \aw k} \mod \R
    \quad \txt{and} \quad
    V_{jk} = \sum_{i' \neq j} \ph_{i'k \aw k} + \sum_{i \neq k} \ph_{ij \aw j} 
    \mod \R
    \end{equation}
    neglecting only currents of the form $\ph_{ij \aw \vide}$ and 
    $\ph_{i \aw \vide}$, 
    so that $V_k = \E^{jk}[V_{jk} \st k]$ does simplify to (\ref{kirchhoff}) 
    in $A_k \mod \R$. 
    To prove the second statement, observe that the equations 
    $\nabla(V)_{k \aw \vide} = 0$ and $\nabla(V)_{jk \aw \vide} = 0$ 
    yield after a few obvious simplifications: 
    \begin{equation} \ph_{k \aw \vide} = \E\Big[ \sum_{i'} \ph_{i'k \aw k} \Big] 
    \quad \txt{and} \quad 
    \ph_{jk \aw \vide} = - \E\big[ \ph_{jk \aw j} + \ph_{jk \aw k} \big] 
    \end{equation}
    The first kind of equations grouping pairs $ij \aw j$ by the vertex $j$, 
    and the second kind grouping pairs $ij \aw j$ by the edge $ij$,
    it follows that:
    \begin{equation} \div(\ph)_\vide = \sum_{i} \ph_{i \aw \vide} 
    + \sum_{ij} \ph_{ij \aw \vide} = 0 \end{equation}
    Hence $V_\vide = v_\vide = \div(\ph)_\vide = 0$, 
    giving $\E[V_{ij}] = \E[V_i] = 0$ for all $i, j$  
    by $\nabla(V)_{ij \aw \vide} = \nabla(V)_{i \aw \vide} = 0$. 
    The zero-mean statement on interaction potentials finally follows by 
    Möbius inversion.
\end{proof}


The conditional expectation $\E^{jk}[\,-\,|k]$, 
which couples $\ph_{jk \aw k}$ with incoming currents $\ph_{i'j \aw j}$, 
consists of the orthogonal projection of $A_j$ onto $A_k$ for 
the metric induced  by a local probability $p_{jk}$ on $A_{jk}$. 
Note that if $p = [1]$ were the uniform distribution, 
we would have $A_j \perp A_k$ for all $jk$ and (\ref{kirchhoff}) 
would obviously not have any non-trivial solution, 
as we already know from the transversality $Z_0$ with $\div A_1$.
However in general, $A_j$ is not orthogonal to $A_k$ as 
interaction brings correlation accross variables. 

\begin{defn}
Given a chain $i_0, \dots, i_n$ and a loop $j_0, \dots, j_n$ in the underlying graph,
we introduce the following definitions for applying
successive conditional expectations along edges:
\bi
    \iii $\Edge^{kj} = \E^{jk}[ \, - \,|k]$ for the 
    {\rm edge operator} projecting $A_j$ to $A_k$,
    \iii $\Chain^{i_n \dots i_0} = 
    \Edge^{i_n i_{n-1}} \circ \dots \circ \Edge^{i_1 i_0}$ 
    for the {\rm chain operator}
    mapping $A_{i_0}$ to $A_{i_n}$,
    \iii $\Loop^{j_n \dots j_0} = \Edge^{j_0 j_n} \circ \dots \circ \Edge^{j_1 j_0}$ 
    for the {\rm loop operator}
    mapping $A_{j_0}$ to itself.
    \ei
As projector of $A_{jk}$, we have 
$\Edge^{jk} = (\Edge^{kj})^*$ so that 
$\Chain^{i_0 \dots i_n} = (\Chain^{i_n \dots i_0})^*$ 
and $\Loop^{j_1 \dots j_n j_0} = (\Loop^{j_n \dots j_1 j_0})^*$. 
\end{defn}

The non-degeneracy of $p_{jk}$ implies that $\Edge^{kj}$ 
is of norm strictly smaller than $1$ as operator of $A_{jk}$,
once restricted to the orthogonal supplement of $A_j \cap A_k = A_\vide = \R$.
It follows that the spectrum of a loop operator $\Loop^{j_n \dots j_0}$ 
restricted to $\R^\perp$ is contained in $]0, 1[$, and 
hence in contrast with the scalar case, one shall not observe 
solutions of (\ref{kirchhoff}) corresponding to a conserved 
current running along a single loop. 
Such solutions only appear in the strong coupling limit, 
$p$ going to the boundary of $\Gamma$ and 
$\Z$ getting asymptotically tangent to $\B$. 

True bifurcations in $T\Z \cap \B$ may however occur when two or 
more loops can collaborate to sustain a conserved current. 
We provide with explicit examples of such bifurcations below. 
Their simplicity is remarkable, as previous examples
had only been witnessed numerically,
apparently starting with Weiss \cite{Weiss-97}. 
Before moving on, it seems important to mention that although 
non-constructive, mathematical proofs on the existence of bifurcations 
have already been given by D. Bennequin in \cite{Bennequin-IEM} and 
unpublished work. The first relied on methods of algebraic geometry 
and shows the existence of at least three homologous fixed points on the 
binary "figure eight". 

The second proof is closer to the methods exposed here, 
as it relies on the present homological description of message-passing algorithms.
The idea is to view (\ref{kirchhoff}) 
as an eigenvalue problem: 
\begin{equation} \label{M-phi}
   \ph = \M(\ph)  \quad\txt{with}\quad \M : A_1 \law A_1 
\end{equation}
Applying the Perron-Frobenius therorem,
one then shows that along certain paths in the space $\Z \simeq \Gamma$ 
of parameters, the real largest eigenvalue of $\M$ has to cross $1$. 
This approach proved the existence of at least six bifurcations 
on the simple "figure eight". 
Avoiding the challenge of producing clever bounds on the spectrum of $\M$, 
we satisfy with exhibiting simple solutions of (\ref{M-phi}).

\begin{figure}[H] 
    \sbox0{\includegraphics[width=0.7\textwidth]{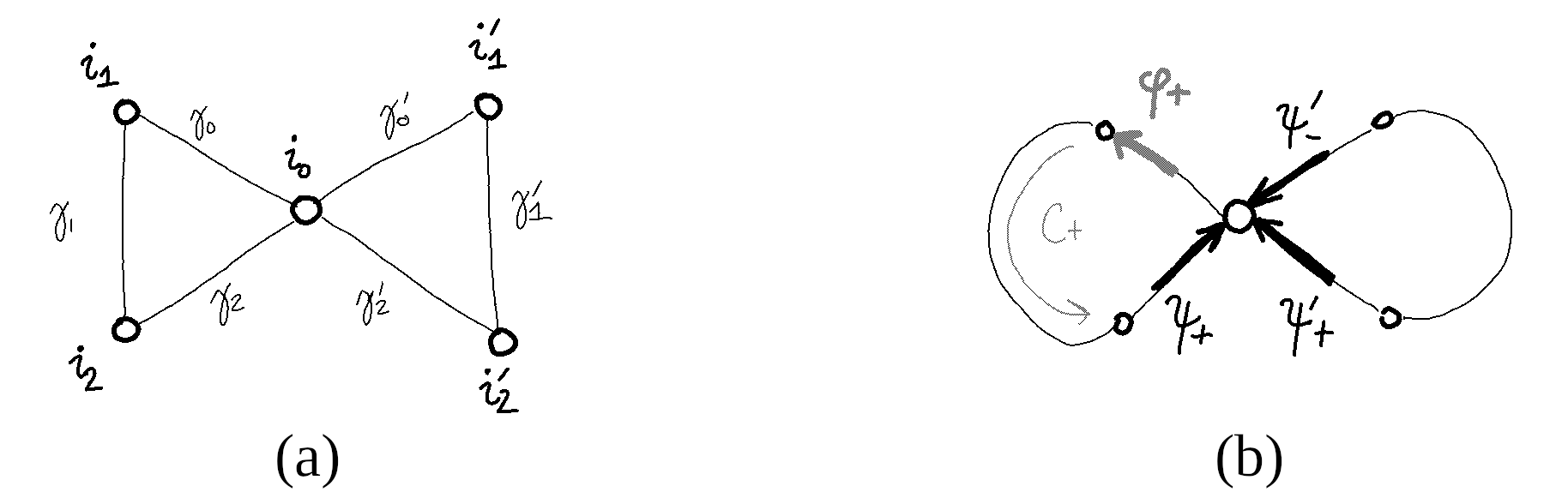}}
\begin{center}
\begin{minipage}
    {\wd0+0.062\textwidth}
\centering
\usebox0
\caption{\label{fig-eight}
    (a) Two triangles joined by a vertex. (b) Currents at the junction.
}
\end{minipage}
\end{center}
\end{figure}

Consider then the "figure eight"
graph obtained by joining two triangles 
$i, i_+, i_-$ and $i', i'_+, i'_-$ by a common vertex $i' = i$,
as depicted in figure \ref{fig-eight}.
Conservation at the junction (i) 
and conservation along each loop (ii) yield
two times four equations of the following form:
\begin{equation} ({\rm i}) \; \ph_+ = \Edge_+ \big( \psi_+ + \psi'_+ + \psi'_- \big) 
\hspace{3cm}
({\rm ii}) \;
\psi_+ = \Chain_+(\ph_+)
\end{equation}
Eliminating the $\psi$'s then reduces the above to four equations of the form:
\begin{equation} \label{junction}
\ph_+ = \Loop_+(\ph_+) + \Edge_+ \Big( \Chain'_+(\ph'_+) + \Chain'_-(\ph'_-) \Big) 
\end{equation}
Identifying each $A_j$ with $\R^2$ and each $A_{jk}$ with $\R^4 \simeq M_2(\R)$,
assume 
that local probabilities on edges and vertices are all of the same following 
form (no magnetic fields) for simplicity: 
\begin{equation} 
p_{jk} = \frac 1 4
\left[ \begin{array}{cc} 
1 + a & 1 - a \\
1 - a & 1 + a
\end{array} \right] 
\quad\txt{and}\quad 
p_j = \frac 1 2
\left[ \begin{array}{c} 
1 \\ 1
\end{array} \right] 
\end{equation}
Now each edge operator  $\E^{jk}[\,-\,|j] = 2 p_{jk}$ is
represented by the self-adjoint matrix $\Edge$ given by:
\begin{equation} 
\Edge = R^{-1}
\left[ \begin{array}{cc} 
a & 0 \\
0 & 1 
\end{array} \right] 
R 
\quad\txt{where}\quad 
R = \frac 1 {\sqrt 2}
\left[ \begin{array}{rc} 
1 & 1  \\
-1 & 1 
\end{array} \right] 
\end{equation}
expressing that we have the non-constant
eigenvector $y = [\,1 \; {-}1\,]^T$ with $\Edge(y) = a \, y$.
In the quotient of $A_0$ by $\R_0$, each edge operator
may hence be represented as multiplication by $a \in \, ] {-}1, 1 [\,$
so that (\ref{junction}) reads: 
\begin{equation} \label{junction-a}
\ph_+ = a^3 \big( \ph_+ + \ph'_+ + \ph'_- \big) 
\end{equation}
Assuming that $a^3 = \frac 1 3$, a bifurcation occurs in the direction of 
$\ph_\pm = \ph'_\pm = y$. 

Let us also prove that this is the only bifurcation in the considered 
one-parameter family of ${\cal Z}$. 
Considering that each of the $\ph$'s is spanned by $[\,1 \; {-}1\,]^T$ 
and using $(\ph_+, \ph_-, \ph'_+, \ph'_-)$ as coordinates, 
the four equations obtained from \ref{junction-a} may be written as: 
\begin{equation} \ph = \M \ph 
\quad\txt{where}\quad 
\M = a^3
\left[ \begin{array}{cccc} 
1 & 0 & 1 & 1 \\
0 & 1 & 1 & 1 \\
1 & 1 & 1 & 0 \\
1 & 1 & 0 & 1 
\end{array} \right] 
\end{equation}
A reduction of $\M$ is easily computed: 
\begin{equation} 
\M = a^3 \: Q^{-1}   
\left[ \begin{array}{rrrr} 
3 & 0 & 0 & 0 \\
0 & 1 & 0 & 0 \\
0 & 0 & 1 & 0 \\
0 & 0 & 0 & -1
\end{array} \right] 
Q
\quad\txt{where}\quad 
Q = \frac 1 {\sqrt 4}
\left[ \begin{array}{rrrr} 
1 & 1  & 1  & 1  \\
1 & -1 & -1 & 1  \\
1 & -1 & 1  & -1 \\
1 & 1  & -1 & -1 
\end{array} \right] 
\end{equation}
Under the constraint $|a| < 1$, 
it follows that 
$1 \in {\rm Spec}(\M)$ if and only if $ 3 a^3 = 1 \eqvl a = 3^{-\frac 1 3}$, 
with corresponding eigenvector $\ph = [\, 1 \; 1 \; 1 \; 1\, ]^T$. 
One may think of the three remaining eigenvectors of $\M$ 
as directions for bifurcations occuring at the boundary of $\Gamma$, 
when $a \to \pm 1$ and the probability of observing aligned/unaligned 
neighbouring spins tends to 0/1.

[[add figure]]

Two triangles joined by an edge lead do the following similar set of equations, 
denoting by $a$ and $b$ the weights on the outter edges and diagonal 
of the obtained square respectively: 
\begin{equation} \ph = \M \ph 
\quad\txt{where}\quad 
\M = a^2
\left[ \begin{array}{cccc} 
b & 0 & b & 1 \\
0 & b & 1 & b \\
b & 1 & b & 0 \\
1 & b & 0 & b 
\end{array} \right] 
\end{equation}
which has the same eigenspaces and
\begin{equation} 
\M = a^2 \: Q^{-1}   
\left[ \begin{array}{cccc} 
2b+1    & 0 & 0     & 0 \\
0       & 1 & 0     & 0 \\
0       & 0 & 2b-1  & 0 \\
0       & 0 & 0     & -1
\end{array} \right] 
Q
\quad\txt{where}\quad 
Q = \frac 1 {\sqrt 4}
\left[ \begin{array}{rrrr} 
1 & 1  & 1  & 1  \\
1 & -1 & -1 & 1  \\
1 & -1 & 1  & -1 \\
1 & 1  & -1 & -1 
\end{array} \right] 
\end{equation}
So that in the considered two-parameter family of ${\cal Z}$, 
diffeomorphic to the square $(a, b) \in \:]{-}1, 1[^2$,
bifurcations occur along the path $a = \pm \frac 1 {\sqrt{2b + 1}}$. 

\egroup

\egroup

\nocite{Mezard-Montanari}
\nocite{Mooij-2007}

\bibliographystyle{siam}
\bibliography{biblio}

\begin{thebibliography}{10}

\bibitem{Abramsky-2011}
{\sc S.~Abramsky and A.~Brandenburger}, {\em {The Sheaf-theoretic structure of
  non-locality and contextuality}}, New Journal of Physics, 13 (2011).

\bibitem{Bennequin-Baudot-2}
{\sc P.~Baudot, M.~Tapia, D.~Bennequin, and J.-M. Goaillard}, {\em {Topological
  Information Data Analysis}}, Entropy, 21 (2019), p.~869.

\bibitem{ExtrafineSheaves}
{\sc D.~Bennequin, O.~Peltre, and G.~Sergeant-Perthuis}, {\em {Extrafine
  Sheaves}}.
\newblock preprint, 2019.

\bibitem{Bennequin-IEM}
{\sc D.~Bennequin, O.~Peltre, G.~Sergeant-Perthuis, and J.~Vigneaux}, {\em
  {Informations, Energies and Messages}}.
\newblock preprint, 2019.

\bibitem{Cartan-64}
{\sc H.~Cartan}, {\em {{\'E}léments d'algèbre homologique: cours aux carrés
  1963-1964}}, {\'E}cole normale supérieure, 1966.

\bibitem{Dwyer-Spalinski}
{\sc W.~G. Dwyer and J.~Spalinski}, {\em Homotopy theories and model
  categories}, 1995.

\bibitem{Friston-Parr-2017}
{\sc K.~J. Friston, T.~Parr, and B.~de~Vries}, {\em {The Graphical Brain:
  Belief Propagation and Active Inference}}, Networks Neuroscience, 1 (2017),
  pp.~381--414.

\bibitem{Gallager-63}
{\sc R.~G. Gallager}, {\em {Low-Density Parity-Check Codes}}, MIT Press, 1963.

\bibitem{Hu-62}
{\sc K.~T. Hu}, {\em {On the Amount of Information}}, Theory of Probability and
  its Applications, 7 (1962), pp.~439--447.

\bibitem{Kellerer-64}
{\sc H.~G. Kellerer}, {\em {Ma{\ss}theoretische Marginalprobleme}},
  Mathematische Annalen, 153 (1964), pp.~168--198.

\bibitem{Kikuchi-51}
{\sc R.~Kikuchi}, {\em {A Theory of Cooperative Phenomena}}, Phys. Rev., 81
  (1951), pp.~988--1003.

\bibitem{Knoll-2017}
{\sc C.~Knoll and F.~Pernkopf}, {\em {On Loopy Belief Propagation -- Local
  Stability Analysis for Non-Vanishing Fields}}, in Uncertainty in Artificial
  Intelligence, 2017.

\bibitem{Kodaira-49}
{\sc K.~{Kodaira}}, {\em {Harmonic fields in Riemannian manifolds (Generalized
  potential theory).}}, {Ann. Math. (2)}, 50 (1949), pp.~587--665.

\bibitem{Leinster-08}
{\sc T.~Leinster}, {\em {The Euler Characteristic of a Category}}, Documenta
  Mathematica, 13 (2008), pp.~21--49.

\bibitem{Eilenberg-MacLane}
{\sc S.~MacLane and S.~Eilenberg}, {\em {General Theory of Natural
  Equivalences}}, Transactions of the American Mathematical Society,  (1945),
  pp.~231--294.

\bibitem{Massieu-69}
{\sc F.~Massieu}, {\em {Sur les Fonctions caractéristiques des divers
  fluides}}, Comptes rendus de l'Acad\'emie des Sciences, 69 (1869),
  pp.~858--862.

\bibitem{Matus-88}
{\sc F.~Mat{\'u}{\v s}}, {\em {Discrete Marginal Problem for Complex
  Measures}}, Kybernetika, 24 (1988), pp.~36--46.

\bibitem{Meyer-86}
{\sc P.-A. Meyer}, {\em {\'El\'ements de probabilit\'es quantiques (expos\'es I
  \`a V)}}, S\'eminaire de probabilit\'es de Strasbourg, 20 (1986),
  pp.~186--312.

\bibitem{Moerdijk}
{\sc I.~Moerdijk}, {\em {Classifying Spaces and Classifying Topoi}}, Springer,
  1995.

\bibitem{Mooij-2007}
{\sc J.~M. Mooij and H.~J. Kappen}, {\em {Sufficient Conditions for Convergence
  of the Sum–Product Algorithm}}, IEEE Transactions on Information Theory, 53
  (2007), p.~4422–4437.

\bibitem{Morita-57}
{\sc T.~Morita}, {\em {Cluster Variation Method of Cooperative Phenomena and
  its Generalization I}}, Journal of the Physical Society of Japan, 12 (1957),
  pp.~753--755.

\bibitem{Murphy-Weiss-99}
{\sc K.~P. Murphy, Y.~Weiss, and M.~I. Jordan}, {\em {Loopy Belief Propagation
  for Approximate Inference: An Empirical Study}}, in UAI, 1999.

\bibitem{Mezard-Montanari}
{\sc M.~Mézard and A.~Montanari}, {\em {Information, Physics and
  Computation}}, Oxford University Press, 2009.

\bibitem{Pearl-82}
{\sc J.~Pearl}, {\em {Reverend Bayes on Inference Engines: A Distributed
  Hierachical Approach}}, in AAAI-82 Proceedings, 1982.

\bibitem{Rota-64}
{\sc G.-C. Rota}, {\em {On the Foundations of Combinatorial Theory - I. Theory
  of Möbius Functions}}, Z. Warscheinlichkeitstheorie, 2 (1964), pp.~340--368.

\bibitem{Schlijper-83}
{\sc A.~G. Schlijper}, {\em {Convergence of the cluster-variation method in the
  thermodynamic limit}}, Phys. Rev. B, 27 (1983), pp.~6841--6848.

\bibitem{Souriau-QG}
{\sc J.-M. Souriau}, {\em {Quantification Géométrique}}, in {Physique
  quantique et géométrie: formulation mathématique cohérente des
  phénomènes quantiques: Colloque géométrie et physique de 1986 en
  l'honneur d'Andr{\'e} Lichnerowicz}, vol.~32, Travaux en cours, Hermann,
  1988.

\bibitem{Speed-79}
{\sc T.~P. Speed}, {\em {A Note on Nearest-Neighbour Gibbs and Markov
  Probabilities}}, Sankhya: The Indian Journal of Statistics, 41 (1979),
  pp.~184--197.

\bibitem{Thom-56}
{\sc R.~Thom}, {\em {Les Singularités des applications différentiables}},
  {Annales de l'institut Fourier}, 6 (1956), pp.~43--87.

\bibitem{SGA-4-V}
{\sc J.-L. Verdier and A.~Grothendieck}, {\em {V: Cohomologie dans les Topos}},
  SGA-4, 2 (1972).

\bibitem{Vigneaux-phd}
{\sc J.-P. Vigneaux}, {\em {Topology of Statistical Systems: A Cohomological
  Approach to Information Theory}}, Université de Paris, 2019.

\bibitem{Vorobev-62}
{\sc N.~Vorob'ev}, {\em {Consistent Families of Measures and their
  Extensions}}, Theory of Probability and its Applications, 7 (1962),
  pp.~147--164.

\bibitem{Weiss-97}
{\sc Y.~Weiss}, {\em {Belief propagation and revision in networks with loops}},
  tech. rep., MIT, 1997.

\bibitem{Yedidia-2005}
{\sc J.~Yedidia, W.~Freeman, and Y.~Weiss}, {\em {Constructing Free Energy
  Approximations and Generalized Belief Propagation Algorithms}}, IEEE
  Transactions on Information Theory, 51 (2005), pp.~2282--2312.

\bibitem{Yedidia-2001}
{\sc J.~S. Yedidia, W.~T. Freeman, and Y.~Weiss}, {\em {Bethe free energy,
  Kikuchi approximations, and belief propagation algorithms}}, Tech. Rep.
  TR2001-16, MERL - Mitsubishi Electric Research Laboratories, Cambridge, MA
  02139, May 2001.

\end{thebibliography}

\end{document}